%% file: 000-main.tex
\newif\ifsubmit 

\newif\ifeprint
\newif\ifconference
\newif\ifproceedings

\input{111-PARAMS}

\ifproceedings
\documentclass[runningheads]{llncs}
\fi
 
\ifconference
\documentclass{llncs}
\fi

\ifeprint
\documentclass[11pt]{article} 
\fi

\input{000-include/macros/config}
\input{000-include/macros/commoncmd}
\input{111-THIS-PROJECT/cmd}

\input{111-THIS-PROJECT/title-authors}

\begin{document}
\maketitle

\ifproceedings
\input{112-PROCEEDINGS}
\fi

\ifconference

\input{112-CONFERENCE}

\fi

\ifeprint

\input{112-EPRINT}
\fi

\bibliographystyle{alpha}
\end{document}

%% file: 111-PARAMS.tex
\submittrue

\eprinttrue

\conferencefalse

\proceedingsfalse

%% file: 000-include/macros/config.tex
\usepackage[utf8]{inputenc}

\ifproceedings
\usepackage[T1]{fontenc}
\fi

\ifeprint
\usepackage[pdftex,left=1in,top=1in,bottom=1in,right=1in]{geometry}
\usepackage{amsthm}
\fi

\usepackage{amsmath,amssymb,dsfont,mathrsfs}

\usepackage{braket}
\usepackage{physics}
\usepackage{mathtools}

\usepackage{xcolor}
\usepackage[framemethod=tikz]{mdframed}
\usepackage{graphicx}
\usepackage{caption}
\usepackage{bm}

\usepackage{verbatim} 
\usepackage{enumitem}

\usepackage{breakcites}

\usepackage[colorlinks=true,citecolor=blue,linkcolor=blue]{hyperref}
\usepackage[nameinlink, noabbrev, capitalize]{cleveref}

\hypersetup{
    colorlinks,
    linkcolor={red!50!black},
    citecolor={blue!50!black},
    urlcolor={blue!80!black}
}

\ifeprint
\usepackage{tocloft}

\fi

\definecolor{darkgreen}{rgb}{0,0.5,0}
\definecolor{darkblue}{rgb}{0,0,0.6}

\allowdisplaybreaks

\usepackage{tabularx}

%% file: 000-include/macros/commoncmd.tex
\newcommand{\addrefs}{\bibliography{000-include/bibs/abbrev0,000-include/bibs/crypto,111-THIS-PROJECT/refs}}

\DeclareMathAlphabet{\mathpzc}{OT1}{pzc}{m}{it}

\newcommand{\varr}[1]{\mathrm{{#1}}} 
\newcommand{\qvarr}[1]{\mathpzc{#1}} 
\newcommand{\sch}[1]{\bm{\mathsf{#1}}} 
\newcommand{\algo}[1]{\mathsf{#1}} 
\newcommand{\gam}[1]{\mathsf{#1}} 

\newcommand{\samp}{\leftarrow} 

\newcommand{\C}{\mathds{C}}

\newcommand{\BN}{\mathbb{N}}
\newcommand{\nat}{\BN}
\newcommand{\U}{\mathcal{U}}
\newcommand{\F}{\mathds{F}}

\newcommand{\R}{\mathbb{R}}
\newcommand{\E}{\mathds{E}}

\newcommand{\zo}{\{0,1\}}
\newcommand{\poly}{\mathsf{poly}}

\newcommand{\regi}{\mathsf{R}} 
\newcommand{\regis}[1]{\regi_{\mathsf{#1}}} 
\newcommand{\trd}[2]{\norm{#1 - #2}_{Tr}} 

\newcommand{\messpa}{\mathcal{M}} 
\newcommand{\meslen}{m}

\newcommand{\ora}{\mathcal{O}}
\newcommand{\obfd}[1]{\hat{#1}} 

\newcommand{\io}{i\mathcal{O}}
\newcommand{\prf}{\sch{PRF}}

\newcommand{\pk}{\varr{pk}}
\newcommand{\sk}{\varr{sk}}

\newcommand{\qsk}{\qvarr{sk}}
\newcommand{\mes}{\varr{m}}
\newcommand{\ct}{\varr{ct}}

\newcommand{\enc}{\algo{Enc}}
\newcommand{\dec}{\algo{Dec}}
\newcommand{\setup}{\algo{Setup}}
\newcommand{\ver}{\algo{Ver}}
\newcommand{\ceval}{\algo{Eval}}
\newcommand{\gen}{\algo{Gen}}

\newcommand{\genkey}{\algo{GenKey}}
\newcommand{\keygen}{\genkey}

\newcommand{\sign}{\algo{Sign}}

\newcommand{\negl}{\mathsf{negl}}
\newcommand{\adve}{\mathcal{A}}

\ifeprint
\newtheorem{theorem}{Theorem}
 
 \newtheorem{claim}{Claim}
 \newtheorem{lemma}{Lemma}
 
 \newtheorem{corollary}{Corollary}
 
 \newtheorem{definition}{Definition}
 
 \newtheorem{remark}{Remark}
 
\fi

\ifconference
\let\originalparagraph\paragraph
\renewcommand{\paragraph}[1]{%
  \originalparagraph{\textbf{#1}}%
}
\fi

%% file: 111-THIS-PROJECT/cmd.tex
\usepackage[normalem]{ulem}
\usepackage{xspace}
\usepackage{tikz}
\usetikzlibrary{arrows.meta,calc,decorations.markings,positioning}
\newcommand{\moe}{\text{Monogamy-of-Entanglement}\xspace}

\newcommand{\dcmoe}{\textsf{Decisional-Coset-MOE}\xspace}

\newcommand{\cosetue}{\textsf{Coset-UE}\xspace}

\newcommand{\advmoe}{{\adve_{\textsf{moe}}}}
\newcommand{\advue}{{\adve_{\textsf{UE}}}}

\newcommand{\keyspace}{{\cal K}}

\newcommand{\bit}{\{0,1\}}

\newcommand{\comp}{\text{Compatible}}

\newcommand{\comparitycoset}{\text{Comp-Parity-Coset}}
\newcommand{\can}{\text{Coset-rep}}

\newcommand{\secparam}{\lambda}

\newcommand{\alice}{\ensuremath{\mathcal{A}}\xspace}
\newcommand{\bob}{\ensuremath{\mathcal{B}}\xspace}
\newcommand{\charlie}{\ensuremath{\mathcal{C}}\xspace}

\newcommand{\col}{\mathsf{Col}}

\newcommand{\regalice}{\ensuremath{\mathsf{A}}\xspace}
\newcommand{\regbob}{\ensuremath{\regi_{1}}\xspace}
\newcommand{\regcharlie}{\ensuremath{\regi_{2}}\xspace}
\newcommand{\regeve}{\ensuremath{\mathsf{E}}\xspace}
\newcommand{\ch}{\textsf{Ch}\xspace}

\newcommand{\FF}{\mathbb{F}}

\newcommand{\NN}{\mathbb{N}}

\newcommand{\CC}{\mathbb{C}}
\newcommand{\EE}{\mathbb{E}}

\newcommand{\getsr}{\xleftarrow{\$}}

\newcommand{\epsilondcmoe}{\epsilon_{\mathsf{DCMOE},n}}
\newcommand{\hybrid}{\mathsf{Hyb}}

\newcommand{\compdcmoe}{\textsf{Comp-Decisional-Coset-MOE}\xspace}
\newcommand{\auxsetup}{\textsf{Aux.Setup}\xspace}
\newcommand{\auxgen}{\textsf{Aux.Gen}\xspace}
\newcommand{\auxsim}{\textsf{Aux.Sim}\xspace}

\providecommand{\iO}{\mathsf{iO}}
\providecommand{\shO}{\mathsf{shO}}

\providecommand{\negl}{\mathsf{negl}}

\newcommand{\sde}{\sch{SDE}}
\newcommand{\corrsde}{\gam{SDE\mbox{-}CORR}\xspace}

\newcommand{\sep}{\mathsf{IND}}
\newcommand{\idmode}{\mathsf{IDEN}}
\newcommand{\projimp}{\mathsf{ProjImp}}
\newcommand{\TI}{\mathsf{TI}}
\providecommand{\sampler}{\mathsf{Samp}}

\newcommand{\upo}{\sch{UPO}}
\newcommand{\Obf}{\mathsf{Obf}}
\newcommand{\Eval}{\mathsf{Eval}}
\newcommand{\acupo}{\textsf{Correlated-UPO}\xspace}
\newcommand{\genpuncture}{\mathsf{GenPuncture}}
\newcommand{\acpuncture}{\textsf{Correlated-Puncture}}
\newcommand{\pprf}{\mathsf{PPRF}}
\newcommand{\advupo}{{\adve_{\mathsf{UPO}}}}
\newcommand{\unifdist}{\mathsf{U}}
\newcommand{\ID}{\mathsf{ID}}
\newcommand{\IDU}{\mathsf{ID}_{\mathsf{U}}}
\newcommand{\Hmin}{\widetilde H_\infty}
\newcommand{\adv}{\adve}

%% file: 111-THIS-PROJECT/title-authors.tex
\title{Decisional Monogamy-of-Entanglement for Coset
States and Applications to Unclonable Cryptography with Correlated Challenges}

\ifproceedings
\fi

\date{}

\ifproceedings
\author{Amit Behera\inst{2}\orcidID{0009-0008-1462-7222} \and \author{Alper {C}akan\inst{1}\orcidID{0000-0003-3567-1704} \and
Vipul Goyal\inst{2,1}}

\authorrunning{A. Behera and A. {C}akan and V. Goyal}

\institute{
NTT Research, Sunnyvale, CA, USA \email{amitbehera1767@gmail.com} \and
Carnegie Mellon University, Pittsburgh PA, USA \email{alpercakan98@gmail.com} \and
NTT Research, Sunnyvale, CA, USA
\email{vipul@vipulgoyal.org}}
\fi

\ifeprint
 \author{Amit Behera\thanks{NTT Research \texttt{amitbehera1767@gmail.com}} \and Alper {C}akan\thanks{Carnegie Mellon University. \texttt{alpercakan98@gmail.com}.} \and Vipul Goyal\thanks{NTT Research \& Carnegie Mellon University.  \texttt{vipul@vipulgoyal.org}}}
\fi

\ifsubmit
\newcommand{\todo}[1]{}
\newcommand{\alper}[1]{}
\newcommand{\amit}[1]{}
\newcommand{\vipul}[1]{}
\newcommand{\alpertodo}[1]{}
\else
\newcommand{\todo}[1]{{\color{blue} \footnotesize(TODO: \uppercase{#1})}}
\newcommand{\alper}[1]{{\color{blue} \footnotesize(ALPER: {#1})}}
\newcommand{\alpertodo}[1]{{\color{blue} \footnotesize(TODO for ALPER: {#1})}}
\newcommand{\amit}[1]{{\color{purple} \footnotesize(AMIT: {#1})}}
\newcommand{\vipul}[1]{{\color{red} \footnotesize(VIPUL: {#1})}}
\fi

%% file: 112-PROCEEDINGS.tex
\input{111-THIS-PROJECT/abstract}

\input{intro}


\begin{credits}
\subsubsection{\ackname} 
\input{111-THIS-PROJECT/acknowledgement}

\end{credits}

\addrefs 

%% file: 111-THIS-PROJECT/abstract.tex
\begin{abstract}
One of the main applications of quantum information in cryptography is \emph{copy-protection}, where we encode a functionality (such as a decryption key or software) into a reusable quantum state
so that it cannot be split into two adversaries (called \emph{freeloaders}) that both remain \emph{useful}. 

Previous works have only been able to show security with respect to \emph{independently sampled challenges} for the two adversaries. However, a competing natural security notion is \emph{identical-challenge} security where the adversaries receive the same challenge. This security notion has many real-life applications, and also has connections to other fundamental primitives such as unclonable bits (i.e. unclonable encryption) and unclonable lockboxes (i.e. copy-protection of point functions). Despite its importance and numerous attempts, the question of achieving identical-challenge security in the plain model has remained open.

In this work, first we make progress on the definitional foundations of copy-protection by introducing natural copy-protection  security definitions that imply the previous ones (including identical-challenge security) and better capture the security intuitions and real-life use cases; and we also characterize the relationship between the previous definitions. Then, we show how to achieve in the \emph{plain model} our new stronger definitions for copy-protection of general classes of functionalities. In particular, we resolve the long-standing open questions of  copy-protection of point functions, copy-protection of compute-and-compare programs, and identical-challenge secure copy-protection of decryption keys and all puncturable functionalities. 

The technical core of our results is a new \textbf{decisional monogamy-of-entanglement result for coset states}, which both allows us to achieve our new results, and also significantly simplifies and unifies unclonable cryptography proofs. We believe this will have further applications and may be of independent interest.

\ifproceedings
\fi
\end{abstract}

\ifeprint
\vfill\par\textbf{Keywords:} Quantum cryptography, Unclonable cryptography, Copy-protection, Unclonable encryption.
\fi

%% file: intro.tex
\section{Introduction}
Starting with the seminal work of Wiesner~\cite{Wie83}, a long line of
research has shown that quantum information can help us achieve remarkable cryptographic guarantees that are simply impossible with classical information.  Perhaps one of the most central applications is \underline{\emph{unclonable cryptography}} (which is based on the No-Cloning principle of quantum states~\cite{WZ82,Dieks82}), which mainly consists of quantum money (\cite{Wie83,EPRINT:AarChr12}), unclonable encryption\footnote{Throughout this paper, we will only consider indistinguishable-secure unclonable encryption where the ciphertext encodes a bit $b \in \zo$ or encodes $m_b$ out of two messages $m_0, m_1$ (with security bound $1/2+\negl(\lambda)$), which is considered the gold standard. We will not consider the weaker notion of search-secure unclonable encryption where the ciphertext encodes a random superlogarithmic-length/unpredictable message (with security bound $\negl(\lambda)$).} \cite{BL20} and \emph{quantum copy-protection}~\cite{Aar09}.

\paragraph{Copy-protection of functionalities} In copy-protection, a functionality is encoded into a quantum state that can be used to evaluate the functionality on any inputs any number of times, but no efficient adversary can split this state into two adversaries (also called \emph{freeloaders}) that both remain useful. The functionality can be a cryptographic object (such as a decryption key for an encryption scheme) or can be a piece of software (such as a video game or a spreadsheet program). Note that a classical program or key can always be copied bit for bit, so copy protection is indeed an inherently quantum security notion.

\paragraph{State-of-the-art} There has been significant recent progress on laying down the definitional foundations of copy protection~\cite{AB24,ABH+26,TCC:KitYam25} and constructing copy protection
for general classes of functionalities~\cite{EPRINT:ColMajPor20,C:ALLZZ21,C:CLLZ21,LLQ+22,C:AKLLZ22,C:AnaKalLiu23,TCC:CakGoy24,AB24,TCC:KitYam25,ABH+26,CG26}.
Nevertheless, as we discuss in more detail below, the state of affairs remains quite unsatisfactory on both the definitional and feasibility fronts.

On the definitional side, while there are various natural definitions, we still do not have a \emph{gold-standard} definition for even copy-protecting decryption keys, which is one of the most fundamental cryptographic applications of copy protection. In fact, even the relationship among the existing definitions is unclear.

On the feasibility front, some of these existing natural security definitions, such as identical-challenge security for decryption keys, still have not been achieved in the plain model\footnote{We note that there have been some non-standard solutions for identical-challenge security, such as in the quantum random oracle model (QROM) \cite{TCC:KitYam25} or that use quantum ciphertexts in the plain model \cite{AKY25}, or that rely on unproven ad-hoc conjectures \cite{AB24}. Identical-challenge \emph{search} security for unpredictable messages distributions has also been achieved \cite{C:CLLZ21,C:AnaKalLiu23}. However, the ideal goal has been achieving plain model indistinguishability security with classical ciphertexts.}; even though (as we discuss below) identical-challenge security has important real-life applications and connections to fundamental primitives like unclonable encryption. Going beyond decryption keys; even for some of the most basic and fundamental functionalities, such as \emph{point functions}, we still do not have a provably secure copy protection scheme in the plain model\footnote{While there are some works that obtain point-function copy-protection \cite{CMP24,TCC:CheHerVu23}, these works achieve it under unnatural security definitions and/or in the QROM. }. Similarly, while independent-challenge copy-protection security for general classes of functionalities such as \emph{all puncturable functions} have been achieved in the plain model \cite{ABH+26,CG26}, identical-challenge security for general classes of functionalities still has not been achieved.

\subsubsection*{Motivation of Our Work}
\paragraph{Defining single-decryptor encryption (SDE) security} The earliest and perhaps the most important primitive in cryptography is \emph{encryption}. Accordingly, one of the earliest and most fundamental questions in quantum copy protection has been how to copy-protect decryption keys. Thus, it has also been one of the most studied questions in copy protection~\cite{GZ20,C:CLLZ21,TCC:AnaKal21,C:AKLLZ22,TCC:KitNis23,C:AnaKalLiu23,TCC:KitYam25,TCC:CakGoy24} and even has a dedicated name: \emph{single-decryptor encryption} (SDE) \cite{GZ20}. In SDE, we encode a decryption key in a quantum state, and we require that no adversary can split it into two \emph{useful} keys. Note that the encryption key and the ciphertexts are still classical. In the case of classical encryption, the notion of a \emph{useful} adversary has been settled decades ago: The seminal work of Goldwasser and Micali~\cite{STOC:GolMic82,GolMic84} introduced the notion of semantic security, which says that the adversary cannot learn anything from the ciphertext beyond what they already knew about the message. This security notion (and the equivalent notions of \emph{chosen-plaintext attack [CPA] or indistinguishability security}) is the \emph{gold standard} for encryption, and also implies various other properties that one would expect from an encryption, such as unpredictability security (also called search security), where an adversary tries to guess an unpredictable long message in full from its encryption.

However, it turns out that the question of defining security for single-decryptor encryption (SDE) is much trickier: There are 8 different security notions currently in the literature! To make things worse, even the relationships among these security notions are unclear, with no \emph{gold-standard} definition: No single definition is known to imply all others. For example, unlike classical encryption, it is not even known if the existing indistinguishability security notions imply unpredictability security (in fact, in this work we show that they do not).

In light of this, we ask the following question.
\begin{mdframed}
\begin{quote}
\emph{Can we define a \textnormal{gold-standard} security notion for single-decryptor encryption that captures the natural security intuition and use cases, and implies the existing definitions (including identical-challenge security)?}
\end{quote}    
\end{mdframed}

\paragraph{Achieving the natural security notions} As mentioned above, while some of the seven definitions for SDE have been achieved; for some natural ones, we still do not have provably secure constructions in the plain model. Perhaps the most important example is the \emph{identical-challenge} security \cite{GZ20} where the two freeloaders receive the same ciphertext.
Identical-challenge security naturally arises in many real-world scenarios, for example, when some encrypted content is broadcast to multiple recipients.
In particular, consider the following real-world scenario: a television company broadcasts encrypted transmissions of a live event to all their subscribers, who decrypt using their (quantum) decryption key. To capture anti-piracy security for this application, we must consider identical-challenge security. 

Identical-challenge security also has direct connections to other fundamental primitives, such as unclonable encryption and copy-protection of point functions, which are always of \emph{identical-challenge} flavor by their nature. Indeed, there exists a simple compiler from any identical-challenge secure SDE to unclonable encryption \cite{GZ20} (which does not apply to, e.g., independent-challenge secure SDE).

Given its significance, there have been many attempts at constructing identical-challenge SDE  \cite{AKY25,C:AnaKalLiu23,C:AKLLZ22,TCC:CheHerVu23,TCC:KitYam25,AB24,ABH+26}. However, the question of achieving this notion in the plain model has been open since Georgiou and Zhandry (2020) \cite{GZ20} introduced it. Thus, we ask the following question.
\begin{mdframed}
\begin{quote}
\emph{Can we achieve identical-challenge single-decryptor encryption in the plain model?}
\end{quote}    
\end{mdframed}

\paragraph{Beyond decryption keys} Similar to SDE, for general classes of functionalities; identical-challenge and correlated challenge distributions are important and natural challenge distributions. For example, consider copy-protection of point functions, which are perhaps the simplest nontrivial programs and have natural applications such as digital lockboxes and password-verification software. Indeed, the conceptual simplicity of point functions and concrete applications led Aaronson to single out their copy protection when introducing quantum copy protection in 2009~\cite{Aar09}. For point functions, any non-trivial copy-protection security notion must consider a challenge distribution that uses the marked point itself. This leads to correlation between the challenges, and hence, every meaningful challenge distribution for point functions is correlated! Among them, the most natural challenge distribution is the following: for a point function $P_y$ marked at a point $y$; with probability $1/2$, give the point $y$ to the freeloaders $\bob$ and $\charlie$, and with probability $1/2$, give a random point $x$ (sampled from the input space) to both the freeloaders. This distribution captures real-world applications, and it also has connections to other important primitives: For example, a point-function copy-protection scheme immediately implies an unclonable encryption scheme.\footnote{In particular, given such a copy-protection scheme for point function for a high-min entropic distribution $D$ on the points, the unclonable encryption scheme does the following. Key-generation samples a point $y$ associate with a function $P_y$ from the $D$. Encryption is the same as the copy-protection compilation on the function $P_y$ if the message is $1$ and $P_{y'}$ if the message is $0$ where $y'$ is independently sampled from $D$. Decryption algorithm just runs the secret key as an input on the quantum program outputs the output. Clearly, the scheme is statistically correct since $D$ has high-min entropy, and it holds that up to negligible statistical distance, the unclonable indistinguishability game of this scheme is the same as the copy-protection anti-piracy game with the natural challenge distribution.}

Despite a series of attempts to obtain copy-protection for point functions~\cite{CMP24,TCC:CheHerVu23}, achieving security with respect to the natural challenge distribution in the plain model has remained open. A similar situation exists for other important functionalities like $k$-point functions and compute-and-compare functions, and for general puncturable functionalities with identical-challenge security (which has been left as an open question by Ananth, Behera, Huang, Kitagawa, Yamakawa (EUROCRYPT'26) \cite{ABH+26} and  Cakan-Goyal (EUROCRYPT'26) \cite{CG26}). Thus, we ask the following question.
\begin{mdframed}
\begin{quote}
  \emph{Can we achieve copy-protection for point functions, $k$-point functions, compute-and-compare functions, and for other general classes of functionalities such as puncturable functionalities in the identical-challenge setting? }
\end{quote}
\end{mdframed}

\subsection{Our Results}
In this work, we make progress towards a gold-standard definition by giving a new natural security definition for SDE and copy-protection of general functionalities that imply the previous ones, including identical-challenge security. Then, we positively answer our construction/feasibility questions by achieving our new security definitions in the plain model.

Before giving an exposition of our feasibility results, we first present our key technical contribution that enables all our feasibility results. We show a novel \emph{monogamy-of-entanglement (MoE)} property for coset states, which not only underpins our new feasibility results, but also simplifies and unifies unclonable cryptography proofs such as copy-protection of functionalities and unclonable encryption.

\subsubsection{A New Monogamy-of-Entanglement Game: Achieving Identical-Challenge Security}\mbox{}\\
As our key technical contribution, we prove the following novel \emph{decisional} monogamy-of-entanglement (MoE) property for coset states, which underlies all our new feasibility results. 
Informally, MoE \cite{EC:TFKW13} limits how well two entangled (but non-communicating) parties can simultaneously predict a certain measurement outcome of a quantum state split between them.

\begin{theorem}[Decisional MoE for Coset States (Informal)]
Let $(\alice, \bob, \charlie)$ be any tuple of (unbounded) adversaries. Suppose $\alice$ receives the coset state
$$
\ket{A_{s,t}} \propto \sum_{a\in A}(-1)^{\langle a,t\rangle}\ket{a+s},
$$
where  $A\leq\FF_2^{n}$ is a random $n/2$-dimensional subspace and
$s,t\in\FF_2^{n}$ are random vectors, and outputs (entangled) registers $\regi_1, \regi_2$ to $\bob$ and $\charlie$. 

Then, $\bob, \charlie$ both receive $(u,v)$ where $u\in A$ and $v\in A^\perp$ are uniformly random. Then,
$$
\Pr[\bob(\regi_1, u,v) = \charlie(\regi_2, u,v) = \langle u || v, t || s\rangle] \leq \frac{1}{2}+2^{-\Omega(n)}.
$$
\end{theorem}

We remark the following about our MoE game.
\begin{enumerate}
    \item Given $\ket{A_{s,t}}$ and $(u,v)$, one can efficiently and perfectly obtain this value (which shows the non-triviality of the security guarantee). 
    \item Observe that $\langle u || v, t || s\rangle =  \langle u, t \rangle \oplus \langle v, s\rangle$, and that for all vectors $w' \in A^\perp + t$ and $u \in A$, the inner product $\langle u, w' \rangle$ is the same value, $\langle u,t\rangle$. Similarly, for all vectors $w \in A + s$ and $v \in A^\perp$, the inner product $\langle v, w \rangle$ is the same value, $\langle v,s\rangle$. Hence the correct answer $\langle u || v, t || s\rangle$ can be computed given $u,v$ and \emph{any} two coset elements $w \in A^\perp + t$ and $w' \in A + s$.
    This property will be quite useful in copy-protection proofs.\label{it:lack-of-canonicalization}
    \item 
While we do not need it for our copy-protection proofs, we even show a stronger version of this result where $\bob$ and $\charlie$ receive the subspace description $A$ (but not $s,t$).    
\end{enumerate}

We also show a computational version of this property.
\begin{theorem}[Computational Decisional MoE for Coset States (Informal)]
Assuming indistinguishability obfuscation (iO) and one-way functions, the same security holds for quantum polynomial time (QPT) adversaries $(\adve, \bob, \charlie)$ where $\adve$ even receives obfuscated membership programs $\io(A+s), \io(A^\perp+t)$ and (simulatable) auxiliary
information about $A,u,v$.
\end{theorem}

We remark that our MoE results also attain optimal bit-security: We can achieve $2^{-\Omega(n)}$-security using $n$ qubits\footnote{In the computational MoE case, to achieve $2^{-\Omega(n)}$-security, we need to assume (any) subexponential security of $\io$ rather than polynomially secure $\io$ and use sufficient padding for the obfuscated programs}.

\paragraph{Comparison to previous MoE results}
As we discuss in \cref{sec:tech-overview}, our new MoE results significantly differ from previous MoE results shown in the literature, since our game is (i) a decisional game (ii) for an unlearnable state. We also emphasize that the two adversaries receive the same input and are required to output the same output. None of the previous MoE results achieved all these properties at the same time, which is why they were not useful for achieving results such as identical-challenge copy-protection or point function copy-protection.

In more detail, prior MoE results differ from ours in at least one of these respects.
\begin{enumerate}
    \item \emph{Search rather than decision:} Coset MoE result of Coladangelo, Liu, Liu, Zhandry \cite{C:CLLZ21} and Culf and Vidick \cite{CulfVidick22} is a search game, where the adversaries must output a vector in the appropriate coset ($A+s$ or $A^\perp+t$), rather than answer a decisional challenge. Thus, this result is not useful for achieving decisional security in the identical-challenge case or for achieving unclonable encryption\footnote{While their MoE game does imply the weaker notion of search-secure unclonable encryption; as mentioned before, we always mean the indistinguishable-secure version when we say unclonable encryption in this paper.}.
    \item \emph{Learnable rather than unlearnable states:}  The celebrated result of Tomamichel, Fehr, Kaniewski, Wehner \cite{EC:TFKW13}  and the very recent unclonable encryption results of Ananth and Sahai \cite{AS26} and Ragavan \cite{Rag26}, and works of Ananth \cite{Ana26} and Coladangelo, Liu and Xie \cite{CLX26}\footnote{\cite{Ana26} and \cite{CLX26} are concurrent to our work} consider Pauli eigenstates or BB84 states, which can be fully learned/copied given a verification oracle for the state (e.g. a projector) and thus cannot be used for copy-protection of functionalities (since evaluating the functionality and checking the output already provides such a state verification ability).
    \item \emph{Independent rather than identical challenges:} Some works such as \cite{C:CLLZ21,AB24} consider a relaxed\footnote{In actual MoE games, the two entangled adversaries (Bob and Charlie) try to guess the measurement outcome of a third adversary (Alice) in a random basis after learning this basis; thus, they should receive the same input and output the same output.} version of monogamy-of-entanglement where the adversaries receive independently sampled challenges and try to predict different measurement outcomes. Consequently, these results are not useful for achieving identical-challenge security or unclonable encryption.
    \item \emph{Need for canonical coset vectors:} In some of the previous decisional coset MOE games, such as the one based on simultaneous (independent) Goldreich-Levin, considered in~\cite{ABH+26}, the correct answers in the MOE game for $\bob$ and $\charlie$ are computed based on the canonical vector of the primal and dual cosets. However, for our new decisional MOE game, the correct answers can be computed based on \emph{any} vector from the respective coset, without needing any canonical vector for the cosets (see \Cref{it:lack-of-canonicalization}). This absence of canonicalization significantly simplifies the existing proofs of feasibility results in the copy-protection literature, as discussed below.
\end{enumerate}

\paragraph{Simplifying and Unifying Unclonable Cryptography} As we discuss in \cref{sec:tech-overview}, our novel MoE property also \emph{significantly simplifies} copy-protection/unclonable cryptography proofs compared to prior work and also \emph{unifies} them under a single technique: For example, unclonable encryption\footnote{It is easy to see that our MoE result immediately gives an indistinguishable-secure unclonable encryption scheme.}, identical-challenge and independent-challenge copy-protection for puncturable functionalities, point function copy-protection and all our new results can all be proven in a straightforward manner using this single MoE game; whereas previously these all required different techniques. In particular, our MoE game eliminates the need for complicated quantum rewinding \cite{MW05,TCC:Zhandry20} and eigenvalue estimation techniques (such as \emph{threshold implementations}) \cite{TCC:Zhandry20,C:ALLZZ21,C:CLLZ21}, and \emph{simultaneous extraction} techniques that formed the crux of many prior works in unclonable cryptography \cite{C:CLLZ21,C:ALLZZ21,CMP24,AB24,CG26,ABH+26,C:AnaKalLiu23,C:AKLLZ22,TCC:CakGoy24}.

\subsubsection{Progress Towards a Gold-Standard Definition}\label{sec:main-feasibility results}
\paragraph{Correlated-challenge security for single-decryptor encryption}
 We introduce a new natural security definition for SDE called \emph{correlated-challenge security} (SDE-CORR) that we believe better captures the security intuition and use cases for SDE, and we also show that it implies all existing SDE security definitions.

\begin{theorem}[Informal]
SDE-CORR implies all seven prior security notions for SDE,
including identical-challenge security.
\end{theorem}

\paragraph{Characterizing the prior SDE definitions} As mentioned before, the relationship of the prior SDE security notions among themselves is also unclear. In this work, we establish an almost complete characterization of the relationship of the prior definitions, which also shows their inadequacy: For example, with respect to independent challenges, we show that {indistinguishability} SDE security definitions does not even imply search security for SDE! Further, we show that independent-challenge indistinguishability SDE security does not imply identical-challenge indistinguishability SDE security, which refutes a folklore conjecture and explains the inadequacy of previous techniques for proving identical-challenge indistinguishability security. 

We refer the reader to \Cref{fig:sde-hierarchy} and \Cref{sec:sde-hier-defn} for the complete picture of the relationships between the SDE definitions.

\paragraph{Security definitions for copy-protection of general functionalities} Similar to SDE, we also introduce the natural notion of correlated-challenge security for copy-protection of general functionalities. We believe our definition better models real-life attacks, and it also implies the prior definitions including identical-challenge security.

\subsubsection{New Feasibilities}
\paragraph{New SDE feasibilities} Using our new MoE game, we achieve our correlated-challenge SDE security notion in the plain model, with the same cryptographic assumptions as previous works.
\begin{theorem}[Informal]
Assuming indistinguishability obfuscation (iO) and one-way functions, there exists an SDE scheme satisfying SDE-CORR.
\end{theorem}
In particular, this is the first SDE construction satisfying identical-challenge security in the plain model\footnote{As mentioned before, we note that there have been other non-standard solutions, such as schemes in the QROM \cite{TCC:KitYam25}, schemes with quantum ciphertexts \cite{AKY25} and schemes that rely on unproven conjectures \cite{AB24}.}, resolving the question introduced by Georgiou and Zhandry (2020) and left open throughout subsequent work~\cite{C:CLLZ21,C:AnaKalLiu23,EPRINT:CGLR23a,TCC:KitNis23,TCC:CheHerVu23,AB24,AKY25,TCC:KitYam25}.

\paragraph{New feasibilities for copy-protecting general classes of functionalities} Again using our new MoE game; we achieve correlated-challenge copy-protection security for puncturable functionalities (which also achieves identical-challenge security for the first time), and also copy-protection for point functions, k-point functions and compute-and-compare programs under the natural security definitions for the first time. Our results also improve prior work by allowing the adversary to receive auxiliary information regarding the copy-protected key and the challenges.

\begin{theorem}[Informal]
Assuming indistinguishability obfuscation (iO) and hardness of Learning with Errors (LWE), there exists \emph{correlated-challenge secure} copy-protection for all puncturable functionalities in the plain model.
\end{theorem}
\begin{theorem}[Informal]
Assuming iO and LWE, there exists plain-model copy-protection for point functions, $k$-point functions, and compute-and-compare programs with respect to the natural challenge distributions.
\end{theorem}
The formal results are given in \Cref{sec:correlated-upo-copy-protection-applications}.

These answer the open questions of \cite{Aar09,C:CLLZ21,C:ALLZZ21,CMP24,AB24,ABH+26,CG26}. 

To achieve these results, we follow the modular approach of \cite{AB24}, which uses a primitive called \emph{unclonable puncturable obfuscation (UPO)}. We first define \emph{correlated-challenge security} for UPO, and then show how to achieve it.

\begin{theorem}[Informal]
Assuming polynomially secure indistinguishability obfuscation and hardness of LWE, there exists an unclonable puncturable obfuscation (UPO) scheme that satisfies correlated-challenge security in the plain model.
\end{theorem}

This in particular gives the first provably secure identical-challenge UPO in the plain model, answering the open questions of \cite{AB24,ABH+26}.

\paragraph{Additional results} We note that for all our results (SDE, general copy-protection and UPO), we also achieve various stronger versions such as \emph{oracular security} considered by previous works \cite{ABH+26} (where even after splitting, freeloader adversaries additionally obtain query access to the functionality (though punctured at the challenge point)), and all our results can be made to have optimal bit-security of $2^{-\Omega(\lambda)}$-security using $O(\lambda)$ qubits\footnote{For this case, rather than just polynomial security, we need to assume (any) subexponential security level for $\io$ and LWE, and use sufficient padding for them.}. We refer the reader to the relevant sections for more details and complete statements of our results (see \cref{sec:organization} for the organization).

\subsection{Organization}\label{sec:organization}
\ifconference
In \cref{sec:toc}, we provide a table of contents.
\paragraph{Main body}
In   \Cref{sec:tech-overview}, we provide a technical overview of our work. In \Cref{sec:sde-conf} and \Cref{sec:cp-conf}, we provide our definitions and new feasibility results for SDE and copy-protection of puncturable functionalities, respectively. In \Cref{sec:it-moe}, we present our decisional MoE game for coset states along with its proof.

\paragraph{Supplementary materials} In \cref{sec:prelims} we provide the notations, prior cryptographic definitions, and standard quantum information-theoretic facts and lemmas needed for the results and proofs. In \Cref{sec:lem-it-moe}, we provide proofs of some technical lemmas from \Cref{sec:it-moe}. In \Cref{sec:comp-moe}, we state and prove the computational version of our MoE game. In \Cref{sec:other-prior-sde-def}, we recall the other prior definitions for SDE and show the relationship between the existing definitions as well as our new definition. In \Cref{sec:UPO-def} we discuss the (old and new) definitions of UPO, and in \Cref{sec:other-prior-cp-def}, we discuss the other copy-protection definitions.
    In \Cref{sec:sde-construction}, \Cref{sec:upo-construction} and \Cref{sec:correlated-upo-copy-protection-applications}, we state and prove our main feasibility result for SDE, UPO and copy-protection, respectively.

\else
    \Cref{sec:prelims} recalls notation and preliminaries.
    \Cref{sec:sde-conf} define the SDE
    notions and establish the relationship between the different notions, after which
    \cref{sec:UPO-def,sec:copy-protection-security-definitions} recall the copy-protection definitions and  define correlated-challenge UPO, respectively.
    Then \Cref{sec:it-moe,sec:comp-moe} proves the information-theoretic and computational decisional
    coset-monogamy theorems.  Finally, \cref{sec:sde-construction} constructs a correlated-challenge SDE, and \cref{sec:upo-construction} constructs a correlated secure UPO, and \cref{sec:correlated-upo-copy-protection-applications} derives the copy protection applications.
\fi

%% file: 111-THIS-PROJECT/acknowledgement.tex
A{C} was supported by a CyLab Presidential Fellowship and by the following grants of VG: NSF award 1916939, DARPA SIEVE program, a gift from Ripple, a DoE NETL award, a JP Morgan Faculty Fellowship, a PNC Center for Financial Services Innovation Award, and a CyLab seed funding award.

AC thanks Omri Shmueli for helpful discussions regarding bit-security and quantum memory in unclonable cryptography.

\subsubsection*{AI Acknowledgment}
Proof of our information-theoretic MoE theorem, and the definition of the SDE-CORR game was obtained with the help of discussions with ChatGPT-5.6 Sol. Rest of the results were fully human-generated directly by the authors. Later, for the purpose of speeding up the writing, ChatGPT-5.6 Sol was also used for drafting the initial versions of some text. The authors take full responsibility for the whole paper.

%% file: 112-CONFERENCE.tex
\input{111-THIS-PROJECT/abstract}
\input{intro}
\subsection{AI Acknowledgment}
The proof of the generalized unclonable encryption lemma (\Cref{lem:general-UE}), and the definition of the SDE-CORR game was obtained with the help of discussions with ChatGPT-5.6 Sol. Rest of the results were fully human-generated directly by the authors. ChatGPT-5.6 Sol was also used to verify our results and perform grammatical checks on the writing. 

The authors take full responsibility for the whole paper.

\input{tech-overview}

\input{defns/main}

\input{moe/it-moe}

\addrefs 
\clearpage
\appendix

\phantomsection
\addcontentsline{toc}{part}{Supplementary Material}

\begin{center}
    {\Huge\bfseries Supplementary Material\par}
\end{center}
\vspace{1em}

\setcounter{tocdepth}{2} 
\tableofcontents\label{sec:toc}

\clearpage
\section{Preliminaries}
\label{sec:prelims}
\input{000-include/boilerplate/notation}

\input{000-include/boilerplate/crypto-definitions}
\subsection{Quantum Information}\label{sec:quant-inf}
We recall some facts about quantum information that we use later.
\input{000-include/boilerplate/quantum-results/gentlemes}
\input{000-include/boilerplate/quantum-definitions/threshold-implementations}

\input{moe/it-moe-general-ue-proof}
\input{moe/comp-moe}

\input{defns/old-sde-defn}
\input{defns/sde-hierearchy}
\input{defns/UPO-defn-and-properties}

\input{defns/old-cp-defn}
\input{sde/main}
\input{upo/main}

\input{upo/applications}

%% file: tech-overview.tex
\section{Technical Overview}\label{sec:tech-overview}
In this section, we give a technical overview of our work. We first recall the relevant definitions, and then introduce our new definitions. We then give an overview of our main technical tool - our new monogamy-of-entanglement (MoE) game; compare it with previous MoE results and discuss how it helps us achieve our results and simplifies proofs, and finally explain the proof ideas behind proving the new MoE game.  

Formal definitions, constructions and proofs can be found in the relevant sections (see \cref{sec:organization} for the organization).

\subsection{Recalling Definitions}
We briefly recall unclonable encryption (UE), single-decryptor encryption (SDE), and copy-protection of general functionalities. In the security experiment of each notion, an adversary $\alice$ receives a quantum state and splits it into two registers, which it sends to $\bob$ and $\charlie$. These registers may be entangled, but $\bob$ and $\charlie$ cannot communicate after the split.

\paragraph{Unclonable encryption (\cref{def:UE-security})}
In an unclonable encryption scheme, we encrypt a classical message into a quantum ciphertext $\ket{\psi_{ct}}$ using a classical secret key $\sk$. In the security experiment, $\alice$ receives an encryption $\ket{\psi_{ct}} \samp \enc(\sk, m_b)$ for a random bit $b$. After the split, both $\bob$ and $\charlie$ receive the secret key $\sk$, and try to guess $b$. We require that the probability of both parties guessing correctly is at most $1/2 + \negl(\lambda)$.

\paragraph{Single-decryptor encryption (\cref{sec:sde-conf})}
In an SDE scheme, we have a quantum decryption key $\ket{\psi_{sk}}$, and a classical public key $\pk$ and classical ciphertexts. In the security experiment, $\alice$ receives $\ket{\psi_{sk}}$ (and $\pk$). After the split, both $\bob$ and $\charlie$ receive challenge ciphertexts. As we discuss in detail later, precise winning condition depends on the particular definition: For example, in \emph{independent-challenge (indistinguishability) security}, $\bob$ and $\charlie$ receive independent ciphertexts $\ct_{\bob} \samp \enc(\pk, m_{b_1})$ and $\ct_{\charlie} \samp \enc(\pk, m_{b_2})$, and try to guess their respective challenge bits $b_1, b_2$, and we require that probability of both parties guessing correctly is at most $1/2 + \negl(\lambda)$. In \emph{identical-challenge (indistinguishability) security}, $\bob$ and $\charlie$ receive the same challenge ciphertext $\ct \samp \enc(\pk, m_b)$ and both try to guess $b$, and we similarly require that they cannot do so with probability better than $1/2 + \negl(\lambda)$.

\paragraph{Copy-protection of general functionalities \ifconference(\cref{sec:cp-conf})\else(\Cref{sec:copy-protection-security-definitions})\fi}
In a copy-protection scheme for a functionality $\mathfrak{C}$, we encode a classical key or circuit $C \samp \mathfrak{C}$ into a quantum state $\ket{\psi_C}$ that allows evaluating $C$ any number of times\footnote{Reusability follows by correctness due to the \emph{gentle measurement technique} \cite{awinter,Aar16}.}  on any choice of input values. In the security experiment, $\alice$ receives $\ket{\psi_C}$, and after the split, $\bob$ and $\charlie$ receive challenge inputs $x_{\bob}$ and $x_{\charlie}$. Then, they are required to \emph{answer} these challenges, and we require that the probability of both succeeding is at most $p_{trivial} + \negl(\lambda)$ where $p_{trivial}$ is the highest probability of an adversary answering a challenge without ever receiving $C$ or $\ket{\psi_C}$. The notion of \emph{answering the challenge} depends on the application; for example, for a circuit $C$, it will be outputting $C(x)$, whereas in the case of a signing key, it will be outputting any valid signature for $x$, or for a decryption key, it will be guessing the challenge bit correctly from $x \samp \enc(\pk, m_b)$. As we discuss in detail later on, similar to SDE, the challenges  $x_{\bob}$ and $x_{\charlie}$ can be the same or independently sampled.

\subsection{New Definitions}
\paragraph{Our new SDE-CORR definition}
We define a new security notion for SDE schemes, which we call \emph{correlated-challenge security} (SDE-CORR). In this overview, we present its single-ciphertext version; see \cref{sec:sde-new-defn} for the full definition. In this security game, at the splitting stage, $\alice$ chooses one of two modes: \emph{independent} or \emph{identical}. In \emph{independent mode}, $\alice$ chooses two message pairs $(m_{\bob,0},m_{\bob,1})$ and $(m_{\charlie,0},m_{\charlie,1})$, and then
the challenger samples $(b_\bob,b_\charlie)$ from an arbitrary efficiently sampleable distribution with uniform marginals, and then independently samples
    $\ct_X \samp \enc(\pk,m_{X,b_X})$ for $X\in\{\bob,\charlie\}$, which are then given to the corresponding parties.
In \emph{identical mode}, $\alice$ chooses one pair $(m_0,m_1)$, and then the challenger samples $b\samp\zo$, computes one ciphertext $\ct\samp\enc(\pk,m_b)$, and gives the exact same $\ct$ to both parties $\bob$ and $\charlie$. In either mode, security requires
\[
    \Pr[
        \widetilde b_\bob=b_\bob
        \ \wedge\
        \widetilde b_\charlie=b_\charlie
    ]
    \leq \frac12+\negl(\lambda),
\]
where $b_\bob=b_\charlie=b$ in identical mode.

The uniform-marginal restriction in the separate mode preserves the standard $1/2$ trivial success probability while allowing every possible correlation between the two challenge bits. Indeed, for every fixed $\pk$, any such distribution can equivalently be sampled by choosing a uniform bit $b$ and a possibly biased bit $\Delta$, independently of $b$, and setting $b_\bob=b,$ and $b_\charlie=b\oplus\Delta$.
 In particular, the above notion restricted to the separate mode captures independent challenges ($\Delta$ uniform), identical challenge bits with independently generated ciphertexts ($\Delta=0$), opposite challenge bits ($\Delta=1$), and everything in between. On the other hand, identical mode captures the stronger requirement that the ciphertext itself is the same. 
 We note that in both modes, the message pair(s) are chosen jointly and may be arbitrarily correlated. Thus, the SDE-CORR definition covers all kinds of meaningful correlation between challenge ciphertexts.\footnote{In particular, along with correlation in messages, ciphertexts are either separately generated (which is covered in the separate mode of our definition), or identical ciphertexts are given to every user (which is covered in the identical mode of our definition).} We show that the general version of our definition (with multiple ciphertexts) implies all previous SDE security notions.
 
\subsection{A New Monogamy-of-Entanglement Game: Decisional and Unlearnable (\cref{sec:it-moe})}
We provide an overview of the main technical tool underlying our results: a novel monogamy-of-entanglement (MoE) property we prove for coset states. 

\paragraph{Coset states} We first recall coset states\footnote{Sometimes also called subspace-coset states.} \cite{EC:VidZha21,C:CLLZ21}.  Let $A$ be a subspace of $\F_2^{2n}$, and let $s, t$ be two vectors in $\F_2^{2n}$. Then, the associated coset state $\ket{A_{s,t}}$ is defined to be 
\begin{equation*}
    \ket{A_{s,t}} := X^sZ^t \sum_{a \in A}\frac{1}{\sqrt{{|A|}}} \ket{a} = \sum_{a\in A} \frac{1}{\sqrt{{|A|}}}(-1)^{\langle a,t\rangle}\ket{a+s}
\end{equation*}
We also recall that $H^{\otimes 2n}\ket{A_{s,t}} = \ket{(A^\perp)_{t,s}}$.

\paragraph{\underline{Our New MoE Game}}
We prove that for any (unbounded) adversary $(\alice, \bob, \charlie)$, their probability of winning the following game $\gam{\dcmoe}$ is at most \begin{equation}\label{eq:overview-dcmoe-bound}
\frac12+\frac12\sqrt{\frac{2^{2n}}{(2^{2n}-1)(2^{2n-1}-1)}}\leq \frac12+2^{-\Omega(n)}.
\end{equation}

\begin{mdframed}
\noindent\underline{\textbf{$\gam{\dcmoe}$}}
\begin{enumerate}
\item Sample uniformly at random an $n$-dimensional subspace $A\leq\FF_2^{2n}$ and vectors $s,t\in\FF_2^{2n}$.
\item $\alice$ receives $\ket{A_{s,t}}$ and outputs two (entangled) registers $\regi_1, \regi_2$.
\item Sample $u\getsr A\setminus\{0\}$ and $v\getsr A^\perp\setminus\{0\}$, and set
\begin{equation*}
   b^* :=  \langle u || v, t || s \rangle =  \langle u,t\rangle\oplus\langle v,s\rangle \in \zo
\end{equation*}
\item $\bob$ receives $\regi_1, (u, v)$ and $\charlie$ receives  $\regi_2, (u, v)$. They output bits $b_\bob$ and $b_\charlie$, respectively. 
\item The adversary wins if $b_\bob=b_\charlie= b^*$.
\end{enumerate}
\end{mdframed}
We now point out a nice property about this game, which will be quite useful in copy-protection security proofs (as we discuss later on). Observe that for all vectors $w' \in A^\perp + t$ and $u \in A$, the inner product $\langle u, w' \rangle$ is the same value, $\langle u,t\rangle$, since $u \in A$ \emph{annihilates} the subspace $A^\perp$ component of $w'$. Similarly, for all vectors $w \in A + s$ and $v \in A^\perp$, the inner product $\langle v, w \rangle$ is the same value, $\langle v,s\rangle$.
On a related note, observe that it is easy to construct an adversary $(\alice, \bob, \charlie)$ where $\bob$ (or $\charlie$) wins - but they cannot win at the same time; showing that our security guarantee is non-trivial. In more detail, given $\ket{A_{s,t}}$ and $(u,v)$, one can obtain $\langle v,s\rangle$ without disturbing the state since any computational basis measurement outcome $a+s$ of $\ket{A_{s,t}}$ satisfies $\langle v, a+s \rangle = \langle v, s \rangle$, and we can similarly obtain $\langle u,t\rangle$ after applying $H^{\otimes 2n}$ to $\ket{A_{s,t}}$.

\paragraph{Computational MoE} Similar to previous MoE results for coset states \cite{C:CLLZ21,CulfVidick22,ABH+26}; we also prove that, assuming $\io$\footnote{It also follows that the same guarantee holds unconditionally against query-bounded adversaries when the adversary quantum-query access to these programs as oracles instead of $\io$.} and one-way functions, our MoE property holds (against efficient/QPT adversaries) even when the adversary $\alice$ is given the obfuscated membership programs $\io(A+s)$ and $\io(A^\perp+t)$ and even some auxiliary information about $A, u, v$, satisfying certain simulatable properties. As we will see later on, this will be crucial in copy-protection constructions/proofs, see \Cref{sec:upo-construction,sec:correlated-upo-copy-protection-applications}. 

Note that these programs can be used to verify computational/Hadamard basis measurement outcomes of $\ket{A_{s,t}}$, and can also be used to implement the projector $\ket{A_{s,t}}\bra{A_{s,t}}$. Therefore, our computational MoE in particular, show that the state in our game is \emph{unlearnable}, which is crucial for copy-protection applications.

\paragraph{Comparison to prior MoE results} Our MoE result is the first to achieve the following three properties at the same time: (i) it is a decisional game (with 1-bit output) rather than a search game (with a long output), (ii) it is for an unlearnable state and (iii) both adversaries $\bob, \charlie$ receive the same input and try to predict the same answer. This particular combination is indeed crucial for achieving our results such as \emph{identical-challenge copy-protection} and \emph{point function copy-protection}, and it is also crucial for having a unified solution/technique for unclonable encryption and copy-protection.
We now give an overview of some prior MoE results and how they lack at least one of these properties, explaining why prior works were not able to achieve our copy-protection results.

\begin{itemize}
\item \emph{Search MoE for coset states:} \cite{C:CLLZ21} and \cite{CulfVidick22} considered an MoE game for coset states similar to ours, where $\alice$ receives $\ket{A_{s,t}}$. However, in their game, $\bob$ and $\charlie$ receive the subspace description $A$ (rather than vectors $u,v$), and are required to output vectors $w \in A+s$ and $w'\in A^\perp+t$, respectively, making it a search game. Thus, this MoE game is not helpful for proving (identical-challenge\footnote{Note that, as we discuss later on, in the independent-challenge case \cite{C:CLLZ21} has been able to utilize this search MoE game for copy-protection using complicated rewinding-based eigenvalue estimation \cite{MW05,TCC:Zhandry20} and simultaneous extraction techniques}) copy-protection or unclonable encryption since they are inherently of decisional nature.

\item \emph{Decisional MoE for learnable states.}
Starting with the seminal work of \cite{EC:TFKW13}, and culminating in the very recent works of Ananth and Sahai \cite{AS26} and Ragavan \cite{Rag26}, various decisional MoE guarantees for BB84 states and Pauli eigenstates have been shown. However, these states are unentangled and are learnable/clonable given access to a verification oracle (e.g. a projector) for them. This makes these results unsuitable for copy-protection of functionalities, since an adversary can implement a verification oracle simply by running the copy-protection evaluation algorithm on a random input on the state and checking the output (e.g. for a signing key, it can verify the resulting signature).

\item \emph{Decisional MoE with independent challenges.} Inspired by the Goldreich-Levin extractor \cite{STOC:GolLev89} based search-to-decision reductions ubiquitous in cryptography, various works \cite{C:CLLZ21,AB24,AKY25,TCC:KitYam25,CG26} considered decisional variants of the coset search MoE game of \cite{C:CLLZ21,CulfVidick22}. However, due to possible entanglement between $\bob$ and $\charlie$, all these works were only able to prove security in the case where the two adversaries $\bob$ and $\charlie$ receive independent \emph{seeds} $r_1, r_2$ and try to predict different outcomes, making the resulting statements not useful for identical-challenge copy-protection or unclonable encryption\footnote{Concurrent to our work, two recent breakthroughs~\cite{CLX26,Ana26} have provided an alternate avenues to deal with this problem via generic search-to-decision compilers in the identical challenge setting.}. 
\end{itemize}


\paragraph{A unifying and simplifying technique}\label{pg:nice-properties} Previously, works on different areas of unclonable cryptography, such as unclonable encryption~\cite{AS26,Rag26} and copy-protection of functionalities \cite{C:CLLZ21,CulfVidick22,AB24,ABH+26}, relied on different techniques and quantum states to overcome similar barriers, because (as explained above) previous techniques provided either decisional security, or, unlearnability, but no single proof technique provided all the guarantees; on the contrary, our new MoE game satisfies all of these guarantees simultaneously. Moreover, most MoE games require canonicalization, and hence the existing non-local security proofs of unclonable primitives based on them, relied on complicated and delicate simultaneous extractions to deal with canonicalization required to answer correctly for the MOE game. Since our new MOE game removes the need for canonicalization, it enables us to bypass simultaneous extractions, thereby significantly simplifying and unifying much of the existing proofs. We explain these simplifications in more detail in \cref{sec:techoverview-proving-sde}.

Furthermore, our MoE game already embodies an unclonable encryption scheme! In particular, our MoE game can be reformulated as the security experiment of a new UE scheme, that we call the \cosetue scheme, and is detailed below. 

\subsubsection{Proof of Our New MOE Game}\label{sec:overview-dcmoe-proof}

We first express the game above as the security experiment of an unclonable bit-encryption scheme based on coset states. Let $q=2n$, and for any $a\in\FF_2^{q}$, define the Pauli operators $X^a$ and $Z^a$ by $X^a\ket{x}:=\ket{x\oplus a}$ and $Z^a\ket{x}:=(-1)^{\langle a,x\rangle}\ket{x}$.
\begin{mdframed}
\noindent\textbf{\cosetue}
\begin{enumerate}
\item $\keygen$: Sample a uniformly random nonzero orthogonal pair $k=(u,v)$, meaning that $u,v\in\FF_2^{q}\setminus\{0\}$ and $\langle u,v\rangle=0$.
\item $\enc(k,e)$: Sample a uniformly random $q/2$-dimensional subspace $A$ satisfying $u\in A$ and $v\in A^\perp$. Sample $s,t\in\FF_2^{q}$ uniformly subject to $\langle u,t\rangle\oplus\langle v,s\rangle=e$, and output $\ket{A_{s,t}}$.
\item $\dec(k,\rho)$: Measure $X^uZ^v$ on $\rho$ and output the bit $e$ corresponding to the measured eigenvalue $(-1)^e$.
\item Security: Give an encryption of a uniform bit $e$ to $\alice$. After $\alice$ sends one register to each recipient, reveal the same key $(u,v)$ to both. They win only if both output $e$.
\end{enumerate}
\end{mdframed}
Note that every nonzero orthogonal pair $(u,v)$ belongs to the same number of subspaces $A$ satisfying $u\in A$ and $v\in A^\perp$, and for any fixed $A,u,v$, the bit $\langle u,t\rangle\oplus\langle v,s\rangle$ is uniform over $s,t$. Consequently, sampling the coset state first and $(u,v)$ afterwards, as in the monogamy game, gives the same joint distribution as sampling $(u,v)$ first and then encrypting, as in \cosetue. Hence, the success probability of any adversary in \dcmoe, is the same as its success probability in the unclonable indistinguishable security experiment against \cosetue. (For the formal equivalence, see \cref{lem:dcmoe-ue}.) Thus, it is enough to show that \cosetue in the unclonable indistinguishable security experiment, has advantage at most $\frac12\sqrt{\frac{2^{q}}{(2^{q}-1)(2^{q-1}-1)}}=\frac12\sqrt{\frac{2^{2n}}{(2^{2n}-1)(2^{2n-1}-1)}}$, i.e., for any adversary $(\alice,\bob,\charlie)$, the success probability is at most \[\frac{1}{2}+\frac12\sqrt{\frac{2^{q}}{(2^{q}-1)(2^{q-1}-1)}}=\frac12+\frac12\sqrt{\frac{2^{2n}}{(2^{2n}-1)(2^{2n-1}-1)}}.\]
To show this, we first consider the mixed state describing the ciphertext,  for a fixed key and message $e\in \bit$, under \cosetue.
\paragraph{Conditional state of the ciphertext given the message and the key.} Fix a key $(u,v)$, with $u\in A$ and $v\in A^\perp$, for the \cosetue, and a message bit $e$, we will show that the ciphertext has a particular form, which will be beneficial for showing unclonable indistinguishable security.
Let $\ket{A}:=\ket{A_{0,0}}$, and let $s,t$ be sampled during encryption according to the message bit $e$. 

Directly from the definition of a coset state and the Pauli commutation rule, we have
\begin{equation}\label{eq:overview-pauli-identities}
X^xZ^z\ket{A}=\ket{A_{x,z}}
\quad\text{and}\quad
Z^vX^s=(-1)^{\langle v,s\rangle}X^sZ^v.
\end{equation}
Since $\ket{A_{s,t}}=X^sZ^t\ket{A}$, the above identities yield,
\begin{equation}\label{eq:overview-pauli-shifted-coset}
X^uZ^v\ket{A_{s,t}}=(-1)^{\langle v,s\rangle}\ket{A_{s+u,t+v}}.
\end{equation}

Next, expanding the state on the right and changing the summation variable from $a$ to $a+u$ gives
\begin{equation}\label{eq:overview-coset-change-variable}
\ket{A_{s+u,t+v}}
=\frac{1}{\sqrt{|A|}}\sum_{a\in A}(-1)^{\langle a,t+v\rangle}\ket{a+s+u}
=(-1)^{\langle u,t\rangle}\ket{A_{s,t}}.
\end{equation}
The condition $v\in A^\perp$ is used here to remove $\langle a,v\rangle$, while $u\in A$ ensures that $a\mapsto a+u$ is a change of variable within $A$. Combining the last two equations, we get
\begin{equation}\label{eq:overview-coset-pauli-eigenstate}
X^uZ^v\ket{A_{s,t}}=(-1)^{\langle u,t\rangle\oplus\langle v,s\rangle}\ket{A_{s,t}}.
\end{equation}
Thus, we conclude that for any key $(u,v)$ and message $e$ in \cosetue, the cipher of message $e$ lies in the$(-1)^e$-eigenspace of $X^uZ^v$, \emph{oblivious of the choice of $A$}. We next show that for a fixed key $(u,v)$ and message $e$, the ciphertext is in fact uniformly distributed over this entire eigenspace. Fix the key $(u,v)$ and message $e$.


Fix any subspace $A$ satisfying $u\in A$ and $v\in A^\perp$. Note that any phase shift $t$, two coset shifts $s,s'$ that differ by an element of $A$ describe the same shifted subspace $A+s$ and hence $\ket{A_{s,t}}=\ket{A_{s',t}}$, while if their difference is outside $A$, then the corresponding coset states $\ket{A_{s,t}},\ket{A_{s',t}}$ are supported on disjoint sets of computational-basis vectors and hence are orthogonal. Similarly for the dual subspace $A^\perp$, and primal shift $s$, two choices of the dual shift $t,t'$ that differ by an element of $A^\perp$ give the same dual coset state $\ket{A^\perp_{t,s}}=\ket{A^\perp_{t',s}}$, while if $t-t'\not\in A^\perp$, then $\ket{A^\perp_{t,s}}$ and $\ket{A^\perp_{t',s}}$ would be orthogonal; since $\ket{A_{s,t}}=H^{\otimes q}\ket{A^\perp_{t,s}}$, these equality and orthogonality conditions are preserved when comparing the corresponding primal coset states $\ket{A_{s,t}}$ and $\ket{A_{s,t'}}$. 

Hence, over all possible vectors $s,t$, we get $2^{q}$ distinct coset states that form an orthonormal set inside $(\CC^2)^{\otimes q}$. Thus, by standard dimensionality argument, they must form an orthonormal basis of the full $2^{q}$-dimensional space. Since $X^uZ^v$ is a non-identity Pauli operator, each of its two eigenspaces has dimension $2^{q-1}$. Exactly half of this coset basis has eigenvalue $(-1)^e$, hence if we restrict the vectors $s,t$ to satisfy $\langle u,t\rangle\oplus\langle v,s\rangle=e$, the allowed ciphertexts form an orthonormal basis of that eigenspace. Moreover, the uniform choices of $s,t$ make the mixture uniform on this basis. Thus the ciphertext forms a maximally mixed state over the $(-1)^e$ eigenspace of $X^uZ^v$, for any fixed key $k=(u,v)$, message $e$, and any fixed (admissible) subspace $A$.

Formally, writing $N:=2^{q}$ and letting $\Pi_{k,e}$ be the projector onto the $(-1)^e$-eigenspace, we obtain that the conditional cipher text $\rho^A_{k,e}$ satisfies the following.
\begin{equation}\label{eq:overview-coset-mixture}
\Pi_{k,e}=\frac{I+(-1)^eX^uZ^v}{2}
\quad\text{and}\quad
\rho^A_{k,e}=\frac{2}{N}\Pi_{k,e}=\frac1N\bigl(I+(-1)^eX^uZ^v\bigr).
\end{equation}
Thus, averaging over the randomly chosen subspaces satisfying $u\in A$ and $v\in A^\perp$ leaves the same distribution or mixed state. Hence, we conclude our main technical observation: conditioned on the key and message bit $e$, a \cosetue ciphertext is simply the maximally mixed state over the $(-1)^e$ eigenspace of $X^uZ^v$. (For the complete calculation, see \cref{claim:conditional-cipher-state}.)

\paragraph{Using the argument of Ananth and Sahai~\cite{AS26}.} The form of the conditional ciphertext of \cosetue in \eqref{eq:overview-coset-mixture} is exactly the structure used in the unclonable indistinguishable security proof of the UE construction in~\cite{AS26}, which enables us to prove security for \cosetue using similar techniques. In particular, we extend their result to a general class of UE. Each UE scheme in this general class, is associated with a family of Hermitian operators $\{\Theta_k\}_k$ indexed by keys in a key space, on an $N$-dimensional space such that each $\Theta_k$ has eigenvalues $+1$ and $-1$ with equally large eigenspaces, and $\Tr(\Theta_k\Theta_{k'})=0$ for distinct keys. A key is sampled from a distribution $P$ on the keyspace and the encryption of $e$ is the maximally mixed state over the $(-1)^e$-eigenspace of $\Theta_k$, given as $N^{-1}(I+(-1)^e\Theta_k)$. We show that any such scheme, in the unclonable indistinguishable security game, has advantage at most
\begin{equation}\label{eq:overview-general-operator-bound}
\frac12\sqrt{N\col(P)}.
\end{equation}
where the quantity $\col(P)$ is the probability of finding collision between two independently sampled keys from the distribution $P$. The proof of unclonable indistinguishable security for this general class of UE schemes, is based on similar arguments as in~\cite{AS26}.  (For the general statement and proof, see \cref{lem:general-UE}.)

It is easy to see that the unclonable-encryption scheme of~\cite{AS26} is also captured by this general class of UE schemes. Due to the Pauli property that $\Tr((X^uZ^v)(X^{u'}Z^{v'}))=0$ for distinct Pauli operators, \cosetue also falls in this general class of UE with $\theta_k=X^uZ^v$ for any key $k=(u,v)$. Instantiating $P$ as the uniform distribution over admissible keys of \cosetue, i.e., the set of nonzero orthogonal pairs gives the required bound in \eqref{eq:overview-dcmoe-bound}, proving our information-theoretic decisional coset monogamy of entanglement theorem.

\paragraph{Computational MoE.} We reduce our computational MoE game to the information-theoretic MOE, based on the subspace-hiding obfuscation and simulatable property of the auxiliary information, by adapting the proof template used in~\cite{C:CLLZ21,ABH+26,CG26}. The formal theorem is given in \Cref{thm:comp-dcmoe-aux}.

\subsection{Proving Copy-Protection Using Our New MoE Game}\label{sec:techoverview-proving-sde}
We now give an overview of how our copy-protection proofs follow from our MoE result.
\paragraph{Single-Decryptor Encryption} For SDE, we use the coset-state based construction of Kitagawa and Yamakawa \cite{TCC:KitYam25}, which we briefly recall. In this construction, the quantum decryption key is the coset state $\ket{A_{s,t}} := \ket{\psi_{sk}}$, and the public-key is the obfuscated membership programs $(\io(A+s), \io(A^\perp+t)) := \pk$. To encrypt a message $m \in \zo$, we do a 2-out-of-2 secret sharing of it as $m \oplus r$ and $r$, and output as the ciphertext the obfuscated program $\io(P) := \ct$ that does the following. On input $(\mathsf{mode},w)$, if $\mathsf{mode}=0$, then $P$ outputs $m \oplus r$ if $(\io(A+s))(w) = \mathsf{TRUE}$ (i.e. if $w \in A + s$), and outputs $\bot$ otherwise. If $\mathsf{mode}=1$, then $P$ outputs $r$ if $(\io(A^\perp+t))(w) = \mathsf{TRUE}$ (i.e. if $w \in A^\perp + t$), and outputs $\bot$ otherwise.

Now we discuss our reduction of SDE security to our MoE game. For simplicity, consider a single-bit challenge message and identical-challenge security for SDE, the general case of correlated challenge security with arbitrary adversary-chosen messages (SDE-CORR) follows similarly (see \Cref{sec:sde-construction}). We write
\[
e_0=\langle v,s\rangle,\qquad
e_1=\langle u,t\rangle.
\]
In the first hybrid of SDE security, we sample independent uniform bits $q,\zeta$, and we set the SDE challenge message $m$ and the encryption pad $r$ as
\[
r=q\oplus e_1
\qquad\text{ }\qquad
m=e_0 \oplus e_1 \oplus\zeta.
\]
Note that these are distributed exactly as a fresh uniform challenge message/bit and a fresh uniform pad. Also observe that we have
\[
m \oplus r = q\oplus\zeta\oplus e_0.
\]
Crucially, recall that for all vectors $w' \in A^\perp + t$, the inner product $\langle u, w' \rangle$ is the same value, $e_1 = \langle u,t\rangle$. Thus, in the next hybrid by $\io$ security, we can change our ciphertext program $P$ so that, on input $(\mathsf{mode}=1,w')$, it outputs $q \oplus \langle u, w' \rangle$ (if $(\io(A^\perp+t))(w') = \mathsf{TRUE}$) instead of $r = q\oplus e_1$. Similarly, on inputs $(\mathsf{mode}=0,w)$ such that $(\io(A+s))(w) = \mathsf{TRUE}$, we can change $P$ to output $q \oplus \zeta \oplus \langle v, w \rangle$ instead of $m \oplus r = q\oplus\zeta\oplus e_0$.

The nice property of this hybrid is that now our challenge ciphertext only uses (i) $u,v$, which will be provided by the MoE challenger at the point when we need to give the ciphertext to the SDE adversaries $\bob, \charlie$, and (ii) $q, \zeta$, which we can sample as part of the reduction. Finally, note that for a (successful) SDE adversary, the outputs of the freeloaders $\bob, \charlie$ will both be $m = e_0 \oplus e_1 \oplus\zeta$, from which we can obtain $e_0 \oplus e_1$ and submit it to the MoE challenger. Hence, we can reduce this hybrid to our computational MoE game, which essentially completes the proof.

\paragraph{Copy-protecting general functionalities} For achieving identical-challenge and correlated-challenge copy-protection for all puncturable functionalities, we follow the modular approach of \cite{AB24}, which uses a primitive called \emph{unclonable puncturable obfuscation} (UPO). We define correlated-challenge security for UPO. The new definition of correlated-UPO plays a similar unifying role as our definition of SDE-CORR security does in the SDE framework, by implying both correlated and identical secure UPO. The new definition also has desirable properties such as allowing for adversaries to receive auxiliary information, which makes UPO even more application friendly. 

Based on our new MoE result, we prove that the construction of \cite{ABH+26,CG26} achieves correlated UPO security. Again due to the properties of our new MoE result (discussed on \Cpageref{pg:nice-properties}), our proofs are significantly easier than previous works. We refer the reader to the related section for the full proof, however, we note that on a high level it follows similarly to the SDE proof discussed above (while dealing with some additional subtleties involving simulatable auxiliary information).

Next, we show that our definition of correlated UPO implies both copy-protection with correlated challenges for puncturable functionalities (as well as stronger guarantees like oracular and pseudorandomness-style security), as well as implies copy-protection of point functions and compute-and-compare functions. This is also the first implication of UPO to compute-and-compare functions, which crucially uses the fact that correlated-UPO security allows leakage of auxiliary information to the adversaries. Here we note that any copy-protection scheme for point functions naturally extends to a copy-protection scheme for compute-and-compare functions: given the function in the clear and copy-protect the point function where the point is the compute-and-compare target value. However, the security reduction requires that the reduction can sample from the preimage of the target value under the function, which may not be efficient. We circumvent this issue by giving the preimage to the adversary as auxiliary information, permitted by the correlated-UPO security definition. (For more details, see \Cref{sec:correlated-upo-copy-protection-applications}.)

\paragraph{Comparison to prior copy-protection proofs}\label{pg:copy-protection-security-no-need-of-canonicalization}
We now explain why our MoE game also significantly simplifies copy-protection proofs compared to prior work. The main reason is that prior works \cite{C:ALLZZ21,C:CLLZ21,C:AKLLZ22,C:AnaKalLiu23,TCC:KitYam25,CG26,ABH+26}  had to rely on \emph{(i) canonicalization and (ii) simultaneous extraction}. This is because they ultimately reduce copy-protection security to the \emph{search} MoE game for coset states (since this was the only unlearnable MoE available), but the copy-protection adversaries in general solve a (decisional) task (e.g. guessing $b$ from $\enc(pk, m_b)$) rather than outputting a coset vector. Thus, these proofs need to convert a successful decisional copy-protection adversary into one that outputs full elements of $A+s$ and $A^\perp+t$. Moreover, since each coset contains superpolynomially many vectors; due to the well-known barriers on $\io$-based extraction techniques, these proofs first need to \emph{canonicalize} the relevant accepting inputs, so that there is a fixed coset representative that the reduction can aim to extract.

Even after these challenges, extraction is still non-trivial because the two pirate registers can be entangled: extracting a vector $w \in A+s$ from $\bob$'s register may collapse $\charlie$'s register to a non-successful adversary, killing any hopes of extracting $w \in A^\perp + t$ from $\charlie$. Earlier proofs handle this using quantum rewinding/eigenvalue-estimation techniques such as \emph{threshold implementations}~\cite{C:ALLZZ21,C:CLLZ21}, while more recent proofs use simultaneous Goldreich--Levin extraction techniques~\cite{AKY25,CG26,ABH+26}. However, both of these techniques crucially only work for independent challenges.

Our new decisional MoE game removes these steps. Recall, as outlined in \Cref{it:lack-of-canonicalization}, for arbitrary vector representative $w\in A+s$ and $w'\in A^\perp+t$, we have
    $\langle v,w\rangle\oplus\langle u,w'\rangle
    =
    \langle v,s\rangle\oplus\langle u,t\rangle$

Thus, the MoE answer depends only on the two cosets, and not on which vector representatives are chosen; and thus there is no need to canonicalize the cosets. Furthermore, after the $\io$ reprogramming step, the output bit of each successful freeloader is already, up to the known mask $\zeta$, exactly this MoE bit. Hence, the reduction does not extract any coset element from either register and it can simply convert adversaries' decisional answers to answer for the MoE task of guessing $\langle v,s\rangle\oplus\langle u,t\rangle$. Thus, we completely avoid the canonicalization, rewinding/eigenvalue-estimation, and simultaneous-extraction machinery required in prior proofs.

\subsection{Concurrent Work}\label{sec:relatedwork}
In this section, we give an overview of concurrent and independent work.  

\paragraph{Identical-challenge search-to-decision} Concurrent works of Coladangelo, Liu, and Xie~\cite{CLX26} and Ananth~\cite{Ana26} prove an \emph{identical-challenge simultaneous Goldreich--Levin theorem.} \cite{CLX26} shows an efficient reduction from indistinguishability security for UE with Goldreich-Levin extractor to the corresponding search security game for UE, whereas Ananth~\cite{Ana26} proves a similar result with an inefficent reduction. Both works focus only on unclonable encryption and do not consider copy-protection.

\paragraph{Unconditional unclonable encryption}
Two concurrent works, Bhattacharyya, Broadbent, and Culf~\cite{BBC26} and Botteron, Designolle, and Fawzi~\cite{BDF26}, consider generalizations of recent unclonable encryption results of \cite{AS26}. \cite{BBC26} shows the existence of statistically secure one-time UE for arbitrary-length messages based on Clifford encodings. \cite{BDF26} show constructions and lower bounds for general classes of unclonable encryption schemes based on Pauli encodings. Concurrent to us, \cite{BBC26,BDF26} also generalizes the proof of \cite{AS26}, thereby obtaining results similar to \Cref{lem:general-UE}. We note that these works only consider unclonable encryption, and do not have any results on copy-protection.

%% file: defns/main.tex
\section{Single-Decryptor Encryption}\label{sec:sde-conf}
In this section, we present security definitions for single-decryptor encryption (SDE). We first define syntax and correctness. We then recall the prior security definitions in \cref{sec:sde-prior-defn} and present our new definition in \cref{sec:sde-new-defn}. Finally, in \cref{sec:sde-hier-defn}, we discuss the relationships between these notions.

Now we present the syntax.
\begin{definition}[Single-Decryptor Encryption (SDE)]
A public-key single-decryptor encryption scheme
$\sde=(\keygen,\enc,\dec)$ for the message space $\messpa$ consists of the following QPT algorithms.
\begin{itemize}
    \item $\keygen(1^\secparam)$: Takes as input a security parameter and outputs a classical public key $\pk$ and a quantum decryption key $\qsk$.

    \item $\enc(\pk, \mes)$: Takes as input a public-key $\pk$ and a message $\mes \in \messpa$ and outputs a classical ciphertext $\ct$.

    \item $\dec(\qsk, \ct)$: Takes as input a decryption key $\qsk$ and a ciphertext $\ct$, and outputs a message $\mes \in \messpa$ or $\bot$.
\end{itemize}
\paragraph{Correctness:} For any message $\mes \in \messpa$, there exists a negligible function $\negl(\secparam)$ such that \[\Pr[\mes' = \mes : \begin{array}{c}
      (\pk,\qsk)\gets\keygen(1^\secparam)  \\
      \ct \samp \enc(\pk, \mes) \\
\mes' \samp \dec(\qsk,\ct)
\end{array}]\geq 1-\negl(\secparam).\]
\end{definition}

\subsection{Prior Definitions}\label{sec:sde-prior-defn}
\ifconference
We recall the existing CPA-style security definitions for SDE. Each security experiment is between a challenger and an adversary tuple $\adve=(\alice,\bob,\charlie)$.\footnote{In the independent-challenge mode, we use the convention that the adversary may choose a different message pair for each party (as done in~\cite{TCC:KitYam25}). The original definition of Coladangelo et al. (\cite{C:CLLZ21}) requires the two pairs to be equal.}
All other existing SDE security notions, including search, CPA$^{+}$, and strong CPA security, are described in \Cref{sec:other-prior-sde-def}.

\begin{definition}[CPA Anti-Piracy
\cite{C:CLLZ21,TCC:KitYam25,GZ20}]
\label{def:sde-cpa}
Let $\sde$ be an SDE scheme with message space $\messpa$, and let
$\adve=(\alice,\bob,\charlie)$ be an adversary tuple. Consider the following experiment for independent and identical CPA security.

\paragraph{\underline{$\gam{SDE\mbox{-}CPA}_{\sde,\adve}
(1^\secparam)$}}
\begin{enumerate}
    \item The challenger samples
    $(\pk,\qsk)\gets\sde.\keygen(1^\secparam)$ and submits $(\pk,\qsk)$ to
    $\alice$.
    \item If it is the independent CPA security experiment denotes as $\gam{SDE\mbox{-}IND\mbox{-}CPA}_{\sde,\adve}
    (1^\secparam)$, $\alice$ outputs message pairs $(\mes_{\bob,0},\mes_{\bob,1})\in\messpa^2$ and $(\mes_{\charlie,0},\mes_{\charlie,1})\in\messpa^2$, and a pair of registers $(\regbob,\regcharlie)$. Otherwise, if it is the identical CPA security experiment denoted as $\gam{SDE\mbox{-}IDEN\mbox{-}CPA}_{\sde,\adve}
    (1^\secparam)$, $\alice$ outputs a message pair
    $(\mes_0,\mes_1)\in\messpa^2$ along with the pair of registers
    $(\regbob,\regcharlie)$. The challenger submits $\regbob$ to
    $\bob$ and $\regcharlie$ to $\charlie$.
    \item If it is the independent CPA security experiment (i.e., $\gam{SDE\mbox{-}IND\mbox{-}CPA}_{\sde,\adve}
    (1^\secparam)$), the challenger samples $b_\bob,b_\charlie\samp\bit$ independently, and computes $\ct_\bob\samp\sde.\enc(\pk,\mes_{\bob,b_\bob})$ and
    $\ct_\charlie\samp\sde.\enc(\pk,\mes_{\charlie,b_\charlie})$, independently. If it is the identical CPA security experiment (i.e., $\gam{SDE\mbox{-}IDEN\mbox{-}CPA}_{\sde,\adve}
    (1^\secparam)$), the challenger samples $b\samp\bit$ and computes a single ciphertext $\ct\samp\sde.\enc(\pk,\mes_b)$, and sets $\ct_\bob=\ct_\charlie=\ct$.
    \item The challenger submits $\ct_\bob$ to $\bob$ and $\ct_\charlie$ to $\charlie$. Then, $\bob$ outputs $\widetilde{b}_\bob$ and $\charlie$ outputs $\widetilde{b}_\charlie$.
    \item The challenger outputs $1$ if and only if
    $\widetilde{b}_\bob=b_\bob$ and
    $\widetilde{b}_\charlie=b_\charlie$ where $b_\bob,b_\charlie$ are defined as $b_\bob=b=b_\charlie$ in the identical CPA security experiment.
\end{enumerate}

We say that the scheme $\sde$ satisfies \emph{independent-challenge CPA anti-piracy security} if,
for every QPT adversary $\adve$ in the independent CP A security experiment,
\begin{equation*}
    \Pr\left[\gam{SDE\mbox{-}IND\mbox{-}CPA}_{\sde,\adve}
    (1^\secparam)=1\right]
    \leq \frac{1}{2}+\negl(\secparam),
\end{equation*}

and similarly, we say that it satisfies \emph{identical-challenge CPA anti-piracy security} if,
for every QPT adversary $\adve$ in the identical CPA security experiment,
\begin{equation*}
    \Pr\left[\gam{SDE\mbox{-}IDEN\mbox{-}CPA}_{\sde,\adve}
    (1^\secparam)=1\right]
    \leq \frac{1}{2}+\negl(\secparam),
\end{equation*}
\end{definition}




\else

We recall the prior security definitions for SDE. Each security experiment is between a challenger and an adversary tuple $\adve=(\alice,\bob,\charlie)$. The adversary $\alice$ receives the quantum decryption key and splits it into two, possibly entangled, registers for $\bob$ and $\charlie$, who do not communicate after receiving their registers.

We begin with the CPA-style notions.
In the independent-challenge experiment, the adversary may choose a different message pair for each party (\cite{TCC:KitYam25}). The original definition of Coladangelo et al. (\cite{C:CLLZ21}) requires the two pairs to be equal.

\begin{definition}[Independent-Challenge CPA Anti-Piracy
\cite{C:CLLZ21,TCC:KitYam25}]
\label{def:sde-ind-cpa}
Let $\sde$ be an SDE scheme with message space $\messpa$, and let
$\adve=(\alice,\bob,\charlie)$ be an adversary tuple. Consider the following experiment.

\paragraph{\underline{$\gam{SDE\mbox{-}IND\mbox{-}CPA}_{\sde,\adve}
(1^\secparam)$}}
\begin{enumerate}
    \item The challenger samples
    $(\pk,\qsk)\gets\sde.\keygen(1^\secparam)$ and submits $(\pk,\qsk)$ to
    $\alice$.
    \item $\alice$ outputs message pairs
    $(\mes_{\bob,0},\mes_{\bob,1})\in\messpa^2$ and
    $(\mes_{\charlie,0},\mes_{\charlie,1})\in\messpa^2$, and a pair of
    registers $(\regbob,\regcharlie)$. The challenger submits $\regbob$ to
    $\bob$ and $\regcharlie$ to $\charlie$.
    \item The challenger samples $b_\bob,b_\charlie\samp\bit$ independently.
    In two independent executions of encryption, it computes
    $\ct_\bob\samp\sde.\enc(\pk,\mes_{\bob,b_\bob})$ and
    $\ct_\charlie\samp\sde.\enc(\pk,\mes_{\charlie,b_\charlie})$.
    \item The challenger submits $\ct_\bob$ to $\bob$ and $\ct_\charlie$ to
    $\charlie$. Then, $\bob$ outputs $\widetilde{b}_\bob$ and $\charlie$
    outputs $\widetilde{b}_\charlie$.
    \item The challenger outputs $1$ if and only if
    $\widetilde{b}_\bob=b_\bob$ and
    $\widetilde{b}_\charlie=b_\charlie$.
\end{enumerate}

We say that the scheme $\sde$ satisfies \emph{independent-challenge CPA anti-piracy security} if,
for every QPT adversary $\adve$,
\begin{equation*}
    \Pr\left[\gam{SDE\mbox{-}IND\mbox{-}CPA}_{\sde,\adve}
    (1^\secparam)=1\right]
    \leq \frac{1}{2}+\negl(\secparam).
\end{equation*}
\end{definition}

\begin{definition}[Identical-Challenge CPA Anti-Piracy
\cite{GZ20}]
\label{def:sde-iden-cpa}
Let $\sde$ be an SDE scheme with message space $\messpa$, and let
$\adve=(\alice,\bob,\charlie)$ be an adversary tuple. Consider the following experiment.

\paragraph{\underline{$\gam{SDE\mbox{-}IDEN\mbox{-}CPA}_{\sde,\adve}
(1^\secparam)$}}
\begin{enumerate}
    \item The challenger samples
    $(\pk,\qsk)\gets\sde.\keygen(1^\secparam)$ and submits $(\pk,\qsk)$ to
    $\alice$.
    \item $\alice$ outputs a message pair
    $(\mes_0,\mes_1)\in\messpa^2$ and a pair of registers
    $(\regbob,\regcharlie)$. The challenger submits $\regbob$ to $\bob$ and
    $\regcharlie$ to $\charlie$.
    \item The challenger samples $b\samp\bit$ and computes a single
    ciphertext $\ct\samp\sde.\enc(\pk,\mes_b)$.
    \item The challenger submits this same ciphertext $\ct$ to both $\bob$
    and $\charlie$. Then,
    $\bob$ outputs $\widetilde{b}_\bob$ and $\charlie$ outputs
    $\widetilde{b}_\charlie$.
    \item The challenger outputs $1$ if and only if
    $\widetilde{b}_\bob=b$ and $\widetilde{b}_\charlie=b$.
\end{enumerate}

We say that the scheme $\sde$ satisfies \emph{identical-challenge CPA anti-piracy security} if, for
every QPT adversary $\adve$,
\begin{equation*}
    \Pr\left[\gam{SDE\mbox{-}IDEN\mbox{-}CPA}_{\sde,\adve}
    (1^\secparam)=1\right]
    \leq \frac{1}{2}+\negl(\secparam).
\end{equation*}
\end{definition}

\fi

\ifconference

\else
    \input{defns/old-sde-defn}

\fi

\subsection{Correlated-Challenge Decision Security}\label{sec:sde-new-defn}
We first define the sampler that determines the correlation between the challenge bits.

\begin{definition}[Admissible Challenge Sampler]
\label{def:correlated-sde-sampler}
Fix an SDE scheme $\sde$. An admissible challenge sampler $\sampler$ is a QPT algorithm with classical output such that, for every public key $\pk$ that $\sde.\keygen(1^\secparam)$ can output: on input $(\pk,\sep)$ it returns $(b_\bob,b_\charlie)$ with uniform marginals, and on input $(\pk,\idmode)$ it returns one uniform bit $b$. Equivalently, for every $X\in\{\bob,\charlie\}$ and $\beta\in\bit$,
\begin{align}
    \Pr_{(b_\bob,b_\charlie)\gets\sampler(\pk,\sep)}
        [b_X=\beta]&=\frac{1}{2},
    &\Pr_{b\gets\sampler(\pk,\idmode)}[b=\beta]&=\frac{1}{2}.
\end{align}
\end{definition}
\begin{definition}[SDE-CORR Security]
\label{def:sde-corr-game}
\label{defn:sde-new-defn}
Let $\sde$ be an SDE scheme, $\sampler$ an admissible challenge sampler, and $\adve=(\alice,\bob,\charlie)$ a QPT adversary tuple. Consider the following security experiment. 

\paragraph{\underline{${\corrsde_{\sde,\adve,\sampler}}(1^\secparam)$}}
\begin{enumerate}
    \item The challenger samples $(\pk,\qsk)\gets \sde.\keygen(1^\secparam)$
    and submits $(\pk,\qsk)$ to $\alice$.
    \item $\alice$ outputs $1^N$, for $N\geq1$, a mode
    $\mathsf{mode}\in\{\sep,\idmode\}$, registers
    $(\regbob,\regcharlie)$, and challenge-message pairs:
    \begin{itemize}
        \item if $\mathsf{mode}=\sep$, it outputs $4N$ messages $\bigl(\mes_{\bob,i,0},\mes_{\bob,i,1},
          \mes_{\charlie,i,0},\mes_{\charlie,i,1}\bigr)_{i\in[N]}
          \in(\messpa^4)^N;$
        \item if $\mathsf{mode}=\idmode$, it outputs $2N$ messages $\bigl(\mes_{i,0},\mes_{i,1}\bigr)_{i\in[N]}\in(\messpa^2)^N.$
    \end{itemize}\label{it:unary-output-Alice}
    \item The challenger sends $\regbob$ to $\bob$ and $\regcharlie$ to $\charlie$, and creates the challenge batches as follows.
    \begin{enumerate}
        \item If $\mathsf{mode}=\sep$, sample
        $(b_\bob,b_\charlie)\gets\sampler(\pk,\sep)$. For
        every $X\in\{\bob,\charlie\}$ and $i\in[N]$, independently compute $\ct_{X,i}\gets\sde.\enc(\pk,\mes_{X,i,b_X}).$
        \item If $\mathsf{mode}=\idmode$, sample
        $b\gets\sampler(\pk,\idmode)$ and set
        $b_\bob=b_\charlie:=b$. For every $i\in[N]$, independently compute one ciphertext $\ct_i\gets\sde.\enc(\pk,\mes_{i,b}),$ and give that same ciphertext to both parties by setting $\ct_{\bob,i}=\ct_{\charlie,i}:=\ct_i$.
    \end{enumerate}
    \item The challenger sends
    $\overline{\ct}_X:=(\ct_{X,1},\ldots,\ct_{X,N})$ to each
    $X\in\{\bob,\charlie\}$.
    \item After receiving their batches, $\bob$ and $\charlie$ output
    $\widetilde b_\bob$ and $\widetilde b_\charlie$, respectively. The
    challenger outputs $1$ if and only if
    $\widetilde b_\bob=b_\bob$ and $\widetilde b_\charlie=b_\charlie$.
\end{enumerate}
 We say that the scheme $\sde$ is \emph{SDE-CORR secure} if, for every admissible $\sampler$ and every QPT $\adve$,
\[\Pr\left[
        \corrsde_{\sde,\adve,\sampler}(1^\secparam) = 1
    \right]
    \leq
    \frac{1}{2}+\negl(\secparam).\]
    We demand the output of $\alice$ in \Cref{it:unary-output-Alice} to be unary, in order to make $N$ polynomially bounded. Clearly, if $\alice$ outputs $N=1$, then we recover the single-ciphertext SDE-CORR security definition.
\end{definition}
\ifconference
    In \Cref{sec:sde-hier-defn}, we show that the SDE-CORR security notion implies all existing SDE security notions and illuminate a clear picture of the relationships among them. The main relationships that we prove can be summarized as follows. Every statement involving search security assumes superlogarithmic message length.\footnote{When we say definition A does not imply definition B, we mean that there exists a construction that satisfies A but not B, and by A implies B, we mean that every construction that satisfies A also satisfies B.}The separations below hold assuming post-quantum-secure $\io$ and one-way functions.
\begin{enumerate}
\item SDE-CORR security implies strong-CPA security, and its restriction to separate mode is equivalent to strong-CPA and CPA${}^+$ security (see \Cref{thm:sde-corr-implies-strong-cpa,thm:sde-sep-corr-iff-cpa-plus}). Single-ciphertext SDE-CORR, and hence SDE-CORR, security implies IND-CPA and IDEN-CPA security (see \Cref{lem:sde-corr-immediate}).
\item Strong-CPA security is strictly stronger than strong-search security (\cref{cor:sde-strong-cpa-implies-strong-search,thm:sde-strong-cpa-strictly-implies-strong-search}), and IDEN-CPA security does not imply CPA${}^+$ security (see \Cref{thm:sde-iden-cpa-not-cpa-plus}).
\item IND-CPA security implies none of IDEN-CPA, IND-SEARCH, IDEN-SEARCH, CPA${}^+$, or strong-CPA security (see \Cref{thm:sde-ind-cpa-not-strong-cpa}).
\item The latest revision of \cite[Section 11]{TCC:KitYam25} also introduces the notion of correlated-challenge-bits CPA security. They show that this notion is implied (see~\cite[Theorem 11.3]{TCC:KitYam25}) by identical-challenge-bit CPA security (where the challenge bits are identical but the challenge ciphertexts are otherwise generated separately and independently for $\bob$ and $\charlie$), which is a special case of SDE-CORR security. Consequently, SDE-CORR security implies this notion as well.
\end{enumerate}

\subsection{New Feasibilities}
We now state our main feasibility result for SDE.
    \begin{theorem}[Correlated-challenge-secure SDE]
    \label{thm:sde-security-restated}
    Assuming post-quantum secure indistinguishability obfuscation and post-quantum one-way functions, there exists an SDE scheme that is perfectly correct and satisfies correlated-challenge decision security (\cref{defn:sde-new-defn}). Consequently, it satisfies strong CPA and $\mathsf{CPA}^{+}$ anti-piracy security.
    \end{theorem}
    The theorem is restated as \Cref{thm:sde-security} and proven in \Cref{sec:sde-construction}.
    Since SDE-CORR security implies all existing SDE security notions, \Cref{thm:sde-security-restated} implies construction for all SDE notions, under the same assumptions as in \Cref{thm:sde-security-restated}.
\else
    \input{defns/sde-hierearchy}

    \clearpage
\fi

\ifconference
    \section{Copy-Protection of Puncturable Functionalities}\label{sec:cp-conf}
\else
    \section{Copy-Protection Definitions}

\label{sec:copy-protection-security-definitions}
\fi

\ifconference
    
\else
    In this section, we present the definitions of copy protection for general functionalities. 
\fi

We first recall the copy-protection syntax, and then state the standard unpredictability-based security notions. 

Let $\mathsf{Circ}=\{\mathsf{Circ}_\secparam\}_{\secparam\in\NN}$ be a class of polynomial-size circuits $C:\mathcal X_\secparam\rightarrow\mathcal Y_\secparam$, where $\mathcal X_\secparam=\zo^{\ell_{\mathsf{in}}(\secparam)}$ and $\mathcal Y_\secparam=\zo^{\ell_{\mathsf{out}}(\secparam)}$. Let $\mathcal D_{\mathsf{Circ}}$ be a QPT-samplable distribution over $\mathsf{Circ}_\secparam$.

\begin{definition}[Copy protection, adapted from~{\cite[Appendix~A.1]{AB24}} and~{\cite[Definition~3.1]{ABH+26}}]
A copy-protection scheme $\mathsf{CP}=(\mathsf{CopyProtect},\mathsf{Eval})$ for $\mathsf{Circ}$ consists of a pair of efficient algorithms as follows.
\begin{enumerate}
    \item $\rho_C\gets\mathsf{CopyProtect}(1^\secparam,C):$ is a QPT algorithm that takes as input a circuit $C\in\mathsf{Circ}_\secparam$ and outputs a quantum program $\rho_C$.
    \item $y\gets\mathsf{Eval}(\rho_C,x):$ is a QPT algorithm that takes as input a quantum program $\rho_C$ and an input $x\in\mathcal X_\secparam$, and outputs $y\in\mathcal Y_\secparam$.
\end{enumerate}
\paragraph{Correctness.} We require that there exists a negligible function $\negl(\cdot)$ such that, for every $\secparam\in\NN$, every $C\in\mathsf{Circ}_\secparam$, and every $x\in\mathcal X_\secparam$,
\begin{equation}
\Pr\left[
\mathsf{Eval}(\rho_C,x)=C(x):
\rho_C\gets\mathsf{CopyProtect}(1^\secparam,C)
\right]
\geq
1-\negl(\secparam).
\end{equation}
When polynomial reusability is required, the evaluation algorithm is understood to return its residual program state; the concrete scheme used below restores that state perfectly.
\end{definition}


\begin{definition}[Correlated copy-protection challenge sampler]
\label{def:correlated-copy-protection-sampler}
A correlated copy-protection challenge sampler $\mathsf{Samp}_{\mathsf{CP}}$ is a QPT algorithm that, on input $1^\secparam$ and a circuit $C$ from $\mathsf{Circ}_\secparam$, outputs $(x_\bob,x_\charlie)\gets\mathsf{Samp}_{\mathsf{CP}}(1^\secparam)$, where $x_\bob,x_\charlie\in\mathcal X_\secparam$. The two challenge points may be arbitrarily correlated. We say that $\mathsf{Samp}_{\mathsf{CP}}$ has $\kappa$-entropic marginals if $\Hmin(x_X)\geq\kappa(\secparam)$ for every $X\in\{\bob,\charlie\}$.
\end{definition}

We use the notion of trivial adversaries from previous works~\cite[Definition~3.1]{ABH+26}. The first-stage adversary $\alice^{\mathsf{triv}}_\bob$ sends the entire protected program $\rho_C$ to $\bob$ and the fixed dummy state $\ketbra{\bot}{\bot}$ to $\charlie$. The first-stage adversary $\alice^{\mathsf{triv}}_\charlie$ is defined symmetrically. Define
\begin{equation}
\mathsf{Triv}:=
\left\{(\alice^{\mathsf{triv}}_\bob,\bob,\charlie):
\bob,\charlie\text{ are QPT}\right\}\mathbin{\cup}
\left\{(\alice^{\mathsf{triv}}_\charlie,\bob,\charlie):
\bob,\charlie\text{ are QPT}\right\}.
\label{eq:copy-protection-trivial-adversaries}
\end{equation}
Thus, the recipient not receiving the protected program receives no auxiliary state prepared by $\alice$.

For any QPT adversarial triplet $\mathcal A_{\mathsf{CP}}=(\alice,\bob,\charlie)$, consider the following experiment.

\par\noindent
$\gam{\mathsf{UnpCP}_{\mathsf{CP},\mathcal A_{\mathsf{CP}},\mathcal D_{\mathsf{Circ}},\mathsf{Samp}_{\mathsf{CP}}}}(1^\secparam)$:
\begin{enumerate}
    \item $\ch$ samples $C\gets\mathcal D_{\mathsf{Circ}}(1^\secparam)$ and $\rho_C\gets\mathsf{CopyProtect}(1^\secparam,C)$, and send $\rho_C$ to $\alice$.
    \item $\alice$ outputs joint state $\sigma_{\regbob,\regcharlie}$ and sends register $\regbob$ to $\bob$ and register $\regcharlie$ to $\charlie$.
    \item After the split, $\ch$ samples $(x_\bob,x_\charlie)\gets\mathsf{Samp}_{\mathsf{CP}}(1^\secparam,C)$, gives $x_\bob$ to $\bob$, and gives $x_\charlie$ to $\charlie$.
    \item $\bob$ and $\charlie$ output $\widehat y_\bob,\widehat y_\charlie$.
    \item The output of the experiment is $1$ if $\widehat y_\bob=C(x_\bob)$ and $\widehat y_\charlie=C(x_\charlie)$.
\end{enumerate}

\begin{definition}[Unpredictability-style copy-protection security with correlated challenges, adapted from~{\cite[Appendix~A.1]{AB24}} and~{\cite[Definition~3.1]{ABH+26}}]
\label{def:unpredictability-copy-protection}
\label{def:non-oracular-unpredictability-copy-protection}
A copy-protection scheme $\mathsf{CP}$ satisfies unpredictability-style copy-protection security with respect to $(\mathcal D_{\mathsf{Circ}},\mathsf{Samp}_{\mathsf{CP}})$ if, for every QPT adversarial triplet $\mathcal A_{\mathsf{CP}}$, there exist a trivial adversary $\mathcal A_{\mathsf{triv}}\in\mathsf{Triv}$ and a negligible function $\negl(\cdot)$ such that
\begin{equation}
\Pr\left[
\gam{\mathsf{UnpCP}_{\mathsf{CP},\mathcal A_{\mathsf{CP}},\mathcal D_{\mathsf{Circ}},\mathsf{Samp}_{\mathsf{CP}}}}
(1^\secparam)=1
\right]
\leq
\Pr\left[
\gam{\mathsf{UnpCP}_{\mathsf{CP},\mathcal A_{\mathsf{triv}},\mathcal D_{\mathsf{Circ}},\mathsf{Samp}_{\mathsf{CP}}}}
(1^\secparam)=1
\right]
+\negl(\secparam).
\end{equation}
\end{definition}

\ifconference
    In \Cref{sec:other-prior-cp-def}, we describe the natural instantiations of the unpredictability-style copy-protection security definition for specialized classes of functions, as well as the pseudorandomness-style copy-protection security notions.
\subsection{New Feasibilities}
    We now present our main feasibility results for correlated-challenge secure copy protection. In the theorem below, all function and challenge distributions are QPT-samplable, and high min-entropy means min-entropy at least $\lambda^c$ for some constant $c>0$ and all sufficiently large $\lambda$.
    \begin{theorem}[Correlated-challenge Copy Protection]
    In all but the first three cases (concerning evasive functionalities), the circuit and challenges are sampled independently. Then, assuming post-quantum $\io$ and the quantum hardness of LWE, the following copy-protection guarantees hold.
    \begin{enumerate}
    \item \textbf{Point functions.} Identical-challenge copy-protection for point functions under the natural challenge distribution for any high-min-entropy distribution $\mathcal D$ of the marked point\footnote{i.e., the copy-protected point function is $P_y$, where $y\gets\mathcal D_\lambda$, and both $\bob$ and $\charlie$ receive the same challenges $x$ sampled as follows: with probability $1/2$, set $x=y$; and with the rest of the probability, sample $x\gets\mathcal D_\lambda$ independently of $y$.} (formally stated and proven as \Cref{cor:point-function-copy-protection}).
    \item \textbf{$k$-point functions.} Identical-challenge copy-protection for $k$-point functions with polynomially bounded $k$, provided the marked-set distribution and its induced uniformly sampled marked-point distribution both have high min-entropy, and the function distribution is preimage-samplable with respect to the latter (formally stated and proven as \Cref{cor:k-point-function-copy-protection}). The special case $k=1$ recovers the result for point functions as described in the previous item.
    \item \textbf{Compute-and-compare functions.} Identical-challenge copy-protection where the compute-and-compare function is sampled as a pair $(F,Y)$, where $F$ is a polynomial-size circuit and $Y$ is the target value, provided the target $Y$ conditioned on $F$ can be efficiently sampled and has high min-entropy, and the required preimages under $F$ can be efficiently sampled given $F,Y$ (formally stated and proven as \Cref{thm:compute-and-compare-copy-protection}).
    \item \textbf{Pseudorandomness-style security.} Copy-protection with respect to correlated challenges with high-min-entropy point marginals and each challenge bit uniform conditioned on both challenge points, for circuit classes that are two-point pseudorandomness-style puncturable (formally stated and proven as \Cref{cor:non-oracular-correlated-pr-copy-protection}).\label{it:pseudorandom}
    \item \textbf{Oracular security.} The guarantees in \Cref{it:pseudorandom} extend to oracular security defined in~\cite{ABH+26} when the challenge points are identical; for pseudorandomness-style security, the challenge bits must also be identical (formally stated and proven as 
    \Cref{cor:oracular-identical-high-entropy-pr-copy-protection}). 
    \end{enumerate}\amit{Removed unpredictability style result since they follow from prior results as well. This is better space-wise}
    For the special case of uniform challenge distributions as specified in the formal statements of the above items (\Cref{cor:oracular-idu-unpredictability-copy-protection,cor:oracular-idu-pr-copy-protection}), the post-quantum LWE assumption can be replaced by post-quantum one-way functions.
    \end{theorem}
\else
    \input{defns/old-cp-defn}

\fi

%% file: defns/old-sde-defn.tex
\ifconference
    \section{Other prior security definitions of SDE}\label{sec:other-prior-sde-def}
\fi

We next recall the two search-style notions.

\begin{definition}[Independent-Challenge Search Anti-Piracy
\cite{C:CLLZ21,TCC:KitYam25}]
\label{def:sde-ind-search}
Let $\sde$ be an SDE scheme with message space $\messpa$, and let
$\adve=(\alice,\bob,\charlie)$ be an adversary tuple. Consider the following experiment.

\paragraph{\underline{$\gam{SDE\mbox{-}IND\mbox{-}SEARCH}_{\sde,\adve}
(1^\secparam)$}}
\begin{enumerate}
    \item The challenger samples
    $(\pk,\qsk)\gets\sde.\keygen(1^\secparam)$ and submits $(\pk,\qsk)$ to
    $\alice$.
    \item $\alice$ outputs a pair of registers $(\regbob,\regcharlie)$. The
    challenger submits $\regbob$ to $\bob$ and $\regcharlie$ to $\charlie$.
    \item The challenger independently samples
    $\mes_\bob,\mes_\charlie\samp\messpa$ and, in two independent executions
    of encryption, computes
    $\ct_\bob\samp\sde.\enc(\pk,\mes_\bob)$ and
    $\ct_\charlie\samp\sde.\enc(\pk,\mes_\charlie)$.
    \item The challenger submits $\ct_\bob$ to $\bob$ and $\ct_\charlie$ to
    $\charlie$. Then, $\bob$ outputs $\widetilde{\mes}_\bob$ and $\charlie$
    outputs $\widetilde{\mes}_\charlie$.
    \item The challenger outputs $1$ if and only if
    $\widetilde{\mes}_\bob=\mes_\bob$ and
    $\widetilde{\mes}_\charlie=\mes_\charlie$.
\end{enumerate}

We say that the scheme $\sde$ satisfies \emph{independent-challenge search anti-piracy security} if,
for every QPT adversary $\adve$,
\begin{equation*}
    \Pr\left[\gam{SDE\mbox{-}IND\mbox{-}SEARCH}_{\sde,\adve}
    (1^\secparam)=1\right]
    \leq \frac{1}{|\messpa|}+\negl(\secparam).
\end{equation*}
\end{definition}

\begin{definition}[Identical-Challenge Search Anti-Piracy
\cite{TCC:KitYam25}]
\label{def:sde-iden-search}
Let $\sde$ be an SDE scheme with message space $\messpa$, and let
$\adve=(\alice,\bob,\charlie)$ be an adversary tuple. Consider the following experiment.

\paragraph{\underline{$\gam{SDE\mbox{-}IDEN\mbox{-}SEARCH}_{\sde,\adve}
(1^\secparam)$}}
\begin{enumerate}
    \item The challenger samples
    $(\pk,\qsk)\gets\sde.\keygen(1^\secparam)$ and submits $(\pk,\qsk)$ to
    $\alice$.
    \item $\alice$ outputs a pair of registers $(\regbob,\regcharlie)$. The
    challenger submits $\regbob$ to $\bob$ and $\regcharlie$ to $\charlie$.
    \item The challenger samples $\mes\samp\messpa$ and computes a single
    ciphertext $\ct\samp\sde.\enc(\pk,\mes)$.
    \item The challenger submits this same ciphertext $\ct$ to both $\bob$
    and $\charlie$. Then,
    $\bob$ outputs $\widetilde{\mes}_\bob$ and $\charlie$ outputs
    $\widetilde{\mes}_\charlie$.
    \item The challenger outputs $1$ if and only if
    $\widetilde{\mes}_\bob=\mes$ and
    $\widetilde{\mes}_\charlie=\mes$.
\end{enumerate}

We say that the scheme $\sde$ satisfies \emph{identical-challenge search anti-piracy security} if,
for every QPT adversary $\adve$,
\begin{equation*}
    \Pr\left[\gam{SDE\mbox{-}IDEN\mbox{-}SEARCH}_{\sde,\adve}
    (1^\secparam)=1\right]
    \leq \frac{1}{|\messpa|}+\negl(\secparam).
\end{equation*}
\end{definition}

CPA${}^+$ anti-piracy uses the independent-challenge CPA experiment with a
stronger winning condition.

\begin{definition}[CPA${}^+$ Anti-Piracy Security \cite{TCC:KitYam25}]
\label{def:sde-cpa-plus}
Let $\sde$ be an SDE scheme and let $\adve=(\alice,\bob,\charlie)$ be an adversary tuple. The experiment $\gam{SDE\mbox{-}CPA^{+}}_{\sde,\adve}(1^\secparam)$ is identical to the independent CPA experiment given in \ifconference \Cref{def:sde-cpa}, \else \cref{def:sde-ind-cpa}, \fi except that the challenger outputs $1$ if and only if
\begin{equation*}
    \widetilde{b}_\bob\oplus\widetilde{b}_\charlie
    =b_\bob\oplus b_\charlie.
\end{equation*}
We say that the scheme $\sde$ satisfies \emph{CPA${}^+$ anti-piracy security} if, for every QPT
adversary $\adve$,
\begin{equation*}
    \Pr\left[\gam{SDE\mbox{-}CPA^{+}}_{\sde,\adve}(1^\secparam)=1\right]
    \leq \frac{1}{2}+\negl(\secparam).
\end{equation*}
\end{definition}

We finally recall strong anti-piracy security
(\cite{C:CLLZ21,TCC:KitYam25}). A \emph{quantum decryptor} is
a pair $\mathsf{D}=(\rho,U)$, where $\rho$ is a quantum state and $U$ is a unitary circuit that takes a ciphertext and $\rho$ as input and produces a classical output. We define the two tests used in the strong definitions.

\paragraph{CPA test.}
For a public key $\pk$, a message pair $(\mes_0,\mes_1)$, and a quantum
decryptor $\mathsf{D}$, let
$\mathcal{P}^{\mathsf{cpa}}_{\mathsf{D}}[\pk;\mes_0,\mes_1]$ denote the binary-outcome POVM corresponding to the following test: sample
$b\samp\bit$, compute
$\ct\samp\sde.\enc(\pk,\mes_b)$, run $\mathsf{D}(\ct)$, and accept if and only
if the output is $b$. The $\gamma$-good CPA test applies
\begin{equation*}
    \TI_{1/2+\gamma}
    (\mathcal{P}^{\mathsf{cpa}}_{\mathsf{D}}[\pk;\mes_0,\mes_1])
\end{equation*}
to the quantum state of $\mathsf{D}$, where for any $\epsilon$ and binary outcome POVM $P$, $\TI_\epsilon(P)$ denotes the threshold implementation (with threshold $\epsilon$) of $P$, as defined in \Cref{sec:threshold-implementations}.

\paragraph{Search test.}
Fix a public key $\pk$ and a quantum decryptor $\mathsf{D}$. Let
$\mathcal{P}^{\mathsf{search}}_{\mathsf{D}}[\pk]$ denote the binary-outcome POVM corresponding to the following test. Sample $\mes\samp\messpa$, compute
$\ct\samp\sde.\enc(\pk,\mes)$, and run $\mathsf{D}(\ct)$. Accept if and only
if the output is $\mes$. The $\gamma$-good search test applies
\begin{equation*}
    \TI_{1/|\messpa|+\gamma}
    (\mathcal{P}^{\mathsf{search}}_{\mathsf{D}}[\pk])
\end{equation*}
to the quantum state of $\mathsf{D}$.

\begin{definition}[Strong CPA Anti-Piracy Security
\cite{C:CLLZ21,TCC:KitYam25}]
\label{def:sde-strong-cpa}
Let $\gamma:\NN\rightarrow[0,1]$, let $\sde$ be an SDE scheme with message
space $\messpa$, and let $\alice$ be an adversary. Consider the following experiment.

\paragraph{\underline{$\gam{SDE\mbox{-}STRONG\mbox{-}CPA}_{\sde,\alice}
(1^\secparam,\gamma)$}}
\begin{enumerate}
    \item The challenger samples
    $(\pk,\qsk)\gets\sde.\keygen(1^\secparam)$ and submits $(\pk,\qsk)$ to
    $\alice$.
    \item $\alice$ outputs message pairs
    $(\mes_{\bob,0},\mes_{\bob,1})\in\messpa^2$ and
    $(\mes_{\charlie,0},\mes_{\charlie,1})\in\messpa^2$, and two possibly
    entangled quantum decryptors
    $\mathsf{D}_\bob$ and $\mathsf{D}_\charlie$.
    \item The challenger applies the $\gamma$-good CPA test with respect to
    $(\pk,\mes_{X,0},\mes_{X,1})$ to $\mathsf{D}_X$, for each
    $X\in\{\bob,\charlie\}$. The challenger outputs $1$ if and only if both
    tests pass.
\end{enumerate}

We say that the scheme $\sde$ satisfies \emph{strong CPA anti-piracy security} if, for every inverse
polynomial $\gamma$ and every QPT adversary $\alice$,
\begin{equation*}
    \Pr\left[\gam{SDE\mbox{-}STRONG\mbox{-}CPA}_{\sde,\alice}
    (1^\secparam,\gamma(\secparam))=1\right]
    \leq\negl(\secparam).
\end{equation*}
\end{definition}

\begin{definition}[Strong Search Anti-Piracy Security \cite{C:CLLZ21}]
\label{def:sde-strong-search}
Let $\gamma:\NN\rightarrow[0,1]$, let $\sde$ be an SDE scheme with message space $\messpa$, and let $\alice$ be an adversary. Consider the following experiment.

\paragraph{\underline{$\gam{SDE\mbox{-}STRONG\mbox{-}SEARCH}_{\sde,\alice}
(1^\secparam,\gamma)$}}
\begin{enumerate}
    \item The challenger samples
    $(\pk,\qsk)\gets\sde.\keygen(1^\secparam)$ and submits $(\pk,\qsk)$ to
    $\alice$.
    \item $\alice$ outputs two possibly entangled quantum decryptors
    $\mathsf{D}_\bob$ and $\mathsf{D}_\charlie$.
    \item The challenger applies the $\gamma$-good search test with respect
    to $\pk$ to each decryptor. The challenger outputs $1$ if and only if
    both tests pass.
\end{enumerate}

We say that the scheme $\sde$ satisfies \emph{strong search anti-piracy security} if, for every
inverse polynomial $\gamma$ and every QPT adversary $\alice$,
\begin{equation*}
    \Pr\left[\gam{SDE\mbox{-}STRONG\mbox{-}SEARCH}_{\sde,\alice}
    (1^\secparam,\gamma(\secparam))=1\right]
    \leq\negl(\secparam).
\end{equation*}
\end{definition}

%% file: defns/sde-hierearchy.tex
\subsection{Hierarchy of SDE Definitions}\label{sec:sde-hier-defn}
All statements below concerning search security assume bit-string messages of length $\meslen(\secparam)=\omega(\log\secparam)$. Equivalently, $1/|\messpa|$ is negligible, so uniform guessing has only negligible success probability in the search games.

\Cref{fig:sde-hierarchy} summarizes the implications and separations proved or recalled below. For notions $\mathsf X$ and $\mathsf Y$, the notation $\mathsf X\nRightarrow\mathsf Y$ means that some SDE scheme satisfies $\mathsf X$ but not $\mathsf Y$.

\input{defns/sde-hierarchy-figure}
\ifconference

\else
    \clearpage
\fi

We first recall the relations between the prior definitions and show that strong CPA is strictly stronger than strong search. We then prove that SDE-CORR implies strong CPA. This last implication requires multiple ciphertexts: independently generated ciphertexts let each party estimate a quantum test, and an approximate threshold measurement converts that estimate into a guess of the common challenge bit.

\paragraph{A perturbation bound.}
We will use the following lemma when recovering the challenge bit in \cref{thm:sde-corr-implies-strong-cpa}. It says that when two Hermitian operators are close, a state supported where one has large eigenvalues has little weight where the other has sufficiently small eigenvalues.

\begin{lemma}[Spectral perturbation]\label{lem:spectral-cross-projection}
Let $A,B$ be Hermitian operators on the same finite-dimensional Hilbert space, let $\kappa>\ell$, and set $\Pi:=\mathbf 1[A\geq\kappa]$ and $Q:=\mathbf 1[B\leq\ell]$. Here $\mathbf 1[A\geq\kappa]$ denotes the orthogonal projector onto the subspace spanned by the eigenvectors of $A$ with eigenvalues at least $\kappa$. Similarly, $\mathbf 1[B\leq\ell]$ projects onto the subspace spanned by the eigenvectors of $B$ with eigenvalues at most $\ell$. Then
\[
 \|Q\Pi\|\leq\frac{\|A-B\|}{\kappa-\ell}.
\]
The same bound holds after tensoring $\Pi$ and $Q$ with the identity on an arbitrary additional register.
\end{lemma}

\begin{proof}
Let $s:=\|Q\Pi\|$. If $s=0$, the claim is immediate. Otherwise, since $\Pi$ and $Q$ are orthogonal projectors,
\[
 \Pi Q\Pi=(Q\Pi)^\dagger(Q\Pi)\succeq0,
 \qquad
 \|\Pi Q\Pi\|=\|Q\Pi\|^2=s^2.
\]
A positive semidefinite operator on a finite-dimensional space has a largest eigenvalue equal to its operator norm. Thus $\Pi Q\Pi$ has a unit eigenvector $v$ with eigenvalue $s^2$. Since $s>0$ and $\Pi Q\Pi$ maps into the range of $\Pi$, we have $\Pi v=v$. Set $u:=Qv/s$. Then $Qu=u$, $\|u\|=1$, and $\Pi u=sv$.

Since $A$ commutes with $\Pi$ and $B$ commutes with $Q$,
\[
 \langle u,(A-B)v\rangle
 =s\bigl(\langle v,Av\rangle-\langle u,Bu\rangle\bigr)
 \geq s(\kappa-\ell).
\]
The inequality uses $v\in\operatorname{Ran}(\Pi)$ and $u\in\operatorname{Ran}(Q)$. Since $u$ and $v$ are unit vectors, the left-hand side is at most $\|A-B\|$. Dividing by $\kappa-\ell$ proves the bound. Tensoring with the identity preserves the operator norm, giving the assertion for an additional register.
\end{proof}

We also use the following matrix concentration bound to control the average of the operators associated with independently generated ciphertexts.

\begin{lemma}[Matrix concentration~{\cite[Theorem~1.6]{Tropp12}}]
\label{lem:matrix-concentration}
Let $Y_1,\ldots,Y_N$ be Hermitian random matrices of dimension $d$, each sampled independently from some (possibly different) distribution. Suppose that, for every $i$, we have $\E[Y_i]=0$, $\E[Y_i^2]\preceq I/4$, and $\|Y_i\|\leq1$ for every possible value of $Y_i$. Then, for every $0<t\leq1/2$,
\[
 \Pr\!\left[\left\|\frac{1}{N}\sum_{i=1}^N Y_i\right\|\geq t\right]
 \leq 2d\exp(-Nt^2).
\]
\end{lemma}

The two ordinary CPA notions already follow from single-ciphertext SDE-CORR security.

\begin{lemma}[Single-ciphertext SDE-CORR implies CPA security]\label{lem:sde-corr-immediate}
Single-ciphertext SDE-CORR security implies $\mathsf{SDE\mbox{-}IND\mbox{-}CPA}$ and $\mathsf{SDE\mbox{-}IDEN\mbox{-}CPA}$ security. Consequently, both notions also follow from SDE-CORR security.
\end{lemma}

\begin{proof}
Consider the single-ciphertext, or $N=1$ version of the SDE-CORR security game of \Cref{def:sde-corr-game}. Separate mode with independent uniform bits $b_\bob,b_\charlie$ gives exactly the independent-challenge CPA experiment of \ifconference \Cref{def:sde-cpa} \else \Cref{def:sde-ind-cpa}. \fi  The identical mode gives the identical-challenge CPA experiment of \ifconference \Cref{def:sde-cpa} \else \Cref{def:sde-iden-cpa}. \fi: the sampler returns one uniform bit, and the challenger copies the resulting ciphertext to both parties. Finally, the general SDE-CORR game includes this single-ciphertext game as a special case, so SDE-CORR security implies single-ciphertext SDE-CORR security.
\end{proof}

We also use the following relations from prior work.

\begin{theorem}[Prior relations \cite{TCC:KitYam25,C:AnaKalLiu23}]
\label{thm:sde-prior-relations}
The following relations hold.
\begin{enumerate}
    \item $\mathsf{SDE\mbox{-}CPA}^{+}$ and $\mathsf{SDE\mbox{-}STRONG\mbox{-}CPA}$ security are equivalent.
    \item $\mathsf{SDE\mbox{-}CPA}^{+}$ security implies $\mathsf{SDE\mbox{-}IND\mbox{-}CPA}$ and $\mathsf{SDE\mbox{-}IND\mbox{-}SEARCH}$ security.
    \item $\mathsf{SDE\mbox{-}STRONG\mbox{-}SEARCH}$ and $\mathsf{SDE\mbox{-}IND\mbox{-}SEARCH}$ security are equivalent.
    \item $\mathsf{SDE\mbox{-}IND\mbox{-}SEARCH}$ security implies $\mathsf{SDE\mbox{-}IDEN\mbox{-}SEARCH}$ security.
\end{enumerate}
\end{theorem}

\begin{proof}
The first three statements follow from \cite[Theorems~6.20, 6.22, 6.23 and Remark~6.19]{TCC:KitYam25}. The last is the independent-to-identical search implication of \cite[Section~6.1]{TCC:KitYam25}, which uses the product-to-correlated comparison of Ananth, Kaleoglu, and Liu~\cite{C:AnaKalLiu23}.
\end{proof}

The implication from strong CPA to strong search follows directly from \cref{thm:sde-prior-relations}.

\begin{corollary}[Strong CPA implies strong search]
\label{cor:sde-strong-cpa-implies-strong-search}
For superlogarithmic message length, strong-CPA security implies strong-search security.
\end{corollary}

\begin{proof}
By \cref{thm:sde-prior-relations}, strong CPA is equivalent to CPA${}^+$, which implies IND-SEARCH, which in turn is equivalent to strong search.
\end{proof}

The next separation, in the converse direction, follows the usual search-versus-CPA separation: reveal the first message bit and encrypt the rest. Recovering the whole message remains hard, but distinguishing messages that differ only in the first bit is easy. However, to instantiate such a separation in the plain model, we need a strong-CPA secure SDE scheme, for which we use our SDE construction, proven in \Cref{thm:sde-security}.

\begin{theorem}[Strong search does not imply strong CPA]
\label{thm:sde-strong-cpa-strictly-implies-strong-search}
Assuming post-quantum-secure indistinguishability obfuscation (\cref{predef:io}) and post-quantum-secure one-way functions, there exists a strong-search-secure SDE scheme that is not strong-CPA secure.
\end{theorem}

\begin{proof}
Under the stated assumptions, \cref{thm:sde-security} provides an SDE-CORR-secure scheme. It is strong-CPA secure by \cref{thm:sde-corr-implies-strong-cpa}, and hence strong-search secure by \cref{cor:sde-strong-cpa-implies-strong-search}. We use this scheme as $\Pi_0$, with message space $\mathcal M_0$ of superlogarithmic message length. More generally, the construction below works for any strong-search-secure $\Pi_0$ with $1/|\mathcal M_0|=\negl(\lambda)$.

The new scheme $\Pi_{\mathsf{clr}}$ has message space $\bit\times\mathcal M_0$, uses the same key generation as $\Pi_0$, and encrypts $(a,m)$ as $(a,\Pi_0.\enc(\pk,m))$. Decryption keeps the first component and decrypts the second using $\Pi_0$.

To prove strong-search security, fix a public key and a decryptor $\mathsf D'$ for $\Pi_{\mathsf{clr}}$. Construct a decryptor $\mathsf D$ for $\Pi_0$ as follows. On input $\ct$, it generates a fresh uniform bit $a$ and runs $\mathsf D'(a,\ct)$. If the output is $(a,m')$, it returns $m'$; otherwise, it returns $\bot\notin\mathcal M_0$. The circuit prepares the random bit and its work registers internally and can run reversibly, with measurements deferred until the final output.

Let $P'$ and $P$ be the acceptance operators of the ordinary search tests for $\mathsf D'$ and $\mathsf D$, respectively. These tests choose a uniform message, encrypt it, and accept exactly when the decryptor recovers that message. Both operators act on the same quantum input register: for any state $\rho$ in this register, $\Tr(P'\rho)$ and $\Tr(P\rho)$ are the corresponding average success probabilities. The strong-search test of \cref{def:sde-strong-search} uses the spectral projector of this operator above the prescribed threshold.

The two ordinary search tests sample the same distribution of $(a,m,\ct)$. The test for $\mathsf D'$ chooses a uniform message $(a,m)$ and encrypts it as $(a,\ct)$. The test for $\mathsf D$ chooses a uniform $m$, encrypts it as $\ct$, and generates the independent uniform bit $a$ inside the decryptor. In either case, the test accepts exactly when $\mathsf D'(a,\ct)$ outputs $(a,m)$. Thus $\Tr(P\rho)=\Tr(P'\rho)$ for every input state $\rho$, which proves $P=P'$.

Although the acceptance operators are equal, the strong-search thresholds differ because they are measured relative to uniform guessing. Write $K:=|\mathcal M_0|$. The original message space has size $K$, so its threshold with gap $\gamma/2$ is $1/K+\gamma/2$. The new message space has size $2K$, so its threshold with gap $\gamma$ is $1/(2K)+\gamma$. For every inverse polynomial $\gamma$, negligibility of $1/K$ gives $1/K\leq\gamma$ for all sufficiently large $\lambda$. Hence
\[
 \left(\frac{1}{2K}+\gamma\right)-\left(\frac{1}{K}+\frac{\gamma}{2}\right)
 =\frac{\gamma-1/K}{2}\geq0.
\]
Thus every eigenvalue that meets the new scheme's threshold also meets the original scheme's threshold with gap $\gamma/2$. Equivalently,
\[
 \mathbf 1[P'\geq1/(2K)+\gamma]
 \preceq \mathbf 1[P\geq1/K+\gamma/2].
\]
Apply the conversion separately to Bob's and Charlie's decryptors, keeping their joint input state. Each accepting projector can only increase, so their tensor product can only increase as well. Therefore the probability that both threshold tests accept cannot decrease, even when the input registers are entangled. Any strong-search attack on $\Pi_{\mathsf{clr}}$ with gap $\gamma$ would give an attack on $\Pi_0$ with the still inverse-polynomial gap $\gamma/2$, a contradiction.

To break strong CPA, sample $m^\star\getsr\mathcal M_0$. Give both receivers the pair $((0,m^\star),(1,m^\star))$. Each decryptor reads the first ciphertext component and recovers the challenge bit perfectly without using the quantum key. Both CPA acceptance operators are therefore the identity. In particular, both threshold tests accept with probability one when $\gamma(\lambda)=1/\lambda$ and $\lambda$ is sufficiently large. Thus, $\Pi_{\mathsf{clr}}$ is not strong-CPA secure.
\end{proof}

We now prove the main implication: SDE-CORR implies strong CPA. Receiving multiple ciphertexts allows each receiver to infer the common challenge bit from a test whose average behavior is known to favor the correct answer.

\begin{theorem}[Correlated challenges decision security imply strong CPA]
\label{thm:sde-corr-implies-strong-cpa}
Every $\mathsf{SDE\mbox{-}CORR}$-secure scheme satisfies $\mathsf{SDE\mbox{-}STRONG\mbox{-}CPA}$ security.
\end{theorem}

\begin{proof}
Suppose a QPT adversary $\alice$ violates \cref{def:sde-strong-cpa}: for inverse polynomials $\gamma,\varepsilon$, both strong-CPA threshold tests accept with probability at least $\varepsilon$ on an infinite set of security parameters. We may assume $\gamma\leq1/2$, since otherwise the threshold exceeds one. We construct an SDE-CORR attack as follows. First, two CPA measurements on each register prepare states close to ones that pass the converted decryptors' CPA threshold tests at $1/2+\gamma/4$. Then multiple ciphertexts make a test close to its average for each challenge bit, allowing both receivers to recover that bit; if preparation rejects, they share a random guess.

The reduction's first algorithm $\alice_{\mathsf{corr}}$ runs $\alice$ and always chooses $\mathsf{mode}=\sep$. It outputs $1^N$, with $N$ chosen in Stage~3, and $N$ repetitions of each receiver's message pair: $\mes_{X,i,j}=\mes_{X,j}$. The two receivers' pairs may differ. Fix the admissible sampler that returns $(b,b)$ in separate mode and $b$ in identical mode, for uniform $b$. All ciphertexts use independent encryption randomness. Before the challenge, $\alice_{\mathsf{corr}}$ performs the four measurements in Stage~1. Each receiver gets its state, public key, message pair, decryptor circuit, all four outcomes, and a shared uniform bit $g$. If all measurements accept, the receivers follow Stages~2--4; otherwise, both output $g$.

To bound measurement errors, use
\begin{equation}
 \eta:=\left(\frac{\varepsilon\gamma^2}{10^6}\right)^2
 \label{eq:sde-corr-ati-error}
\end{equation}
in every application of \cref{lem:gapped-threshold-instrument}. We suppress the public-key, message-pair, and circuit registers; operations are coherent or conditioned on them, and we average over them at the end.

\paragraph{Stage 1: Prepare states above a CPA threshold.}
We convert the decryptors to binary outputs while keeping the original strong-CPA test. A test's acceptance operator $M$ gives acceptance probability $\Tr(M\rho)$ on state $\rho$. For $X\in\{\bob,\charlie\}$, define $\widehat{\mathsf D}_X$ from the original decryptor $\mathsf D_X$ by replacing every output outside $\bit$ with a fresh uniform bit. Let $P_X$ and $F_X$ be the original and converted CPA acceptance operators, respectively. If $R_X$ is the original CPA test's average operator for nonbinary outputs, then $F_X=P_X+R_X/2\succeq P_X$. Both tests use fresh local encryptions of known message choices, so preparation also works with a single challenge ciphertext.

To express the converted decryptor's CPA advantage, let $D_{X,y}(\ct)$ be its output-$y$ operator and set $Z_X(\ct):=D_{X,0}(\ct)-D_{X,1}(\ct)$. This records the difference between its output probabilities. Define
\[
    Z_{X,j}:=
    \E_{\ct\gets\sde.\enc(\pk,\mes_{X,j})}[Z_X(\ct)]
    \qquad (j\in\bit).
\]
Since $D_{X,0}(\ct)+D_{X,1}(\ct)=I$,
\begin{equation}
    H_X:=\frac{Z_{X,0}-Z_{X,1}}{4}
    =F_X-\frac{I}{2}
    \succeq P_X-\frac{I}{2},
    \label{eq:sde-corr-cpa-effect}
\end{equation}
with $\|H_X\|\leq1/2$. Thus, $H_X$ measures CPA advantage above $1/2$.


The first measurements preserve the original strong-CPA acceptance probability. For the joint decryptor state $\rho$, let
\[
 \Pi_X^+:=\mathbf 1[P_X\geq1/2+\gamma]
 \quad\text{and}\quad
 p_{\mathrm{str}}:=
 \Tr[(\Pi_\bob^+\otimes\Pi_\charlie^+)\rho],
\]
so $\E[p_{\mathrm{str}}]\geq\varepsilon$. Apply \cref{lem:gapped-threshold-instrument} to each $P_X$, with completeness threshold $1/2+\gamma$ and soundness threshold $1/2+3\gamma/4$. Let $\mathcal A_X^P$ be the accepting map and
\[
 \sigma_P:=(\mathcal A_\bob^P\otimes\mathcal A_\charlie^P)(\rho),
 \qquad p_P:=\Tr(\sigma_P).
\]
The trace of this unnormalized state is the probability that both measurements accept. Completeness holds with an entangled reference register, so successive applications give
\begin{equation}
 p_P\geq p_{\mathrm{str}}-2\eta.
 \label{eq:sde-corr-first-filter-mass}
\end{equation}

To analyze the next measurements, project onto the subspaces where the first tests guarantee large eigenvalues. Set $\Pi_X:=\mathbf 1[P_X>1/2+3\gamma/4]$, $\Pi:=\Pi_\bob\otimes\Pi_\charlie$, and $\overline\sigma_P:=\Pi\sigma_P\Pi$. Each local measurement leaves at most $\eta$ weight outside $\Pi_X$; the other party's accepting map cannot increase this weight. The output guarantee and gentle measurement bound from \cref{sec:gentle-measurement} give
\[
 \Tr(\overline\sigma_P)\geq p_P-2\eta,
 \qquad
 \|\sigma_P-\overline\sigma_P\|_1\leq3\sqrt{2\eta}.
\]
The attack retains $\sigma_P$; the projection is only for analysis.

The second measurements select states where $H_X$ is large. Although $F_X\succeq P_X$, their high-eigenvalue subspaces need not be nested. We therefore bound the weight above the second threshold. Set $\Omega_X:=\mathbf 1[H_X\geq\gamma/2]$. By \cref{eq:sde-corr-cpa-effect}, $\Pi_XH_X\Pi_X\succeq(3\gamma/4)\Pi_X$, while $\|H_X\|\leq1/2$ gives
\[
 \Pi_XH_X\Pi_X
 \preceq \frac{\gamma}{2}\Pi_X+
          \left(\frac{1}{2}-\frac{\gamma}{2}\right)
          \Pi_X\Omega_X\Pi_X.
\]
Combining these bounds, with $c:=\gamma/(2(1-\gamma))\geq\gamma/2$, yields
\begin{equation}
 \Pi_X\Omega_X\Pi_X\succeq c\Pi_X.
 \label{eq:sde-corr-overlap-compression}
\end{equation}

We use this weight bound to show that all four measurements accept with non-negligible probability. Apply the second measurements to $F_X=I/2+H_X$, with completeness threshold $1/2+\gamma/2$ and soundness threshold $1/2+\gamma/4$. For their accepting maps $\mathcal A_X^H$, set
\[
 \sigma_H:=(\mathcal A_\bob^H\otimes\mathcal A_\charlie^H)(\sigma_P),
 \qquad p_H:=\Tr(\sigma_H).
\]
Here $p_H$ is the probability that all four measurements accept. Tensoring \cref{eq:sde-corr-overlap-compression} gives
\[
 (\Pi_\bob\Omega_\bob\Pi_\bob)\otimes
 (\Pi_\charlie\Omega_\charlie\Pi_\charlie)
 \succeq c^2\Pi,
\]
also for entangled inputs. By completeness,
\[
 p_H\geq
 \Tr[(\Omega_\bob\otimes\Omega_\charlie)\sigma_P]-2\eta
 \geq c^2\Tr(\overline\sigma_P)-2\eta-\|\sigma_P-\overline\sigma_P\|_1.
\]
The bounds on $\overline\sigma_P$ and \cref{eq:sde-corr-first-filter-mass} imply
\begin{equation}
 p_H\geq\frac{\gamma^2}{4}p_{\mathrm{str}}-6\eta-3\sqrt{2\eta},
 \label{eq:sde-corr-second-filter-mass}
\end{equation}
so $\E[p_H]=\Omega(\varepsilon\gamma^2)$ by our choice of $\eta$.

For bit recovery, we again compare the retained state with a projected state. Let $\Theta_X:=\mathbf 1[H_X>\gamma/4]$, $\Theta:=\Theta_\bob\otimes\Theta_\charlie$, and $\overline\sigma_H:=\Theta\sigma_H\Theta$. The same output and gentle measurement bounds give
\[
 \Tr(\overline\sigma_H)\geq p_H-2\eta,
 \qquad
 \|\sigma_H-\overline\sigma_H\|_1\leq3\sqrt{2\eta}.
\]

\paragraph{Stage 2: Test the common bit.}
We construct a test whose average depends on the bit through $(-1)^bH_X$. For the received ciphertexts $\overline{\ct}_X=(\ct_{X,1},\ldots,\ct_{X,N})$, define
\begin{equation}
    \widehat H_{X,b}:=
    \frac{1}{2}\left(
        \frac{1}{N}\sum_{i=1}^N Z_X(\ct_{X,i})
        -\frac{Z_{X,0}+Z_{X,1}}{2}
    \right).
    \label{eq:sde-corr-empirical-observable}
\end{equation}
Subtracting the second term removes any fixed preference for one output bit, and $\E[\widehat H_{X,b}]=(-1)^bH_X$.

To implement this comparison without knowing $b$, use the binary test with acceptance operator
\begin{equation}
 G_X(\overline\ct_X):=\frac{I}{2}+\frac{\widehat H_{X,b}}{2}.
 \label{eq:sde-corr-batch-effect}
\end{equation}
With probability $1/2$, run $\widehat{\mathsf D}_X$ on a uniform received ciphertext and accept exactly on output $0$. Otherwise, encrypt $\mes_{X,j}$ for fresh uniform $j\getsr\bit$, run $\widehat{\mathsf D}_X$, and accept exactly on output $1$. Writing $S_{X,b}:=N^{-1}\sum_iZ_X(\ct_{X,i})$ and $\overline Z_X:=(Z_{X,0}+Z_{X,1})/2$, the two branches have acceptance operators $(I+S_{X,b})/2$ and $(I-\overline Z_X)/2$, whose average is $G_X$.

\paragraph{Stage 3: Choose enough ciphertexts.}
Multiple challenge ciphertexts are needed here to make $\widehat H_{X,b}$ close to its average on every input state. Let $q_X$ bound the qubits used by $\mathsf D_X$ and $\widehat{\mathsf D}_X$, including workspace. Condition on the public key, message pair, circuit, and $b$. The matrices $Y_i:=(Z_X(\ct_{X,i})-\E[Z_X(\ct_{X,i})])/2$ are independent, have mean zero, satisfy $\|Y_i\|\leq1$ and $\E[Y_i^2]\preceq I/4$, and average to $\widehat H_{X,b}-(-1)^bH_X$. Apply \cref{lem:matrix-concentration} with $d\leq2^{q_X}$ and $t=\gamma/100$, weakening the constant to obtain
\begin{equation}
 \Pr\!\left[
   \|\widehat H_{X,b}-(-1)^bH_X\|>\frac{\gamma}{100}
 \right]
 \leq 2^{q_X+1}\exp\!\left(-\frac{N\gamma^2}{25000}\right).
 \label{eq:sde-corr-matrix-concentration}
\end{equation}
Choose the same power-of-two ciphertext count for both parties:
\begin{equation}
 N:=2^{\left\lceil\log_2\!\left(
 10^5(\max_Xq_X+\secparam)/\gamma^2\right)\right\rceil}.
 \label{eq:sde-corr-batch-size}
\end{equation}
The norm bound then holds for both parties except with negligible probability.

All measurements are efficient: $P_X,F_X,G_X$ have reversible implementations using fresh randomness, public encryption, and the decryptor circuits. Since $q_X$ is polynomial and $\eta$ is inverse polynomial, \cref{lem:gapped-threshold-instrument} gives fixed polynomial worst-case running time.

\paragraph{Stage 4: Recover the bit.}
We bound the probability that the test from Stage~2 gives the wrong sign. Condition on the norm bound from Stage~3 for both parties. Recall that $\overline\sigma_H$ is supported where $H_X>\gamma/4$. Set
\[
 \Lambda_X:=\mathbf 1[H_X\geq\gamma/4],
 \qquad
 Q_{X,b}:=\mathbf 1[(-1)^b\widehat H_{X,b}\leq\gamma/8].
\]
By \cref{lem:spectral-cross-projection},
\[
 \|Q_{X,b}\Lambda_X\|
 \leq\frac{\gamma/100}{\gamma/4-\gamma/8}
 =\frac{8}{100}.
\]
Since $\Theta_X\preceq\Lambda_X$, treating the other receiver's register as a reference gives
\begin{equation}
 \Tr[(Q_{X,b}\otimes I)\overline\sigma_H]
 \leq\|Q_{X,b}\Lambda_X\|^2\Tr(\overline\sigma_H)
 \leq0.0064\Tr(\overline\sigma_H)
 \leq\frac{1}{100}\Tr(\overline\sigma_H).
 \label{eq:sde-corr-sign-error}
\end{equation}

To recover $b$ from this sign, each receiver applies
\begin{equation}
 \mathsf{GATI}^{\,1/2+\gamma/16,\,1/2-\gamma/16}_{\eta}
       \bigl(G_X(\overline\ct_X),I-G_X(\overline\ct_X)\bigr)
 \label{eq:sde-corr-final-ati}
\end{equation}
and guesses $0$ on acceptance and $1$ on rejection. Outside $\operatorname{Ran}(Q_{X,b})$, the eigenvalues of $G_X$ exceed $1/2+\gamma/16$ when $b=0$ and are below $1/2-\gamma/16$ when $b=1$. By \cref{lem:gapped-threshold-instrument,eq:sde-corr-sign-error}, each receiver's error probability on $\overline\sigma_H$ is at most $\Tr(\overline\sigma_H)/100+\eta$. These bounds hold for entangled states, so a union bound gives joint correctness at least
\begin{equation}
 \left(1-\frac{2}{100}\right)\Tr(\overline\sigma_H)-2\eta.
 \label{eq:sde-corr-final-ati-success}
\end{equation}
Returning to the retained state $\sigma_H$ costs at most $\|\sigma_H-\overline\sigma_H\|_1$; averaging over ciphertexts costs the negligible failure probability from Stage~3.

\paragraph{Stage 5: Include preparation failures.}
We finish by including the shared random guess when preparation rejects. Let $s_H$ be the probability that all four measurements accept and both guesses are correct. By \cref{eq:sde-corr-final-ati-success} and the bounds on $\overline\sigma_H$, weakening the constant gives
\[
 s_H\geq\frac{4}{5}p_H-4\eta-3\sqrt{2\eta}
             -\negl(\secparam).
\]
The shared random guess succeeds with probability exactly $1/2$, so
\begin{align*}
 \Pr[\mathsf{SDE\mbox{-}CORR}\text{ wins}]
 &=\frac{1-p_H}{2}+s_H\\
 &\geq\frac{1}{2}+\frac{3}{10}p_H-4\eta-3\sqrt{2\eta}
             -\negl(\secparam).
\end{align*}
Averaging \cref{eq:sde-corr-second-filter-mass} over the classical registers, the advantage is at least
\[
 \frac{3}{40}\varepsilon\gamma^2-6\eta-6\sqrt{2\eta}
       -\negl(\secparam),
\]
where we enlarged the error terms. By \cref{eq:sde-corr-ati-error}, this is at least $\varepsilon\gamma^2/20$ eventually on the infinite set, contradicting SDE-CORR security.
\end{proof}

The proof idea above is quite general and even extends to showing an implication from SDE-CORR security to a two-sided generalization of strong CPA. The two-sided version of strong CPA measures the tendency of each party in $\{\bob,\charlie\}$ to output either the correct challenge bit or its opposite with non-negligible advantage. For the following remark, we restrict the adversarial decryptors of $\bob$ and $\charlie$ to always output in $\bit$.

Fix the public key, message pairs, and decryptor circuits. For each party $X\in\{\bob,\charlie\}$, let $P_X$ be its CPA acceptance operator, and set $H_X:=P_X-I/2$. For $\gamma>0$, define
\[
 T_\gamma:=\mathbf 1[|H_\bob|\geq\gamma]\otimes
 \mathbf 1[|H_\charlie|\geq\gamma].
\]
This test accepts when both local CPA acceptance eigenvalues are at least $\gamma$ away from $1/2$, whether above or below it. For comparison, let $W_+$ be the acceptance operator of the CPA${}^+$ test. With independent uniform challenge bits, it accepts when both guesses are correct or both are incorrect. Since both decryptors always output a bit,
\[
 W_+=P_\bob\otimes P_\charlie+(I-P_\bob)\otimes(I-P_\charlie)
 =\frac{I}{2}+2H_\bob\otimes H_\charlie.
\]

\begin{remark}[A stronger form of strong CPA]\label{rem}
For every SDE-CORR-secure scheme, every QPT adversary whose two decryptors always output bits, and every inverse polynomial $\gamma$, the test with acceptance projector $T_\gamma$ accepts with negligible probability.
\end{remark}

\begin{proof}
Run an adversary for the test $T_\gamma$, then choose independent uniform signs $s_\bob,s_\charlie\in\{+1,-1\}$. For each party $X$, keep its decryptor's output when $s_X=1$, and reverse that output bit when $s_X=-1$. The resulting CPA acceptance operator is $P_X$ or $I-P_X$, respectively, and therefore equals $I/2+s_XH_X$. For any joint decryptor state $\rho$, averaging over the signs gives strong-CPA acceptance probability
\[
 \frac{1}{4}\sum_{s_\bob,s_\charlie\in\{+1,-1\}}
 \Tr\!\left[
 \bigl(\mathbf 1[s_\bob H_\bob\geq\gamma]\otimes
       \mathbf 1[s_\charlie H_\charlie\geq\gamma]\bigr)\rho
 \right]
 =\frac{1}{4}\Tr(T_\gamma\rho).
\]
The equality holds because the positive and negative spectral subspaces together give $\mathbf 1[|H_X|\geq\gamma]$ for each party. Averaging also over the public key and the adversary's classical output preserves the equality. Non-negligible acceptance in the test $T_\gamma$ would therefore give a strong-CPA attack with non-negligible success, contradicting \cref{thm:sde-corr-implies-strong-cpa}.

The operator $W_+-I/2=2H_\bob\otimes H_\charlie$ explains the relation to CPA${}^+$. A positive threshold selects the equal-sign cases $(+1,+1)$ and $(-1,-1)$; a threshold on its absolute value also includes the opposite-sign cases. Since $\|H_X\|\leq1/2$, requiring the absolute product of the two local eigenvalues to be inverse polynomial is equivalent, up to an inverse-polynomial change in threshold, to requiring both absolute eigenvalues to be inverse polynomial. Thus, for binary-output decryptors, the test $T_\gamma$ includes ordinary strong CPA, the threshold form of CPA${}^+$, and the variants obtained by reversing either party's output. We do not introduce a separate security definition.
\end{proof}

The proof of \cref{thm:sde-corr-implies-strong-cpa} uses only separate mode. We now prove the reverse implication for this mode, i.e., that SDE-CORR security restricted to $\alice$ always choosing separate mode is equivalent to strong-CPA security, and hence to CPA${}^+$ security.

\begin{theorem}[Separate mode is exactly CPA${}^+$]
\label{thm:sde-sep-corr-iff-cpa-plus}
Let $\mathsf{SDE\mbox{-}CORR}|_{\sep}$ denote the game of \cref{def:sde-corr-game} restricted to adversaries that choose separate mode. Security in this game is equivalent to CPA${}^+$ security from \cref{def:sde-cpa-plus}, and hence to strong-CPA security from \cref{def:sde-strong-cpa}.
\end{theorem}

\begin{proof}
The proof of \cref{thm:sde-corr-implies-strong-cpa} gives the implication from separate-mode SDE-CORR to strong CPA. By \cref{thm:sde-prior-relations}, strong CPA is equivalent to CPA${}^+$. For the remaining direction, we turn any separate-mode SDE-CORR attack into a CPA${}^+$ attack in two steps. First we replace the correlated challenge bits by independent bits. Then we reduce the ciphertext lists to one challenge ciphertext per receiver.

Fix a separate-mode SDE-CORR adversary $\adve=(\alice,\bob,\charlie)$. We interpret every receiver output outside $\bit$ as $1$, which cannot decrease the probability that both challenge-bit guesses are correct. All acceptance operators below refer to the receivers after this conversion. We first compare its acceptance operators for each fixed public classical output of $\alice$ before the challenge. We omit this output from the notation; the reduction runs $\alice$ honestly and averages over all such outputs. For $X\in\{\bob,\charlie\}$ and $b\in\bit$, let $P_{X,b}$ be the average acceptance operator for the event that $X$ outputs $b$ on independent encryptions of $(\mes_{X,1,b},\ldots,\mes_{X,N,b})$. For fixed challenge bits, the operator for both receivers to be correct is $P_{\bob,b_\bob}\otimes P_{\charlie,b_\charlie}$, because their encryption randomness is independent.

Both challenge bits are individually uniform. Their joint distribution therefore satisfies $\Pr[00]=\Pr[11]=\theta/2$ and $\Pr[01]=\Pr[10]=(1-\theta)/2$ for some $\theta\in[0,1]$, which may depend on the public key. Define the acceptance operators for equal and unequal challenge bits, respectively, by
\begin{align*}
W_{\mathsf{diag}}&:=\frac{1}{2}\bigl(P_{\bob,0}\otimes P_{\charlie,0}+P_{\bob,1}\otimes P_{\charlie,1}\bigr),\\
W_{\mathsf{off}}&:=\frac{1}{2}\bigl(P_{\bob,0}\otimes P_{\charlie,1}+P_{\bob,1}\otimes P_{\charlie,0}\bigr).
\end{align*}
The SDE-CORR acceptance operator is $W_{\sep}=\theta W_{\mathsf{diag}}+(1-\theta)W_{\mathsf{off}}$.

Now use the same state and receivers in CPA${}^+$, with independent challenge bits. The CPA${}^+$ test accepts when both receivers are correct or both are incorrect. Its acceptance operator is
\[
W_+:=\frac{1}{4}\sum_{b,c\in\bit}\bigl(P_{\bob,b}\otimes P_{\charlie,c}+(I-P_{\bob,b})\otimes(I-P_{\charlie,c})\bigr).
\]
Expanding shows that this operator is at least each of the two SDE-CORR operators:
\begin{align}
W_+-W_{\mathsf{diag}}&=\frac{1}{2}\bigl((I-P_{\bob,0})\otimes(I-P_{\charlie,1})+(I-P_{\bob,1})\otimes(I-P_{\charlie,0})\bigr)\succeq0,\label{eq:sde-sep-corr-diag-dominance}\\
W_+-W_{\mathsf{off}}&=\frac{1}{2}\bigl((I-P_{\bob,0})\otimes(I-P_{\charlie,0})+(I-P_{\bob,1})\otimes(I-P_{\charlie,1})\bigr)\succeq0.\label{eq:sde-sep-corr-off-dominance}
\end{align}
Thus $W_{\sep}\preceq W_+$. A successful separate-mode SDE-CORR attack gives a CPA${}^+$ attack with $N$ challenge pairs and at least the same success probability.

To obtain an ordinary CPA${}^+$ attack, we change one ciphertext at a time. For a ciphertext list $\overline\ct$, let $Z_X(\overline\ct)$ be the difference between the receiver's acceptance operators for outputs $0$ and $1$. For $j\in\{0,\ldots,N\}$, let $Z_X^{(j)}$ be its expectation when the first $j$ coordinates encrypt their messages indexed by $1$, and all later coordinates encrypt their messages indexed by $0$. For $j\in[N]$, set $\Delta_{X,j}:=Z_X^{(j-1)}-Z_X^{(j)}$. The advantage of the $N$-pair CPA${}^+$ attack over $1/2$ can then be written as a sum over one coordinate for each receiver:
\begin{equation}
\frac{1}{8}\Tr\!\left[(Z_\bob^{(0)}-Z_\bob^{(N)})\otimes(Z_\charlie^{(0)}-Z_\charlie^{(N)})\rho\right]=\frac{1}{8}\sum_{j,k\in[N]}\Tr[(\Delta_{\bob,j}\otimes\Delta_{\charlie,k})\rho].
\label{eq:sde-sep-corr-telescoping}
\end{equation}

The reduction chooses these coordinates uniformly. Let $T(\secparam)$ be a polynomial upper bound on the number $N$ chosen by the fixed adversary. An ordinary CPA${}^+$ adversary runs $\alice$ and samples $j,k\getsr[T]$. If $j>N$ or $k>N$, it submits arbitrary fixed valid message pairs and has both receivers return fixed answers. Over the independent challenge bits, this succeeds with probability exactly $1/2$. Otherwise, it submits Bob's coordinate-$j$ pair and Charlie's coordinate-$k$ pair. It gives the receivers copies of $\pk,N$, both message lists, the selected indices, and descriptions of the original receiver algorithms.

Each receiver places its challenge ciphertext in the selected coordinate. It independently encrypts the message indexed by $1$ in each earlier coordinate and the message indexed by $0$ in each later coordinate, then runs its original receiver on the resulting list. For valid $j,k$, the difference between its average output operators for challenge bits $0$ and $1$ is exactly $\Delta_{\bob,j}$ or $\Delta_{\charlie,k}$. By \cref{eq:sde-sep-corr-telescoping}, averaging over $j,k$ gives an ordinary CPA${}^+$ advantage equal to $1/T^2$ times the $N$-pair advantage. A non-negligible $N$-pair advantage therefore contradicts CPA${}^+$ security.
\end{proof}

It remains open whether CPA${}^+$ security implies full SDE-CORR security. By \cref{thm:sde-sep-corr-iff-cpa-plus}, separate mode is already equivalent to CPA${}^+$. Any separation between full SDE-CORR and CPA${}^+$, or equivalently strong CPA, must therefore use identical mode, where both receivers get the same ciphertext. The identical mode of single-ciphertext SDE-CORR is already IDEN-CPA, and the implications of Kitagawa and Yamakawa~\cite{TCC:KitYam25} do not cover this case.

In particular, it is open whether CPA${}^+$ implies IDEN-CPA security. The reverse implication fails, as shown by the two-key construction of \cite[Theorem~13]{EPRINT:CGLR23a}.

\begin{theorem}[IDEN-CPA does not imply CPA${}^+$]
\label{thm:sde-iden-cpa-not-cpa-plus}
Assuming post-quantum-secure indistinguishability obfuscation and post-quantum-secure one-way functions, there exists a perfectly correct IDEN-CPA-secure SDE scheme with classical encryption that is not CPA${}^+$ secure.
\end{theorem}

\begin{proof}
Under the stated assumptions, \cref{thm:sde-security} provides a perfectly correct SDE-CORR-secure scheme with classical encryption. By \cref{lem:sde-corr-immediate}, it is IDEN-CPA secure. More generally, the construction below applies to any perfectly correct IDEN-CPA-secure scheme $\Pi$ whose message space contains two distinct messages, and preserves classical encryption.

The scheme $\Pi_{\mathsf{2key}}$ uses two independent keys and chooses one uniformly for each encryption. Generate $(\pk_t,\qsk_t)\gets\Pi.\keygen(1^\lambda)$ for $t\in\bit$, publish $(\pk_0,\pk_1)$, and use $\qsk_0\otimes\qsk_1$ as the quantum key. To encrypt $m$, sample $t\getsr\bit$, compute $\ct\gets\Pi.\enc(\pk_t,m)$, and output $(t,\ct)$. Decryption reads $t$ and uses the corresponding key, so correctness is perfect.

Suppose an adversary $\adve=(\alice,\bob,\charlie)$ violates IDEN-CPA security of $\Pi_{\mathsf{2key}}$. We construct an adversary against $\Pi$. Given $(\pk^\star,\qsk^\star)$, it samples $J\getsr\bit$ and generates an independent key pair. It places the given pair in position $J$ and the new pair in the other position, then runs $\alice$ for $\Pi_{\mathsf{2key}}$. It forwards the common message pair to its challenger and gives each receiver a classical copy of $J$. On receiving the copied ciphertext $\ct^\star$, each receiver runs its original circuit on $(J,\ct^\star)$.

The two key pairs are independent and identically distributed, so the ordered public and quantum keys reveal no information about the uniform position $J$. This exactly simulates the IDEN-CPA experiment for $\Pi_{\mathsf{2key}}$, even if $\alice$ processes the two component keys jointly. Thus IDEN-CPA security of $\Pi$ proves IDEN-CPA security of $\Pi_{\mathsf{2key}}$.

To violate CPA${}^+$ security, fix distinct messages $(\mu_0,\mu_1)$ and use this pair for both receivers. Give $\qsk_0$ to Bob and $\qsk_1$ to Charlie. Bob decrypts when the first ciphertext component is $0$ and otherwise returns a private fair bit. Charlie acts analogously when that component is $1$. Each receiver is correct with probability $3/4$ and incorrect with probability $1/4$, independently of the other receiver. Their correctness bits agree with probability $(3/4)^2+(1/4)^2=5/8$, so the scheme is not CPA${}^+$ secure.
\end{proof}

The CPA${}^+$ attack above gives one key to each receiver, leaving them with a product state. We now consider the reverse implication under a restriction on the IDEN-CPA adversary. We call an IDEN-CPA attack \emph{separable} if, conditioned on the public key and all classical information output by its splitter before the challenge, the joint state sent to Bob and Charlie is a mixture of product states. More precisely, for each value $u$ of this classical information that occurs with nonzero probability, the normalized state has the form $\rho^u_{\bob\charlie}=\sum_z p_{u,z}\rho_{u,z}\otimes\sigma_{u,z}$, where $(p_{u,z})_z$ is a probability distribution and $\rho_{u,z},\sigma_{u,z}$ are local states.

\begin{lemma}[Strong CPA implies IDEN-CPA against separable attacks]
\label{lem:sde-separable-iden-implies-strong-break}
Every strong-CPA-secure SDE scheme satisfies IDEN-CPA security against separable attacks.
\end{lemma}

\begin{proof}
We turn any successful separable IDEN-CPA attack into a strong-CPA attack. More precisely, let $\adve=(\alice,\bob,\charlie)$ be a QPT separable IDEN-CPA adversary against an SDE scheme $\sde$, with
\begin{equation}
\Pr\left[\gam{SDE\mbox{-}IDEN\mbox{-}CPA}_{\sde,\adve}(1^\lambda)=1\right]\geq\frac{1}{2}+\varepsilon,
\label{eq:sde-separable-iden-success}
\end{equation}
for some $\varepsilon=\varepsilon(\lambda)\in(0,1/2]$. We construct a strong-CPA adversary $\alice_{\mathsf{str}}$ satisfying
\begin{equation}
\Pr\left[\gam{SDE\mbox{-}STRONG\mbox{-}CPA}_{\sde,\alice_{\mathsf{str}}}(1^\lambda,\varepsilon/4)=1\right]\geq\frac{\varepsilon^3}{4}.
\label{eq:sde-separable-strong-success}
\end{equation}
If the separable IDEN-CPA attack has inverse-polynomial advantage $\varepsilon$, this bound contradicts strong-CPA security.

The strong-CPA adversary $\alice_{\mathsf{str}}$ runs the splitter $\alice$ of the separable IDEN-CPA attack. It uses the common message pair chosen by $\alice$ for both strong-CPA decryptors, keeps the joint state produced by $\alice$, and uses the receiver circuits $\bob$ and $\charlie$ as the decryptor circuits. The key point is that a successful separable IDEN-CPA attack gives non-negligible weight to product states on which both receivers individually succeed above the strong-CPA threshold. The following operator claim tracks the average over the public classical information.

\smallskip
\noindent\emph{Claim.}
Let $u$ be a classical random variable with distribution $(q_u)_u$. For each $u$, let $i$ be drawn from a distribution $\mathcal D_u$, and let $0\preceq A_{u,i},B_{u,i}\preceq I$ be operators on two registers. Define $P_u:=\E_{i\gets\mathcal D_u}[A_{u,i}]$, $Q_u:=\E_{i\gets\mathcal D_u}[B_{u,i}]$, and $W_u:=\E_{i\gets\mathcal D_u}[A_{u,i}\otimes B_{u,i}]$. Suppose $\rho_u=\sum_zp_{u,z}\rho_{u,z}\otimes\sigma_{u,z}$ is separable and $\sum_uq_u\Tr(W_u\rho_u)\geq1/2+\varepsilon$. Set $\Pi_u:=\mathbf 1[P_u\geq1/2+\varepsilon/4]$ and $\Gamma_u:=\mathbf 1[Q_u\geq1/2+\varepsilon/4]$. Then
\begin{equation}
\sum_uq_u\Tr\bigl((\Pi_u\otimes\Gamma_u)\rho_u\bigr)\geq\frac{\varepsilon^3}{4}.
\label{eq:sde-separable-operator-claim}
\end{equation}

To prove the claim, write $\tau_{u,z}:=\rho_{u,z}\otimes\sigma_{u,z}$, and define $a_{u,z}:=\Tr(P_u\rho_{u,z})$, $b_{u,z}:=\Tr(Q_u\sigma_{u,z})$, and $s_{u,z}:=\Tr(W_u\tau_{u,z})$. Since $s_{u,z}=\E_i[\Tr(A_{u,i}\rho_{u,z})\Tr(B_{u,i}\sigma_{u,z})]$ and each factor lies in $[0,1]$, we have $s_{u,z}\leq\min\{a_{u,z},b_{u,z}\}$. Give each pair $(u,z)$ probability $q_up_{u,z}$, and let $G$ be the set of pairs with $s_{u,z}\geq1/2+\varepsilon/2$. Because $0\leq s_{u,z}\leq1$, the premise implies that $G$ has probability at least $\varepsilon/(1-\varepsilon)$.

For every $(u,z)\in G$, both $a_{u,z}$ and $b_{u,z}$ are at least $1/2+\varepsilon/2$. The inequality $P_u\preceq(1/2+\varepsilon/4)I+(1/2-\varepsilon/4)\Pi_u$ then gives $\Tr(\Pi_u\rho_{u,z})\geq\varepsilon/(2-\varepsilon)$. The same argument gives $\Tr(\Gamma_u\sigma_{u,z})\geq\varepsilon/(2-\varepsilon)$. The two passing probabilities multiply within each product state. Averaging over $G$ therefore yields
\begin{equation}
\sum_uq_u\Tr\bigl((\Pi_u\otimes\Gamma_u)\rho_u\bigr)\geq\frac{\varepsilon}{1-\varepsilon}\left(\frac{\varepsilon}{2-\varepsilon}\right)^2\geq\frac{\varepsilon^3}{4},
\end{equation}
where the last inequality uses $0<\varepsilon\leq1/2$. This proves the claim.

We apply the claim to the separable IDEN-CPA attack $\adve$. Let $u$ contain $\pk$, the common message pair, and all classical information produced by $\alice$ before the challenge ciphertext, and let $q_u$ be the probability of this classical output. By the restriction to separable IDEN-CPA attacks, the conditional receiver state $\rho^u_{\bob\charlie}$ is a mixture of product states. Conditioned on $u$, let $i=(b,\ct)$, where $b\getsr\bit$ and $\ct\gets\sde.\enc(\pk,\mes_b)$. Let $A_{u,i}$ be the measurement operator for the IDEN-CPA receiver $\bob$ to output $b$ on input $\ct$, and define $B_{u,i}$ for the IDEN-CPA receiver $\charlie$ in the same way. Then $\sum_uq_u\Tr(W_u\rho^u_{\bob\charlie})$ is the success probability of the separable IDEN-CPA attack in \cref{eq:sde-separable-iden-success}. Moreover, $P_u$ and $Q_u$ are the local CPA acceptance operators for the two decryptors used by $\alice_{\mathsf{str}}$. Thus $\Pi_u$ and $\Gamma_u$ are the accepting projectors in the strong-CPA experiment with gap $\varepsilon/4$. The claim proves \cref{eq:sde-separable-strong-success}, completing the reduction from a separable IDEN-CPA attack to a strong-CPA attack.
\end{proof}

\Cref{lem:sde-separable-iden-implies-strong-break} rules out separable IDEN-CPA attacks against every strong-CPA-secure scheme, including attacks based on classical mixtures of product states indexed by public indices or challenge bits. Since strong CPA and CPA${}^+$ are equivalent, separating CPA${}^+$ from IDEN-CPA would require an IDEN-CPA attack with an entangled conditional receiver state for some value of the public key and the splitter's classical output. Constructing an IDEN-CPA attack that gives this separation remains open.

We next recall the implication from identical CPA to identical search. It follows by applying the reduction in \cite[proof of Theorem~29]{C:AnaKalLiu23} to identical challenges and uniform messages.

\begin{theorem}[Identical CPA implies identical search~\cite{C:AnaKalLiu23}]
\label{thm:sde-iden-cpa-implies-search}
For superlogarithmic message length, identical-CPA security implies identical-search security.
\end{theorem}

In this specialization, the CPA receivers run the search receivers and test whether the recovered message equals the second challenge message. The resulting CPA advantage is at least half the search advantage minus $1/(2|\messpa|)$. This loss is negligible for superlogarithmic message length.

Identical-challenge search and CPA security do not imply their independent-challenge counterparts; the corresponding two-key constructions appear in \cite[Theorems~12 and~13]{EPRINT:CGLR23a}. The CPA construction also fails independent-search security. Giving one component key to each receiver lets both decrypt whenever their independently sampled ciphertext indices select their assigned keys, an event of constant probability. We omit this additional separation between identical and independent challenges from the figure for readability.

Search security does not imply CPA security either. Given a search-secure scheme over $\mathcal M_0$, enlarge its message space to $\bit\times\mathcal M_0$ and reveal the first component in every ciphertext. Recovering the entire message still requires recovering the $\mathcal M_0$ component. Superlogarithmic message length makes the change in the uniform-guessing probability negligible, but two messages that differ only in the revealed component are perfectly distinguishable. This gives the search-to-CPA separations in the figure, including those between identical- and independent-challenge notions.
It remains to prove the separations with IND-CPA on the left. One scheme suffices: if it is IND-CPA secure but not IDEN-SEARCH secure, the implications above show that it is neither IDEN-CPA nor IND-SEARCH secure. It is also neither CPA${}^+$ nor strong-CPA secure, since both imply IND-SEARCH. The next construction, due to Coladangelo, Liu, Liu, and Zhandry, uses coset states and a logarithmic-length public index in each ciphertext. We prove its security for the IND-CPA game in which the two receivers may choose unrelated message pairs.

\begin{theorem}[Independent CPA does not imply identical CPA or strong CPA]
\label{thm:sde-ind-cpa-not-strong-cpa}
Assuming post-quantum-secure indistinguishability obfuscation and post-quantum-secure one-way functions, there exists a perfectly correct IND-CPA-secure SDE scheme that is neither IDEN-SEARCH secure nor strong-CPA secure.
The same scheme establishes every displayed separation in \cref{fig:sde-hierarchy} whose left-hand notion is $\mathsf{SDE\mbox{-}IND\mbox{-}CPA}$.
\end{theorem}

\begin{proof}
An IDEN-SEARCH adversary can prepare for one of the polynomially many public indices and recover the message whenever that index occurs. An IND-CPA adversary with non-negligible advantage, however, would let us find two different indices by threshold measurements and recover a valid point for each using quantum Goldreich--Levin. A coordinate where the indices differ would then violate computational strong monogamy.

Let $n=2\lambda$, let $L=\lceil c\log\lambda\rceil$ for any fixed constant $c>0$, set $R:=2^L$, and take $\mathcal M=\zo^\lambda$. For an $n/2$-dimensional subspace $A\leq\FF_2^n$ and shifts $s,t\in\FF_2^n$, define
\begin{equation*}
\ket{A_{s,t}}:=\frac{1}{\sqrt{|A|}}\sum_{a\in A}(-1)^{\langle a,t\rangle}\ket{a+s}.
\end{equation*}
For all sufficiently large $\lambda$, we have $R\geq4$.

\paragraph{Construction and correctness.}
For each $i\in[L]$, key generation independently samples a uniform $n/2$-dimensional subspace $A_i\leq\FF_2^n$ and uniform shifts $s_i,t_i\getsr\FF_2^n$. It publishes obfuscated membership programs for $A_i+s_i$ and $A_i^\perp+t_i$, and sets $\qsk:=\bigotimes_{i=1}^L\ket{(A_i)_{s_i,t_i}}$. For $a\in\bit$, define $B_{i,0}:=A_i$, $B_{i,1}:=A_i^\perp$, $c_{i,0}:=s_i$, and $c_{i,1}:=t_i$. For $r\in\bit^L$, put $B_r:=\prod_{i=1}^L B_{i,r_i}$, $c_r:=(c_{1,r_1},\ldots,c_{L,r_L})$, and $S_r:=B_r+c_r$. The public programs determine membership in $S_r$.

For $m\in\mathcal M$, define
\begin{equation*}
P_{r,m}(z):=
\begin{cases}
m,&z\in S_r,\\
\bot,&z\notin S_r.
\end{cases}
\end{equation*}
Encryption samples $r\getsr\bit^L$ and outputs $\ct=(r,\iO(P_{r,m}))$. To decrypt, apply a Hadamard transform to the $i$th coset-state register exactly when $r_i=1$, evaluate the obfuscated program coherently, copy its output, and undo the evaluation and the Hadamard transforms. The transformed key is supported on $S_r$, where the program has constant value $m$. Decryption therefore outputs $m$ and restores the key, proving perfect correctness. This is the SDE construction of~\cite{C:CLLZ21}, with $L=\Theta(\log\lambda)$ coset states and a $\lambda$-bit program output.

\paragraph{Failure of identical-challenge security.}
Fix $r^\star\in\bit^L$. Before the challenge, measure coset state $i$ in the computational basis when $r_i^\star=0$ and the Hadamard basis when $r_i^\star=1$, obtaining a classical point $z^\star\in S_{r^\star}$. Give both recipients $z^\star$ and a shared uniform fallback value $g\getsr\mathcal M$. Both recover the IDEN-SEARCH message when the ciphertext index is $r^\star$, and otherwise output $g$. Their success probability is
\begin{equation}
\frac{1}{R}+\left(1-\frac{1}{R}\right)\frac{1}{|\mathcal M|}
=
\frac{1}{|\mathcal M|}+\frac{1}{R}\left(1-\frac{1}{|\mathcal M|}\right).
\label{eq:sde-separation-iden-search-attack}
\end{equation}
Since $R=2^L$ is polynomial, this gives an inverse-polynomial IDEN-SEARCH advantage. The same strategy breaks IDEN-CPA security: for any fixed distinct messages $(\mu_0,\mu_1)$, both recipients recover the common challenge bit when $r=r^\star$ and use one shared uniform bit otherwise. Their success probability is $1/2+1/(2R)$.

\paragraph{Finding two different public indices.}
Suppose a QPT adversary $\adve=(\alice,\bob,\charlie)$ wins IND-CPA with probability at least $1/2+\varepsilon$ for an inverse polynomial $\varepsilon$. We interpret receiver outputs outside $\bit$ as $1$, which cannot decrease success; all receiver circuits and operators below use this convention. We retain, but suppress from the notation, the classical message-pair and circuit-description registers output by $\alice$. If either message pair repeats a message, that receiver's challenge bit is independent of its view and joint success is at most $1/2$. We therefore restrict attention to branches where both pairs are distinct.

We need two different indices at which the respective receivers retain an advantage. Joint success gives both receivers some advantage and at least one a larger advantage, allowing us to exclude the other receiver's chosen index.

For $X\in\{\bob,\charlie\}$ and $r\in\bit^L$, let $P_{X,r}$ be the correct-bit acceptance operator for fixed index $r$ and otherwise honest encryption, and set $P_X:=R^{-1}\sum_{r\in\bit^L}P_{X,r}$. Independence of the challenge bits, indices, and encryption coins gives joint acceptance operator $P_\bob\otimes P_\charlie$. Thus the state $\rho$ produced by $\alice$ satisfies
\begin{equation}
\Tr\!\left[(P_\bob\otimes P_\charlie)\rho\right]\geq\frac{1}{2}+\varepsilon.
\label{eq:sde-separation-ind-win}
\end{equation}

Choose all approximate-threshold errors below to be sufficiently small fixed polynomials in $\varepsilon/L$. Apply the gapped threshold measurement of \cref{lem:gapped-threshold-instrument} to $P_\bob\otimes P_\charlie$, with completeness threshold $1/2+3\varepsilon/4$ and soundness threshold $1/2+2\varepsilon/3$. The bound $M\preceq\tau I+(1-\tau)\mathbf 1[M\geq\tau]$ and \cref{eq:sde-separation-ind-win} give acceptance probability $\Omega(\varepsilon)$. Up to the chosen approximation error, the accepted state is supported on product eigenspaces satisfying
\begin{equation}
\lambda_\bob\lambda_\charlie>\frac{1}{2}+\frac{2\varepsilon}{3},
\label{eq:sde-separation-product-sector}
\end{equation}
where $\lambda_X$ is an eigenvalue of $P_X$. Since $\lambda_X\leq1$, both local eigenvalues exceed $1/2+2\varepsilon/3$, and at least one exceeds $1/\sqrt2>11/16$.

Choose one party uniformly as $H$ and call the other $W$. Measure $P_H$ with completeness threshold $11/16$ and soundness threshold $43/64$. On each product eigenspace in \cref{eq:sde-separation-product-sector}, at least one party is above the completeness threshold. Averaging over $H$ therefore gives constant acceptance probability. Up to the chosen approximation error, the accepted state is supported where
\begin{equation}
P_H\succeq\frac{43}{64}I
\qquad\text{and}\qquad
P_W\succeq\left(\frac{1}{2}+\frac{2\varepsilon}{3}\right)I.
\label{eq:sde-separation-high-low}
\end{equation}
Here and below, displayed operator inequalities hold on the support established by the preceding measurements.

For $W$, sample $r_W\getsr\bit^L$ and measure $P_{W,r_W}$ with completeness threshold $1/2+\varepsilon/4$ and soundness threshold $1/2+7\varepsilon/32$. Set $\Pi_{W,r}:=\mathbf 1[P_{W,r}\geq1/2+\varepsilon/4]$. Since
\begin{equation*}
P_{W,r}\preceq\left(\frac{1}{2}+\frac{\varepsilon}{4}\right)I+\left(\frac{1}{2}-\frac{\varepsilon}{4}\right)\Pi_{W,r},
\end{equation*}
averaging over $r$ and using \cref{eq:sde-separation-high-low} gives
\begin{equation*}
\E_r[\Pi_{W,r}]
\succeq
\frac{5\varepsilon/12}{1/2-\varepsilon/4}I
\succeq
\frac{5\varepsilon}{6}I.
\end{equation*}
The measurement therefore accepts with probability $\Omega(\varepsilon)$, leaving $W$ supported where $P_{W,r_W}>1/2+7\varepsilon/32$.

For $H$, sample $r_H$ uniformly from $\bit^L\setminus\{r_W\}$. The preceding local operation on $W$ preserves the support condition for $H$ in \cref{eq:sde-separation-high-low}. Since $P_{H,r_W}\preceq I$,
\begin{equation}
\frac{1}{R-1}\sum_{r\neq r_W}P_{H,r}
=
\frac{RP_H-P_{H,r_W}}{R-1}
\succeq
\frac{R(43/64)-1}{R-1}I
\succeq
\frac{9}{16}I,
\label{eq:sde-separation-excluded-tag}
\end{equation}
where the last inequality uses $R\geq4$. Let $\Pi_{H,r}:=\mathbf 1[P_{H,r}\geq17/32]$. Since $P_{H,r}\preceq(17/32)I+(15/32)\Pi_{H,r}$, \cref{eq:sde-separation-excluded-tag} implies that the average weight on $\Pi_{H,r}$, over $r\neq r_W$, is at least $1/15$. Applying the threshold measurement to $P_{H,r_H}$ with completeness threshold $17/32$ and soundness threshold $67/128$ therefore accepts with constant probability.

Before distributing the registers, the reduction therefore obtains, with probability $\Omega(\varepsilon^2)$, distinct indices $r_W\neq r_H$ and a state supported, up to the chosen approximation errors, where
\begin{equation}
P_{W,r_W}>\frac{1}{2}+\frac{7\varepsilon}{32}
\qquad\text{and}\qquad
P_{H,r_H}>\frac{1}{2}+\frac{3}{128}.
\label{eq:sde-separation-fixed-tag-sectors}
\end{equation}

\paragraph{Recovering a point for a fixed index.}
Quantum Goldreich--Levin turns a receiver's advantage at a fixed index $r$ into a point in $S_r$, as follows. Let $\tau$ be an efficiently generated subnormalized joint state of a receiver register and an arbitrary reference register, together with an honestly generated public setup, a fixed public index $r$, and a distinct message pair $(\mu_0,\mu_1)$. Write $p:=\Tr(\tau)$, and assume that $p$ is bounded below by an inverse polynomial on the security parameters under consideration. Suppose
\begin{equation}
\Tr(P_{X,r}\tau)\geq\left(\frac{1}{2}+\delta\right)p
\label{eq:sde-separation-fixed-index-advantage}
\end{equation}
for an inverse polynomial $\delta$. After descriptions of the $A_i$ are revealed, there is a local QPT algorithm, acting only on the receiver register and its work registers, that outputs $z\in S_r$ with probability $\Omega(p\delta^2)$.

To find a point in the coset, we learn linear information about $c_r$. We rewrite the honest ciphertext program using an inner product with $c_r$, without changing its function. Let $\Delta:=\mu_0\oplus\mu_1$, where multiplication of $\Delta$ by a bit denotes either $0^\lambda$ or $\Delta$. For $q\getsr B_r^\perp$ and $d\getsr\bit$, define
\begin{equation}
Q_{r,q,d}(z):=
\begin{cases}
\mu_0\oplus(d\oplus\langle q,z\rangle)\cdot\Delta,&z\in S_r,\\
\bot,&z\notin S_r.
\end{cases}
\label{eq:sde-separation-linear-program}
\end{equation}
For $z\in S_r=B_r+c_r$, we have $\langle q,z\rangle=\langle q,c_r\rangle$. Thus, with $b:=d\oplus\langle q,c_r\rangle$, programs $Q_{r,q,d}$ and $P_{r,\mu_b}$ compute the same function, and $b$ is uniform because $d$ is uniform.

The iO reduction generates the honest setup and prepares $\tau$, so it knows $c_r$; the recovery algorithm will not need the shifts. It samples $q,d$, computes $b=d\oplus\langle q,c_r\rangle$, and submits equally padded descriptions of $P_{r,\mu_b}$ and $Q_{r,q,d}$ to the challenger. It tests whether the receiver outputs $b$, returning an independent fair bit if preparation of $\tau$ fails. The first circuit gives the honest fixed-$r$ experiment. With the second, XORing the receiver's output with $d$ predicts $\langle q,c_r\rangle$. Post-quantum iO preserves the total prediction advantage in \cref{eq:sde-separation-fixed-index-advantage} up to a negligible term.

After the $A_i$ are revealed, let $G_r$ be a full-rank matrix whose rows form a basis of $B_r^\perp$, and write $q=G_r^{\mathsf T}y$ for uniform $y$. The predictor then guesses
\begin{equation*}
\langle q,c_r\rangle=\langle y,G_rc_r\rangle.
\end{equation*}
Let $\Gamma$ denote twice this prediction advantage, weighted by the trace $p$ of $\tau$. The iO hybrid gives $\Gamma\geq2\delta p-\negl(\lambda)$. Since both $p$ and $\delta$ have inverse-polynomial lower bounds, the negligible loss is at most $\delta p$ for all sufficiently large security parameters under consideration. Hence $\Gamma\geq\delta p>0$. The quantum Goldreich--Levin lemma with quantum auxiliary input~\cite[Lemma~B.12]{C:CLLZ21}, together with Cauchy--Schwarz over the classical setup and coin variables, gives total recovery probability at least $\Gamma^2/p=\Omega(p\delta^2)$. Once $G_rc_r$ is recovered, solve $G_rz=G_rc_r$. Since $\ker G_r=B_r$, every solution satisfies $z\in B_r+c_r=S_r$.

Recovery uses only the public membership programs, the revealed $A_i$, the receiver circuit, and its register. Each Goldreich--Levin execution samples $d$ and one obfuscator random tape classically and fixes these coins across all coherent queries. Averaging over them and applying Jensen's inequality preserves the quadratic recovery probability. The recovery map is local: conditioning on an outcome may change the reference state, but preserves its previously established local spectral support.

\paragraph{Reduction to computational strong monogamy.}
The two recovery algorithms give an adversary against computational strong monogamy of entanglement~\cite[Theorem~4.18]{C:CLLZ21}, whose information-theoretic premise was proved in~\cite{CulfVidick22}. The first-stage adversary receives a challenge state $\ket{A_{s,t}}$ and public membership obfuscations for $A+s$ and $A^\perp+t$. It samples $J\getsr[L]$, places the challenge state and membership programs in coordinate $J$, and generates the other coordinates honestly. This gives exactly the scheme's public and quantum key distribution.

It runs $\alice$ and the four threshold measurements without revealing or using $J$ in the index-selection procedure. All coordinates have the same distribution, so $J$ is independent of the setup, measurement outcomes, and selected indices. If any measurement rejects, the reduction outputs arbitrary registers. Otherwise it has distinct indices $r_W,r_H$ before splitting, and
\begin{equation*}
\Pr_J[(r_W)_J\neq(r_H)_J]
=
\frac{\operatorname{dist}(r_W,r_H)}{L}
\geq
\frac{1}{L}.
\end{equation*}
If the $J$th bits agree, the reduction aborts. Otherwise it sends the register with $J$th index bit $0$ to the primal CLLZ recipient and the register with $J$th index bit $1$ to the dual recipient. Each receives its circuit, message pair, selected index, the public key, and descriptions of all honestly generated subspaces.

After the challenger reveals $A=A_J$, both recipients know every subspace and run fixed-index recovery. Let $\tau$ be the subnormalized branch where all four measurements accept and $(r_W)_J\neq(r_H)_J$. The acceptance bound and independence of $J$ give
\begin{equation*}
\Tr(\tau)=\Omega\!\left(\frac{\varepsilon^2}{L}\right).
\end{equation*}
This trace has an inverse-polynomial lower bound, so the fixed-index recovery claim applies. Recovery on $W$ with $\delta=7\varepsilon/32$ gives a valid point in $S_{r_W}$ with total probability $\Omega(\varepsilon^4/L)$. Both local recovery algorithms may run concurrently; for the analysis, apply $W$'s map first and condition on public verification of its output. The successful branch still has inverse-polynomial trace, and this local operation on $W$ leaves $H$ supported where $P_{H,r_H}>1/2+3/128$. Applying the same recovery claim to $H$ therefore gives constant conditional success probability. Joint recovery in $S_{r_W}$ and $S_{r_H}$ succeeds with probability $\Omega(\varepsilon^4/L)$.

At coordinate $J$, index bit $0$ gives a point in $A+s$, and index bit $1$ gives a point in $A^\perp+t$. The recipients output these $J$th components, winning computational strong monogamy with probability $\Omega(\varepsilon^4/L)$. Since $L=O(\log\lambda)$, this contradicts~\cite[Theorem~4.18]{C:CLLZ21} and proves IND-CPA security.

The IDEN-SEARCH attack in \cref{eq:sde-separation-iden-search-attack} rules out IDEN-SEARCH security, and hence IDEN-CPA and IND-SEARCH security. It also rules out CPA${}^+$ and strong-CPA security, which imply IND-SEARCH by \cref{thm:sde-prior-relations}.
\end{proof}
\subsubsection{Summary of the Main Relations}
\label{sec:sde-hierarchy-summary}
We summarize the main relations proved or recalled above. As throughout \cref{sec:sde-hier-defn}, every statement involving search security assumes superlogarithmic message length.

\begin{enumerate}
 \item SDE-CORR security implies strong-CPA security by \cref{thm:sde-corr-implies-strong-cpa}, and hence CPA${}^+$ security by~\cite{TCC:KitYam25}. Single-ciphertext SDE-CORR security already implies IND-CPA and IDEN-CPA security by \cref{lem:sde-corr-immediate}, so full SDE-CORR security does as well.
 \item Separate-mode SDE-CORR security is equivalent to CPA${}^+$ security, and hence to strong-CPA security. This holds for polynomially many challenge-message pairs and every admissible sampler; see \cref{thm:sde-sep-corr-iff-cpa-plus}.
 \item Strong-CPA security, equivalently CPA${}^+$ security, implies IND-CPA and IND-SEARCH security. IND-SEARCH and strong-search security are equivalent, and both imply IDEN-SEARCH security; see \cref{thm:sde-prior-relations}. Strong-CPA security also implies IDEN-CPA security against separable attacks by \cref{lem:sde-separable-iden-implies-strong-break}.
 \item Strong-CPA security implies strong-search security by \cref{cor:sde-strong-cpa-implies-strong-search}. The converse fails under the assumptions of \cref{thm:sde-strong-cpa-strictly-implies-strong-search}. IDEN-CPA security implies IDEN-SEARCH security by \cref{thm:sde-iden-cpa-implies-search}.
 \item IDEN-CPA security does not imply CPA${}^+$ security by \cref{thm:sde-iden-cpa-not-cpa-plus}. Assuming post-quantum-secure iO and post-quantum-secure one-way functions, IND-CPA security implies none of IDEN-CPA, IND-SEARCH, IDEN-SEARCH, CPA${}^+$, or strong-CPA security. The scheme in \cref{thm:sde-ind-cpa-not-strong-cpa} gives all these separations.
 \item IDEN-SEARCH security does not imply IND-SEARCH security, and IDEN-CPA security implies neither IND-CPA nor IND-SEARCH security~\cite[Theorems~12 and~13]{EPRINT:CGLR23a}. Revealing one component of each message also shows that IND-SEARCH implies neither IND-CPA nor IDEN-CPA, and that IDEN-SEARCH implies neither IDEN-CPA nor IND-CPA.
\end{enumerate}

\clearpage

%% file: defns/sde-hierarchy-figure.tex
\begin{center}
    \centering
    \definecolor{sdenew}{RGB}{111,45,189}
    \definecolor{sdenewdark}{RGB}{78,23,140}
    \makebox[\textwidth][c]{%
    \begin{tikzpicture}[
        x=1cm,
        y=.88cm,
        >=Stealth,
        notion/.style={
            draw=black!42,
            fill=white,
            rounded corners=1pt,
            minimum height=9mm,
            inner xsep=5pt,
            font=\small
        },
        corr notion/.style={
            draw=sdenewdark,
            fill=sdenewdark,
            text=white,
            rounded corners=1.5pt,
            minimum height=13mm,
            inner xsep=10pt,
            line width=1.35pt,
            font=\small
        },
        cross color/.initial=black!82,
        prior/.style={
            draw=black!82,
            line width=.48pt,
            line cap=round,
            cross color=black!82
        },
        new/.style={
            draw=sdenew,
            line width=1.35pt,
            line cap=round,
            cross color=sdenew
        },
        implication/.style={
            solid,
            -{Stealth[length=2mm,width=1.5mm]}
        },
        equivalence/.style={solid,{Stealth[length=1.8mm,width=1.35mm]}-{Stealth[length=1.8mm,width=1.35mm]}
        },
        separation/.style={
            dash pattern=on 3.8pt off 2.3pt,
            -{Stealth[length=1.8mm,width=1.35mm]},
            postaction={decorate},
            decoration={
                markings,
                mark=at position .50 with {
                    \node[
                        inner sep=.35pt,
                        outer sep=0pt,
                        fill=white,
                        text=\pgfkeysvalueof{/tikz/cross color},
                        font=\scriptsize\bfseries
                    ] {$\times$};
                }
            }
        }
    ]
        \node[
            corr notion,
            align=center,
            minimum width=52mm
        ] (corr) at (0,11.0)
        {
            \textsf{SDE-CORR}\\[-1pt]
            {\footnotesize\bfseries implies all notions below}
        };

        \node[notion,minimum width=26mm]
            (plus) at (-6.2,7.8)
            {$\textsf{SDE-CPA}^{+}$};

        \node[notion,minimum width=37mm]
            (strongcpa) at (-2.4,7.3)
            {\textsf{SDE-STRONG-CPA}};

        \node[notion,minimum width=33mm]
            (indcpa) at (-1.3,4.1)
            {\textsf{SDE-IND-CPA}};

        \node[notion,minimum width=35mm]
            (idencpa) at (6.0,5.1)
            {\textsf{SDE-IDEN-CPA}};

        \node[notion,minimum width=40mm]
            (strongsearch) at (-6.2,-3.2)
            {\textsf{SDE-STRONG-SEARCH}};

        \node[notion,minimum width=35mm]
            (indsearch) at (-1.0,-.1)
            {\textsf{SDE-IND-SEARCH}};

        \node[notion,minimum width=36mm]
            (idensearch) at (4.2,-2.8)
            {\textsf{SDE-IDEN-SEARCH}};

        \draw[
            prior,
            separation,
            preaction={draw=white,line width=3.0pt}
        ]
            ([xshift=-4mm]idencpa.north)
            .. controls (4.7,9.35) and (-4.8,9.35) ..
            (plus.north);

        \draw[new,implication]
            ([xshift=-13mm]corr.south)
            to[bend right=10]
            (strongcpa.north);

        \draw[new,implication]
            ([xshift=-2mm]corr.south)
            to[bend left=3]
            ([xshift=8mm]indcpa.north);

        \draw[new,implication]
            ([xshift=13mm]corr.south)
            to[bend left=7]
            ([xshift=5mm]idencpa.north);

        \draw[prior,equivalence]
            (plus.east) -- (strongcpa.west);

        \draw[prior,implication]
            ([xshift=-2mm]strongcpa.south)
            to[bend left=3]
            ([xshift=-2mm]indcpa.north);

        \draw[new,separation]
            ([xshift=2mm]indcpa.north)
            to[bend left=3]
            ([xshift=2mm]strongcpa.south);

        \draw[prior,implication]
            ([yshift=-2mm]strongcpa.west)
            to[bend right=4]
            ([xshift=-2mm]strongsearch.north);

        \draw[prior,separation]
            ([xshift=2mm]strongsearch.north)
            to[bend right=4]
            ([yshift=2mm]strongcpa.west);

        \draw[prior,equivalence]
            (strongsearch.east) -- (indsearch.west);

        \draw[new,separation]
            (indcpa.east)
            to[bend left=4]
            (idencpa.west);

        \draw[prior,separation]
            (idencpa.west)
            to[bend left=4]
            (indcpa.east);

        \draw[new,separation]
            ([xshift=-6mm]indcpa.south)
            to[bend left=4]
            (indsearch.north west);

        \draw[prior,separation]
            (indsearch.north west)
            to[bend left=4]
            ([xshift=-6mm]indcpa.south);

        \draw[
            prior,
            separation,
            rounded corners=5pt
        ]
            ([yshift=-2mm]idencpa.east)
            -- (8.25,4.9)
            -- (8.25,-3.65)
            -- (-1.9,-3.65)
            -- ([xshift=6mm]indsearch.south);

        \draw[
            prior,
            separation,
            rounded corners=5pt
        ]
            ([xshift=-6mm]indsearch.south)
            -- (-3.1,-4.0)
            -- (8.55,-4.0)
            -- (8.55,5.3)
            -- ([yshift=2mm]idencpa.east);

        \draw[new,separation]
            ([xshift=7mm]indcpa.south)
            to[bend left=4]
            ([xshift=-6mm]idensearch.north);

        \draw[prior,separation]
            ([xshift=-6mm]idensearch.north)
            to[bend left=4]
            ([xshift=7mm]indcpa.south);

        \draw[prior,implication]
            ([xshift=7mm]idencpa.south)
            to[bend left=4]
            (idensearch.north east);

        \draw[prior,separation]
            (idensearch.north east)
            to[bend left=4]
            ([xshift=7mm]idencpa.south);

        \draw[prior,implication]
            (indsearch.south east)
            to[bend left=4]
            (idensearch.west);

        \draw[prior,separation]
            (idensearch.west)
            to[bend left=4]
            (indsearch.south east);

        \coordinate (legone) at (-7.2,-4.75);

        \node[
            anchor=west,
            font=\footnotesize\bfseries,
            text=black!72
        ] at (legone) {Provenance:};

        \draw[new]
            ($(legone)+(2.3,0)$) -- ++(.8,0);

        \node[
            anchor=west,
            font=\footnotesize,
            text=sdenewdark
        ] at ($(legone)+(3.25,0)$)
            {new in this work};

        \draw[prior]
            ($(legone)+(5.7,0)$) -- ++(.8,0);

        \node[
            anchor=west,
            font=\footnotesize,
            text=black!82
        ] at ($(legone)+(6.65,0)$)
            {prior or elementary};

        \coordinate (legtwo) at (-7.2,-5.38);

        \node[
            anchor=west,
            font=\footnotesize\bfseries,
            text=black!72
        ] at (legtwo) {Relation:};

        \draw[prior,implication]
            ($(legtwo)+(2.3,0)$) -- ++(.8,0);

        \node[anchor=west,font=\footnotesize]
            at ($(legtwo)+(3.25,0)$)
            {implies};

        \draw[prior,equivalence]
            ($(legtwo)+(4.55,0)$) -- ++(.8,0);

        \node[anchor=west,font=\footnotesize]
            at ($(legtwo)+(5.5,0)$)
            {equivalent};

        \draw[prior,separation]
            ($(legtwo)+(7.35,0)$) -- ++(.9,0);

        \node[anchor=west,font=\footnotesize]
            at ($(legtwo)+(8.4,0)$)
            {does not imply};
    \end{tikzpicture}
    }%

    
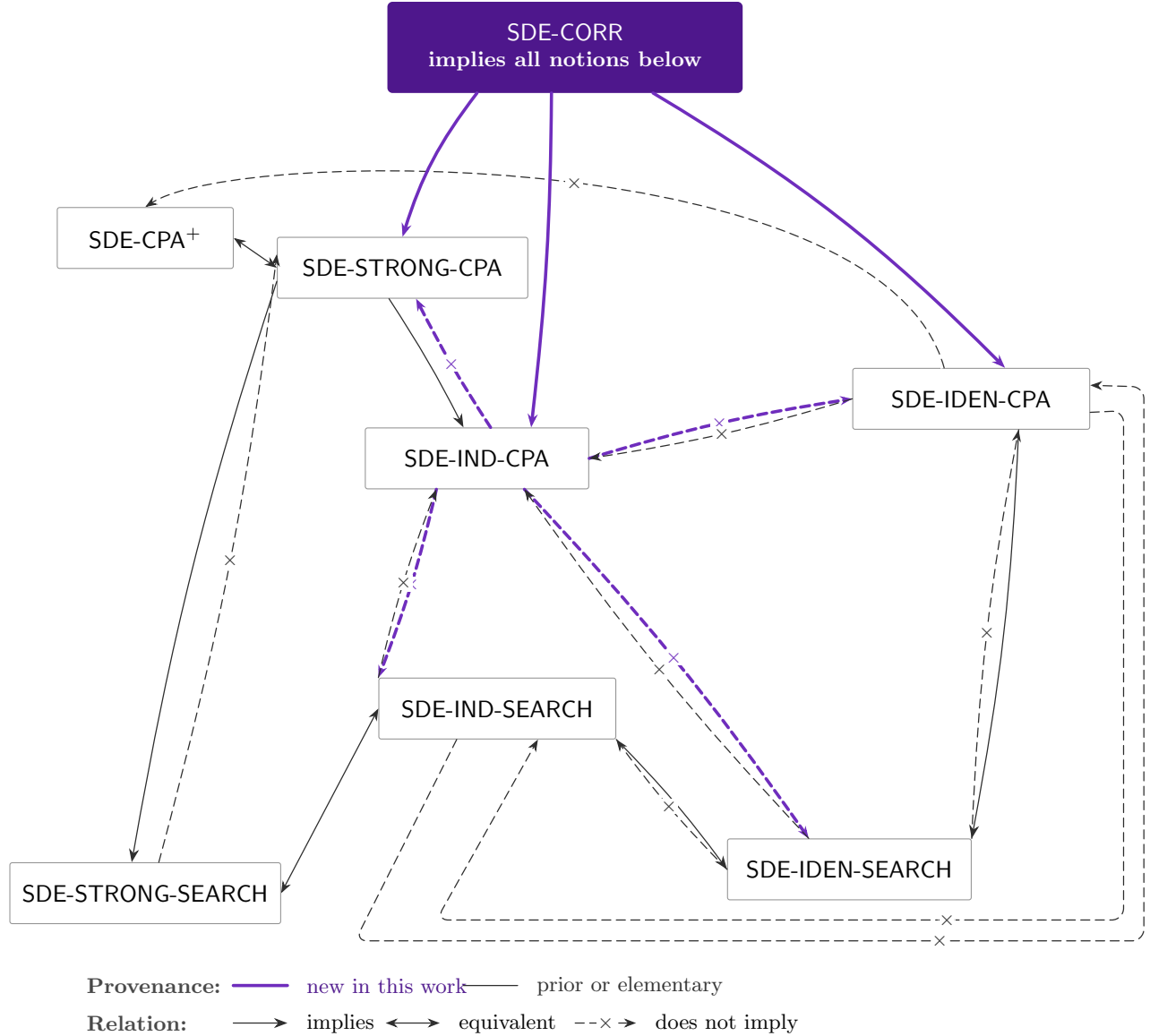
\captionof{figure}{
        Hierarchy of SDE security notions under the standing assumptions of
        this subsection. Solid arrows denote implication, solid double-headed
        lines denote equivalence, and dashed directed arrows marked by a
        midpoint $\times$ denote failure of the implication in the arrow
        direction. Violet heavy lines are new in this work; black lines are
        prior or elementary. Unmarked directions are not claimed; in
        particular, the implication from SDE-CPA${}^+$ to SDE-IDEN-CPA
        remains open.
    }
    \label{fig:sde-hierarchy}
\end{center}


%% file: defns/old-cp-defn.tex
\ifconference
    \section{Other prior security definitions of copy-protection}\label{sec:other-prior-cp-def}
\fi

Next, we describe the natural instantiations of the unpredictability-style copy-protection security definition for specialized classes of functions, as well as the pseudorandomness-style copy-protection security notions. In the following definitions, except in the case of point-function, $k$-point-function, and compute-and-compare experiments stated separately below, the circuit and the challenge samplers are independent, while the two challenge points themselves may be arbitrarily correlated.

\begin{definition}[Identical-challenge unpredictability-style copy protection for point functions, adapted from~{\cite[Theorem~89 and Corollary~90]{AB24}}]
\label{def:point-function-copy-protection}
Let $\mathcal X_\secparam=\zo^{\ell_{\mathsf{in}}(\secparam)}$, let $\mathcal D=\{\mathcal D_\secparam\}_{\secparam\in\NN}$ be a QPT-samplable distribution over $\mathcal X_\secparam$, and for every $y\in\mathcal X_\secparam$ define $P_y(x):=[x=y]$. For a copy-protection scheme $\mathsf{CP}$ for point functions and any QPT adversarial triplet $\mathcal A_{\mathsf{PF}}=(\alice,\bob,\charlie)$, consider the following experiment.

\par\noindent
$\gam{\mathsf{PFCP}_{\mathsf{CP},\mathcal A_{\mathsf{PF}},\mathcal D}}(1^\secparam)$:
\begin{enumerate}
    \item $\ch$ samples $y\gets\mathcal D_\secparam$ and $\rho_y\gets\mathsf{CopyProtect}(1^\secparam,P_y)$, and sends $\rho_y$ to $\alice$.
    \item $\alice$ outputs a joint state $\sigma_{\regbob,\regcharlie}$ and sends register $\regbob$ to $\bob$ and register $\regcharlie$ to $\charlie$.
    \item After the split, $\ch$ samples $b\getsr\bit$. If $b=0$, it samples an independent $x\gets\mathcal D_\secparam$; if $b=1$, it sets $x:=y$. It sends the same $x$ to both $\bob$ and $\charlie$.
    \item $\bob$ and $\charlie$ output $\widehat v_\bob,\widehat v_\charlie\in\bit$.
    \item The output of the experiment is $1$ if $\widehat v_\bob=\widehat v_\charlie=P_y(x)$.
\end{enumerate}

A copy-protection scheme for point functions satisfies identical-challenge point-function security with respect to $\mathcal D$ if, for every QPT adversarial triplet $\mathcal A_{\mathsf{PF}}$, there exist a trivial adversary $\mathcal A_{\mathsf{triv}}\in\mathsf{Triv}$ and a negligible function $\negl(\cdot)$ such that
\begin{equation}
\Pr[
\gam{\mathsf{PFCP}_{\mathsf{CP},\mathcal A_{\mathsf{PF}},\mathcal D}}(1^\secparam)=1
]
\leq
p_{\mathsf{triv}}^{\mathsf{PF},\mathcal D}(\secparam)
+\negl(\secparam),
\end{equation}
where
\begin{equation}
p_{\mathsf{triv}}^{\mathsf{PF},\mathcal D}(\secparam):=
\sup_{\mathcal A_{\mathsf{triv}}\in\mathsf{Triv}}
\Pr\left[
\gam{\mathsf{PFCP}_{\mathsf{CP},\mathcal A_{\mathsf{triv}},\mathcal D}}(1^\secparam)=1
\right].
\label{eq:point-function-trivial-success}
\end{equation}
The uniform choice of $\mathcal D_\secparam$ recovers the experiment of~\cite[Theorem~89 and Corollary~90]{AB24}.
\end{definition}

We note that, assuming the point distribution has sufficiently high min-entropy, the trivial success probability is negligibly close to $1/2$.
\begin{lemma}[Trivial success in the point-function experiment]
\label{lem:point-function-trivial-success}
Let $\operatorname{cp}(\mathcal D_\secparam):=\sum_x\Pr[\mathcal D_\secparam=x]^2$. For the experiment in \Cref{def:point-function-copy-protection},
\begin{equation}
p_{\mathsf{triv}}^{\mathsf{PF},\mathcal D}(\secparam)
=
\frac{1}{2}+\frac{1}{2}\operatorname{cp}(\mathcal D_\secparam)
\leq
\frac{1}{2}+\frac{1}{2}\,2^{-\Hmin(\mathcal D_\secparam)}.
\end{equation}
Consequently, if $\Hmin(\mathcal D_\secparam)\geq\secparam^c$ for some constant $c>0$, then
\begin{equation}
p_{\mathsf{triv}}^{\mathsf{PF},\mathcal D}(\secparam)
=\frac{1}{2}+\negl'(\secparam).
\end{equation}
\end{lemma}

\begin{proof}
Both events induce the marginal distribution $x\gets\mathcal D_\secparam$. For every $x$ in its support,
\begin{equation}
\Pr[P_y(x)=1\mid x]
=
\frac{1}{2}\bigl(1+\Pr[\mathcal D_\secparam=x]\bigr).
\end{equation}
Under a trivial strategy, the recipient receiving the dummy register has no information about $y$ beyond $x$, and hence cannot do better than always outputting $1$. This strategy is attained when both recipients output $1$ and has average success $1/2+\operatorname{cp}(\mathcal D_\secparam)/2$. Finally, $\operatorname{cp}(\mathcal D_\secparam)\leq\max_x\Pr[\mathcal D_\secparam=x]=2^{-\Hmin(\mathcal D_\secparam)}$.
\end{proof}
\begin{remark}[Copy protection definition for $k$-point functions]\label{rem:k-point-functions-def}
The copy protection definition for point functions, as outlined above, also extends to the fixed-$k$-point functions considered in~\cite[Theorem~85 and Corollary~90]{AB24}. Let $\mathcal D^{(k)}=\{\mathcal D^{(k)}_\secparam\}_{\secparam\in\NN}$ be any QPT-samplable distribution over the $k=k(\secparam)$ element subsets of $\mathcal X_\secparam$ such that $\Hmin(\mathcal D^{(k)}_\secparam)\geq\secparam^c$ for some constant $c>0$. For a set $S\subseteq\mathcal X_\secparam$ of size $k$, define $P_S(x):=[x\in S]$, and let $\mathcal D^{\mathsf{pt}}_\secparam$ denote the distribution obtained by sampling $S\gets\mathcal D^{(k)}_\secparam$ and then $x\getsr S$, such that $\Hmin(\mathcal D^{\mathsf{pt}}_\secparam)\geq\secparam^{c'}$ for some constant $c'>0$. The identical-challenge experiment samples $S\gets\mathcal D^{(k)}_\secparam$ and protects $P_S$. After the split, it samples $b\getsr\bit$. If $b=0$, it samples an independent $S'\gets\mathcal D^{(k)}_\secparam$ and then $x\getsr S'$; if $b=1$, it samples $x\getsr S$. In either case, it gives the same $x$ to both recipients. Security uses the same trivial-adversary comparison as above.

For $q_\secparam(x):=\Pr_{S\gets\mathcal D^{(k)}_\secparam}[x\in S]$, the trivial success probability is
\begin{equation}
p_{\mathsf{triv}}^{\mathsf{PF}(k),\mathcal D^{(k)}}(\secparam)
=\frac{1}{2}+\frac{1}{2k(\secparam)}\sum_{x\in\mathcal X_\secparam}q_\secparam(x)^2
=\frac{1}{2}+\frac{k(\secparam)}2\operatorname{cp}(\mathcal D^{\mathsf{pt}}_\secparam),
\label{eq:k-point-function-trivial-success}
\end{equation}
which is $1/2+\negl(\secparam)$ whenever $k$ is polynomially bounded. 
\end{remark}

\begin{definition}[Identical-challenge copy protection for compute-and-compare functions, adapted from~\cite{CMP24}]
\label{def:compute-and-compare-functions-copy-protection}
\label{rem:compute-and-compare-functions-def}
Following~\cite[Construction~4 and Theorem~7]{CMP24}, the definition above extends to compute-and-compare functions
\[
\mathsf{CC}_{f,y}(x):=[f(x)=y],
\]
where $f:\mathcal X_\secparam\rightarrow\mathcal Z_\secparam$, $\mathcal Z_\secparam=\zo^{\ell_{\mathsf{cc}}(\secparam)}$, and $y\in\mathcal Z_\secparam$. 
Let $\mathsf{Samp}$ be a QPT algorithm that samples $(f,y)\gets\mathsf{Samp}(1^\secparam)$, interpreted as sampling $\mathsf{CC}_{f,y}$, and let $(F,Y)$ denote its output distribution. We require $\Hmin(Y\mid F)\geq\secparam^c$ for some constant $c>0$.

For every $f$ in the support of $F$ and every $z$ in the support of the conditional distribution of $Y$ given $F=f$, fix a distribution $\mathcal Q_{f,z}$ supported on $f^{-1}(z)$. The identical-challenge copy-protection experiment for compute-and-compare functions samples $(f,y)\gets\mathsf{Samp}(1^\secparam)$ and protects $\mathsf{CC}_{f,y}$ as $(f,\rho_y)$, where $\rho_y$ is a copy protection of $P_y$. After the split, it samples $b\getsr\bit$. If $b=0$, it samples a fresh $z$ independently according to the conditional distribution of $Y$ given $F=f$; if $b=1$, it sets $z:=y$. It then samples $x\gets\mathcal Q_{f,z}$ and gives the same $x$ to both recipients. The same family of distributions $\{\mathcal Q_{f,z}\}_{f,z}$ is used for both values of $b$. We say that the scheme satisfies identical-challenge compute-and-compare copy-protection security if, for every QPT adversarial triplet, there exists a negligible function $\negl$ such that its success probability in the experiment above is at most
\[
p_{\mathsf{triv}}^{\mathsf{CC},\mathsf{Samp}}(\secparam)+\negl(\secparam),
\]
where $p_{\mathsf{triv}}^{\mathsf{CC},\mathsf{Samp}}(\secparam)$ denotes the supremum success probability in this experiment over the trivial adversaries in $\mathsf{Triv}$, as defined in Definition~\ref{def:point-function-copy-protection}.

We say that this identical-challenge experiment is efficiently realizable if there exist QPT procedures that:
\begin{enumerate}
    \item on input $(1^\secparam,f)$, sample a fresh $z$ according to the conditional distribution of $Y$ given $F=f$; and
    \item on input $(1^\secparam,f,z)$, sample $x\gets\mathcal Q_{f,z}$.
\end{enumerate}
\end{definition}

Next, we recall the pseudorandomness-style copy protection notion from previous works. 
\begin{definition}[Correlated pseudorandomness-style copy-protection challenge sampler]
\label{def:correlated-pr-copy-protection-sampler}
A correlated pseudorandomness-style copy-protection challenge sampler $\mathsf{Samp}_{\mathsf{PRCP}}$ with $\kappa$-entropic marginals is a QPT algorithm that, on input $1^\secparam$, outputs
$(x_\bob,x_\charlie,b_\bob,b_\charlie)\gets\mathsf{Samp}_{\mathsf{PRCP}}(1^\secparam)$,
where $x_\bob,x_\charlie\in\mathcal X_\secparam$ and $b_\bob,b_\charlie\in\bit$, such that, for every $(\bar x_\bob,\bar x_\charlie)$ in the support of $(x_\bob,x_\charlie)$ and every $X\in\{\bob,\charlie\}$,
\begin{equation}
\Pr[
b_X=0\mid
x_\bob=\bar x_\bob,x_\charlie=\bar x_\charlie
]
=
\Pr[
b_X=1\mid
x_\bob=\bar x_\bob,x_\charlie=\bar x_\charlie
]
=
\frac{1}{2},
\end{equation}
and $\Hmin(x_X)\geq\kappa(\secparam)$ for every $X\in\{\bob,\charlie\}$.
\end{definition}

We write $\mathsf{Samp}^{\mathsf{pr}}_{\IDU}$ for the sampler that samples $x\getsr\mathcal X_\secparam$ and $b\getsr\bit$, and outputs $(x_\bob,x_\charlie,b_\bob,b_\charlie):=(x,x,b,b)$.

Given $C$, a sampled tuple $(x_\bob,x_\charlie,b_\bob,b_\charlie)$, and independent $y_\bob,y_\charlie\getsr\mathcal Y_\secparam$, define
\begin{equation}
z_\bob:=
\begin{cases}
C(x_\bob),&b_\bob=0,\\
y_\bob,&b_\bob=1,
\end{cases}
\qquad
z_\charlie:=
\begin{cases}
z_\bob,&x_\charlie=x_\bob,\\
C(x_\charlie),&x_\charlie\neq x_\bob\text{ and }b_\charlie=0,\\
y_\charlie,&x_\charlie\neq x_\bob\text{ and }b_\charlie=1.
\end{cases}
\label{eq:pr-copy-protection-challenge-values}
\end{equation}
Thus, a potential tie (i.e., $x_\bob=x_\charlie$ but $b_\bob\neq b_\charlie$) is resolved by $\bob$'s challenge bit the precedence.

For any QPT adversarial triplet $\mathcal A_{\mathsf{PRCP}}=(\alice,\bob,\charlie)$, consider the following experiment.

\par\noindent
$\gam{\mathsf{PRCP}_{\mathsf{CP},\mathcal A_{\mathsf{PRCP}},\mathcal D_{\mathsf{Circ}},\mathsf{Samp}_{\mathsf{PRCP}}}}(1^\secparam)$:
\begin{enumerate}
    \item $\ch$ samples $C\gets\mathcal D_{\mathsf{Circ}}(1^\secparam)$ and $\rho_C\gets\mathsf{CopyProtect}(1^\secparam,C)$, and gives $\rho_C$ to $\alice$.
    \item $\alice$ outputs a joint state $\sigma_{\regbob,\regcharlie}$ and sends register $\regbob$ to $\bob$ and register $\regcharlie$ to $\charlie$.
    \item After the split, $\ch$ samples $(x_\bob,x_\charlie,b_\bob,b_\charlie)\gets\mathsf{Samp}_{\mathsf{PRCP}}(1^\secparam)$ and independent $y_\bob,y_\charlie\getsr\mathcal Y_\secparam$. It defines $z_\bob,z_\charlie$ according to \Cref{eq:pr-copy-protection-challenge-values}.
    \item $\ch$ gives $(x_\bob,z_\bob)$ to $\bob$ and $(x_\charlie,z_\charlie)$ to $\charlie$.
    \item $\bob$ and $\charlie$ output $\widehat b_\bob,\widehat b_\charlie$.
    \item The output of the experiment is $1$ if $\widehat b_\bob=b_\bob$ and $\widehat b_\charlie=b_\charlie$.
\end{enumerate}

\begin{definition}[Pseudorandomness-style copy-protection security with correlated challenges, adapted from~{\cite[Definition~3.3]{ABH+26}}]
\label{def:pr-copy-protection}
\label{def:non-oracular-pr-copy-protection}
A copy-protection scheme $\mathsf{CP}$ satisfies $\kappa$-entropic pseudorandomness-style copy-protection security with respect to $(\mathcal D_{\mathsf{Circ}},\mathsf{Samp}_{\mathsf{PRCP}})$ if, for every QPT adversarial triplet $\mathcal A_{\mathsf{PRCP}}$, there exists a negligible function $\negl(\cdot)$ such that
\begin{equation}
\Pr[
\gam{\mathsf{PRCP}_{\mathsf{CP},\mathcal A_{\mathsf{PRCP}},\mathcal D_{\mathsf{Circ}},\mathsf{Samp}_{\mathsf{PRCP}}}}(1^\secparam)=1
]
\leq
\frac{1}{2}+\negl(\secparam).
\end{equation}
\end{definition}

\subsection{Oracular security notions}
Next we recall the oracular versions of the copy-protection security definitions defined in~\cite{ABH+26}.
The experiments below are identical to the experiments above, except that $\alice$ receives oracle access to the sampled circuit before the split and each recipient receives oracle access punctured at that recipient's challenge point after the split.

For every circuit $C:\mathcal X_\secparam\rightarrow\mathcal Y_\secparam$ and $x^\star\in\mathcal X_\secparam$, define
\begin{equation}
C^{\setminus x^\star}(x):=
\begin{cases}
\bot,&x=x^\star,\\
C(x),&x\neq x^\star.
\end{cases}
\label{eq:single-point-punctured-oracle}
\end{equation}
All oracle accesses below are quantum superposition accesses.

For any QPT adversarial triplet $\mathcal A_{\mathsf{CP}}=(\alice,\bob,\charlie)$, consider the following experiment.

\par\noindent
$\gam{\mathsf{OraUnpCP}_{\mathsf{CP},\mathcal A_{\mathsf{CP}},\mathcal D_{\mathsf{Circ}},\mathsf{Samp}_{\mathsf{CP}}}}(1^\secparam)$:
\begin{enumerate}
    \item $\ch$ samples $C\gets\mathcal D_{\mathsf{Circ}}(1^\secparam)$ and $\rho_C\gets\mathsf{CopyProtect}(1^\secparam,C)$, and gives $\rho_C$ and oracle access to $C$ to $\alice$.
    \item $\alice^C$ outputs a joint state $\sigma_{\regbob,\regcharlie}$ and sends register $\regbob$ to $\bob$ and register $\regcharlie$ to $\charlie$.
    \item After the split, $\ch$ samples $(x_\bob,x_\charlie)\gets\mathsf{Samp}_{\mathsf{CP}}(1^\secparam)$, gives $x_\bob$ and oracle access to $C^{\setminus x_\bob}$ to $\bob$, and gives $x_\charlie$ and oracle access to $C^{\setminus x_\charlie}$ to $\charlie$.
    \item $\bob^{C^{\setminus x_\bob}}$ and $\charlie^{C^{\setminus x_\charlie}}$ output $\widehat y_\bob,\widehat y_\charlie$.
    \item The output of the experiment is $1$ if $\widehat y_\bob=C(x_\bob)$ and $\widehat y_\charlie=C(x_\charlie)$.
\end{enumerate}

The notion of trivial strategies in this experiment is defined similarly to the case of no oracle access. For $X\in \{\bob,\charlie\}$, rhe first-stage adversary $\alice^{\mathsf{triv}}_X$ with access to the oracle sends the entire protected program $\rho_C$ to $X$ and the fixed dummy state $\ketbra{\bot}{\bot}$ to the other party, and, we define
\begin{equation}
\mathsf{Triv}:=
\left\{(\alice^{\mathsf{triv}}_\bob,\bob,\charlie):
\bob,\charlie\text{ are QPT}\right\}\mathbin{\cup}
\left\{(\alice^{\mathsf{triv}}_\charlie,\bob,\charlie):
\bob,\charlie\text{ are QPT}\right\}.
\label{eq:copy-protection-trivial-adversaries-oracular}
\end{equation}

\begin{definition}[Oracular unpredictability-style copy-protection security with correlated challenges, adapted from~{\cite[Definition~3.1]{ABH+26}}]
\label{def:oracular-unpredictability-copy-protection}
\label{def:oracular-copy-protection-security}
A copy-protection scheme $\mathsf{CP}$ satisfies oracular unpredictability-style copy-protection security with respect to $(\mathcal D_{\mathsf{Circ}},\mathsf{Samp}_{\mathsf{CP}})$ if, for every QPT adversarial triplet $\mathcal A_{\mathsf{CP}}$, there exist a trivial adversary $\mathcal A_{\mathsf{triv}}\in\mathsf{Triv}$ (defined in \Cref{eq:copy-protection-trivial-adversaries-oracular}) and a negligible function $\negl(\cdot)$ such that
\begin{equation}
\Pr\left[
\gam{\mathsf{OraUnpCP}_{\mathsf{CP},\mathcal A_{\mathsf{CP}},\mathcal D_{\mathsf{Circ}},\mathsf{Samp}_{\mathsf{CP}}}}
(1^\secparam)=1
\right]
\leq
\Pr\left[
\gam{\mathsf{OraUnpCP}_{\mathsf{CP},\mathcal A_{\mathsf{triv}},\mathcal D_{\mathsf{Circ}},\mathsf{Samp}_{\mathsf{CP}}}}
(1^\secparam)=1
\right]
+\negl(\secparam).
\end{equation}
\end{definition}

For any QPT adversarial triplet $\mathcal A_{\mathsf{PRCP}}=(\alice,\bob,\charlie)$, consider the following experiment.

\par\noindent
$\gam{\mathsf{OraPRCP}_{\mathsf{CP},\mathcal A_{\mathsf{PRCP}},\mathcal D_{\mathsf{Circ}},\mathsf{Samp}_{\mathsf{PRCP}}}}(1^\secparam)$:
\begin{enumerate}
    \item $\ch$ samples $C\gets\mathcal D_{\mathsf{Circ}}(1^\secparam)$ and $\rho_C\gets\mathsf{CopyProtect}(1^\secparam,C)$, and gives $\rho_C$ and oracle access to $C$ to $\alice$.
    \item $\alice^C$ outputs a joint state $\sigma_{\regbob,\regcharlie}$ and sends register $\regbob$ to $\bob$ and register $\regcharlie$ to $\charlie$.
    \item After the split, $\ch$ samples $(x_\bob,x_\charlie,b_\bob,b_\charlie)\gets\mathsf{Samp}_{\mathsf{PRCP}}(1^\secparam)$ and independent $y_\bob,y_\charlie\getsr\mathcal Y_\secparam$. It defines $z_\bob,z_\charlie$ according to \Cref{eq:pr-copy-protection-challenge-values}.
    \item $\ch$ gives $(x_\bob,z_\bob)$ and oracle access to $C^{\setminus x_\bob}$ to $\bob$, and gives $(x_\charlie,z_\charlie)$ and oracle access to $C^{\setminus x_\charlie}$ to $\charlie$.
    \item $\bob^{C^{\setminus x_\bob}}$ and $\charlie^{C^{\setminus x_\charlie}}$ output $\widehat b_\bob,\widehat b_\charlie$.
    \item The output of the experiment is $1$ if $\widehat b_\bob=b_\bob$ and $\widehat b_\charlie=b_\charlie$.
\end{enumerate}

\begin{definition}[Oracular pseudorandomness-style copy-protection security with correlated challenges, adapted from~{\cite[Definition~3.3]{ABH+26}}]
\label{def:oracular-pr-copy-protection}
A copy-protection scheme $\mathsf{CP}$ satisfies $\kappa$-entropic oracular pseudorandomness-style copy-protection security with respect to $(\mathcal D_{\mathsf{Circ}},\mathsf{Samp}_{\mathsf{PRCP}})$ if, for every QPT adversarial triplet $\mathcal A_{\mathsf{PRCP}}$, there exists a negligible function $\negl(\cdot)$ such that
\begin{equation}
\Pr[
\gam{\mathsf{OraPRCP}_{\mathsf{CP},\mathcal A_{\mathsf{PRCP}},\mathcal D_{\mathsf{Circ}},\mathsf{Samp}_{\mathsf{PRCP}}}}(1^\secparam)=1
]
\leq
\frac{1}{2}+\negl(\secparam).
\end{equation}
\end{definition}

%% file: moe/it-moe.tex
\section{Decisional Monogamy-of-Entanglement for Coset States}\label{sec:it-moe}
\ifconference
In this section, we state and prove a novel information-theoretic decisional \moe game for coset states that serves as the foundation (under other cryptographic assumptions) for all the feasibility results in this article. \else 
All existing constructions of copy-protection schemes, including single-decryptor schemes, are based on an information-theoretic \moe game. However, existing \moe games have been unable to yield results in settings with correlated challenges due to various obstructions, including the fact that they were mainly search-based \moe games. In this section, we introduce and prove a novel decisional \moe game for coset states that yields feasibility results in the correlated challenge case, which we describe next.
 \fi 
 We note that all adversaries in this section are computationally unbounded and can be modeled as arbitrary (finite-dimensional) quantum channels.  \ifconference \else All the results hold unconditionally, without any computational assumptions. \fi

\par\noindent
\begin{definition}[\dcmoe]\label{def:dcmoe}
Let $n:=n(\secparam)$ be a polynomial function in $\secparam$.  For any triplet of adversaries $\advmoe=(\alice,\bob,\charlie)$, consider the following experiment between a challenger $\ch$ and $\advmoe$.
\par\noindent $\underline{\gam{\dcmoe_\advmoe}(1^\secparam)}$:
\begin{enumerate}
    \item $\ch$ samples a uniformly random $n$-dimensional subspace $A\leq \FF_2^{2n}$ and samples $s,t\getsr \FF_2^{2n}$ independently, and prepares and sends the state $\ket{A_{s,t}}=\frac{1}{\sqrt{|A|}}\sum_{a\in A}(-1)^{\langle a,t\rangle }\ket{a+ s}$.
    \item $\alice$ outputs a joint state $\sigma_{\regbob,\regcharlie}$ and sends registers $\regbob$ to $\bob$ and $\regcharlie$ to $\charlie$.
    \item After the split, $\ch$ samples $u\getsr A\setminus \{0\}$ and $v\getsr A^\perp\setminus \{0\}$ and sends the pair $(u,v)$ to $\bob$ and $\charlie$.
    \item $\bob$ and $\charlie$ output $b_\bob,b_\charlie$.
    \item The output of the experiment $1$ if $b_\bob=b_\charlie=\langle u,t\rangle+ \langle v, s\rangle$.
\end{enumerate}
\end{definition}
The next lemma (\Cref{lem:dcmoe-ue}) shows that the $\dcmoe$ experiment can be seen as an unclonable-indistinguishable experiment (see \Cref{def:UE-security}) as follows.
\begin{lemma}\label{lem:dcmoe-ue}
    The \dcmoe experiment is equivalent to the unclonable-indistinguishable security experiment (see \Cref{def:UE-security}) for the following unclonable encryption scheme.
\par\noindent \underline{$\cosetue$}:     
\begin{enumerate}
    \item $\keygen(1^\secparam)$: Sample $u,v\gets \FF_2^{2n}\setminus \{0\}$ subject to $\langle u,v\rangle=0$. Output $k=(u,v)$.
    \item $\enc(k,m)$: Sample a uniformly random $n$-dimensional subspace $ A\leq \FF_2^{2n}$ subject to $u\in A$ and $v\in A^\perp$ where $k=(u,v)$. Sample $s,t\getsr \FF_2^{2n}$ subject to $m=\langle u,t\rangle+ \langle v, s\rangle$. Prepare and output the state $\rho_{ct}=\ket{A_{s,t}}$.
    \item $\dec(k,\rho_{ct})$: Interpret $k:=(u,v)$, compute $m_1=\langle v,s\rangle$, by computing the inner product of $v$ coherently with all strings in superposition in the state $\ket{A_{s,t}}$ and then measure the outcome, which does not collapse the state (since $\langle a+s, v=\langle s,v\rangle$ for every $a\in A$). Then apply the $QFT$ on the state to get the state $\ket{A^\perp_{t,s}}$, and then compute $m_2=\langle u,t\rangle$ similarly using the state $\ket{A^\perp_{t,s}}$ and $u$. Output $m=m_1+ m_2$.
\end{enumerate}
\end{lemma}
\ifconference
    The proof of the lemma is given in \Cref{sec:lem-it-moe}.
\else
    \begin{proof}[Proof of \Cref{lem:dcmoe-ue}]
    For every fixed nonzero orthogonal pair $(u,v)\in \FF_2^{2n}\times\FF_2^{2n}$, let the set of $n$-dimensional subspaces $A$ satisfying $u\in A$ and $v\in A^\perp$ be defined as $\comp(u,v)$. It is easy to see that $|\comp(u,v)|$ depends only on $n$ and is independent of the choice of $(u,v)$. Hence the marginal distribution on the pair $(u,v)$ induced by the $\dcmoe$ experiment is uniform over the key space of $\cosetue$, given as $\keyspace:=\{(u,v)\mid u,v \in \FF_2^{2n}\setminus\{0\}, \langle u,v\rangle =0\}$. Thus sampling $A$ uniformly at random, and then sampling $u\getsr A\setminus\{0\}, v\getsr A^\perp\setminus\{0\}$, is the same as first sampling a uniformly random nonzero orthogonal pair $(u,v)\in \FF_2^{2n}\times\FF_2^{2n}$, and then sampling $A\in \comp(u,v)$, and therefore the marginal distribution of $(A,u,v))$ in the \dcmoe experiment and the unclonable-indistinguishable experiment (see \Cref{def:UE-security}) for \cosetue, is the same. Next, for a fixed non-zero orthogonal pair $(u,v)$ and a fixed subspace $A\in \comp(u,v)$, the map \[(s,t)\longmapsto\langle u,t\rangle+\langle v,s\rangle\] is a nonzero linear function and hence is balanced; therefore, sampling $s,t$ uniformly and setting the message to this parity is equivalent to first sampling a uniform message $m$ and then sampling $s,t$ uniformly conditioned on \[m=\langle u,t\rangle+\langle v,s\rangle.\] Therefore, conditioned on $A,(u,v)$, the distribution of the subspace $s,t$ and the bit $\langle u,t\rangle+ \langle v, s\rangle$ in the \dcmoe experiment, is the same as that of $s,t$ and the message bit $m$ in the unclonable indistinguishable experiment for \cosetue, which concludes the proof of the lemma.
    \end{proof}

\fi
\begin{theorem}[Decisional Coset Monogamy]\label{thm:dcmoe}
For any $\secparam\in \NN$, and any triplet of adversary $\advmoe=(\alice,\bob,\charlie)$,
\[\Pr[\gam{\dcmoe_\advmoe}(1^\secparam)=1]\leq \frac{1}{2}+\frac{1}{2}\sqrt{\frac{2^{2n}}{(2^{2n}-1)(2^{2n-1}-1)}}.\]
\end{theorem}

\begin{remark}
It is easy to verify that the conclusion of \Cref{thm:dcmoe} holds even if we reveal $A$ to $\bob$ and $\charlie$ in the \dcmoe game.
\end{remark}

\ifconference

\else
Due to \Cref{lem:dcmoe-ue}, \Cref{thm:dcmoe} gives the following immediate corollary.

\begin{corollary}\label{cor:dcmoe-ue}
The encryption scheme $\cosetue$ is an unconditional $\frac{1}{2}\sqrt{\frac{2^{2n}}{(2^{2n}-1)(2^{2n-1}-1)}}$-unclonable indistinguishable encryption scheme.
\end{corollary}

\fi
In order to prove \Cref{thm:dcmoe}, we will use its interpretation as a unclonable indistinguishable experiment, given in \Cref{lem:dcmoe-ue}. In particular, we mildly extend the proof of \cite{AS26}\footnote{A similar generalization of~\cite{AS26} have been independently shown in the concurrent works of \cite{BBC26,BDF26}.} to get a general lemma showing unclonable-indistinguishable security of a class of unclonable bit-encryption schemes, that captures both $\cosetue$ and the construction of~\cite{AS26}.

\begin{lemma}[General class of unclonable indistinguishable encryption schemes]\label{lem:general-UE}
    Suppose for every $\ell\in \NN$, there exists a set of keys $\keyspace$ and a set of hermitian operators  $\{\Theta_k\}_{k\in \cal{K}}$ in $\mathcal{L}((\CC^2)^{\otimes \ell})$, indexed by $\keyspace$, satisfying that for every $k$, $\tr[\Theta_k]=0$, $\Theta_k^2=I$, and for every $k,k'\in \cal{K}$, $\tr[\Theta_k\Theta_{k'}]=N\delta_{k,k'}$ where $N=2^\ell$ is the dimension of the space, and $\delta_{ij}$ is the Kronecker Delta function, then for any distribution $P=\{p_k\}_{k\in \cal{K}}$ on $\cal{K}$, the following quantum encryption scheme given by
    \begin{enumerate}
        \item $\keygen(1^\secparam)$: Sample $k\gets P$, and output $k$.
        \item $\enc(k,m)$: Prepare and output the state $\rho^{k,m}_{ct}=1/N\left(I+ (-1)^m\Theta_k\right)$.
        \item $\dec(k,\rho)$: Measure $\rho$ in the eigen basis of $\Theta_k$, i.e., the POVM $(\Pi_0,\Pi_1):=\left(\frac{I + \Theta_k}{2},\frac{I - \Theta_k}{2}\right)$ and output the outcome.
    \end{enumerate}
     satisfies that for every $\secparam$, and adversarial triplet $\advue=(\alice,\bob,\charlie)$ against it in the unclonable-indistinguishable experiment (see \Cref{def:UE-security}), the adversarial success probability is at most $\frac{1}{2}+\frac{1}{2}\sqrt{N\col(P)}$, where $\col(P)=\sum_{k\in \cal{K}}p_k^2$' i.e., the above scheme satisfies $\frac{1}{2}\sqrt{N\col(P)}$-unclonable-indistinguishability\footnote{We clarify that for certain choices of the operator family $\{\Theta_k\}_{k\in \cal{K}}$ and distribution $P$, the obtained scheme may not be efficient. }.
\end{lemma}

The class of unclonable encryption schemes in the above lemma is a generalization of~\cite{AS26} in two directions: first, in terms of considering an arbitrary Hermitian operator family as opposed to Paulis, and second, in terms of considering an arbitrary distribution on the keyspace as opposed to the uniformly random distribution.
Next we prove \Cref{thm:dcmoe} assuming \Cref{lem:general-UE}. We defer the proof of \Cref{lem:general-UE} to \Cref{sec:general-UE-proof}.
\begin{proof}[Proof of \Cref{thm:dcmoe}]
We will use the interpretation of \Cref{lem:dcmoe-ue} and prove $\left(\frac{1}{2}\sqrt{\frac{2^{2n}}{(2^{2n}-1)(2^{2n-1}-1)}}\right)$-unclonable-indistinguishable security of $\cosetue$, which by \Cref{lem:dcmoe-ue} would imply \Cref{thm:dcmoe}. 

We will show that $\cosetue$ is an instance of the general class of unclonable indistinguishable encryption schemes in \Cref{lem:general-UE}.
Let $\keyspace:=\{(u,v)\mid u,v \in \FF_2^{2n}\setminus\{0\}, \langle u,v\rangle =0\}$ denote the key space for \cosetue.
For any key $k=(u,v)\in \keyspace$, let $\Theta_k=X^uZ^v\in \mathcal{L}((\CC^2)^{\otimes 2n})$.
Since for every $k$, $\Theta_k=X^uZ^v$ is a product of Pauli, it is a unitary and it holds that $\Theta_k^\dagger=Z^vX^u=(-1)^{\langle u,v\rangle}X^uZ^v=X^uZ^v=\Theta_k,$ i.e., $\Theta_k$ is Hermitian, and hence, due to unitary condition, $\Theta_k^2=\Theta_k\Theta^\dagger_k=I$. Moreover, since $(u,v)\neq 0$, $\Theta_k=X^uZ^v$ is a non-identity product of Paulis and therefore $\Tr[\Theta_k]=0$ and $\Theta_k^2=I$.
Similarly, for any pair of keys, $k=(u,v)\in \keyspace$ and $k'=(u',v')\keyspace$, \[\tr\left(\Theta_k\Theta_{k'}\right)=\tr\left(X^{u}Z^v X^{u'} Z^{v'}\right)=\tr\left((-1)^{\langle v,u'\rangle}X^{u+ u'}Z^{v+ v'}\right)=(-1)^{\langle v,u'\rangle}\tr\left(X^{u+ u'}Z^{v+ v'}\right),\]
which is $\tr(I)=2^{2n}$ if $u=u', v=v'$, and $0$ otherwise. Hence, we conclude that for any pair of keys $k,k'\in \keyspace$, $\tr(\Theta_k\Theta_{k'})=2^{2n}\delta_{k,k'}$. Hence $\{\Theta_k\}_{k\in \keyspace}$ satisfies the hypothesis of \Cref{lem:general-UE} with $N=2^{2n}$. 

Next, let $P$ denote the distribution induced by the $\keygen$ algorithm of $\cosetue$ is just the uniformly random distribution on $\keyspace$. Hence, $\col(P)=\frac{1}{|\keyspace|}$. It is easy to check that $|K|=(2^{2n}-1)(2^{2n-1}-1)$. Hence, we conclude that 
\begin{equation}\label{eq:adv-term}
\col(P):=\frac{1}{(2^{2n}-1)(2^{2n-1}-1)}.
\end{equation}

Next, we claim that for any message $m\in \bit$ and key $k\in \keyspace$, the cipher is of the form required in the hypothesis of \Cref{lem:general-UE}.

\begin{claim}\label{claim:conditional-cipher-state}
    For any message $m\in \bit$, and key $k\in \keyspace$, let the ciphertext produced by $\cosetue.\enc(k,m)$ be $\rho_{ct}^{k,m}$, then it holds that $\rho_{ct}^{k,m}=\frac{1}{N}\left(I + (-1)^m \Theta_k\right).$
\end{claim}
Next, we complete the proof of the theorem assuming Claim~\ref{claim:conditional-cipher-state} before proving the claim.
Let $\sf{UE}$ be a quantum encryption schme that is the same as \cosetue, but with a different decryption algorithm, namely the new decryption algorithm $\dec'$ on input $k,\rho_ct$ simply measure in the eigenbasis of $\Theta_k$, i.e., the POVM $(\Pi_0,\Pi_1):=\left(\frac{I + \Theta_k}{2},\frac{I - \Theta_k}{2}\right)$ and outputs the outcome\footnote{Nevertheless, it can be checked that the two decryption algorithms are functionally the same.}. Clearly, the fact that the operator family $\{\Theta_k\}_{k\in\keyspace}$ satisfies the hypothesis of \Cref{lem:general-UE}, combined with Claim~\ref{claim:conditional-cipher-state}, implies that $\sf{UE}$ is an instance of the general class of encryption schemes described in \Cref{lem:general-UE}. Hence by \Cref{lem:general-UE}, we conclude that $\sf{UE}$ satisfies $\frac{1}{2}\sqrt{N\col(P)}$-unclonable indistinguishable security. Since the unclonable-indistinguishable experiment (see \Cref{def:UE-security}) is independent of the decryption algorithm of the encryption scheme, and $\cosetue$ and $\sf{UE}$ only differ in decryption algorithms, $\cosetue$ must also satisfy $\frac{1}{2}\sqrt{N\col(P)}$-unclonable indistinguishable security. By \Cref{eq:adv-term}, this implies that $\cosetue$ satisfies
\[\frac{1}{2}\sqrt{N\col(P)}=\frac{1}{2}\sqrt{\frac{2^{2n}}{(2^{2n}-1)(2^{2n-1}-1)}}\text{-unclonable indistinguishable security.}\]

Next we prove Claim~\ref{claim:conditional-cipher-state} to complete the proof.
\begin{proof}[Proof of Claim~\ref{claim:conditional-cipher-state}]
For any $(u,v)\in \keyspace$, let $\comp(u,v)$ denote the set of all subspace $A\leq \FF_2^{2n}$ with $dim(A)=n$ such that $u\in A$ and $v\in A^\perp$.

For a fixed $A\in\comp(u,v)$, let $\can(A)$ and
$\can(A^\perp)$ denote the sets of canonical representatives of
the cosets of $A$ and $A^\perp$, respectively, where a canonical representative of a coset is the lexicographically first vector in the coset, and also define
\begin{align*}
    &\comparitycoset((u,v),m,A)\\
    &:=
    \left\{
        (s,t)\in \can(A)\times \can(A^\perp)
        \;\middle|\;
        \langle v,s\rangle+\langle u,t\rangle=m
    \right\}.
\end{align*}

There are exactly $2^n$ cosets of $A$ and $2^n$ cosets of $A^\perp$ and the inner product of every vector in a fixed coset $A+s'$ with $v$ is the same as $\langle v,s\rangle$, and similarly, the inner product of every vector in a fixed coset $A^\perp+t'$ is the same as $\langle u,t\rangle$.
Hence, the map $\mathcal{F}:(s+A,t+A^\perp)\longmapsto\langle v,s\rangle+\langle u,t\rangle$ is well-defined on $\left(\FF_2^{2n}/A\right) \times \left(\FF_2^{2n}/A^\perp\right)$. Since for both $u$ and $v$, and any bit $b\in \bit$, there are exactly $2^{2n-1}$ vectors that have inner product $b$ with $u$ and $2^{2n-1}$ vectors that have inner product $b$ with $v$, it holds that $\mathcal{F}$ is a balanced function, i.e., has the same number of preimages of $0$ as that for $1$. Moreover, for any $m\in\bit$, the set $\comparitycoset((u,v),m,A)$ denotes the preimage set of $m$ under $\mathcal{F}$.
Hence
\begin{equation}\label{eq:comparity-counting}
    |\comparitycoset((u,v),m,A)|
    = \frac{\left|\frac{\FF_2^{2n}}{A}\times \frac{\FF_2^{2n}}{A^\perp}\right|}{2}=\frac{2^{2n}}{2}=  2^{2n-1}.
\end{equation}

Since every pair of cosets has the same number of vector representatives, and the parity is constant on each pair of cosets, sampling $s,t\getsr\FF_2^{2n}$ conditioned on $\langle v,s\rangle+\langle u,t\rangle=m$ induces the uniform distribution over $\comparitycoset((u,v),m,A)$.

For any fixed key $(u,v)\in\keyspace$, $A\in\comp(u,v)$, and
message $m\in\bit$, let $\rho_{ct}^{(u,v),m,A}$ denote the cipher state, conditioned on the sampled subspace being $A$. Hence, by the above conclusion, we can write
\[
    \rho_{ct}^{(u,v),m,A}
    =
    \frac{1}{
        |\comparitycoset((u,v),m,A)|
    }
    \sum_{
        (s,t)\in
        \comparitycoset((u,v),m,A)
    }
    \ketbra{A_{s,t}}{A_{s,t}}.
\]

    
     

    
    The ciphertext $\rho_{ct}^{(u,v),m}=\cosetue.\enc((u,v),m)$ can then be written as $\rho_{ct}^{(u,v),m}=\EE_{A\getsr \comp(u,v)}\rho_{ct}^{(u,v),m,A}.$
    Next, we will prove that for any fixed key $(u,v)\in \keyspace$, $A\in \comp(u,v)$ and message $m\in \bit$, it holds that
    \begin{equation}\label{eq:to-be-proven}
    \rho_{ct}^{(u,v),m,A}=\frac{1}{N}\left(I+(-1)^m\Theta_{(u,v)}\right),
    \end{equation}
    since that would imply
     \begin{align*}    \rho_{ct}^{(u,v),m}&=\EE_{A\getsr\comp(u,v)}\rho_{ct}^{(u,v),m,A}\\
     &=\EE_{A\getsr \comp(u,v)}\left[\frac{1}{N}\left(I+(-1)^m\Theta_{(u,v)}\right)\right]=\frac{1}{N}\left(I+(-1)^m\Theta_{(u,v)}\right),
    \end{align*}
    which proves the claim.  Next, we prove \Cref{eq:to-be-proven} to complete the proof of the claim.
    Fix a key $k=(u,v)$, $m\in \bit$ and $A\in \comp(k)=\comp(u,v)$ arbitrarily. Since $\Theta_k$ is Hermitian, and $\Theta_k^2=I$, it holds that $\Theta_k$ only has eigen values in $\{-1,+1\}$. For any bit $b\in\bit$, let the eigenspace of $(-1)^b$, i.e., the subspace of eigenvectors with eigenvalue $(-1)^b$ of $\Theta_k$, be denoted as $V^k_b$. Clearly, the Hilbert space $(\CC^2)^{\otimes 2n}$ can be written as the direct sum of $V^k_0$ and $V^k_1$. Since $\Tr(\Theta_k)=0$, this implies that $dim(V^k_0)=2^{2n-1}=dim(V^k_1)$. Moreover, for any bit $b\in \bit$, the projection on to the eigenspace $V^k_b$, denoted as $\Pi_{V^k_b}$ can be written as
    \begin{equation}\label{eq:Projector}
        \Pi_{V^k_b}=\frac{1}{2}\left(I+(-1)^b\Theta_k\right).
    \end{equation}

    Next, for every $(s,t)\in\comparitycoset((u,v),m,A)$, it holds that
    \begin{align}
        \Theta_k \ket{A_{s,t}}&=X^u Z^v\ket{A_{s,t}}=X^uZ^vX^sZ^t\ket{A}&\text{since $X^sZ^t\ket{A}=\ket{A_{s,t}}$.}\notag\\
        &=(-1)^{\langle v,s\rangle} X^uX^sZ^vZ^t\ket{A}=(-1)^{\langle v,s\rangle} X^{u+ s}Z^{v+ t}\ket{A}&\text{by Pauli anti-commutation.}\notag\\
        &=(-1)^{\langle v,s\rangle}\ket{A_{u+ s,v+ t}}. \label{eq:till-shifted-coset}
    \end{align}
    Moreover, note that
    \begin{align*}
        \ket{A_{u+ s,v+ t}}&=\frac{1}{\sqrt{|A|}}\sum_{a\in A}(-1)^{\langle v+ t,a\rangle}\ket{a+s+u}=\frac{1}{\sqrt{|A|}}\sum_{a\in A}(-1)^{\langle t,a\rangle}\ket{a+s+u}&\text{since $v\in A^\perp$.}\\
        &=\frac{1}{\sqrt{|A|}}\sum_{a'\in A}(-1)^{\langle u,t\rangle}(-1)^{\langle a',t\rangle}\ket{a'+s}&\text{where $a'=u+a$, $a=u+a'$.}\\
        &=\frac{(-1)^{\langle u,t\rangle}}{\sqrt{|A|}}\sum_{a'\in A}(-1)^{\langle a',t\rangle}\ket{a'+s}=(-1)^{\langle u,t\rangle}\ket{A_{s,t}}.
    \end{align*}
 Substituting the last equation in \Cref{eq:till-shifted-coset}, we get
 \[ \Theta_k \ket{A_{s,t}}=(-1)^{\langle v,s\rangle}\ket{A_{u+ s,v+ t}}=(-1)^{\langle v,s\rangle}(-1)^{\langle u,t\rangle}\ket{A_{s,t}}=(-1)^{\langle v,s\rangle + \langle u,t\rangle}\ket{A_{s,t}}=(-1)^m\ket{A_{s,t}}.\]

    Therefore the set of vectors $\{\ket{A_{s,t}}\}_{(s,t)\in \comparitycoset(k,m,A)}\subset V^k_m$. But clearly, for any two pairs of $(s,t),(s',t')\in \comparitycoset(k,m,A)$, such that $(s,t)\neq (s',t')$, $\ip{A_{s,t}}{A_{s',t'}}=0$. This follows from the fact that if $s\neq s'$ then they are simply supported on different cosets of $A$, and hence on different subsets of computational basis elements. Similarly, if $t\neq t'$, then the same holds after applying QFT to both states, as $QFT\ket{A_{s,t}}=\ket{A^\perp_{t,s}}$ and $QFT\ket{A_{s',t'}}=\ket{A^\perp_{t',s'}}$. Combined with the fact that each element in $\{\ket{A_{s,t}}\}_{(s,t)\in \comparitycoset(k,m,A)}$ is a normalized state, we conclude that $\{\ket{A_{s,t}}\}_{(s,t)\in \comparitycoset(k,m,A)}$ forms an orthonormal set inside $V^k_m$. However, $dim(V^k_m)=2^{2n-1}=|\comparitycoset((u,v),m,A)|=|\comparitycoset(k,m,A)|$ (see \Cref{eq:comparity-counting}). 
    Therefore it must hold that $\{\ket{A_{s,t}}\}_{(s,t)\in \comparitycoset(k,m,A)}$ is an orthonormal basis for $V^k_m$. Hence, it holds that
    \begin{equation}\label{eq:projector-orthonormal-basis}
    \Pi_{V^k_m}=\sum_{(s,t)\in \comparitycoset(k,m,A)}\ketbra{A_{s,t}}{A_{s,t}}.
    \end{equation} 
    Combined with \Cref{eq:comparity-counting,eq:Projector,eq:projector-orthonormal-basis} implies
    \begin{align*}
    \rho_{ct}^{k,m,A}
    &=\frac{1}{|\comparitycoset(k,m,A)|}\sum_{(s,t)\in\comparitycoset(k,m,A)}\ketbra{A_{s,t}}{A_{s,t}}\\
    &=\frac{1}{2^{2n-1}}\Pi_{V^k_m}&\text{by \Cref{eq:comparity-counting,eq:projector-orthonormal-basis}}\\
    &=\frac{2}{N}\Pi_{V^k_m}&\text{by definition of $N$.}\\
    &=\frac{2}{N}\cdot \frac{1}{2}\left[I+(-1)^m\Theta_k\right]=\frac{1}{N}\left[I+(-1)^m\Theta_k\right]&\text{by \Cref{eq:Projector}},
    \end{align*}
    which completes the proof of \Cref{eq:to-be-proven} and hence the claim.
\end{proof}
\end{proof}

\ifconference
\else

\input{moe/it-moe-general-ue-proof}

\fi

%% file: moe/it-moe-general-ue-proof.tex
\ifconference
\section{Proof of the lemmas in \texorpdfstring{\Cref{sec:it-moe}}{Section 5}}\label{sec:lem-it-moe}
We first prove \Cref{lem:dcmoe-ue}, and then proceed to proving \Cref{lem:general-UE} in the next subsection.

\begin{proof}[Proof of \Cref{lem:dcmoe-ue}]
    For every fixed nonzero orthogonal pair $(u,v)\in \FF_2^{2n}\times\FF_2^{2n}$, let the set of $n$-dimensional subspaces $A$ satisfying $u\in A$ and $v\in A^\perp$ be defined as $\comp(u,v)$. It is easy to see that $|\comp(u,v)|$ depends only on $n$ and is independent of the choice of $(u,v)$. Hence the marginal distribution on the pair $(u,v)$ induced by the $\dcmoe$ experiment is uniform over the key space of $\cosetue$, given as $\keyspace:=\{(u,v)\mid u,v \in \FF_2^{2n}\setminus\{0\}, \langle u,v\rangle =0\}$. Thus sampling $A$ uniformly at random, and then sampling $u\getsr A\setminus\{0\}, v\getsr A^\perp\setminus\{0\}$, is the same as first sampling a uniformly random nonzero orthogonal pair $(u,v)\in \FF_2^{2n}\times\FF_2^{2n}$, and then sampling $A\in \comp(u,v)$, and therefore the marginal distribution of $(A,u,v))$ in the \dcmoe experiment and the unclonable-indistinguishable experiment (see \Cref{def:UE-security}) for \cosetue, is the same. Next, for a fixed non-zero orthogonal pair $(u,v)$ and a fixed subspace $A\in \comp(u,v)$, the map \[(s,t)\longmapsto\langle u,t\rangle+\langle v,s\rangle\] is a nonzero linear function and hence is balanced; therefore, sampling $s,t$ uniformly and setting the message to this parity is equivalent to first sampling a uniform message $m$ and then sampling $s,t$ uniformly conditioned on \[m=\langle u,t\rangle+\langle v,s\rangle.\] Therefore, conditioned on $A,(u,v)$, the distribution of the subspace $s,t$ and the bit $\langle u,t\rangle+ \langle v, s\rangle$ in the \dcmoe experiment, is the same as that of $s,t$ and the message bit $m$ in the unclonable indistinguishable experiment for \cosetue, which concludes the proof of the lemma.
\end{proof}

\subsection{Unclonable indistinguishable security of the general class of encryption schemes (Proof of \texorpdfstring{\Cref{lem:general-UE}}{Lemma 2})}\label{sec:general-UE-proof}
\else
\subsection{Unclonable indistinguishable security of the general class of encryption schemes (Proof of \texorpdfstring{\Cref{lem:general-UE}}{Lemma 2})}\label{sec:general-UE-proof}
\fi
We first note that the scheme is perfectly correct. Since $\Theta_k$ is a Hermitian operator, with $\Theta_k^2=I$, all its eigenvalues are in $\{+1,-1\}$. Hence
\[
    I+(-1)^m\Theta_k\succeq 0.
\]
Moreover, since $\tr[\Theta_k]=0$,
\[
    \tr\left[\rho^{k,m}_{ct}\right]=\frac{1}{N}\left( \tr[I]+(-1)^m\tr[\Theta_k]\right)=1.
\]
Thus $\rho^{k,m}_{ct}$ is a valid density operator. Furthermore, it is supported entirely on the $(-1)^m$-eigenspace of $\Theta_k$, and hence the decryption measurement outputs $m$ with probability one.


Next, we prove unclonable-indistinguishability security.
We closely follow the notations and the proof template of~\cite{AS26}. In particular, we reprove the Parts 2 and 3 of their main security proof~\cite[Theorem 3.2]{AS26} for the general case of arbitrary Hermitian operators, and arbitrary distributions on the key space. Next, we recall the following results from~\cite{AS26} that will be required for our proof.  

Let $\|\|_\infty$ denote the operator norm, i.e., for any operator $M$, we denote $\|M\|_\infty:=\sup_{|\psi|=1}\|M\ket{\psi}\|$, where $\|\|$ is the $L_2$ norm for vectors. Similarly, or any Hermitian operator $M$, that can be spectral decomposed as $M=\sum_{j}\lambda_j\ketbra{v_j}{v_j}$, we denote $M_+:=\sum_{j}\max(\lambda_j,0)\ketbra{v_j}{v_j}$.
\begin{theorem}[Normalized Choi identity,
{\cite[Theorem~2.1]{AS26}}]
\label{thm:AS-choi}
Let
$\Phi:   \mathcal L(\mathcal H_{\regeve})\rightarrow\mathcal L(\mathcal H_{\regbob}\otimes  \mathcal H_{\regcharlie})$ be a quantum channel, where
$\dim(\mathcal H_{\regeve})=N$. Let
$\mathcal H_{\regalice}\simeq\mathcal H_{\regeve}$, and define \[\ket{\Omega}_{\regalice\regeve}:=\frac{1}{\sqrt N}\sum_{i=1}^{N}\ket{i}_{\regalice}\ket{i}_{\regeve}.\]
Then, for every
$M\in\mathcal L(\mathcal H_{\regeve})$ and
$Q\in
\mathcal L(
    \mathcal H_{\regbob}\otimes\mathcal H_{\regcharlie}
)$,
\[
    \frac{1}{N}
    \tr[
        Q\Phi(M)
    ]
    =
    \tr\left[
        (M^T\otimes Q)
        (I_{\regalice}\otimes\Phi)
        (\ket{\Omega}\!\bra{\Omega})
    \right],
\]
where the transpose is taken in the basis used to define
$\ket{\Omega}$.
\end{theorem}

\begin{lemma}[Conditional overlap,
{\cite[Lemma~3.3]{AS26}}]
\label{lem:AS-conditional-overlap}
Let
$ U_{\regalice\regbob}\in \mathcal L(\mathcal H_{\regalice}\otimes \mathcal H_{\regbob})$
and $V_{\regalice\regcharlie}\in\mathcal L(\mathcal H_{\regalice}\otimes\mathcal H_{\regcharlie}).$

Viewing them canonically (by tensoring with identity operators) as operators on
$\mathcal H_{\regalice}\otimes
 \mathcal H_{\regbob}\otimes
 \mathcal H_{\regcharlie}$, we have
\[
    \|UV\|_\infty
    \leq
    \sqrt{
        \left\|
            \tr_{\regalice}[U^\dagger U]
        \right\|_\infty
    }
    \sqrt{
        \left\|
            \tr_{\regalice}[VV^\dagger]
        \right\|_\infty
    }.
\]
\end{lemma}

Next, we abstract out an operator bound implicit in the proof of \cite[Proposition~3.8]{AS26}.

\begin{lemma}[{Abstracted from the Positive-part bound,~\cite[Proposition~3.8]{AS26}}]
\label{lem:AS-positive-part}
Let $E_B,E_C,E_{BC}$ be Hermitian operators on a
finite-dimensional Hilbert space, and define the Hermitian operator
\[
    G
    :=
    \frac{1}{2}
    \left(
        E_B+E_C+E_{BC}-I
    \right).
\]
Suppose that, 
\begin{enumerate}
\item for every $a,c\in\{+1,-1\}$,
$\Gamma_{a,c}
    :=
    \frac{1}{4}
    \left(
        I+aE_B+cE_C+acE_{BC}
    \right)
    \succeq0,$
and
\item there exists $\gamma\geq0$ such that for every integer
$\ell\geq0$,
\[
    \left\|
        E_BE_{BC}^{\ell}E_C
    \right\|_\infty
    \leq
    \gamma.
\]
\end{enumerate}
Then, given the spectral decomposition $G=\sum_{j}\lambda_j\ketbra{v_j}{v_j}$ and $G_+:=\sum_{j}\max(\lambda_j,0)\ketbra{v_j}{v_j}$, it holds that
\[
    \|G_+\|_\infty^2
    \leq
    \gamma.
\]
\end{lemma}

We defer the proof of \Cref{lem:AS-positive-part}  to \Cref{app:AS-positive-part}, which follows from similar arguments as used in the proof of \cite[Proposition~3.8]{AS26}. Next we complete the security proof assuming \Cref{lem:AS-positive-part}.


Fix an
arbitrary adversarial triplet $
    \advue=(\alice,\bob,\charlie).$

Let $\Phi:
    \mathcal L((\CC^2)^{\otimes \ell})
    \rightarrow
    \mathcal L(
        \mathcal H_{\regbob}\otimes
        \mathcal H_{\regcharlie})$
denote the splitting channel applied by $\alice$ to produce the registers $\regbob$ and $\regcharlie$ for $\bob$ and $\charlie$ respectively, before the key is revealed to them.

Next, for every key $k\in\mathcal K$, let
\[
    \{M^k_{\bob,0},M^k_{\bob,1}\}
    \qquad\text{and}\qquad
    \{M^k_{\charlie,0},M^k_{\charlie,1}\}
\] denote the binary POVM corresponding to the measurements applied by $\bob$ and $\charlie$ respectively,
after receiving $k$, and define
\[
    B_k:=M^k_{\bob,0}-M^k_{\bob,1},
    \qquad
    C_k:=M^k_{\charlie,0}-M^k_{\charlie,1}.
\]
It is easy to see that $B_k$ and $C_k$ are Hermitian contractions, and we can rewrite
\[
    M^k_{\bob,m}
    =
    \frac{I+(-1)^mB_k}{2},
    \qquad
    M^k_{\charlie,m}
    =
    \frac{I+(-1)^mC_k}{2}.
\]

Let $\ket{\Omega}_{\regalice\regeve}
    :=
    \frac{1}{\sqrt N}
    \sum_{i=1}^{N}
    \ket{i}_{\regalice}\ket{i}_{\regeve},$ 
where $\regeve$ is isomorphic to the ciphertext register $\regalice$, and
let
\[
    \tau_{\regalice\regbob\regcharlie}
    :=
    (I_{\regalice}\otimes\Phi_{\regeve})
    \left(
        \ket{\Omega}\!\bra{\Omega}
    \right)
\]
be the normalized Choi state of $\Phi$.

For every $k$, define
$\widehat{\Theta}_k
    :=
    (\Theta_k^T)_{\regalice},$ and define
Define
\begin{align}
    E_B
    &:=
    \sum_{k\in\mathcal K}
    p_k\,
    \widehat{\Theta}_k
    \otimes B_k\otimes I,
    \label{eq:EB-general}
    \\
    E_C
    &:=
    \sum_{k\in\mathcal K}
    p_k\,
    \widehat{\Theta}_k
    \otimes I\otimes C_k,
    \label{eq:EC-general}
    \\
    E_{BC}
    &:=
    \sum_{k\in\mathcal K}
    p_k\,
    I\otimes B_k\otimes C_k.
    \label{eq:EBC-general}
\end{align}
Also define
\[
    G
    :=
    \frac{1}{2}
    \left(
        E_B+E_C+E_{BC}-I
    \right).
\]

By \Cref{thm:AS-choi},
\begin{align}
    &\Pr[
        b_{\bob}=b_{\charlie}=m
    ]
    \notag\\
    &=
    \mathbb E_{k\gets P,\;m\getsr\{0,1\}}
    \tr\Bigg[
        \left(
            I+(-1)^m\widehat{\Theta}_k
        \right)
        \otimes
        \frac{I+(-1)^mB_k}{2}
        \otimes
        \frac{I+(-1)^mC_k}{2}
        \;\tau
    \Bigg].
    \label{eq:win-choi-general}
\end{align}
Expanding the three factors and averaging over the uniform bit
$m$ removes all terms containing an odd power of $(-1)^m$.
Hence
\begin{align}
    \Pr[
        b_{\bob}=b_{\charlie}=m
    ]
    &=
    \frac{1}{4}
    \left(
        1+
        \tr[E_B\tau]+
        \tr[E_C\tau]+
        \tr[E_{BC}\tau]
    \right)
    \notag\\
    &=
    \frac{1}{2}
    \left(
        1+\tr[G\tau]
    \right)
    \notag\\
    &\leq
    \frac{1}{2}
    \left(
        1+\|G_+\|_\infty
    \right),
    \label{eq:win-G-general}
\end{align}
where the last inequality follows from
$ G\preceq G_+
    \preceq
    \|G_+\|_\infty I$, 
and the fact that $\tau$ is a density operator.

We now show that the operator bounds used in~\cite{AS26}
continue to hold when the key is sampled from the possibly
nonuniform distribution $P$.

Let    $\gamma
    :=
    N\sum_{k\in\mathcal K}p_k^2
    =
    N\col(P),$ and define
\[
    U_{\regalice\regbob}
    :=
    \sum_{k\in\mathcal K}
    p_k\,
    \widehat{\Theta}_k\otimes B_k
,
    V_{\regalice\regcharlie}
    :=
    \sum_{k\in\mathcal K}
    p_k\,
    \widehat{\Theta}_k\otimes C_k.
\]
Their canonical extensions to the three registers are
$E_B$ and $E_C$, respectively.

Using the Hilbert--Schmidt orthogonality of
$\{\Theta_k\}_{k\in\mathcal K}$, we obtain
\begin{align}
    \tr_{\regalice}
    \left[
        U_{\regalice\regbob}^\dagger
        U_{\regalice\regbob}
    \right]
    &=
    N\sum_{k\in\mathcal K}
    p_k^2 B_k^2
    \preceq
    \gamma I_{\regbob},
    \label{eq:EB-general-bound}
    \\
    \tr_{\regalice}
    \left[
        V_{\regalice\regcharlie}
        V_{\regalice\regcharlie}^\dagger
    \right]
    &=
    N\sum_{k\in\mathcal K}
    p_k^2 C_k^2
    \preceq
    \gamma I_{\regcharlie}.
    \label{eq:EC-general-bound}
\end{align}
Hence, by \Cref{lem:AS-conditional-overlap},
\begin{equation}
    \|E_BE_C\|_\infty
    \leq
    \gamma.
    \label{eq:endpoint-general}
\end{equation}

We next verify the corresponding extension of
\cite[Proposition~3.6]{AS26}. For every integer $\ell\geq0$,
\begin{equation}
    \left\|
        E_BE_{BC}^{\ell}E_C
    \right\|_\infty
    \leq
    \gamma.
    \label{eq:moment-general}
\end{equation}
For $\ell=0$, this follows from
\cref{eq:endpoint-general}. For $\ell\geq1$, expanding
$E_{BC}^{\ell}$ gives
\begin{align}
    E_BE_{BC}^{\ell}E_C
    &=
    \sum_{k_1,\ldots,k_\ell}
    p_{k_1}\cdots p_{k_\ell}
    \left(
        I\otimes I\otimes
        C_{k_1}\cdots C_{k_\ell}
    \right)
    E_BE_C
    \left(
        I\otimes
        B_{k_1}\cdots B_{k_\ell}
        \otimes I
    \right).
    \label{eq:moment-expansion-general}
\end{align}
The two exterior operators are contractions, while the
nonnegative coefficients sum to one. Therefore
\cref{eq:moment-general} follows from
\cref{eq:endpoint-general}.

Finally, for $a,c\in\{+1,-1\}$, define
\[
    \Gamma_{a,c}
    :=
    \frac{1}{4}
    \left(
        I+aE_B+cE_C+acE_{BC}
    \right).
\]
For every fixed $k$, the operators
$\widehat{\Theta}_k\otimes B_k\otimes I$  and 
    $\widehat{\Theta}_k\otimes I\otimes C_k$
are commuting Hermitian contractions, and their product is
    $I\otimes B_k\otimes C_k.$
Therefore
\[
    \frac{1}{4}
    \left(
        I+
        a\widehat{\Theta}_k\otimes B_k\otimes I
    \right)
    \left(
        I+
        c\widehat{\Theta}_k\otimes I\otimes C_k
    \right)
    \succeq0.
\]
Averaging with weights $p_k$ shows that
\begin{equation}
    \Gamma_{a,c}\succeq0
    \qquad
    \text{for every }a,c\in\{+1,-1\}.
    \label{eq:Gamma-general}
\end{equation}

By
\cref{eq:moment-general,eq:Gamma-general},
the hypotheses of \Cref{lem:AS-positive-part} hold with
$\gamma=N\col(P).$
Therefore
\[
    \|G_+\|_\infty^2
    \leq
    N\col(P).
\]
Combining this with~\cref{eq:win-G-general} gives
$\Pr[
        b_{\bob}=b_{\charlie}=m
    ]
    \leq
    \frac{1}{2}
    +
    \frac{1}{2}
    \sqrt{N\col(P)},$
which completes the proof of the lemma.

\subsubsection{Proof of \texorpdfstring{\Cref{lem:AS-positive-part}}{Lemma 8}}\label{app:AS-positive-part}

\begin{proof}[Proof of \Cref{lem:AS-positive-part}]
The proof follows the same arguments as in the proof of
\cite[Proposition~3.8]{AS26}.

Let $t:=\|G_+\|_\infty.$
If $t=0$, there is nothing to prove. Suppose $t>0$. Since the
Hilbert space is finite dimensional and $G$ is Hermitian,
there exists a unit vector $\ket{\psi}$ such that
\[
    G\ket{\psi}=t\ket{\psi}.
\]

First note that the positivity of the four operators $\Gamma_{a,c}$ implies that
\[-I  \preceq E_B,E_C,E_{BC} \preceq I.
\]
Indeed, suitable pairs of the four operators
$\Gamma_{a,c}$ sum to
\[
    \frac{1}{2}(I\pm E_B),
    \qquad
    \frac{1}{2}(I\pm E_C),
    \qquad
    \frac{1}{2}(I\pm E_{BC}).
\]

Define $\Delta:=I-E_{BC}, \qquad R:=E_B-E_C.$
Then it is easy to see that $\Delta\succeq0$. Moreover, positivity of
$\Gamma_{+1,-1}$ and $\Gamma_{-1,+1}$ gives
\[\Delta+R\succeq0,  \qquad \Delta-R\succeq0,\]
and hence
\begin{equation}
    -\Delta
    \preceq
    R
    \preceq
    \Delta.
    \label{eq:AS-R-order}
\end{equation}

Next, the eigenvector equation for $G$ gives
\begin{equation}
    (E_B+E_C)\ket{\psi}
    =
    (2tI+\Delta)\ket{\psi}.
    \label{eq:AS-eigenvector}
\end{equation}

Define $F:=2t(2tI+\Delta)^{-1}.$
Since $t>0$ and $\Delta\succeq0$, $F$ is well defined and
\[
    0\prec F\preceq I,
    \qquad
    F(2tI+\Delta)=2tI.
\]

We will first show that
\begin{equation}
    RFR
    \preceq
    2t\Delta.
    \label{eq:AS-RFR}
\end{equation}
For arbitrary vectors $\ket{\xi},\ket{\eta}$, positivity of
$\Delta+R$ and $\Delta-R$ gives
\begin{align}
    0&\leq \frac{1}{2} \bra{\xi+\eta} (\Delta+R)\ket{\xi+\eta}
    +
    \frac{1}{2}\bra{\xi-\eta}(\Delta-R)\ket{\xi-\eta}
    +
    2t\|\eta\|^2\notag\\
    &=
    \bra{\xi}\Delta\ket{\xi}
    +
    2\operatorname{Re}
    \bra{\xi}R\ket{\eta}
    +
    \bra{\eta} (\Delta+2tI)\ket{\eta}.
    \label{eq:AS-completion}
\end{align}
Choosing $\ket{\eta}=-(\Delta+2tI)^{-1}R\ket{\xi}$ in \cref{eq:AS-completion} yields
\[
    R(\Delta+2tI)^{-1}R
    \preceq
    \Delta.
\]
Multiplying by $2t$ proves
\cref{eq:AS-RFR}.

For any vector $\ket{w}$, define
\[
    Q_F(w)
    :=
    \langle w|F|w\rangle.
\]
Using \cref{eq:AS-eigenvector} and
$F(2tI+\Delta)=2tI$, we have
\begin{align}
    Q_F\bigl((E_B+E_C)\psi\bigr)
    &=
    4t^2
    +
    2t
    \langle\psi|\Delta|\psi\rangle.
    \label{eq:AS-plus}
\end{align}
On the other hand, by \cref{eq:AS-RFR},
\begin{align}
    Q_F\bigl((E_B-E_C)\psi\bigr)
    &=
    \langle\psi|RFR|\psi\rangle
    \notag\\
    &\leq
    2t
    \langle\psi|\Delta|\psi\rangle.
    \label{eq:AS-minus}
\end{align}
Subtracting \cref{eq:AS-minus} from
\cref{eq:AS-plus} and expanding the cross terms gives
\begin{equation}
    t^2
    \leq
    \operatorname{Re}
    \langle\psi|
        E_BFE_C
    |\psi\rangle
    \leq
    \|E_BFE_C\|_\infty.
    \label{eq:AS-real-part}
\end{equation}

It remains to upper bound the last norm. Set
$q:=(1+2t)^{-1}\in(0,1).$
Since $\Delta=I-E_{BC}$,
\[
    F
    =
    (1-q)(I-qE_{BC})^{-1}.
\]
For an integer $L\geq1$, define
$F_L   :=(1-q)\sum_{\ell=0}^{L-1}q^\ell E_{BC}^{\ell}.$
By the hypothesis of the lemma,
\begin{align}
    \|E_BF_LE_C\|_\infty
    &\leq
    (1-q)
    \sum_{\ell=0}^{L-1}
    q^\ell
    \left\|
        E_BE_{BC}^{\ell}E_C
    \right\|_\infty
    \notag\\
    &\leq
    (1-q)
    \sum_{\ell=0}^{L-1}
    q^\ell\gamma
    \notag\\
    &=
    (1-q^L)\gamma.
    \label{eq:AS-truncated}
\end{align}

Moreover,
$F-F_L=q^L E_{BC}^L F.$

As observed above,
\[\|E_B\|_\infty,\|E_C\|_\infty,\|E_{BC}\|_\infty,\|F\|_\infty\leq1.\]
Therefore
\[
    \|E_B(F-F_L)E_C\|_\infty=\|E_B(q^L E_{BC}^L F)E_C\|_\infty= q^L\|E_BE_{BC}^L FE_C\|_\infty\leq q^L.
\]
Letting $L\rightarrow\infty$ in conjunction with
\cref{eq:AS-truncated} gives
\[\|E_BFE_C\|_\infty \leq \gamma.\]
Combining this with \cref{eq:AS-real-part}, we obtain $t^2\leq\gamma.$
Since $t=\|G_+\|_\infty$, this concludes the proof of the lemma.
\end{proof}

%% file: 000-include/boilerplate/notation.tex
\subsection{Notation and Conventions}\label{sec:notation}
 We refer the reader to \cite{Goldreich_2001,Goldreich_2004} for a preliminary on cryptography and to \cite{NC10} for a preliminary on quantum information and computation. We follow the notations and conventions commonly used in (quantum) cryptography, theoretical computer science and mathematics.  In this section, we non-exhaustively highlight some of them. All of our conventions are implicitly followed unless otherwise specified, e.g. our cryptographic assumptions are always post-quantum even though most of the time we will not write this explicitly, and if there is a rare occasion where we assume classical security only, we will specify this explicitly.

\paragraph{Functions} For a function $f: \nat^+ \to \nat^+$, we will say that it is an efficiently computable function if it is polynomially bounded, that is, if there is some $n_0 \in \nat^+$  and $c \in \R^+$ such that $f(n) \leq n^c$ for all $n > n_0$. Sometimes we will say polynomial to also mean polynomially bounded instead (e.g. we will say polynomial to refer to $O(\poly(\lambda))$, which includes, say, $\sqrt{n}$) - the precise meaning will be clear from context. $f(\lambda)$ is said to be negligible if for any $c \in \nat^+$, there is $n_0 \in \nat^+$ such that $f(\lambda) \leq \frac{1}{n^c}$ for all $n \geq n_0$. $f(\lambda)$ is said to superlogarithmic if $f = \omega(\log(\lambda))$ or equivalently, if $2^{-f(\lambda)}$ is negligible.

When we describe a function, such as a polynomial, if no input is mentioned, it will implicitly be the security parameter $\lambda$.

\paragraph{Classical and Quantum Information and Oracles} Variables written in lower-case,  such as $pk,sk$ or $\pk, \sk$, are classical objects. Variables written in the font $\qsk$ are quantum variables. We write $\regi$ to mean a quantum register, which keeps a quantum state that will evolve when the register is acted on, and it can be entangled with other registers. 

We write $\regi \samp \rho$ to mean that the register $\regi$ is initialized with a sample from the quantum distribution (i.e. mixed state) $\rho$. For simplicity, we will usually write non-normalized versions of the quantum states, however, they are always implicitly normalized.

For mixed states $\rho$ and $\sigma$, we write $\trd{\rho}{\sigma}:=\frac{1}{2}\norm{\rho-\sigma}_1$ for their trace distance.

We write $\ora$ to denote a classical oracle, and write $\adve^{\ora}$ to denote an algorithm that has quantum query-access to $\ora$. Similarly, for unitary oracles we write $U$ and $\adve^U$. All adversaries always have quantum access to all of the oracles.

\paragraph{Distributions} When $S$ is a set, we write $x \samp S$ to mean that $x$ is sampled uniformly at random from $S$. When $\mathcal{D}$ is a distribution, we write $x \samp \mathcal{D}$ to mean that $x$ is sampled from $\mathcal{D}$. Finally, we write $x \samp \mathcal{B}(\regi)$ or $x \samp \mathcal{B}(a)$ to mean that $x$ is a sample as output by the (quantum or classical randomized) algorithm $\mathcal{B}$ run on the quantum input register $\regi$ or classical input $a$. We use similar notation for quantum registers, e.g., $\regi \samp \mathcal{B}(a)$. 

\paragraph{Adversaries and Constructions/Algorithms} We write \emph{QPT} to mean \emph{quantum polynomial time} and \emph{PPT} to mean \emph{probabilistic (classical) polynomial time}. We say query-bounded to mean that the algorithm (e.g. construction, adversary) makes polynomially many queries to the oracle, but it can possibly take unbounded time outside the queries. All of these are polynomial in the security parameter $\lambda$. Note that for QPT and PPT, this also means that the inputs to these algorithms are also of polynomial size in $\lambda$. 

We say \emph{efficient} to mean QPT or PPT (will be clear from context) in the plain model, or query-bounded in the oracle model. For constructions, efficient means QPT or PPT, even in the oracle model (though the oracles themselves might not be efficiently sampleable or implementable).

In the plain model, adversaries will be implicitly QPT. In the oracle security model, adversaries will be implicitly \emph{query-bounded}. An admissible adversary means an adversary that conforms to the (implicit) requirements of the security game (e.g. polynomial time or query bounded, has the correct interaction model and has correct input-output format and so on). All of our cryptographical assumptions are post-quantum, e.g., \emph{one-way functions} means \emph{post-quantum secure one-way functions}. Adversaries are stateful uniform QPT.  

When we say that a primitive is subexponentially secure, we mean that any  polynomial time adversary has subexponentially small advantage. Sometimes we require security against subexponential time adversaries, in which case we will specify this explicitly.

All algorithms of a cryptographic scheme are uniform (stateless) PPT or QPT (will be clear from context).

\paragraph{Security Games} A security game is a binary random variable which is the outcome of an experiment between a \emph{challenger} and an adversary. For a security game, we say that the adversary has won if the experiment output is $1$, and otherwise we say that the adversary has lost. When we say that an adversary has negligible advantage, we mean that the probability of the adversary winning (or distinguishing, depending on context) is upper bounded by a negligible function of $\lambda$.

\paragraph{Domains} We write $\messpa$ to denote the message space for various cryptographic primitives, and it will be equal to $\zo^{p(\lambda)}$ for some polynomial $p(\cdot)$ (will be clear from context).

%% file: 000-include/boilerplate/crypto-definitions.tex
\newcommand{\prelimlev}{\subsection}

\prelimlev{Pseudorandom Functions}
In this section, we recall pseudorandom functions (PRF) and puncturable PRFs.
\begin{definition}[Pseudorandom Functions]\label{predef:prf}
    A pseudorandom function (PRF) scheme $\prf$ is a family of functions $\{F: \zo^{c(
\lambda)} \times \zo^{m(\lambda)} \to \zo^{n(\lambda)}\}_{\lambda \in \nat^+}$ along with the following efficient algorithms.
    \begin{itemize}
        \item $\prf.\gen(1^\lambda):$ Takes in the unary representation of the security parameter and outputs a key $K$ in $\zo^{c(\lambda)}$.
        \item $\prf.\ceval(K, x):$ Takes in a key $K \in \zo^{c(\lambda)}$ and an input $x \in \zo^{m(\lambda)}$, outputs an evaluation of the PRF in $\zo^{n(\lambda)}$.
    \end{itemize} 
    We require the following.
    \paragraph{Correctness.}
    \begin{equation*}
        \Pr[\forall x~\prf.\ceval(K, x) = F(K, x): \begin{array}{c}
              K \samp \prf.\gen(1^\lambda)
        \end{array}] = 1.
    \end{equation*}
    \paragraph{Security} We require that for any QPT adversary $\adve$,
    \begin{equation*}
        \left|\Pr_{H \samp \mathrm{H}}[\adve^{H}(1^\lambda) = 1] - \Pr_{K \samp \prf.\gen(1^\lambda)}[\adve^{F(K,\cdot)}(1^\lambda) = 1] \right| \leq \negl(\lambda).
    \end{equation*}
    where $\mathrm{H}$ is the uniform distribution over the functions $\zo^{m(\lambda)} \to \zo^{n(\lambda)}$.
\end{definition}
 We will usually write $F_K(x)$ or $F(K, x)$ in algorithms and this will implicitly mean $\prf.\ceval(K, x)$.

\begin{definition}[Puncturable PRF]
\label{predef:puncprf}
A puncturable PRF scheme $\prf$ is a PRF scheme together with the
following efficient algorithm.
\begin{itemize}
    \item
    $K_{\{\mathcal S\}}\gets
    \prf.\mathsf{Puncture}(K,\mathcal S)$ takes a key
    $K\in\zo^{c(\lambda)}$ and a polynomial-size set
    $\mathcal S\subseteq\zo^{m(\lambda)}$, and outputs a punctured key
    $K_{\{\mathcal S\}}$.
\end{itemize}

We require the following properties.

\paragraph{Punctured correctness.}
For every polynomial-size set
$\mathcal S\subseteq\zo^{m(\lambda)}$ and every
$x\in\zo^{m(\lambda)}\setminus\mathcal S$,
\[
\Pr\!\left[
F_{K_{\{\mathcal S\}}}(x)=F_K(x):
K\gets\prf.\gen(1^\lambda),\
K_{\{\mathcal S\}}\gets
\prf.\mathsf{Puncture}(K,\mathcal S)
\right]=1.
\]

\paragraph{Pseudorandom at punctured point.}
For every polynomial-size set
$\mathcal S\subseteq\zo^{m(\lambda)}$ and every QPT distinguisher
$\adve$,
\[
\begin{aligned}
\Bigl|
&\Pr\!\left[
\adve\!\left(
F_{K_{\{\mathcal S\}}},
\{F_K(x)\}_{x\in\mathcal S}
\right)=1
\right]\\
&\qquad -
\Pr\!\left[
\adve\!\left(
F_{K_{\{\mathcal S\}}},
(\U_{n(\lambda)})^{|\mathcal S|}
\right)=1
\right]
\Bigr|
\leq\negl(\lambda),
\end{aligned}
\]
where in both probabilities
\[
K\gets\prf.\gen(1^\lambda),
\qquad
K_{\{\mathcal S\}}\gets
\prf.\mathsf{Puncture}(K,\mathcal S),
\]
and $\U_{n(\lambda)}$ is the uniform distribution over
$\zo^{n(\lambda)}$.

If
$\mathcal S=\{x_1^\star,\ldots,x_q^\star\}$, we write
$K_{\{x_1^\star,\ldots,x_q^\star\}}$ for
$K_{\{\mathcal S\}}$. For brevity, we sometimes write
$K_{\mathcal S}$ as simply $K_{\{\mathcal S\}}$.
\end{definition}

It is easy to see that any puncturable PRF scheme that satisfies puncturing security also satisfies usual PRF security.

\begin{theorem}[\cite{STOC:SahWat14,FOCS:Zhandry12}]\label{prethm:puncprfexists}
Let $n(\cdot), m(\cdot)$ be polynomially bounded.
\begin{itemize}
    \item If (post-quantum) one-way functions exist, then there exists a (post-quantum) puncturable PRF with input space $\zo^{m(\lambda)}$ and output space $\zo^{n(\lambda)}$.

\item If subexponentially-secure (post-quantum) one-way functions exist, then for any $c > 0$, there exists a (post-quantum) $2^{-\lambda^c}$-secure\footnote{While the original results are for negligible security against polynomial time adversaries, it is easy to see that they carry over to subexponential security. Further, by scaling the security parameter by a polynomial and simple input/output conversions, subexponentially secure (for any exponent $c'$) one-way functions is sufficient to construct for any $c$ a puncturable PRF that is $2^{-\lambda^c}$-secure.} puncturable PRF with input space $\zo^{m(\lambda)}$ and output space $\zo^{n(\lambda)}$.
\end{itemize}
\end{theorem}

\prelimlev{Digital Signature Schemes}
In this section we recall the definition of signatures schemes.

\begin{definition}\label{predef:digsig}
A digital signature scheme with message space $\messpa$ consists of the following efficient algorithms that satisfy the correctness and security guarantees below.
\begin{itemize}
    \item $\setup(1^\lambda):$ Takes the unary representation of the security parameter and outputs a signing key $sk$ and a verification key $vk$.
    \item $\sign(sk, m):$ Takes a signing key $sk$ and a message $m \in \messpa$, returns a signature for $m$.
    \item $\ver(vk, m, sig):$ Takes the public verification key $vk$, a message $m \in \messpa$ and an alleged signature $sig$ for $m$, outputs $1$ if $s$ is a valid signature for $m$.
\end{itemize}
\paragraph{Correctness}
We require the following for all messages $m \in \messpa$.
\begin{equation*}
    \Pr[\ver(vk, m, s) = 1 : \begin{array}{c}
         sk, vk \samp \mathsf{Setup}(1^\lambda) \\
         s \samp \mathsf{Sign}(sk, m)
    \end{array}] = 1.
\end{equation*}
\paragraph{Adaptive existential-unforgeability security under chosen message attack (EUF-CMA)}
We require that any QPT adversary $\adve$ wins the following game with negligible probability.
\begin{enumerate}
    \item The challenger samples the keys $sk, vk \samp \mathsf{Setup}(1)$.
    \item The adversary $\adve$ receives $vk$, and interacts with the signing oracle by sending classical messages $m \in \messpa$ and receiving the corresponding signatures that is computed by the challenger as $sig \samp \sign(sk, m)$.
    \item The adversary $\adve$ outputs a message $m^* \in \messpa$ that it has not queried the oracle with and a forged signature $sig^*$.
    \item The challenger outputs $1$ if and only if $\ver(vk, m^*, sig^*) = 1$.
\end{enumerate}
If the adversary $\adve$ is required to output the message $m^*$ before receiving $vk$, we call it \emph{selective EUF-CMA} security.
\end{definition}

\prelimlev{Indistinguishability Obfuscation}
In this section, we recall indistinguishability obfuscation (iO).
\begin{definition}[Indistinguishability Obfuscation]\label{predef:io}
    An indistinguishability obfuscation (iO) scheme for a class of circuits $\mathcal{C} = \{\mathcal{C}_\lambda\}_\lambda$ is an efficient algorithm $\io$ that satisfies the following.
    \paragraph{Correctness.} For all $\lambda \in \nat^+, C \in \mathcal{C}_\lambda$ and all inputs $x$ to $C$,
    $$\Pr[\Tilde{C}(x) = C(x): \Tilde{C} \samp \io(1^\lambda, C)] = 1.$$

    \paragraph{Security.} Let $\mathcal{B}$ be any QPT algorithm that outputs two circuits $C_0, C_1 \in \mathcal{C}$ of the same size, along with quantum auxiliary information $\regis{aux}$, such that $\Pr[\forall x ~ C_0(x)=C_1(x) : (C_0, C_1, \regis{aux}) \samp \mathcal{B}(1^\lambda)] \geq 1 - \negl(\lambda)$. Then, for any QPT adversary $\mathcal{A}$,
    \begin{align*}
      \bigg|&\Pr[\adve(\io(1^\lambda, C_0), \regis{aux}) = 1 :  (C_0, C_1, \regis{aux}) \samp \mathcal{B}(1^\lambda)] -\\ &\Pr[\adve(\io(1^\lambda, C_1), \regis{aux}) = 1 : (C_0, C_1, \regis{aux}) \samp \mathcal{B}(1^\lambda)]\bigg| \leq \negl(\lambda).  
    \end{align*}

    For subexponentially secure obfuscation, we require security for circuit samplers $\mathcal{B}$ such that $\Pr[\forall x ~ C_0(x)=C_1(x) : (C_0, C_1, \regis{aux}) \samp \mathcal{B}(1^\lambda)] = 1$.
\end{definition}
    
Fix an arbitrary constant $c>0$ and set $n:=\lceil\lambda^{c/3}\rceil$.

\prelimlev{Subspace-Hiding Obfuscation}

\begin{definition}[Subspace-hiding obfuscation]
\label{def:subspace-hiding-obfuscation}
Let $N=N(\lambda)$ and $d_0=d_0(\lambda)<d_1=d_1(\lambda)\leq N(\lambda)$ be polynomially bounded. A subspace-hiding obfuscator $\shO$ is a PPT algorithm that, on input $1^\lambda$ and a subspace $S\leq\FF_2^N$ of dimension $d_0$ or $d_1$, outputs a circuit $\widehat S\gets\shO(1^\lambda,S)$.

\paragraph{Correctness.} For every such subspace $S$, except with negligible probability over the randomness of $\shO$, $\widehat S(x)=1$ if and only if $x\in S$ for every $x\in\FF_2^N$.

\paragraph{Security.} For every QPT algorithm $\mathcal B$ that outputs a $d_0$-dimensional subspace $S_0\leq\FF_2^N$ together with quantum auxiliary information $R$, the following distributions are computationally indistinguishable:
\begin{equation}
\left\{(R,\widehat S_0):(S_0,R)\gets\mathcal B(1^\lambda),\ \widehat S_0\gets\shO(1^\lambda,S_0)\right\}
\approx_c
\left\{(R,\widehat S_1):
\begin{array}{l}
(S_0,R)\gets\mathcal B(1^\lambda),\\
S_1\getsr\{S\leq\FF_2^N:\dim(S)=d_1,\ S_0\subseteq S\},\\
\widehat S_1\gets\shO(1^\lambda,S_1)
\end{array}
\right\}.
\label{eq:subspace-hiding-security}
\end{equation}
\end{definition}

\begin{theorem}[Subspace-hiding obfuscation from $\io$ and one-way functions{~\cite{Zha21,C:CLLZ21,ABH+26}}]
\label{thm:subspace-hiding-instantiation}
Assume the existence of post-quantum polynomially secure $\io$ and post-quantum one-way functions. Then, for every polynomially bounded $n=n(\lambda)=\lambda^\epsilon$ for some constant $\epsilon>0$, there exists a post-quantum secure subspace-hiding obfuscator for $N=4n$, $d_0=2n$, and $d_1=3n$. These are the parameters used in \Cref{thm:comp-dcmoe-aux}.
\end{theorem}

\prelimlev{Universal Hash Functions and Average Conditional Min-Entropy}

\begin{definition}[Average conditional min-entropy]
\label{def:average-conditional-min-entropy}
For jointly distributed classical random variables $X,Z$, define
\begin{equation}
\Hmin(X\mid Z):=-\log\left(\sum_z\Pr[Z=z]\max_x\Pr[X=x\mid Z=z]\right).
\label{eq:average-conditional-min-entropy}
\end{equation}
When $Z$ is empty, this is the usual min-entropy of $X$.
\end{definition}

\begin{definition}[Universal hash family]
A family $\mathcal H$ of functions from $\zo^\ell$ to $\zo^m$ is universal if it is efficiently sampleable and, for every distinct $x,x'\in\zo^\ell$, $\Pr_{h\getsr\mathcal H}[h(x)=h(x')]\leq 2^{-m}$.
\end{definition}

\begin{theorem}[Universal hash families]
\label{thm:universal-hash-instantiation}
For every polynomially bounded $\ell=\ell(\lambda)$ and $m=m(\lambda)$, there exists an efficiently sampleable universal hash family from $\zo^\ell$ to $\zo^m$. In \Cref{thm:correlated-upo-construction}, we use such a family from $\mathcal X_\lambda=\zo^{\ell_{\mathsf{in}}(\lambda)}$ to $\zo^{8n}$.
\end{theorem}

\begin{lemma}[Leakage bound and leftover hash lemma]
\label{lem:average-min-entropy-facts}
Let $X,Y,Z$ be jointly distributed classical random variables. If $|\operatorname{Supp}(Y)|\leq 2^q$, then $\Hmin(X\mid Z,Y)\geq\Hmin(X\mid Z)-q$. Moreover, let $\mathcal H$ be a universal hash family from the support of $X$ to $\zo^m$, let $h\getsr\mathcal H$ be independent of $(X,Z)$, and let $U_m\getsr\zo^m$ be independent of $(h,X,Z)$. If $\Hmin(X\mid Z)\geq m+2\log(1/\varepsilon)$, then $(h,Z,h(X))$ and $(h,Z,U_m)$ are within statistical distance at most $\varepsilon$.
\end{lemma}

\prelimlev{Injective/Lossy Functions}

\begin{definition}[Injective/lossy functions]
\label{def:injective-lossy-functions}
An injective/lossy function family $\mathsf{LF}=(\mathsf{LF.Gen}_{\mathsf{inj}},\mathsf{LF.Gen}_{\mathsf{loss}})$ with domain $\mathcal X_\lambda$ and range $\mathcal T_\lambda$ consists of two PPT algorithms that output descriptions of deterministic functions $L:\mathcal X_\lambda\rightarrow\mathcal T_\lambda$. We require the following.
\begin{enumerate}
    \item Every $L\gets\mathsf{LF.Gen}_{\mathsf{inj}}(1^\lambda)$ is injective.
    \item There exists a polynomially bounded function $\ell_{\mathsf{img}}=\ell_{\mathsf{img}}(\lambda)$ such that every $L\gets\mathsf{LF.Gen}_{\mathsf{loss}}(1^\lambda)$ satisfies $|\operatorname{Im}(L)|\leq 2^{\ell_{\mathsf{img}}}$.
    \item The distributions $\{\mathsf{LF.Gen}_{\mathsf{inj}}(1^\lambda)\}_\lambda$ and $\{\mathsf{LF.Gen}_{\mathsf{loss}}(1^\lambda)\}_\lambda$ are computationally indistinguishable against QPT algorithms.
\end{enumerate}
\end{definition}

\begin{theorem}[Injective/lossy functions from LWE~\cite{ABH+26}]
\label{thm:lossy-function-instantiation}
Assume the quantum hardness of LWE. Then, for every constant $c>0$ and every polynomially bounded input length $\ell_{\mathsf{in}}=\ell_{\mathsf{in}}(\lambda)$ satisfying $\ell_{\mathsf{in}}(\lambda)\geq\lambda^c$, there exists a post-quantum secure injective/lossy function family with domain $\zo^{\ell_{\mathsf{in}}(\lambda)}$ for which $\ell_{\mathsf{img}}\leq\lambda^{2c/3}$. This is the parameter setting used in the proof of \Cref{thm:correlated-upo-construction}.
\end{theorem}

\subsection{Unclonable Encryption}
Next, we define (one-time) unclonable encryption for single bit messages.
\begin{definition}[Unclonable encryption for bits]
An unclonable encryption scheme for bits is a tuple
$\mathsf{UE}=(\mathsf{KeyGen},\mathsf{Enc},\mathsf{Dec})$ with the following syntax.
\begin{itemize}
    \item $k\gets\mathsf{KeyGen}(1^\secparam)$ outputs a classical key $k\in\mathcal K_\secparam$.
    \item $\rho_{\mathsf{ct}}\gets\mathsf{Enc}(k,m)$, on input $m\in\bit$, outputs a quantum ciphertext in a finite-dimensional Hilbert space $\mathcal H_{\mathsf{ct}}$.
    \item $m'\gets\mathsf{Dec}(k,\rho_{\mathsf{ct}})$ outputs a bit.
\end{itemize}
For every $m\in\bit$, correctness requires
\begin{equation}
\Pr\left[
m'=m:
\begin{array}{c}
k\gets\mathsf{KeyGen}(1^\secparam),\\
\rho_{\mathsf{ct}}\gets\mathsf{Enc}(k,m),\\
m'\gets\mathsf{Dec}(k,\rho_{\mathsf{ct}})
\end{array}
\right]
\geq 1-\negl(\secparam).
\end{equation}
\end{definition}

For arbitrary finite-dimensional Hilbert spaces $\mathcal H_\bob$ and $\mathcal H_\charlie$, let an adversary $\advue$ consist of a quantum channel
\begin{equation}
\Phi_\alice:\mathcal L(\mathcal H_{\mathsf{ct}})\longrightarrow
\mathcal L(\mathcal H_\bob\otimes\mathcal H_\charlie)
\end{equation}
and, for every $k\in\mathcal K_\secparam$, binary POVMs
$\{B_0^k,B_1^k\}$ on $\mathcal H_\bob$ and
$\{C_0^k,C_1^k\}$ on $\mathcal H_\charlie$. Consider the following experiment.

\par\noindent
$\gam{\mathsf{UE\mbox{-}Unclonable\mbox{-}Indistinguishable}_{\mathsf{UE},\advue}}(1^\secparam)$:
\begin{enumerate}
    \item $\ch$ samples $k\gets\mathsf{KeyGen}(1^\secparam)$ and $m\getsr\bit$, computes $\rho_{\mathsf{ct}}\gets\mathsf{Enc}(k,m)$, and sends $\rho_{\mathsf{ct}}$ to $\alice$ while retaining $k$.
    \item $\alice$ applies $\Phi_\alice$ and sends the resulting registers to $\bob$ and $\charlie$.
    \item After the split, $\ch$ sends $k$ to both $\bob$ and $\charlie$.
    \item $\bob$ measures with $\{B_0^k,B_1^k\}$ and outputs $\widehat m_\bob$, while $\charlie$ measures with $\{C_0^k,C_1^k\}$ and outputs $\widehat m_\charlie$.
    \item The output of the experiment is $1$ if $\widehat m_\bob=\widehat m_\charlie=m$.
\end{enumerate}

\begin{definition}[$\epsilon$-unclonable indistinguishable security]\label{def:UE-security}
For $\epsilon=\epsilon(\secparam)\geq0$, define the winning advantage of $\advue$ by
\begin{equation}
\operatorname{Adv}^{\mathsf{UI}}_{\mathsf{UE},\advue}(\secparam)
:=
\Pr\left[
\gam{\mathsf{UE\mbox{-}UI}_{\mathsf{UE},\advue}}(1^\secparam)=1
\right]-\frac{1}{2}.
\end{equation}
The scheme $\mathsf{UE}$ satisfies $\epsilon$-unclonable indistinguishable security if, for every (potentially unbounded) adversary $\advue$ as above,
\begin{equation}
\operatorname{Adv}^{\mathsf{UI}}_{\mathsf{UE},\advue}(\secparam)
\leq\epsilon(\secparam),
\end{equation}
or equivalently,
\begin{equation}
\Pr\left[
\gam{\mathsf{UE\mbox{-}UI}_{\mathsf{UE},\advue}}(1^\secparam)=1
\right]
\leq\frac{1}{2}+\epsilon(\secparam).
\end{equation}
\end{definition}

%% file: 000-include/boilerplate/quantum-results/gentlemes.tex
\paragraph{Gentle measurement and rewinding.}
\label{sec:gentle-measurement}
The almost-as-good-as-new lemma shows that a measurement that accepts with
high probability can be gently rewound.

\begin{lemma}[Almost As Good As New Lemma \cite{Aar16}; see also \cite{awinter}]
\label{prelem:gentlemes}
Let $d,d'\in\nat^+$, let $\rho$ be a mixed state over $\C^d$, and let $U$ be a unitary over $\C^d\otimes\C^{d'}$. Let $(\Pi_0,\Pi_1=I-\Pi_0)$ be a projective measurement over $\C^d\otimes\C^{d'}$. Consider the measurement that appends an ancillary register initialized to $\ketbra{0}{0}$, applies $U$, and then measures $(\Pi_0,\Pi_1)$. Suppose that, for some $\epsilon\in[0,1]$,
\begin{equation*}
    \Tr\left[\Pi_0 U\left(\rho\otimes\ketbra{0}{0}\right)U^\dagger\right]=1-\epsilon.
\end{equation*}
Define
\begin{equation*}
\widetilde{\rho}:=
\Tr_{\C^{d'}}\left[
U^\dagger\left(
\sum_{b\in\zo}\Pi_b U\left(\rho\otimes\ketbra{0}{0}\right)U^\dagger\Pi_b
\right)U
\right].
\end{equation*}
Then,
\begin{equation*}
    \trd{\rho}{\widetilde{\rho}}\leq\sqrt{\epsilon}.
\end{equation*}
\end{lemma}
We refer to the operation of undoing $U$ after the measurement and then tracing out the ancillary register as \emph{rewinding} the measured register. Thus, $\widetilde{\rho}$ in \cref{prelem:gentlemes} is the state obtained after rewinding.

%% file: 000-include/boilerplate/quantum-definitions/threshold-implementations.tex
\paragraph{Projective and threshold implementations.}
\label{sec:threshold-implementations}

\begin{definition}[Projective implementation \cite{TCC:Zhandry20}]
Let $\mathcal{P}=(P,I-P)$ be a binary-outcome POVM. A projective
implementation of $\mathcal{P}$ is a projective measurement
$\mathcal{E}=(\Pi_p)_{p\in S}$, for a finite set $S\subseteq[0,1]$, such that
the following measurement is equivalent to $\mathcal{P}$: apply $\mathcal{E}$
to obtain $p$, and output $1$ with probability $p$ and $0$ otherwise. We
denote the unique projective implementation of $\mathcal{P}$ by
$\projimp(\mathcal{P})$.
\end{definition}

We use the following efficient, gapped form of approximate threshold
implementation for a general efficient POVM.  The fixed running-time and
reference-system guarantees are recorded explicitly because both are needed
in our reductions.

\begin{lemma}[Fixed-time gapped approximate threshold instrument
\cite{TCC:Zhandry20,C:ALLZZ21}]
\label{lem:gapped-threshold-instrument}
Let $\mathcal P=(P,I-P)$ be an efficient binary-outcome POVM on a register
$\mathsf A$, where an efficient implementation means a uniform
polynomial-size unitary dilation (with its inverse) followed by a one-qubit
projective measurement.  Let $0\leq s<c\leq1$, where
$\Delta:=c-s$ is inverse polynomial, and let $\eta\in(0,1]$ be inverse
polynomial.
There is a uniform, fixed-polynomial-time two-outcome instrument
\begin{equation*}
  \mathsf{GATI}^{c,s}_{\eta}(\mathcal P)
  =\bigl(\mathcal A,\mathcal R\bigr)
\end{equation*}
with the following properties.  Write
\begin{equation*}
  \Pi_{\mathsf{hi}}:=\mathbf 1[P\geq c],
  \qquad
  \Pi_{\mathsf{lo}}:=\mathbf 1[P\leq s].
\end{equation*}
For every subnormalized state $\rho_{\mathsf{AR}}$, with an arbitrary
reference register $\mathsf R$, and
$\sigma:=(\mathcal A\otimes\mathrm{Id}_{\mathsf R})(\rho)$, we have
\begin{align}
 \Tr\!\left[(\Pi_{\mathsf{hi}}\otimes I)\rho\right]
       -\eta\Tr(\rho)
 &\leq \Tr(\sigma),
 \label{eq:gati-completeness}\\
 \Tr(\sigma)
 &\leq
 \Tr\!\left[(\mathbf 1[P>s]\otimes I)\rho\right]
       +\eta\Tr(\rho),
 \label{eq:gati-soundness}\\
 \Tr\!\left[(\Pi_{\mathsf{lo}}\otimes I)\sigma\right]
 &\leq \eta\Tr(\rho).
 \label{eq:gati-poststate}
\end{align}
Consequently, for
\begin{equation*}
 \overline\sigma
 :=(\mathbf 1[P>s]\otimes I)\sigma
   (\mathbf 1[P>s]\otimes I),
\end{equation*}
we have
\begin{equation}
 \Tr(\overline\sigma)\geq\Tr(\sigma)-\eta\Tr(\rho),
 \qquad
 \|\sigma-\overline\sigma\|_1
 \leq3\sqrt\eta\,\Tr(\rho).
 \label{eq:gati-gentle-poststate}
\end{equation}
The running time is
\begin{equation*}
 T_{\mathcal P}\cdot
 \poly\!\left(\Delta^{-1},\log(1/\eta),\eta^{-1}\right),
\end{equation*}
where $T_{\mathcal P}$ is the running time of the dilation of $\mathcal P$.
The same circuit and bounds apply uniformly for every reference register.
\end{lemma}

\begin{proof}
Purify all classical coins and deferred measurements in the implementation
of $\mathcal P$.  To recall why the cited threshold construction applies to
a general efficient POVM, let $\Pi_0$ project the dilation workspace onto its
initial state and let $\Pi_1$ be the accepting projector pulled back through
the dilation unitary.  On the image of $\Pi_0$,
\begin{equation*}
  \Pi_0\Pi_1\Pi_0=P.
\end{equation*}
The standard two-projector (Jordan-block) construction underlying the
approximate projective implementation of Zhandry~\cite{TCC:Zhandry20} and
the approximate threshold implementation of Aaronson et
al.~\cite{C:ALLZZ21} therefore measures the spectral parameter of $P$ and
returns the dilation workspace to its initial state.  This gives an
instrument on $\mathsf A$, not merely on the dilation space.

Set $\xi:=\Delta/4$ and run that approximate threshold implementation at
threshold $c-\xi$, with shift and almost-projectivity errors at most
$\eta/20$.  Its shift guarantee compares its acceptance probability from
below with the exact threshold projector $\mathbf 1[P\geq c]$ and from
above with $\mathbf 1[P>c-2\xi]\preceq\mathbf 1[P>s]$, up to error
$\eta/20$.  Its post-measurement guarantee places all but $\eta/10$ of the
accepted subnormalized state above threshold $c-3\xi>s$.  These are
operator and local post-state guarantees on $\mathsf A$.  For an input
entangled with $\mathsf R$, the reduced accepted state on $\mathsf A$ is
exactly $\mathcal A(\rho_{\mathsf A})$; hence the same probability and
low-subspace bounds hold after tensoring the construction with
$I_{\mathsf R}$.

The cited implementation has expected running time
$T_{\mathcal P}\poly(\Delta^{-1},\log(1/\eta))$, uniformly over its input
state.  Fix a corresponding polynomial upper bound $T_{\mathsf{exp}}$ and
truncate after $20T_{\mathsf{exp}}/\eta$ steps, outputting reject if the
cutoff is reached.  Markov's inequality bounds the cutoff probability by
$\eta/20$ for every input, including one entangled with a reference.
Absorbing the approximation and truncation errors proves
\cref{eq:gati-completeness,eq:gati-soundness,eq:gati-poststate} and gives the
claimed fixed polynomial running time.  Finally,
\cref{eq:gati-gentle-poststate} follows from
\cref{eq:gati-poststate} and the gentle-measurement inequality, after
normalizing $\rho$; the displayed constant also absorbs the discarded
low-subspace block.
\end{proof}

\begin{definition}[Threshold implementation \cite{TCC:Zhandry20,C:ALLZZ21}]
Let $\mathcal{P}$ and $\projimp(\mathcal{P})=(\Pi_p)_{p\in S}$ be as above.
For a threshold $\tau\in[0,1]$, the threshold implementation of
$\mathcal{P}$ with threshold $\tau$ is the binary projective measurement
\begin{equation*}
    \TI_\tau(\mathcal{P})
    :=\left(\sum_{p\geq\tau}\Pi_p,
    I-\sum_{p\geq\tau}\Pi_p\right).
\end{equation*}
When clear from context, we also use $\TI_\tau(\mathcal{P})$ to denote its
projector corresponding to outcome $1$.
\end{definition}

%% file: moe/comp-moe.tex
\section{Computational Decisional Coset Monogamy-of-Entanglement Game}\label{sec:comp-moe}

In this subsection, we prove the computational versions of the \dcmoe game that will be used in the security proofs of SDE and UPO. The proof follows the computational subspace-hiding proof of~\cite{C:CLLZ21}, and its extension to simulatable auxiliary information in~\cite{ABH+26,CG26}. We deviate from the formalism~\cite{ABH+26,CG26} who considered independent auxiliary -information  generators for $\bob$ and $\charlie$ and instead use a joint auxiliary information generator for \bob and \charlie, rather than two independent generators, since the applications in this article allow the two auxiliary strings to be arbitrarily correlated.

For the computational games, let
\[
    V:=\FF_2^{4n}.
\]
For any subspace $S\leq V$ and $z\in V$, let $P_{S+z}$ denote
the canonical circuit that, on input $x\in V$, outputs $1$ if
and only if $x\in S+z$. All circuits compared using $\iO$ below
are padded to the same size. We condition on the event that all
the generated membership programs are correct on every input;
the complement of this event contributes a negligible term.

We first record the following immediate extension of
\Cref{thm:dcmoe}, which will be used at the end of the
computational reduction.

\begin{corollary}[Post-split auxiliary information]
\label{cor:dcmoe-post-aux}
Consider the \dcmoe experiment, except that after sampling
$(u,v)$, the challenger applies an arbitrary, possibly
inefficient, classical algorithm
\[
    (z_{\bob},z_{\charlie})
    \leftarrow
    \mathsf{Post}(A,u,v),
\]
and sends $(u,v,z_{\bob})$ to \bob and
$(u,v,z_{\charlie})$ to \charlie. The algorithm
$\mathsf{Post}$ does not receive $s,t$ or the split registers,
and its two outputs may be arbitrarily correlated. Then, for
every triplet of adversaries
$\advmoe=(\alice,\bob,\charlie)$,
\[
    \Pr[\mathsf{win}]
    \leq
    \frac{1}{2}
    +
    \frac{1}{2}
    \sqrt{
        \frac{2^{2n}}
        {(2^{2n}-1)(2^{2n-1}-1)}
    }.
\]
\end{corollary}

\begin{proof}
Fix a nonzero orthogonal pair $k=(u,v)$ and a bit $m$. By
\Cref{claim:conditional-cipher-state}, for every compatible
subspace $A$, the state conditioned on $k,m,A$ is
\[
    \rho_{ct}^{k,m,A}
    =
    \frac{1}{2^{2n}}
    \left(I+(-1)^m\Theta_k\right),
\]
which is independent of $A$. Moreover, for every fixed
compatible $A$ and $k$, the bit
\[
    m=\langle u,t\rangle+\langle v,s\rangle
\]
is uniform. Hence, conditioned on $k$, the joint distribution
of $(A,z_{\bob},z_{\charlie})$ is independent of $m$.

The splitting adversary may append a shared classical random
string to the two output registers. After receiving $k$, both
recipients can use this common random string to sample the same
value of $A$ according to its conditional distribution given
$k$, and then sample the correlated pair
$(z_{\bob},z_{\charlie})$ according to
$\mathsf{Post}(A,k)$. This does not require communication and
does not use $m$. Therefore, the post-split auxiliary
information can be absorbed into the local strategies of \bob
and \charlie, and the corollary follows from
\Cref{thm:dcmoe}.
\end{proof}

Next, we define the auxiliary information allowed in the
computational game.

\begin{definition}[Simulatable pre-split auxiliary information]
\label{def:comp-dcmoe-aux}
Let $\mathsf{Aux}=(\auxsetup,\auxgen)$ be a pair of QPT
algorithms with the following syntax.
\begin{enumerate}
    \item $pp\gets \auxsetup(1^\secparam):$ takes as input the security parameter $1^\secparam$, and outputs a classical public parameter $pp$.
    \item $(z_{1,\bob},z_{2,\bob},z_{1,\charlie},z_{2,\charlie})\gets\auxgen(pp,A,u,v):$ takes as input a classical public parameter $pp$, a subspace $A\leq V$, vectors$u\in A\setminus\{0\}$, and $v\in A^\perp\setminus\{0\}$, and outputs pre-split auxiliary information $(z_{1,\bob},z_{1,\charlie})$ and post-split auxiliary information $(z_{2,\bob},z_{2,\charlie})$, for $\bob$ and $\charlie$ respectively\footnote{The four strings $(z_{1,\bob},z_{2,\bob},z_{1,\charlie},z_{2,\charlie})$ can be arbitrarily correlated.}.
\end{enumerate}


There exists a possibly inefficient algorithm $\auxsim$ such that for every subspace $A\leq V$, $u\in A\setminus \{0\},v\in A^\perp\setminus \{0\}$,
\begin{equation}\label{eq:comp-dcmoe-aux-simulation}
\left\{
    (pp,z_{1,\bob},z_{1,\charlie})
    :
    \begin{array}{l}
        pp\gets\auxsetup(1^\secparam),\\
        (z_{1,\bob},z_{2,\bob},z_{1,\charlie},z_{2,\charlie})
        \gets\auxgen(pp,A,u,v)
    \end{array}
\right\}
\approx_s
\left\{
    (pp,\widetilde z_{1,\bob},\widetilde z_{1,\charlie})
    :
    \begin{array}{l}
        pp\gets\auxsetup(1^\secparam),\\
        (\widetilde z_{1,\bob},\widetilde z_{1,\charlie})
        \leftarrow\auxsim(pp)
    \end{array}
\right\}.
\end{equation}
where the first distribution is obtained by sampling
$pp\gets\auxsetup(1^\secparam)$ and then running
$\auxgen(pp,A,u,v)$, and the second is obtained by sampling
$pp\gets\auxsetup(1^\secparam)$ and then running
$\auxsim(pp)$.
\end{definition}

\par\noindent
\begin{definition}[\compdcmoe]\label{def:computational-dcmoe}
Let $n:=n(\secparam)$ be a polynomial function in $\secparam$, and $V=\FF_2^{4n}$.
For a triplet of QPT adversaries
$\advmoe=(\alice,\bob,\charlie)$ and auxiliary algorithms
$\mathsf{Aux}$ as above, consider the following experiment.
$\gam{\compdcmoe_{\advmoe,\mathsf{Aux}}}(1^\secparam)$:
\begin{enumerate}
    \item $\ch$ samples a uniformly random $2n$-dimensional subspace $A\leq V$ and independent $s,t\getsr V$, and prepares
    \[\ket{A_{s,t}}=\frac{1}{\sqrt{|A|}}\sum_{a\in A}(-1)^{\langle a,t\rangle}\ket{a+s}.\]
    It also generates
    $M_0\gets\iO(P_{A+s}),\text{ and } M_1\gets\iO(P_{A^\perp+t}).$
    \item $\ch$ samples $u\getsr A\setminus\{0\}, v\getsr A^\perp\setminus\{0\},$ and generates $pp\gets\auxsetup(1^\secparam),$ and samples $(z_{1,\bob},z_{2,\bob},z_{1,\charlie},z_{2,\charlie})\gets\auxgen(pp,A,u,v).$
    \item $\ch$ sends $\left(\ket{A_{s,t}},M_0,M_1,pp,z_{1,\bob},z_{1,\charlie}\right)$ to \alice. 
    \item $\alice$ outputs a joint state
    $\sigma_{\regbob,\regcharlie}$ and sends registers
    $\regbob$ to \bob and $\regcharlie$ to \charlie.

    \item After the split, $\ch$ sends
    $(u,v,z_{2,\bob})$ to \bob and
    $(u,v,z_{2,\charlie})$ to \charlie.

    \item $\bob$ and $\charlie$ output
    $b_\bob,b_\charlie$.

    \item The output of the experiment is $1$ if
        $b_\bob=b_\charlie= \langle u,t\rangle+\langle v,s\rangle.$    
\end{enumerate}
\end{definition}
Note that the pair $(u,v)$ is sampled before the split only because the
pre-split auxiliary information may depend on it. 


For convenience, define
\[\epsilondcmoe  :=
    \frac{1}{2}
    \sqrt{
        \frac{2^{2n}}
        {(2^{2n}-1)(2^{2n-1}-1)}
    },
    \qquad
    \beta_n
    :=
    \frac{2^n-1}{2^{2n}-1}.
\]

\begin{theorem}[Computational Decisional Coset Monogamy with
auxiliary information]
\label{thm:comp-dcmoe-aux}
Assume post-quantum secure $\iO$ and post-quantum secure subspace-hiding obfuscation. Then, for every function $n=n(\secparam)\in \omega(\log(\secparam))$, and for every QPT triplet $\advmoe=(\alice,\bob,\charlie)$ and every auxiliary-information generator satisfying \Cref{def:comp-dcmoe-aux}, there exists a negligible function $\negl(\secparam)$ such that
\[\Pr\left[ \gam{\compdcmoe_{\advmoe,\mathsf{Aux}}}(1^\secparam)=1\right]\leq \frac{1}{2}+\negl(\secparam).\]
\end{theorem}

\begin{proof}
Fix a QPT triplet of adversaries $\advmoe=(\alice,\bob,\charlie)$. We proceed through a sequence of hybrids, following the computational monogamy proof of~\cite{C:CLLZ21} and the auxiliary-information hybrids of~\cite{ABH+26,CG26}.

\paragraph{$\hybrid_0$.}
This is the experiment
$\gam{\compdcmoe_{\advmoe,\mathsf{Aux}}}(1^\secparam)$.

\paragraph{$\hybrid_1$: Enlarge the primal subspace.}
Replace the first public program $M_0=\iO(P_{A+s})$ by
$\iO(x\mapsto\widehat U(x-s))$, where $U\supset A$ is a uniformly random $3n$-dimensional superspace and $\widehat U\gets\shO(U)$. Formally, this change is performed through the standard intermediate hybrid in which $M_0$ is first replaced by $\iO(x\mapsto\widehat A(x-s))$ for $\widehat A\gets\shO(A)$, followed by the subspace-hiding replacement of $\widehat A$ by $\widehat U$. The reduction knows $(A,s,t,u,v)$ and can therefore prepare the coset state and generate all four auxiliary strings. Hence $\hybrid_0$ and $\hybrid_1$ are computationally indistinguishable.

\paragraph{$\hybrid_2$: Rerandomize the primal shift.}
Sample $w_s\getsr U$, set $\sigma:=s+w_s$, and replace the source shift $s$ in the first public program by $\sigma$. Since $U+s=U+\sigma$, the two source circuits are functionally equivalent, and therefore $\hybrid_1$ and $\hybrid_2$ are computationally indistinguishable by $\iO$-security.

\paragraph{$\hybrid_3$: Enlarge the dual subspace.}
Sample a uniformly random $n$-dimensional subspace $L\subset A$, so that $L^\perp$ is a uniformly random $3n$-dimensional superspace of $A^\perp$, and replace the second public program $M_1=\iO(P_{A^\perp+t})$ by $\iO(y\mapsto\widehat{L^\perp}(y-t))$, where $\widehat{L^\perp}\gets\shO(L^\perp)$. As in $\hybrid_1$, this is performed through the standard intermediate hybrid using an obfuscation of $A^\perp$, and subspace-hiding security gives that $\hybrid_2$ and $\hybrid_3$ are computationally indistinguishable.

\paragraph{$\hybrid_4$: Rerandomize the dual shift.}
Sample $w_t\getsr L^\perp$, set $\tau:=t+w_t$, and replace the source shift $t$ in the second public program by $\tau$. Since $L^\perp+t=L^\perp+\tau$, $\iO$-security gives that $\hybrid_3$ and $\hybrid_4$ are computationally indistinguishable.

The first four hybrids are exactly the primal and dual subspace-hiding and affine-rerandomization hybrids used in~\cite{C:CLLZ21}. Their validity is unaffected by the joint auxiliary information because the reduction can sample the auxiliary information as it knows the smaller subspace and all values required to run $\auxsetup$ and $\auxgen$. Thus, there exists a negligible function $\negl$ such that
\[
\left|
\Pr[\hybrid_0=1]
-
\Pr[\hybrid_4=1]
\right|
\leq
\negl(\secparam).
\]

\paragraph{$\hybrid_5$: Simulate the pre-split auxiliary information.}
Sample $(\widetilde z_{1,\bob},\widetilde z_{1,\charlie})
\gets\auxsim(pp)$ and use this pair in place of
$(z_{1,\bob},z_{1,\charlie})$. Then sample $(z_{2,\bob},z_{2,\charlie})$ jointly from the conditional distribution given by
\[(z_{1,\bob},z_{2,\bob},z_{1,\charlie},z_{2,\charlie})\gets\auxgen(pp,A,u,v)),\]
conditioned on the event
\[(z_{1,\bob},z_{1,\charlie})=(\widetilde z_{1,\bob},\widetilde z_{1,\charlie}).\]
If this event $(z_{1,\bob},z_{1,\charlie})=(\widetilde z_{1,\bob},\widetilde z_{1,\charlie})$ has probability zero, define the conditional
distribution arbitrarily. By \Cref{def:comp-dcmoe-aux} and the fact that statistical distance does not increase under a classical channel, there exists a negligible function $\negl'(\secparam)$ such that
\[
\left|
\Pr[\hybrid_4=1]
-
\Pr[\hybrid_5=1]
\right|
\leq
\negl'(\secparam).
\]
The experiment $\hybrid_5$ may be inefficient, but this is harmless since the transition is statistical and no further computational indistinguishability argument will be used.

In $\hybrid_5$, the entire input given to \alice before the split is independent of $(u,v)$ conditioned on $(A,s,t)$ and the two public programs. Hence, the experiment is identically distributed if $(u,v)$ is sampled only after \alice outputs the two registers, and $(z_{2,\bob},z_{2,\charlie})$ is subsequently sampled from the conditional distribution above.

\paragraph{$\hybrid_6$: Reveal $L,U,\sigma,\tau$.}
In addition to the information already given before the split, give $(L,U,\sigma,\tau)$ to \alice. This can only increase the winning probability, and therefore
\[
\Pr[\hybrid_5=1]
\leq
\Pr[\hybrid_6=1].
\]

\paragraph{Analysis of $\hybrid_6$.}
Condition on fixed values of $(L,U,\sigma,\tau,pp,\widetilde z_{1,\bob},\widetilde z_{1,\charlie})$ and on the two public-program strings. Conditioned on these values, $A$ is uniformly distributed among the $2n$-dimensional subspaces satisfying $L\subset A\subset U$. Moreover, $s+\sigma$ is uniform over $U$, $t+\tau$ is uniform over $L^\perp$, and these two random variables are independent.

Choose an invertible linear transformation $T$, as a deterministic function of $(L,U)$, such that
\[
T(L)=\FF_2^n\oplus 0\oplus 0
\qquad\text{and}\qquad
T(U)=\FF_2^n\oplus\FF_2^{2n}\oplus 0.
\]
We use $T$ on computational-basis coordinates and $T^{-\mathsf T}$ on phase coordinates. For every intermediate subspace $A$, there exists a unique $n$-dimensional subspace $\overline A\leq\FF_2^{2n}$ such that
\[
T(A)
=
\FF_2^n\oplus\overline A\oplus 0,
\qquad
T^{-\mathsf T}(A^\perp)
=
0\oplus\overline A^\perp\oplus\FF_2^n.
\]
Since $A$ is uniform subject to $L\subset A\subset U$, the induced subspace $\overline A$ is uniformly distributed among the $n$-dimensional subspaces of $\FF_2^{2n}$.

By the same change-of-basis calculation as in~\cite{C:CLLZ21,ABH+26,CG26}, after applying the known Pauli correction $X^\sigma Z^\tau$ and the unitary corresponding to $T$, the coset state becomes, up to a global phase,
\[
\ket{+}^{\otimes n}
\otimes
\ket{\overline A_{\overline s,\overline t}}
\otimes
\ket{0}^{\otimes n},
\]
where $\overline s,\overline t\getsr\FF_2^{2n}$ are independent and uniform. The two outer registers are fixed and can be incorporated into \alice's strategy.
Write the transformed vectors as
\[
T(u)=(u_L,\overline u,0),
\qquad
T^{-\mathsf T}(v)=(0,\overline v,v_R),
\]
where $\overline u\in\overline A$ and $\overline v\in\overline A^\perp$. The target bit satisfies
\begin{equation}\label{eq:comp-dcmoe-target-decomposition}
\langle u,t\rangle+\langle v,s\rangle
=
\langle u,\tau\rangle
+
\langle v,\sigma\rangle
+
\langle\overline u,\overline t\rangle
+
\langle\overline v,\overline s\rangle.
\end{equation}
The first two terms depend only on $(\sigma,\tau,u,v)$ and are therefore known to both recipients after the split.
Next, let $\mathsf{Bad}$ denote the event that $\overline u=0$ or $\overline v=0$. Since $\overline u=0$ exactly when $u\in L$, and $\overline v=0$ exactly when $v\in U^\perp$, we have
\[
\Pr[\overline u=0]
=
\Pr[\overline v=0]
=
\frac{2^n-1}{2^{2n}-1}
=
\beta_n.
\]
Therefore,
\begin{equation}\label{eq:comp-dcmoe-bad}
\Pr[\mathsf{Bad}]
\leq
2\beta_n.
\end{equation}
Conditioned on $\neg\mathsf{Bad}$, the subspace $\overline A$ remains uniform and, conditioned on $\overline A$, the vectors
$\overline u\getsr\overline A\setminus\{0\}$ and
$\overline v\getsr\overline A^\perp\setminus\{0\}$ are independent and uniform.

We now reduce $\hybrid_6$, conditioned on the fixed values above and on $\neg\mathsf{Bad}$, to the information-theoretic experiment of \Cref{cor:dcmoe-post-aux}. The information-theoretic \alice receives
$\ket{\overline A_{\overline s,\overline t}}$, appends
$\ket{+}^{\otimes n}$ and $\ket{0}^{\otimes n}$, applies the inverse of the known change of basis and Pauli correction, supplies the fixed public programs and the fixed strings
$(pp,\widetilde z_{1,\bob},\widetilde z_{1,\charlie},L,U,\sigma,\tau)$, and runs the splitting algorithm from $\hybrid_6$.

After the information-theoretic challenger reveals
$(\overline u,\overline v)$, the post-split algorithm samples
$u_L,v_R\getsr\FF_2^n$ and defines
\[
u:=T^{-1}(u_L,\overline u,0),
\qquad
v:=T^{\mathsf T}(0,\overline v,v_R).
\]
These are distributed exactly as the original vectors conditioned on $\neg\mathsf{Bad}$. The algorithm reconstructs
\[
A
=
T^{-1}\left(
\FF_2^n\oplus\overline A\oplus 0
\right),
\]
samples $(z_{2,\bob},z_{2,\charlie})$ jointly according to the conditional distribution used in $\hybrid_5$, and computes
$\eta:=\langle u,\tau\rangle+\langle v,\sigma\rangle$. It gives
$(u,v,z_{2,\bob},\eta)$ to \bob and
$(u,v,z_{2,\charlie},\eta)$ to \charlie. This procedure may be inefficient, but it does not receive $(\overline s,\overline t)$ or the split registers and is therefore allowed by \Cref{cor:dcmoe-post-aux}.

Each recipient runs the corresponding recipient from $\hybrid_6$. If the recipient outputs $b_X$, the information-theoretic recipient outputs $b_X+\eta$. By \Cref{eq:comp-dcmoe-target-decomposition}, the winning events in the two experiments are identical. Therefore,
\[
\Pr[
\hybrid_6=1
\mid
\neg\mathsf{Bad}
]
\leq
\frac{1}{2}+\epsilondcmoe.
\]

Let $\delta:=\Pr[\mathsf{Bad}]$. Upper bounding the winning probability conditioned on $\mathsf{Bad}$ by one and using \Cref{eq:comp-dcmoe-bad}, we obtain
\begin{align*}
\Pr[\hybrid_6=1]
&\leq
\delta
+
(1-\delta)
\left(
\frac{1}{2}+\epsilondcmoe
\right)\\
&\leq
\frac{1}{2}
+
\epsilondcmoe
+
\frac{\delta}{2}\\
&\leq
\frac{1}{2}
+
\epsilondcmoe
+
\beta_n.
\end{align*}

Combining the computational hybrid transitions, the statistical transition from $\hybrid_4$ to $\hybrid_5$, and the inequality between $\hybrid_5$ and $\hybrid_6$, and absorbing all negligible functions into a single negligible function, gives
\[
\Pr[\hybrid_0=1]
\leq
\frac{1}{2}
+
\epsilondcmoe
+
\beta_n
+
\negl(\secparam).
\]
In particular, since $n=\omega(\log\secparam)$, both $\epsilondcmoe$ and $\beta_n$ are negligible, which completes the proof of the theorem.
\end{proof}

\begin{corollary}[Computational Decisional Coset Monogamy without auxiliary information]
\label{cor:comp-dcmoe}
Let $\gam{\compdcmoe_{\advmoe}}(1^\secparam)$ denote the
experiment of \Cref{def:computational-dcmoe} with no auxiliary
information, i.e., with
\[
    pp=z_{1,\bob}=z_{2,\bob}
    =z_{1,\charlie}=z_{2,\charlie}=\bot.
\]
Assume post-quantum secure $\iO$ and post-quantum secure
subspace-hiding obfuscation. Then, for every function
$n=n(\secparam)=\secparam^\epsilon$ for some constant $\epsilon>0$, and every QPT triplet
$\advmoe=(\alice,\bob,\charlie)$, there exists a negligible
function $\negl(\secparam)$ such that
\[
    \Pr\left[
        \gam{\compdcmoe_{\advmoe}}(1^\secparam)=1
    \right]
    \leq
    \frac12+\negl(\secparam).
\]
\end{corollary}

\begin{proof}
Let $\auxsetup(1^\secparam)$ output $\bot$, let
$\auxgen(pp,A,u,v)$ output
$(\bot,\bot,\bot,\bot)$, and let
$\auxsim(pp)$ output $(\bot,\bot)$. The statistical
simulatability condition in
\Cref{def:comp-dcmoe-aux} then holds perfectly, and the result
follows directly from \Cref{thm:comp-dcmoe-aux}.
\end{proof}

The subspace-hiding obfuscators required by \cref{cor:comp-dcmoe} can be obtained from post-quantum secure indistinguishability obfuscation and injective keyed one-way functions (see ~\cite{C:CLLZ21}) which can be constructed from indistinguishability obfuscation and one-way functions~\cite{TCC:BitPanWic16}.

%% file: defns/UPO-defn-and-properties.tex
\section{Unclonable Puncturable Obfuscation definitions}\label{sec:UPO-def}
In this section, we recall the definitions of unclonable puncturable obfuscation introduced in~\cite{AB24}, and also define a new stronger correlated security notion of UPO. As in~\cite{AB24}, UPO and its stronger correlated security enable us to achieve our desired copy-protection feasibility results. 

First, we recall the syntax and correctness of Unclonable Puncturable Obfuscation (UPO) from~\cite{AB24,ABH+26}, then in \cref{sec:upo-prior-defn} we recall the prior security definitions. Then, in \cref{sec:upo-new-defn} we present our new definition and show that our new definition implies all prior UPO security notions considered in previous works, except UPO plus.
Let $\mathcal C=\{\mathcal C_\secparam\}_{\secparam\in\NN}$ be a
keyed circuit class, where
$\mathcal C_\secparam=\{C_k:\mathcal X_\secparam\rightarrow
\mathcal Y_\secparam\}_{k\in\mathcal K_\secparam}$,
$\mathcal X_\secparam:=\{0,1\}^{\ell_{\mathsf{in}}(\secparam)}$,
and
$\mathcal Y_\secparam:=\{0,1\}^{\ell_{\mathsf{out}}(\secparam)}$.

\begin{definition}[Unclonable puncturable obfuscation]
A UPO scheme $\Pi=(\Obf,\Eval)$ for $\mathcal C$ consists of a pair
of efficient algorithms as follows.
\begin{enumerate}
    \item $\rho_k\gets\Obf(1^\secparam,C_k):$ is a QPT algorithm
    that takes as input a circuit $C_k\in\mathcal C_\secparam$ and
    outputs a quantum obfuscation $\rho_k$.

    \item $(\rho'_k,y)\gets\Eval(\rho_k,x):$ is a QPT algorithm that takes as
    input a quantum obfuscation $\rho_k$ and an input
    $x\in\mathcal X_\secparam$, and outputs
    an updated quantum obfuscation $\rho'_k$ and $y\in\mathcal Y_\secparam$.
\end{enumerate}

\paragraph{Correctness} We require that there exists a negligible function $\negl(\cdot)$ such that, for every $\secparam\in\NN$, every $k\in\mathcal K_\secparam$, and every $x\in\mathcal X_\secparam$,
\begin{equation}
\Pr[ y=C_k(x)\mid \rho_k\gets\Obf(1^\secparam,C_k),  (\rho'_k,y)\gets\Eval(\rho_k,x)]\geq 1-\negl(\secparam).
\end{equation}
As in~\cite[Remark~8]{AB24}, the almost-as-good-as-new lemma and a quantum union bound give polynomial reusability. We say that the scheme is perfectly correct when the probability above is one.
\end{definition}

Throughout this section, whenever a UPO security experiment
obfuscates a circuit produced by one of its specified puncturing,
programming, or compilation algorithms, the UPO scheme and its
correctness requirement are understood to apply to that circuit as
well.
\subsection{Prior Definitions:Generalized UPO Security}\label{sec:upo-prior-defn}
We next recall generalized puncturing from~\cite{AB24,ABH+26}.
A keyed circuit class $\mathcal C$ is generalized puncturable if
there exists a deterministic PPT algorithm $\genpuncture$ with the
following syntax.

\begin{definition}[Generalized puncturing]
$G^\star_{k,x_\bob,x_\charlie,\mu_\bob,\mu_\charlie}
\gets\genpuncture(k,x_\bob,x_\charlie,\mu_\bob,\mu_\charlie):$
takes as input a key $k\in\mathcal K_\secparam$, points
$x_\bob,x_\charlie\in\mathcal X_\secparam$, and circuits
$\mu_\bob,\mu_\charlie:\mathcal X_\secparam\rightarrow
\mathcal Y_\secparam$, and outputs a circuit
$G^\star_{k,x_\bob,x_\charlie,\mu_\bob,\mu_\charlie}$ satisfying
\begin{equation}
G^\star_{k,x_\bob,x_\charlie,\mu_\bob,\mu_\charlie}(x)
:=
\begin{cases}
\mu_\bob(x),
    &x=x_\bob,\\
\mu_\charlie(x),
    &x=x_\charlie\text{ and }x_\charlie\neq x_\bob,\\
C_k(x),
    &\text{otherwise.}
\end{cases}
\label{eq:generalized-puncturing}
\end{equation}
All compared circuits are padded to the same size.
\end{definition}

For any distribution $\mathcal D_X$ over
$\mathcal X_\secparam\times\mathcal X_\secparam$ and any QPT
adversarial triplet $\advupo=(\alice,\bob,\charlie)$, consider the
following experiment.

\par\noindent
$\gam{\mathsf{GenUPO}_{\Pi,\advupo,\mathcal{D_X},\mathcal{C}}}(1^\secparam)$:
\begin{enumerate}
    \item $\alice(1^\secparam)$ sends
    $(k,\mu_\bob,\mu_\charlie)$ to $\ch$, where
    $k\in\mathcal K_\secparam$ and
    $\mu_\bob,\mu_\charlie:\mathcal X_\secparam\rightarrow
    \mathcal Y_\secparam$.

    \item $\ch$ samples
    $(x_\bob,x_\charlie)\gets\mathcal D_X(1^\secparam)$ and
    $b\getsr\bit$, and computes
    $G^\star\gets\genpuncture
    (k,x_\bob,x_\charlie,\mu_\bob,\mu_\charlie)$.

    \item $\ch$ generates
    $\rho_b\gets\Obf(1^\secparam,C_k)$ if $b=0$, and
    $\rho_b\gets\Obf(1^\secparam,G^\star)$ if $b=1$, and sends
    $\rho_b$ to $\alice$.

    \item $\alice$ outputs a joint state
    $\sigma_{\regbob,\regcharlie}$ and sends register $\regbob$ to
    $\bob$ and register $\regcharlie$ to $\charlie$.

    \item After the split, $\ch$ sends $x_\bob$ to $\bob$ and
    $x_\charlie$ to $\charlie$.

    \item $\bob$ and $\charlie$ output
    $\widehat b_\bob,\widehat b_\charlie$.

    \item The output of the experiment is $1$ if
    $\widehat b_\bob=\widehat b_\charlie=b$.
\end{enumerate}

\begin{definition}[$\mathcal D_X$-generalized UPO security]
A UPO scheme $\Pi$ satisfies $\mathcal D_X$-generalized UPO
security with respect to $\cal{C}$ if, for every QPT adversarial triplet
$\advupo=(\alice,\bob,\charlie)$, there exists a negligible
function $\negl(\cdot)$ such that
\begin{equation}
\Pr[
    \gam{\mathsf{GenUPO}_{\Pi,\advupo,\mathcal{D_X},\mathcal{C}}}
    (1^\secparam)=1
]
\leq
\frac{1}{2}+\negl(\secparam).
\end{equation}
\end{definition}

We use the following standard instantiations of the challenge
distribution. 

\begin{remark}\label{rem:instantiations}
    For a distribution $\mathcal D$ over
$\mathcal X_\secparam$, let $\mathcal D^2$ denote the product
distribution obtained by sampling $x_\bob,x_\charlie\gets\mathcal D$
independently. Let $\ID_{\mathcal D}$ denote the diagonal
distribution obtained by sampling $x\gets\mathcal D$ and outputting
$(x,x)$. Let $\unifdist$ denote the uniform distribution on
$\mathcal X_\secparam^2$. In particular,
$\IDU:=\ID_{\mathsf U}$ denotes the identical-uniform distribution,
which samples $x\getsr\mathcal X_\secparam$ and outputs $(x,x)$.
Thus $\unifdist$-generalized UPO is the independent-uniform or
product-challenge notion, while $\ID_{\mathcal D}$-generalized UPO
is the diagonal notion and $\IDU$-generalized UPO is the
identical-uniform notion.
\end{remark}

\subsection{Correlated-Challenge Security}\label{sec:upo-new-defn}
We generalize the prior definition of UPO in three ways: first, we allow the challenge bits for $\bob$ and $\charlie$ to be arbitrarily correlated; next, we allow both the challenge points and the two puncturing bits to be sampled jointly, thus allowing them to be correlated; finally, we allow the sampler to output classical auxiliary information $z_{\mathsf{aux}}$ that is given to $\alice$ before she chooses the circuit and classical auxiliary information $w_\bob,w_\charlie$ that is given to $\bob,\charlie$, respectively, after the split\footnote{We note that auxiliary information leakage is a natural generalization which can become crucial for some of the applications of UPO (in fact, this generalization is crucial for the application to copy protection for compute-and-compare functions, see \Cref{thm:compute-and-compare-copy-protection}).}. We require the challenge bits to have uniform marginals conditioned on all the auxiliary information and the challenge points, and require each challenge point to have high average conditional min-entropy, as defined in \Cref{def:average-conditional-min-entropy}, conditioned on the pre-split auxiliary information.\footnote{Arbitrary correlation without these restrictions is not tolerable, since the auxiliary information or the challenge points may trivially reveal a challenge bit.}

\begin{definition}[Correlated High min-entropy UPO sampler]
\label{def:ac-upo-sampler}
Let $\kappa:=\kappa(\secparam)=\secparam^c$ for some constant $c>0$.
A correlated UPO sampler with $\kappa$-entropic marginals for a keyed-circuit class $\mathcal{C}=\{C_\secparam\}_{\secparam\in \NN}$ with input space $\cal{X}$ is a QPT algorithm $\sampler$ that, on input $1^\secparam$, outputs
$(z_{\mathsf{aux}},x_\bob,x_\charlie,b_\bob,b_\charlie,w_\bob,w_\charlie)
\gets\sampler(1^\secparam)$, where
$z_{\mathsf{aux}}\in\mathcal Z^{\mathsf{aux}}_\secparam$ is a polynomial-length classical auxiliary string,
$x_\bob,x_\charlie\in\mathcal X_\secparam$ and
$b_\bob,b_\charlie\in\bit$, and $w_\bob,w_\charlie$ are polynomial-length classical auxiliary strings, such that
\begin{enumerate}
    \item For every $(\bar z_{\mathsf{aux}},\bar x_\bob,\bar x_\charlie,\bar w_\bob,\bar w_\charlie)$ in the support of $(z_{\mathsf{aux}},x_\bob,x_\charlie,w_\bob,w_\charlie)$, every $X\in\{\bob,\charlie\}$, and every $a\in\bit$,
\begin{equation}
\Pr[
    b_X=a
    \mid
    \substack{
    z_{\mathsf{aux}}=\bar z_{\mathsf{aux}},\ x_\bob=\bar x_\bob,\ x_\charlie=\bar x_\charlie,\\
    w_\bob=\bar w_\bob,\ w_\charlie=\bar w_\charlie
    }
]
=
\frac{1}{2}.
\label{eq:ac-upo-uniform-marginals}
\end{equation}
    \item For every $X\in\{\bob,\charlie\}$,
    $\Hmin(x_X\mid z_{\mathsf{aux}})\geq\kappa(\secparam)$.
\end{enumerate}
We identify a sampler with no explicit auxiliary output with the special case $z_{\mathsf{aux}}:=w_\bob:=w_\charlie:=\bot$.
\end{definition}

For $k,x_\bob,x_\charlie,b_\bob,b_\charlie,
\mu_\bob,\mu_\charlie$ as above, define
\begin{equation}
G^{\mathsf{ac}}_{k,x_\bob,x_\charlie,b_\bob,b_\charlie,
\mu_\bob,\mu_\charlie}(x)
:=
\begin{cases}
\mu_\bob(x),
    &x=x_\bob,\ b_\bob=1,\\
\mu_\charlie(x),
    &x=x_\charlie,\ x_\charlie\neq x_\bob,\ b_\charlie=1,\\
C_k(x),
    &\text{otherwise.}
\end{cases}
\label{eq:ac-upo-puncturing}
\end{equation}
Equivalently,
$G^{\mathsf{ac}}\gets\acpuncture
(k,x_\bob,x_\charlie,b_\bob,b_\charlie,
\mu_\bob,\mu_\charlie)$.
If $x_\bob=x_\charlie$ and $b_\bob\neq b_\charlie$, we follow the convention used in~\cite{ABH+26} that $\bob$'s challenge bit takes priority, i.e., the common point $x_\bob=x_\charlie$ is punctured using $\mu_\bob$ if $b_\bob=1$, and is left unpunctured if $b_\bob=0$, oblivious of $\charlie$'s challenge bit $b_\charlie$.

\begin{definition}[Correlated UPO security]
\label{def:ac-upo-security}
For any QPT adversarial triplet
$\advupo=(\alice,\bob,\charlie)$ against a UPO scheme
$\Pi=(\Obf,\Eval)$ and any correlated UPO sampler
$\sampler$ with $\kappa$-entropic marginals, write $\alice=(\alice_0,\alice_1)$, where $\alice_0$ may retain an arbitrary quantum state for $\alice_1$, and consider the following experiment.

\par\noindent
$\gam{\acupo_{\Pi,\advupo,\sampler,\mathcal{C}}}(1^\secparam)$:
\begin{enumerate}
    \item $\ch$ samples
    \begin{equation*}
    (z_{\mathsf{aux}},x_\bob,x_\charlie,b_\bob,b_\charlie,w_\bob,w_\charlie)
    \gets\sampler(1^\secparam)
    \end{equation*}
    and sends $z_{\mathsf{aux}}$ to $\alice_0$.

    \item $\alice_0(1^\secparam,z_{\mathsf{aux}})$ sends
    $(k,\mu_\bob,\mu_\charlie)$ to $\ch$, where
    $k\in\mathcal K_\secparam$ and
    $\mu_\bob,\mu_\charlie:\mathcal X_\secparam\rightarrow
    \mathcal Y_\secparam$, while retaining an arbitrary quantum state.

    \item $\ch$ computes
    $G^{\mathsf{ac}}\gets\acpuncture
    (k,x_\bob,x_\charlie,b_\bob,b_\charlie,
    \mu_\bob,\mu_\charlie)$, generates
    $\rho\gets\Obf(1^\secparam,G^{\mathsf{ac}})$, and sends
    $\rho$ to $\alice_1$.

    \item $\alice_1$ outputs a joint state
    $\sigma_{\regbob,\regcharlie}$ and sends register $\regbob$ to
    $\bob$ and register $\regcharlie$ to $\charlie$.

    \item After the split, $\ch$ sends $(x_\bob,w_\bob)$ to $\bob$ and
    $(x_\charlie,w_\charlie)$ to $\charlie$.

    \item $\bob$ and $\charlie$ output
    $\widehat b_\bob,\widehat b_\charlie$.

    \item The output of the experiment is $1$ if
    $\widehat b_\bob=b_\bob$ and
    $\widehat b_\charlie=b_\charlie$.
\end{enumerate}

A UPO scheme $\Pi$ satisfies $\kappa$-entropic \acupo security with respect to $\cal{C}$ if, for every correlated UPO sampler $\sampler$ with $\kappa$-entropic marginals for $\cal{C}$, and every QPT adversarial triplet $\advupo$, there exists a negligible function $\negl(\cdot)$ such that
\begin{equation}
\Pr[
    \gam{\acupo_{\Pi,\advupo,\sampler,\mathcal{C}}}(1^\secparam)=1
]
\leq
\frac{1}{2}+\negl(\secparam).
\end{equation}
\end{definition}

\begin{definition}[$\io$-security for UPO, {adapted from~\cite[Definition 3.15]{ABH+26}}]
\label{def:upo-io-security}
A UPO scheme $\Pi=(\Obf,\Eval)$ for a circuit class $\mathcal C$ satisfies $\io$-security if, for every QPT adversarial pair $\adv_{\io}=(\adv_{\io,0},\adv_{\io,1})$, where $\adv_{\io,0}(1^\secparam)$ outputs same-size circuits $C_0,C_1\in\mathcal C_\secparam$ and a quantum state $\sigma_{\mathsf{aux}}$ such that $C_0(x)=C_1(x)$ for every input $x$, there exists a negligible function $\negl(\cdot)$ such that
\begin{equation}
\Pr\left[
\widehat b=b:
\begin{array}{l}
(C_0,C_1,\sigma_{\mathsf{aux}})\gets\adv_{\io,0}(1^\secparam),\\
b\getsr\bit,\quad
\rho_b\gets\Obf(1^\secparam,C_b),\\
\widehat b\gets\adv_{\io,1}(\sigma_{\mathsf{aux}},\rho_b)
\end{array}
\right]
\leq
\frac{1}{2}+\negl(\secparam).
\end{equation}
\end{definition}
Next, we show the following composition theorem for correlated security of UPO using the same argument as in the proof of~\cite[Theorem~12]{AB24}.
\begin{theorem}[Adding $\io$ security, adapted from
{\cite[Theorem~12]{AB24},\cite[Lemma~3.16]{ABH+26}}]
\label{thm:correlated-upo-composition}
Let
$\mathsf{Circ}_{\ell_{\mathsf{in}},\ell_{\mathsf{out}}}$
denote the full class of polynomial-size circuits from
$\zo^{\ell_{\mathsf{in}}(\secparam)}$ to
$\zo^{\ell_{\mathsf{out}}(\secparam)}$, regarded as a keyed circuit
class whose keys are circuit descriptions. If there exist a
post-quantum secure indistinguishability obfuscation scheme and a UPO
scheme for
$\mathsf{Circ}_{\ell_{\mathsf{in}},\ell_{\mathsf{out}}}$
satisfying $\kappa$-entropic correlated UPO security, then there
exists a UPO scheme for this circuit class satisfying both
$\kappa$-entropic correlated UPO security and $\io$-security.
Consequently, the same holds for every polynomial-size keyed subclass
$\mathcal C\subseteq
\mathsf{Circ}_{\ell_{\mathsf{in}},\ell_{\mathsf{out}}}$.
\end{theorem}

\begin{proof}
This follows using the same add-$\io$ transformation of
\cite[Lemma~3.16]{ABH+26}, based on the composition theorem of
\cite[Theorem~12]{AB24}. The same reduction applies to the correlated
experiment by using the sampler's entire output tuple, including the
pre-split and post-split auxiliary strings, unchanged.
\end{proof}

\begin{lemma}[Relation to prior UPO notions]
\label{lem:ac-upo-implies-prior}
For any $\kappa:=\kappa(\secparam)=\secparam^c$ for some constant $c> 0$, $\kappa$-entropic \acupo security implies $\mathcal D_X$-generalized UPO security for every QPT-samplable distribution $\mathcal D_X$ over challenge pairs, that induces marginal challenge distribution on each of the two challenges with min-entropy at least $\kappa(\secparam)$. In particular, for every QPT-samplable distribution $\mathcal D$ satisfying $\Hmin(\mathcal D)\geq \kappa(\secparam)$  (such as the uniform  random distribution), it implies $\mathcal D^2$-generalized UPO security and $\ID_{\mathcal D}$-generalized UPO security.\footnote{The only known UPO security notion that is not covered by correlated UPO security is the UPO+ notion defined in~\cite{ABH+26}, which can be seen as a completely different flavor of security notion. In particular, UPO${}^+$ has a different winning
condition and is not an instance of the correlated-UPO experiment.}.

\end{lemma}

\begin{proof}
Given $\mathcal D_X$, define a sampler that sets $z_{\mathsf{aux}}:=w_\bob:=w_\charlie:=\bot$, samples $(x_\bob,x_\charlie)\gets\mathcal D_X$, samples $b\getsr\bit$, and sets $b_\bob=b_\charlie:=b$.
By the hypothesis of the lemma, $\mathcal D_X$, the resulting sampler has $\kappa$-entropic marginals on the individual challenges. Moreover, clearly, this sampler satisfies \Cref{eq:ac-upo-uniform-marginals}.
Since every nonempty sampler output is hidden from $\alice_0$ before it sends $(k,\mu_\bob,\mu_\charlie)$, sampling the tuple before this message does not change the joint distribution.

If $b=0$, the punctured circuit is $C_k$, while if $b=1$, it is exactly the generalized
punctured circuit from \Cref{eq:generalized-puncturing}, including the Bob-priority convention used in the generalized UPO security definition when the points coincide. Hence the two security experiments are identically distributed. Taking
$\mathcal D_X=\mathcal D^2$, $\mathcal D_X=\ID_{\mathcal D}$,
$\mathcal D_X=\unifdist$, or $\mathcal D_X=\IDU$ gives the stated
instantiations. Vanilla UPO security follows since it is implied by generalized UPO security. 
\end{proof}

%% file: sde/main.tex
\section{Single-Decryptor Encryption Construction}
\label{sec:sde-construction}

In this section, we construct an SDE scheme for the message space
$\messpa=\zo^{\meslen(\lambda)}$ and prove that it satisfies correlated-challenge decision security. The construction is a minor variant of the coset-state SDE scheme of Kitagawa and Yamakawa (\cite{TCC:KitYam25}): we combine the two ciphertext
programs in their construction into a single program with a selector bit.

\subsection{Construction}

Let $\io$ be an indistinguishability obfuscation scheme and set
$n=n(\lambda):=2\lambda$. For the ciphertext programs used below, let
$\Omega_{\io}=\zo^{\ell_{\io}(\lambda)}$ be the random-tape space of $\io$.
We regard $\io$ as deterministic once this tape is fixed. All circuits compared in an $\io$ hybrid are padded to the same size, and $\bot$ is represented by a fixed string distinct from the message outputs.

\paragraph{\underline{$\sde.\keygen(1^\lambda)$}}
\begin{enumerate}
    \item Sample a uniformly random $n$-dimensional subspace
    $A\leq\FF_2^{2n}$ and independent $s,t\getsr\FF_2^{2n}$.
    \item Compute $\obfd{M_0}\gets\io(1^\lambda,M_0)$ and
    $\obfd{M_1}\gets\io(1^\lambda,M_1)$, where $M_0,M_1$ are the following
    programs.

    \begin{mdframed}
        {\bf \underline{$M_0(v)$}}

        {\bf Hardcoded: $A+s$}
        \begin{enumerate}[label=\arabic*.]
            \item If $v\in A+s$, output $1$; otherwise, output $0$.
        \end{enumerate}
    \end{mdframed}

    \begin{mdframed}
        {\bf \underline{$M_1(v)$}}

        {\bf Hardcoded: $A^\perp+t$}
        \begin{enumerate}[label=\arabic*.]
            \item If $v\in A^\perp+t$, output $1$; otherwise, output $0$.
        \end{enumerate}
    \end{mdframed}

    \item Set $\pk:=(\obfd{M_0},\obfd{M_1})$ and $\qsk:=\ket{A_{s,t}}
        :=\frac{1}{\sqrt{|A|}}
        \sum_{a\in A}(-1)^{\langle a,t\rangle}\ket{a+s}.$
    Output $(\pk,\qsk)$.
\end{enumerate}

\paragraph{\underline{$\sde.\enc(\pk,\mes)$}}
The randomness space of encryption is
$\mathcal R:=\messpa\times\Omega_{\io}$. On randomness
$\tau=(r,\omega)\getsr\mathcal R$, proceed as follows.
\begin{enumerate}
    \item Parse $\pk=(\obfd{M_0},\obfd{M_1})$.
    \item Compute $\obfd{P}\gets\io(1^\lambda,P_{\mes,r};\omega)$, where
    $P_{\mes,r}$ is the following program.

    \begin{mdframed}
        {\bf \underline{$P_{\mes,r}(c,v)$}}

        {\bf Hardcoded: $\obfd{M_0},\obfd{M_1},\mes,r$}
        \begin{enumerate}[label=\arabic*.]
            \item If $c=0$ and $\obfd{M_0}(v)=1$, output $\mes\oplus r$.
            \item If $c=1$ and $\obfd{M_1}(v)=1$, output $r$.
            \item Otherwise, output $\bot$.
        \end{enumerate}
    \end{mdframed}

    \item Output $\ct:=\obfd{P}$.
\end{enumerate}
Thus, encryption is fully classical and uses an independent uniform pad and
an ordinary fresh $\io$ tape.

\paragraph{\underline{$\sde.\dec(\qsk,\ct)$}}
\begin{enumerate}
    \item Parse $\obfd{P}=\ct$.
    \item Run $\obfd{P}(0,\cdot)$ on $\qsk$ coherently, let $y_0$ denote the outcome. Then, apply the inverse unitary on $\qsk$ to restore the state $\qsk$ (as in \cref{prelem:gentlemes}).
    \item Apply $H^{\otimes 2n(\lambda)}$ to $\qsk$.
    \item Run $\obfd{P}(1,\cdot)$ on $\qsk$ coherently, let $y_1$ denote the
    outcome. Then, again apply the inverse of the last unitary to recover the state $\qsk$ (as in \cref{prelem:gentlemes}).
    \item Apply $H^{\otimes 2n(\lambda)}$ to $\qsk$.
    \item Output $y_0\oplus y_1$.
\end{enumerate}

 \begin{theorem}[Correlated-challenge-secure SDE \ifconference {(Restatement of \Cref{thm:sde-security-restated})}\fi]
  \label{thm:sde-security}
  Assuming post-quantum secure indistinguishability obfuscation and post-quantum one-way functions, the scheme above is perfectly correct and satisfies correlated-challenge decision security (\cref{defn:sde-new-defn}). Consequently, it satisfies strong CPA and $\mathsf{CPA}^{+}$ anti-piracy security.
\end{theorem}

\subsection{Correctness}
\begin{lemma}\label{lem:sde-correctness}
The above construction is perfectly correct.    
\end{lemma}
\begin{proof}[Proof of \Cref{lem:sde-correctness}]
Fix an honestly generated key pair and a ciphertext encrypting $\mes\in\messpa$ with pad $r$. For every $a\in A$, correctness of $\io$ gives $\obfd{P}(0,a+s)=\mes\oplus r$. The first evaluation therefore returns $y_0=\mes\oplus r$ with certainty and leaves the coset state unchanged. Moreover, $H^{\otimes 2n}\ket{A_{s,t}} =(-1)^{\langle s,t\rangle}\ket{A^\perp_{t,s}}$. Since $\obfd{P}(1,a^\perp+t)=r$ for every $a^\perp\in A^\perp$, the second evaluation returns $y_1=r$ with certainty. Thus decryption outputs $(\mes\oplus r)\oplus r=\mes$, and the final Hadamard transform restores $\ket{A_{s,t}}$ up to a global phase.
\end{proof}

\subsection{Proof of security}
\begin{proof}[Proof of \Cref{thm:sde-security}]
Fix an admissible sampler $\sampler$ and a QPT adversary $\adve=(\alice,\bob,\charlie)$ for the correlated-challenge experiment. Let $N$ be the number of challenge pairs per party output by $\alice$; since it is encoded in unary, $N$ is polynomially bounded. For $X\in\{\bob,\charlie\}$ and $i\in[N]$, write $(\mu_{X,i,0},\mu_{X,i,1})$ for $X$'s $i$-th pair, identifying the two parties' pairs in identical mode. Set $\delta_{X,i}:=\mu_{X,i,0}\oplus\mu_{X,i,1}$ and $\mu_{X,i}:=\mu_{X,i,b_X}=\mu_{X,i,0}\oplus b_X\cdot\delta_{X,i}$, where $b\cdot z=z$ if $b=1$ and $b\cdot z=0^{\meslen(\lambda)}$ otherwise.

For the $i^{th}$ challenge ciphertext that party $X\in\{\bob,\charlie\}$ receives, let $r_{X,i}$ denote the uniform pad hardcoded in its obfuscated circuit. In separate mode, the pads $(r_{X,i})_{X\in\{\bob,\charlie\},\,i\in[N]}$ are mutually independent. In identical mode, the challenger samples one pad $r_i$ for each $i\in[N]$, uses it in the single ciphertext for that coordinate, and gives the same ciphertext to both parties; equivalently, $r_{\bob,i}=r_{\charlie,i}:=r_i$.

\paragraph{$\hybrid_0$: Real experiment.}
This is the correlated-challenge experiment of \Cref{def:sde-corr-game}, instantiated with the scheme above.

\paragraph{$\hybrid_1$: Use one uniform challenge bit.}
In separate mode, run $\sampler(\pk,\sep)$ and denote its output by $(b'_\bob,b'_\charlie)$. Set $\Delta:=b'_\bob\oplus b'_\charlie$. In identical mode, set $\Delta:=0$. Next sample $\widetilde e\getsr\bit$. Use $b_\bob:=\widetilde e$ as Bob's challenge bit and $b_\charlie:=\widetilde e\oplus\Delta$ as Charlie's challenge bit. All other steps are unchanged.

We show that $\hybrid_0$ and $\hybrid_1$ are identically distributed. Fix $\pk$. For $i,j\in\bit$, let $p_{ij}$ be the probability that $\sampler(\pk,\sep)$ outputs $(i,j)$. The uniform-marginal condition gives $p_{00}=p_{11}$ and $p_{01}=p_{10}$. Consequently, conditioned on any value of $\Delta$ in its support, $b'_\bob$ is uniform and $b'_\charlie=b'_\bob\oplus\Delta$. Thus $(\widetilde e,\widetilde e\oplus\Delta)$ in $\hybrid_1$ has the same distribution as $(b'_\bob,b'_\charlie)$ in $\hybrid_0$. In identical mode, the sampler's output is a single uniform bit, which has the same distribution as $\widetilde e$. Finally, $\alice$ produces its messages and split registers before the sampler is invoked, and the sampler receives only $\pk$ and uses fresh coins. The equality therefore holds jointly with the complete output of $\alice$.

\paragraph{$\hybrid_2$: Mask the coset parity to get to the challenge bit.}
After sampling $(A,s,t)$, additionally sample $u\getsr A\setminus\{0\}$, $v\getsr A^\perp\setminus\{0\}$, and $\zeta\getsr\bit$. Define $\alpha:=\langle v,s\rangle$, $\beta:=\langle u,t\rangle$, and $e:=\alpha\oplus\beta$, and set $b_\bob:=e\oplus\zeta$ and $b_\charlie:=e\oplus\zeta\oplus\Delta$. For every fixed choice of the remaining variables, $e\oplus\zeta$ is uniform. Hence $\hybrid_1$ and $\hybrid_2$ are identically distributed. For later use, define $d_\bob:=\zeta$ and $d_\charlie:=\zeta\oplus\Delta$; then $b_X=e\oplus d_X$ for $X\in\{\bob,\charlie\}$. This hybrid changes only the sampling of the challenge bits; the pads $r_{X,i}$ are still sampled as specified above.

\paragraph{$\hybrid_3$: Write the encryption pads using uniform strings.}
We change only how the pads, i.e., $(r_{X,i})_{X\in\{\bob,\charlie\},\,i\in[N]}$ from $\hybrid_2$ are sampled. In separate mode, sample mutually independent $q_{X,i}\getsr\messpa$ for every $X\in\{\bob,\charlie\}$ and $i\in[N]$. In identical mode, sample independent $q_i\getsr\messpa$ for $i\in[N]$ and set $q_{\bob,i}=q_{\charlie,i}:=q_i$. In either mode, define $r_{X,i}:=q_{X,i}\oplus\beta\cdot\delta_{X,i}$. Conditioned on all variables sampled before the $q$'s, the string $\beta\cdot\delta_{X,i}$ is fixed, so the map from $q_{X,i}$ to $r_{X,i}$ is a bijection on $\messpa$. Consequently, in separate mode the pads remain mutually independent and uniform. In identical mode, $\delta_{\bob,i}=\delta_{\charlie,i}$ and $q_{\bob,i}=q_{\charlie,i}$, so $r_{\bob,i}=r_{\charlie,i}$ is a single uniform pad, exactly as in $\hybrid_2$. The transformation is independent across coordinates, and therefore the complete joint distribution of the pads is unchanged. Hence $\hybrid_2$ and $\hybrid_3$ are identically distributed.

This way of sampling the pads prepares the next hybrid. Since $b_X=\alpha\oplus\beta\oplus d_X$, the two occurrences of $\beta\cdot\delta_{X,i}$ cancel in $\mu_{X,i,b_X}\oplus r_{X,i}$, giving $\mu_{X,i,b_X}\oplus r_{X,i}=q_{X,i}\oplus\mu_{X,i,0}\oplus(\alpha\oplus d_X)\cdot\delta_{X,i}$. Thus the first branch of the ciphertext program will use the value $\alpha$ available on $A+s$, while its second branch will use the value $\beta$ available on $A^\perp+t$.

\paragraph{$\hybrid_4$: Replace the ciphertext programs.}
For every $X,i$, set $\gamma_{X,i}:=d_X\cdot\delta_{X,i}$ and replace $P_{\mu_{X,i},r_{X,i}}$ by the following program $Q_{X,i}$.

\begin{mdframed}
{\bf \underline{$Q_{X,i}(c,w)$}}

{\bf Hardcoded: $\obfd{M_0},\obfd{M_1},u,v,q_{X,i},\gamma_{X,i}, \mu_{X,i,0},\delta_{X,i}$}
    \begin{enumerate}[label=\arabic*.]
\item If $c=0$ and $\obfd{M_0}(w)=1$, output $q_{X,i}\oplus\mu_{X,i,0}\oplus\gamma_{X,i} \oplus\langle v,w\rangle\cdot\delta_{X,i}$.
\item If $c=1$ and $\obfd{M_1}(w)=1$, output $q_{X,i}\oplus\langle u,w\rangle\cdot\delta_{X,i}$.
\item Otherwise, output $\bot$.
    \end{enumerate}
\end{mdframed}

The old and new programs are functionally equivalent. If $w\in A+s$, then $\langle v,w\rangle=\alpha$, and the first branch returns $q_{X,i}\oplus\mu_{X,i,0} \oplus(d_X\oplus\alpha)\cdot\delta_{X,i} =\mu_{X,i,b_X}\oplus r_{X,i}$. If $w\in A^\perp+t$, then $\langle u,w\rangle=\beta$, and the second branch returns $q_{X,i}\oplus\beta\cdot\delta_{X,i}=r_{X,i}$. Both programs return $\bot$ otherwise.

In separate mode, the $2N$ obfuscations use independent tapes, and hence we change the $2N$ programs inside the obfuscations one at a time using $\io$ security. In identical mode, there is one common old program and one common new program per coordinate $i\in [N]$, and we change the two identical copies of the obfuscations of the $N$ programs. Since $N$ is polynomial, $\lvert\Pr[\hybrid_3=1]-\Pr[\hybrid_4=1]\rvert\leq\negl(\lambda)$.

We reduce $\hybrid_4$ to computational decisional coset monogamy (\cref{cor:comp-dcmoe}), instantiated there with $n=\lambda$. Its challenger samples a $2\lambda$-dimensional subspace of $\FF_2^{4\lambda}$, exactly as in our construction, where $n=2\lambda$.

The splitter $\alice'$ receives the coset state and the two obfuscated membership programs for the cosets. It treats the programs as $\pk$ and runs $\alice$ on $(\pk,\ket{A_{s,t}})$. It samples $\Delta,\zeta,d_\bob,d_\charlie$, the $q$'s, and the $\io$ tapes as in $\hybrid_4$, and appends to party $X$'s register $\pk,d_X$ and $(\mu_{X,i,0},\delta_{X,i},q_{X,i},\omega_{X,i})_{i\in[N]}$. In identical mode, the parties receive the same $q_i$ and tape $\omega_i$ for each coordinate; otherwise, all tapes are independent. These values can be generated before $(u,v)$ are revealed and without knowing $(u,v)$ or $e$.

Upon receiving $(u,v)$, party $X'$ computes $\gamma_{X,i}=d_X\cdot\delta_{X,i}$, constructs and obfuscates every $Q_{X,i}$ using the tape it received for every $i\in[N]$ from $\alice'$, and runs $X$ on the resulting ciphertexts, and obtains $\widetilde b_X$. In particular, in the identical mode, the two parties use identical program descriptions and tapes, so they obtain the same obfuscation, exactly matching a single ciphertext copied to both parties. Finally, $X'$ outputs $z_X:=\widetilde b_X\oplus d_X$. Since $b_X=e\oplus d_X$, we have $z_X=e$ exactly when $\widetilde b_X=b_X$. Thus the reduction wins the coset-monogamy game exactly when $\adve$ wins $\hybrid_4$. By \cref{cor:comp-dcmoe}, $\Pr[\hybrid_4=1]\leq \frac12+\negl(\lambda)$, and hence $\Pr[\corrsde_{\sde,\adve,\sampler}(1^\lambda)=1]\leq \frac12+\negl(\lambda)$. Strong CPA and $\mathsf{CPA}^{+}$ security follow from \cref{thm:sde-corr-implies-strong-cpa,thm:sde-prior-relations}.
\end{proof}

%% file: upo/main.tex
\section{Unclonable Puncturable Obfuscation Construction}
\label{sec:upo-construction}

In this section, we prove the following theorem.
\begin{theorem}[Correlated UPO from $\io$ and LWE]
\label{thm:correlated-upo-construction}
Assume the existence of post-quantum secure $\io$ and the quantum hardness of LWE. Then, for every constant $c>0$ and every polynomial-size keyed circuit class $\mathcal C=\{C_k:\mathcal X_\lambda\rightarrow\mathcal Y_\lambda\}_{k\in\mathcal K_\lambda}$ where $\mathcal{X}=\zo^{\ell_{\mathsf{in}}},\mathcal{Y}:=\zo^{\ell_{\mathsf{out}}}$ and $\ell_{\mathsf{in}}=\ell_{\mathsf{in}}(\secparam)\geq \secparam^c$ for some constant $c>0$, then there exists a UPO scheme for $\mathcal C$ that satisfies $\lambda^c$-entropic correlated UPO security. Moreover, for the special case of $\IDU$-generalized UPO security, the LWE assumption can be replaced with post-quantum one-way functions.
\end{theorem}


\subsection{Construction}
\label{sec:UPO-construction}
Fix an arbitrary constant $c>0$, with respect to which we want to prove \Cref{thm:correlated-upo-construction}. The UPO scheme that we will use to prove \Cref{thm:correlated-upo-construction} is the UPO construction of~\cite{ABH+26,CG26}. Next, we recall the construction of~\cite{ABH+26} (also considered for copy protection in~\cite{CG26}, independently) as follows.

Let $n:=\lceil\lambda^{c/3}\rceil$, let $V:=\FF_2^{4n}$, and let $\mathcal Y_\lambda=\zo^{\ell_{\mathsf{out}}(\lambda)}$. Let $\pprf$ be a puncturable PRF with domain $\mathcal X_\lambda$ and range $\mathcal Y_\lambda$.

\paragraph{\underline{$\upo.\Obf(1^\lambda,C_k)$}}
\begin{enumerate}
    \item Sample a uniformly random $2n$-dimensional subspace $A\leq V$ and independent $s,t\getsr V$.
    \item Sample $\obfd{M_0}\gets\io(M_0)$ and $\obfd{M_1}\gets\io(M_1)$, where $M_0,M_1$ are the following programs.

    \begin{mdframed}
        {\bf \underline{$M_0(a)$}}

        {\bf Hardcoded: $A+s$}
        \begin{enumerate}[label=\arabic*.]
            \item If $a\in A+s$, output $1$; otherwise, output $0$.
        \end{enumerate}
    \end{mdframed}

    \begin{mdframed}
        {\bf \underline{$M_1(a)$}}

        {\bf Hardcoded: $A^\perp+t$}
        \begin{enumerate}[label=\arabic*.]
            \item If $a\in A^\perp+t$, output $1$; otherwise, output $0$.
        \end{enumerate}
    \end{mdframed}

    \item Sample $K\gets\pprf.\gen(1^\lambda)$.
    \item Sample $\obfd{Q_0}\gets\io(Q_0)$ and $\obfd{Q_1}\gets\io(Q_1)$, where $Q_0,Q_1$ are the following programs.

    \begin{mdframed}
        {\bf \underline{$Q_0(a,x)$}}

        {\bf Hardcoded: $\obfd{M_0},C_k,K$}
        \begin{enumerate}[label=\arabic*.]
            \item If $\obfd{M_0}(a)=0$, output $\bot$.
            \item Otherwise, output $F_K(x)\oplus C_k(x)$.
        \end{enumerate}
    \end{mdframed}

    \begin{mdframed}
        {\bf \underline{$Q_1(a,x)$}}

        {\bf Hardcoded: $\obfd{M_1},K$}
        \begin{enumerate}[label=\arabic*.]
            \item If $\obfd{M_1}(a)=0$, output $\bot$.
            \item Otherwise, output $F_K(x)$.
        \end{enumerate}
        ~
    \end{mdframed}

    \item Output $\rho_k:=\left(\ket{A_{s,t}},\obfd{Q_0},\obfd{Q_1}\right)$.
\end{enumerate}

\paragraph{\underline{$\upo.\Eval(\rho_k,x)$}}
\begin{enumerate}
    \item Parse $\rho_k=\left(\ket{A_{s,t}},\obfd{Q_0},\obfd{Q_1}\right)$, and let $\regi$ denote the register containing $\ket{A_{s,t}}$.
    \item Using a reversible implementation of $\obfd{Q_0}$, coherently evaluate $\obfd{Q_0}(\cdot,x)$ on $\regi$ into a fresh output register, and measure the output register to obtain $y_0$.
    \item Apply $H^{\otimes 4n}$ to $\regi$, obtaining $\ket{A^\perp_{t,s}}$ up to a global phase.
    \item Using a reversible implementation of $\obfd{Q_1}$, coherently evaluate $\obfd{Q_1}(\cdot,x)$ on $\regi$ into a fresh output register, and measure the output register to obtain $y_1$.
    \item Apply $H^{\otimes 4n}$ to $\regi$, restoring $\ket{A_{s,t}}$.
    \item Output $(\rho_k,y_0\oplus y_1)$.
\end{enumerate}

\begin{lemma}[Correctness]
The above construction is perfectly correct.
\end{lemma}

\begin{proof}
Every basis vector in the support of $\ket{A_{s,t}}$ belongs to $A+s$. Hence the first coherent evaluation returns $y_0=F_K(x)\oplus C_k(x)$ and does not disturb the coset state. After applying $H^{\otimes 4n}$, every basis vector in the support of $\ket{A^\perp_{t,s}}$ belongs to $A^\perp+t$, and the second coherent evaluation returns $y_1=F_K(x)$. The final Hadamard transform restores the input state, so the evaluation algorithm outputs $(\rho_k,C_k(x))$ with certainty.
\end{proof}

\subsection{Proof of Security}

The proof follows the template used in the proof of~\cite[Theorem 5.2]{ABH+26}. There are two main differences, first we use a probabilistic argument to rewrite the experiment arbitrary correlated challenge bits, as an experiment with a single challenge bit (see the indistinguishability argument in the indistinguishability proof of $\hybrid_0$ and $\hybrid_1$, which crucially relies on the restriction (\Cref{eq:ac-upo-uniform-marginals}) of the UPO sampler definition (\Cref{def:ac-upo-sampler})), which is suitable for reducing to the computational MOE experiment $\compdcmoe$, secondly, analogous to the~\cite[Lemma 5.3]{ABH+26}, we prove a new simulation lemma (see \Cref{lem:correlated-upo-masked-string-simulation}) for the correlated UPO security definition that itself allows auxiliary information leakage to the adversaries, (see \Cref{def:ac-upo-security}), to show that the pre-split auxiliary information in our final hybrid can be simulated to reduce the hybrid to a valid instance of \compdcmoe.

Let $m:=8n$, and for $u,v\in V$ define $\mathsf{msg}(u,v):=u\mathbin\|v\in\zo^m$. Let $\mathcal H$ be an efficiently sampleable universal hash family from $\mathcal X_\lambda$ to $\zo^m$. Let $\mathsf{LF}=(\mathsf{LF.Gen}_{\mathsf{inj}},\mathsf{LF.Gen}_{\mathsf{loss}})$ be the injective/lossy function family from~\cite{ABH+26}, with domain $\mathcal X_\lambda$, such that every function sampled in injective mode is injective, the two modes are computationally indistinguishable against QPT algorithms, and every function sampled in lossy mode has image size at most $2^{\ell_{\mathsf{img}}}$, where $\ell_{\mathsf{img}}\leq\lambda^{2c/3}$. We use the puncturable PRF and subspace-hiding obfuscation instantiated in~\cite{ABH+26}. For $a\in\bit$ and $z\in\mathcal Y_\lambda$, define $a\cdot z:=z$ if $a=1$, and $a\cdot z:=0^{\ell_{\mathsf{out}}}$ if $a=0$. Set $\varepsilon:=2^{-\lambda^{2c/3}}$. For all sufficiently large $\lambda$, $\lambda^c\geq 2m+2\ell_{\mathsf{img}}+2\log(1/\varepsilon)+3$. For jointly distributed classical random variables $X,Z$, we use average conditional min-entropy, defined as
\begin{equation}
\Hmin(X\mid Z)
:=
-\log\left(
\sum_z\Pr[Z=z]\max_x\Pr[X=x\mid Z=z]
\right).
\end{equation}
In particular, if $Y$ has support size at most $2^\ell$, then $\Hmin(X\mid Z,Y)\geq\Hmin(X\mid Z)-\ell$. 

\begin{proof}[Proof of \Cref{thm:correlated-upo-construction}]
We need to prove $\secparam^c$-entropic correlated UPO security for $\upo$.

Fix any correlated UPO sampler $\sampler$ satisfying $\Hmin(x_X\mid z_{\mathsf{aux}})\geq\lambda^c$ for every $X\in\{\bob,\charlie\}$, and a QPT triplet of adversaries $\advupo=(\alice,\bob,\charlie)$. We view $\alice$ as a two-stage algorithm with retained state. The pre-split auxiliary string $z_{\mathsf{aux}}$ is given to the first stage of $\alice$ before it chooses $(k,\mu_\bob,\mu_\charlie)$, while $w_\bob,w_\charlie$ are carried unchanged through all the hybrids and are given to $\bob,\charlie$, respectively, after the split.
  We consider the following set of hybrids.
\paragraph{$\hybrid_0$: Real experiment.} This is the real experiment $\gam{\acupo_{\Pi,\advupo,\sampler,\mathcal C}}(1^\lambda)$ from \Cref{def:ac-upo-security}, instantiated with $\upo$ as the construction.

\paragraph{$\hybrid_1$: Sample the challenge bits using one uniform bit.} Run $(z_{\mathsf{aux}},x_\bob,x_\charlie,b'_\bob,b'_\charlie,w_\bob,w_\charlie)\gets\sampler(1^\lambda)$, set $\Delta:=b'_\bob+b'_\charlie$, discard $b'_\bob,b'_\charlie$, sample $\widetilde e\getsr\bit$, and set $b_\bob:=\widetilde e$ and $b_\charlie:=\widetilde e+\Delta$. The remaining steps are unchanged.

We show that $\hybrid_0$ and $\hybrid_1$ are identically distributed. Fix $(\bar z_{\mathsf{aux}},\bar x_\bob,\bar x_\charlie,\bar w_\bob,\bar w_\charlie)$ in the support of $(z_{\mathsf{aux}},x_\bob,x_\charlie,w_\bob,w_\charlie)$, and let
\begin{equation*}
p_{ij}:=\Pr\!\left[
b_\bob=i,b_\charlie=j
\,\middle|\,
\substack{
z_{\mathsf{aux}}=\bar z_{\mathsf{aux}},\ x_\bob=\bar x_\bob,\ x_\charlie=\bar x_\charlie,\\
w_\bob=\bar w_\bob,\ w_\charlie=\bar w_\charlie
}
\right].
\end{equation*}
By the condition (\Cref{eq:ac-upo-uniform-marginals}) in \Cref{def:ac-upo-sampler}, it is easy to see that $p_{00}=p_{11}$ and $p_{01}=p_{10}$, and hence conditioned on $b_\bob\oplus b_\charlie=b$ for some bit $b$, the marginal distribution on $b_\bob$ as well as $b_\charlie$ remains uniform. Therefore, we conclude that conditioned on $(z_{\mathsf{aux}},x_\bob,x_\charlie,w_\bob,w_\charlie,\Delta)$, the bit $b_\bob$ is uniform and $b_\charlie=b_\bob+\Delta$.
The first-stage output and retained state of $\alice$ are obtained from $z_{\mathsf{aux}}$ and independent randomness, so adjoining them preserves this identical distribution.

\paragraph{$\hybrid_2$: Write the common uniform bit as the coset parity plus a uniform mask.} In the execution of $\Obf$, after sampling $(A,s,t)$, sample $u\getsr A\setminus\{0\}$, $v\getsr A^\perp\setminus\{0\}$, and $\zeta\getsr\bit$, and define $\alpha:=\langle v,s\rangle$, $\beta:=\langle u,t\rangle$, and $e:=\alpha+\beta$. Set $b_\bob:=e+\zeta$ and $b_\charlie:=e+\zeta+\Delta$. Since $e+\zeta$ is uniform for every fixed value of the remaining variables, $\hybrid_1$ and $\hybrid_2$ are identically distributed. Define $d_\bob:=\zeta$, $d_\charlie:=\zeta+\Delta$, and $\mathsf{diff}_X(x):=C_k(x)\oplus\mu_X(x)$ for $X\in\{\bob,\charlie\}$. Then $b_X=\alpha+\beta+d_X$.

\paragraph{$\hybrid_3$: Puncture the PRF and hardcode its values at the challenge points.} Let $\mathcal S:=\{x_\bob,x_\charlie\}$, and let $K_{\mathcal S}$ be the puncturable PRF key punctured at $\mathcal S$. If $x_\bob\neq x_\charlie$, let $y_X:=F_K(x_X)$ for every $X\in\{\bob,\charlie\}$. If $x_\bob=x_\charlie$, let $y_\bob:=F_K(x_\bob)$ and set $y_\charlie:=y_\bob$. Replace $Q_0,Q_1$ by the following programs $P_0,P_1$.

\begin{mdframed}
{\bf \underline{$P_0(a,x)$}}

{\bf Hardcoded: $\obfd{M_0},C_k,\mu_\bob,\mu_\charlie,K_{\mathcal S},x_\bob,x_\charlie,b_\bob,b_\charlie,y_\bob,y_\charlie$}
\begin{enumerate}[label=\arabic*.]
    \item If $\obfd{M_0}(a)=0$, output $\bot$.
    \item If $x=x_\bob$, output
    $y_\bob\oplus C_k(x)\oplus
    b_\bob\cdot\mathsf{diff}_\bob(x)$.
    \item If $x=x_\charlie$ and $x_\charlie\neq x_\bob$, output
    $y_\charlie\oplus C_k(x)\oplus
    b_\charlie\cdot\mathsf{diff}_\charlie(x)$.
    \item Output $F_{K_{\mathcal S}}(x)\oplus C_k(x)$.
\end{enumerate}
\end{mdframed}

\begin{mdframed}
{\bf \underline{$P_1(a,x)$}}

{\bf Hardcoded: $\obfd{M_1},K_{\mathcal S},x_\bob,x_\charlie,y_\bob,y_\charlie$}
\begin{enumerate}[label=\arabic*.]
    \item If $\obfd{M_1}(a)=0$, output $\bot$.
    \item If $x=x_\bob$, output $y_\bob$.
    \item If $x=x_\charlie$ and $x_\charlie\neq x_\bob$, output
    $y_\charlie$.
    \item Output $F_{K_{\mathcal S}}(x)$.
\end{enumerate}
\end{mdframed}

The relevant circuits in $\hybrid_2$ and $\hybrid_3$ have the same functionality. At $x=x_\bob$, the programs in $\hybrid_3$ output $y_\bob\oplus C_k(x_\bob)\oplus b_\bob\cdot\mathsf{diff}_\bob(x_\bob)$ and $y_\bob$, respectively, which are exactly the corresponding outputs of the programs in $\hybrid_2$. This also covers the case $x_\bob=x_\charlie$, since the Bob line is evaluated first. If $x_\bob\neq x_\charlie$, the same argument applies at $x=x_\charlie$. Finally, for $x\notin\mathcal S$, punctured correctness gives $F_{K_{\mathcal S}}(x)=F_K(x)$. Therefore, by $\io$-security, we conclude that $\hybrid_2$ and $\hybrid_3$ are computationally indistinguishable.

\paragraph{$\hybrid_4$: Replace the hardcoded PRF values by uniform strings.} If $x_\bob\neq x_\charlie$, replace $(F_K(x_\bob),F_K(x_\charlie))$ by independent $y_\bob,y_\charlie\getsr\mathcal Y_\lambda$. If $x_\bob=x_\charlie$, sample $y_\bob\getsr\mathcal Y_\lambda$ and set $y_\charlie:=y_\bob$. Puncturable-PRF security implies that $\hybrid_3$ and $\hybrid_4$ are computationally indistinguishable.

\paragraph{$\hybrid_5$: Rewrite the uniform values.} Sample independent $q_\bob,q_\charlie\getsr\mathcal Y_\lambda$. If $x_\bob\neq x_\charlie$, then for every $X\in\{\bob,\charlie\}$ set
\begin{equation}
y_X:=q_X\oplus\beta\cdot\mathsf{diff}_X(x_X).
\end{equation}
If $x_\bob=x_\charlie$, set
\begin{equation}
y_\bob=y_\charlie
:=
q_\bob\oplus\beta\cdot\mathsf{diff}_\bob(x_\bob),
\end{equation}
and let $q_\charlie$ be an unused independent uniform string. Since the values $y_X$ in $\hybrid_4$ are uniform, this is an exact change of variables. Replace $P_0,P_1$ by the following programs.

\begin{mdframed}
{\bf \underline{$P'_0(a,x)$}}

{\bf Hardcoded: $\obfd{M_0},C_k,\mu_\bob,\mu_\charlie,K_{\mathcal S},x_\bob,x_\charlie,q_\bob,q_\charlie,\alpha,d_\bob,d_\charlie$}
\begin{enumerate}[label=\arabic*.]
    \item If $\obfd{M_0}(a)=0$, output $\bot$.
    \item If $x=x_\bob$, output
    $q_\bob\oplus C_k(x)\oplus
    (\alpha+d_\bob)\cdot\mathsf{diff}_\bob(x)$.
    \item If $x=x_\charlie$ and $x_\charlie\neq x_\bob$, output
    $q_\charlie\oplus C_k(x)\oplus
    (\alpha+d_\charlie)\cdot\mathsf{diff}_\charlie(x)$.
    \item Output $F_{K_{\mathcal S}}(x)\oplus C_k(x)$.
\end{enumerate}
\end{mdframed}

\begin{mdframed}
{\bf \underline{$P'_1(a,x)$}}

{\bf Hardcoded: $\obfd{M_1},C_k,\mu_\bob,\mu_\charlie,K_{\mathcal S},x_\bob,x_\charlie,q_\bob,q_\charlie,\beta$}
\begin{enumerate}[label=\arabic*.]
    \item If $\obfd{M_1}(a)=0$, output $\bot$.
    \item If $x=x_\bob$, output
    $q_\bob\oplus
    \beta\cdot\mathsf{diff}_\bob(x)$.
    \item If $x=x_\charlie$ and $x_\charlie\neq x_\bob$, output
    $q_\charlie\oplus
    \beta\cdot\mathsf{diff}_\charlie(x)$.
    \item Output $F_{K_{\mathcal S}}(x)$.
\end{enumerate}
\end{mdframed}

The relevant circuits in $\hybrid_4$ and $\hybrid_5$ have the same functionality. Suppose first that $x_\bob\neq x_\charlie$ and fix $X\in\{\bob,\charlie\}$. By the definition of $y_X$,
\begin{equation}
q_X\oplus\beta\cdot\mathsf{diff}_X(x_X)=y_X.
\end{equation}
Moreover, since $b_X=\alpha+\beta+d_X$,
\begin{align*}
&q_X\oplus C_k(x_X)\oplus
(\alpha+d_X)\cdot\mathsf{diff}_X(x_X)\\
&\qquad=
y_X\oplus C_k(x_X)\oplus
(\alpha+\beta+d_X)\cdot\mathsf{diff}_X(x_X)\\
&\qquad=
y_X\oplus C_k(x_X)\oplus
b_X\cdot\mathsf{diff}_X(x_X).
\end{align*}
Thus, the individual outputs of $P'_0,P'_1$ at $x_X$ are exactly the corresponding outputs of $P_0,P_1$. If $x_\bob=x_\charlie$, the Bob line is evaluated first, and the same calculation applies with $X=\bob$. At every other point, the programs are unchanged. Therefore, by $\io$-security we conclude that $\hybrid_4$ and $\hybrid_5$ are computationally indistinguishable.

\paragraph{$\hybrid_6$: Replace the explicit point comparisons by injective-function comparisons.}
For every $X\in\{\bob,\charlie\}$, sample
$h_X\getsr\mathcal H$ and
$L_X\gets\mathsf{LF.Gen}_{\mathsf{inj}}(1^\lambda)$ independently,
and set
$\eta_X:=L_X(x_X)$ and
$\mathsf{str}_X:=\mathsf{msg}(u,v)\oplus h_X(x_X)$.
Replace $P'_0,P'_1$ by the following programs $R_0,R_1$.

\begin{mdframed}
{\bf \underline{$R_0(a,x)$}}

{\bf Hardcoded: $\obfd{M_0},C_k,\mu_\bob,\mu_\charlie,
K_{\mathcal S},q_\bob,q_\charlie,d_\bob,d_\charlie,
L_\bob,L_\charlie,\eta_\bob,\eta_\charlie,
h_\bob,h_\charlie,\mathsf{str}_\bob,\mathsf{str}_\charlie$}
\begin{enumerate}[label=\arabic*.]
    \item If $\obfd{M_0}(a)=0$, output $\bot$.
    \item If $L_\bob(x)=\eta_\bob$, parse
    $(u'_\bob,v'_\bob):=
    \mathsf{str}_\bob\oplus h_\bob(x)$ and output
    \[
    q_\bob\oplus C_k(x)\oplus
    \bigl(\langle v'_\bob,a\rangle+d_\bob\bigr)
    \cdot\mathsf{diff}_\bob(x).
    \]
    \item If $L_\charlie(x)=\eta_\charlie$, parse
    $(u'_\charlie,v'_\charlie):=
    \mathsf{str}_\charlie\oplus h_\charlie(x)$ and output
    \[
    q_\charlie\oplus C_k(x)\oplus
    \bigl(\langle v'_\charlie,a\rangle+d_\charlie\bigr)
    \cdot\mathsf{diff}_\charlie(x).
    \]
    \item Output $F_{K_{\mathcal S}}(x)\oplus C_k(x)$.
\end{enumerate}
\end{mdframed}

\begin{mdframed}
{\bf \underline{$R_1(a,x)$}}

{\bf Hardcoded: $\obfd{M_1},C_k,\mu_\bob,\mu_\charlie,
K_{\mathcal S},q_\bob,q_\charlie,
L_\bob,L_\charlie,\eta_\bob,\eta_\charlie,
h_\bob,h_\charlie,\mathsf{str}_\bob,\mathsf{str}_\charlie$}
\begin{enumerate}[label=\arabic*.]
    \item If $\obfd{M_1}(a)=0$, output $\bot$.
    \item If $L_\bob(x)=\eta_\bob$, parse
    $(u'_\bob,v'_\bob):=
    \mathsf{str}_\bob\oplus h_\bob(x)$ and output
    \[
    q_\bob\oplus
    \langle u'_\bob,a\rangle
    \cdot\mathsf{diff}_\bob(x).
    \]
    \item If $L_\charlie(x)=\eta_\charlie$, parse
    $(u'_\charlie,v'_\charlie):=
    \mathsf{str}_\charlie\oplus h_\charlie(x)$ and output
    \[
    q_\charlie\oplus
    \langle u'_\charlie,a\rangle
    \cdot\mathsf{diff}_\charlie(x).
    \]
    \item Output $F_{K_{\mathcal S}}(x)$.
\end{enumerate}
\end{mdframed}

The relevant circuits in $\hybrid_5$ and $\hybrid_6$ have the
same functionality. Suppose first that
$x_\bob\neq x_\charlie$. At $x=x_X$, the equality
$L_X(x_X)=\eta_X$ holds, while the equality corresponding to the
other recipient does not hold by injectivity. Moreover,
$\mathsf{str}_X\oplus h_X(x_X)=u\|v$. For every
$a\in A+s$, it holds that
$\langle v,a\rangle=\langle v,s\rangle=\alpha$, and for every
$a\in A^\perp+t$, it holds that
$\langle u,a\rangle=\langle u,t\rangle=\beta$. Therefore, the
individual outputs of $R_0$ and $R_1$ at $x_X$ are exactly the
corresponding outputs of $P'_0$ and $P'_1$.

If $x_\bob=x_\charlie$, both equalities hold at the common point,
but the Bob equality is evaluated first in both programs. Hence,
the programs use the Bob values independently of $\Delta$, exactly
as prescribed by \Cref{eq:ac-upo-puncturing}. At every
$x\notin\mathcal S$, neither equality holds, and the final lines of
the programs are evaluated. Thus, the two pairs of programs are
functionally equivalent. Therefore, by $\io$-security, we conclude that $\hybrid_5$ and $\hybrid_6$ are computationally indistinguishable.

\paragraph{$\hybrid_7$: Replace the punctured PRF key with the unpunctured key.} Replace $K_{\mathcal S}$ by $K$ in the final line of $R_0,R_1$. The relevant circuits have the same functionality, since one of the preceding equality tests holds at every point in $\mathcal S$, and $F_{K_{\mathcal S}}(x)=F_K(x)$ for every $x\notin\mathcal S$. Therefore, by $\io$-security, we conclude that $\hybrid_6$ and $\hybrid_7$ are computationally indistinguishable.

\paragraph{$\hybrid_8$: Switch the two functions to lossy mode.} Replace $L_\bob$ and $L_\charlie$, one at a time, by functions sampled using $\mathsf{LF.Gen}_{\mathsf{loss}}(1^\lambda)$, and recompute $\eta_X:=L_X(x_X)$ after each replacement. Injective/lossy mode indistinguishability implies that $\hybrid_7$ and $\hybrid_8$ are computationally indistinguishable.

Define the auxiliary-information generator (see \Cref{def:comp-dcmoe-aux}) as follows. The algorithm $\auxsetup(1^\lambda)$ samples $h_\bob,h_\charlie\getsr\mathcal H$ and independent $L_\bob,L_\charlie\gets\mathsf{LF.Gen}_{\mathsf{loss}}(1^\lambda)$, and outputs $pp:=(L_\bob,L_\charlie,h_\bob,h_\charlie)$. On input $(pp,A,u,v)$, the algorithm $\auxgen$ runs $(z_{\mathsf{aux}},x_\bob,x_\charlie,b'_\bob,b'_\charlie,w_\bob,w_\charlie)\gets\sampler(1^\lambda)$, sets $\Delta:=b'_\bob+b'_\charlie$, computes $\eta_X:=L_X(x_X)$ and $\mathsf{str}_X:=\mathsf{msg}(u,v)\oplus h_X(x_X)$ for $X\in\{\bob,\charlie\}$, and outputs
\begin{align*}
z_{1,\bob}&:=(z_{\mathsf{aux}},\Delta,\eta_\bob,\eta_\charlie,\mathsf{str}_\bob,\mathsf{str}_\charlie),&z_{2,\bob}&:=(x_\bob,w_\bob),\\
z_{1,\charlie}&:=\bot,&z_{2,\charlie}&:=(x_\charlie,w_\charlie).
\end{align*}
The simulator $\auxsim(pp)$ runs $(z_{\mathsf{aux}},x_\bob,x_\charlie,b'_\bob,b'_\charlie,w_\bob,w_\charlie)\gets\sampler(1^\lambda)$, computes $\Delta,\eta_\bob,\eta_\charlie$, discards $w_\bob,w_\charlie$, samples independent $U_\bob,U_\charlie\getsr\zo^m$, and outputs $\widetilde z_{1,\bob}:=(z_{\mathsf{aux}},\Delta,\eta_\bob,\eta_\charlie,U_\bob,U_\charlie)$ and $\widetilde z_{1,\charlie}:=\bot$. Next we show the statistical simulatability condition for the above auxiliary-information generator, as requried in \Cref{def:comp-dcmoe-aux}.

\begin{lemma}[Statistical simulation of the pre-split auxiliary information]
\label{lem:correlated-upo-masked-string-simulation}
For every fixed subspace $A\leq V$, $u\in A\setminus\{0\}$, and
$v\in A^\perp\setminus\{0\}$, the auxiliary-information algorithms
defined above satisfy the statistical simulatability condition in
\Cref{eq:comp-dcmoe-aux-simulation} with statistical distance at
most $\varepsilon$.
\end{lemma}
\begin{proof}[Proof of \Cref{lem:correlated-upo-masked-string-simulation}]
Let $(Z_{\mathsf{aux}},X_\bob,X_\charlie,b_\bob,b_\charlie,W_\bob,W_\charlie)\gets\sampler(1^\lambda)$ and set $\Delta:=b_\bob+b_\charlie$. The variables $W_\bob,W_\charlie$ are carried as post-split auxiliary information and do not occur in the pre-split strings. Condition on the descriptions of $L_\bob,L_\charlie$, set $\eta_X:=L_X(X_X)$ for every $X\in\{\bob,\charlie\}$, and let $Z:=(\Delta,\eta_\bob,\eta_\charlie)$. Since the descriptions of $L_\bob,L_\charlie$ are independent of $(Z_{\mathsf{aux}},X_\bob,X_\charlie,\Delta)$ and $|\operatorname{Supp}(Z)|\leq 2^{1+2\ell_{\mathsf{img}}}$, the support-size leakage bound for average conditional min-entropy gives $\Hmin(X_X\mid Z_{\mathsf{aux}},Z)\geq\lambda^c-1-2\ell_{\mathsf{img}}$ for every $X\in\{\bob,\charlie\}$.

For every fixed $(Z_{\mathsf{aux}},Z)=(z_{\mathsf{aux}},z)$, let $\operatorname{cp}_\bob(z_{\mathsf{aux}},z)$, $\operatorname{cp}_\charlie(z_{\mathsf{aux}},z)$, and $\operatorname{cp}_{\bob,\charlie}(z_{\mathsf{aux}},z)$ denote the collision probabilities of $X_\bob$, $X_\charlie$, and $(X_\bob,X_\charlie)$, respectively, conditioned on $(Z_{\mathsf{aux}},Z)=(z_{\mathsf{aux}},z)$. For every $X\in\{\bob,\charlie\}$,
\begin{align*}
\EE_{Z_{\mathsf{aux}},Z}[\operatorname{cp}_X(Z_{\mathsf{aux}},Z)]
&\leq\sum_{z_{\mathsf{aux}},z}\Pr[Z_{\mathsf{aux}}=z_{\mathsf{aux}},Z=z]\max_x\Pr[X_X=x\mid Z_{\mathsf{aux}}=z_{\mathsf{aux}},Z=z]\\
&=2^{-\Hmin(X_X\mid Z_{\mathsf{aux}},Z)}
\leq 2^{-(\lambda^c-1-2\ell_{\mathsf{img}})}.
\end{align*}
Moreover, $\operatorname{cp}_{\bob,\charlie}(z_{\mathsf{aux}},z)\leq\min\{\operatorname{cp}_\bob(z_{\mathsf{aux}},z),\operatorname{cp}_\charlie(z_{\mathsf{aux}},z)\}$ for every $(z_{\mathsf{aux}},z)$.

Independence and universality of $h_\bob,h_\charlie$ give, for every fixed $(Z_{\mathsf{aux}},Z)=(z_{\mathsf{aux}},z)$,
\begin{align*}
&\EE_{h_\bob,h_\charlie}\left[\operatorname{cp}\bigl(h_\bob(X_\bob),h_\charlie(X_\charlie)\bigr)\mid Z_{\mathsf{aux}}=z_{\mathsf{aux}},Z=z\right]\\
&\qquad\leq 2^{-2m}+2^{-m}\bigl(\operatorname{cp}_\bob(z_{\mathsf{aux}},z)+\operatorname{cp}_\charlie(z_{\mathsf{aux}},z)\bigr)+\operatorname{cp}_{\bob,\charlie}(z_{\mathsf{aux}},z).
\end{align*}
Averaging over $(Z_{\mathsf{aux}},Z)$ therefore gives
\begin{align*}
&\EE_{Z_{\mathsf{aux}},Z,h_\bob,h_\charlie}\left[\operatorname{cp}\bigl(h_\bob(X_\bob),h_\charlie(X_\charlie)\bigr)\right]\\
&\qquad\leq 2^{-2m}+\left(1+2^{1-m}\right)2^{-(\lambda^c-1-2\ell_{\mathsf{img}})}.
\end{align*}
The collision-probability bound on statistical distance, followed by Jensen's inequality, gives statistical distance at most
\begin{equation}
\frac12\sqrt{2^{2m-(\lambda^c-1-2\ell_{\mathsf{img}})}\left(1+2^{1-m}\right)}\leq\varepsilon,
\end{equation}
where the last inequality follows from the choice of parameters. Finally, XORing both hash values with the fixed string $\mathsf{msg}(u,v)$ does not change statistical distance. Hence the real and simulated pre-split auxiliary-information distributions satisfy the statistical simulatability condition in \Cref{def:comp-dcmoe-aux}.
\end{proof}

We construct a QPT triplet of adversaries $\widehat\advmoe=(\widehat\alice,\widehat\bob,\widehat\charlie)$ for \Cref{thm:comp-dcmoe-aux}. The algorithm $\widehat\alice$ receives $\ket{A_{s,t}}$, $\obfd{M_0},\obfd{M_1}$, $pp$, and the pre-split auxiliary information, from which it parses $z_{\mathsf{aux}},\Delta,\eta_\bob,\eta_\charlie,\mathsf{str}_\bob,\mathsf{str}_\charlie$. It runs the first stage of $\alice$ on $(1^\lambda,z_{\mathsf{aux}})$ to obtain $(k,\mu_\bob,\mu_\charlie)$ and retained state, samples $\zeta\getsr\bit$, $q_\bob,q_\charlie\getsr\mathcal Y_\lambda$, and $K\gets\pprf.\gen(1^\lambda)$, sets $d_\bob:=\zeta$ and $d_\charlie:=\zeta+\Delta$, constructs the two lossy-mode programs from $\hybrid_8$, obfuscates them using $\io$, and gives $\rho_k:=(\ket{A_{s,t}},\obfd{R_0},\obfd{R_1})$ to the continuation of $\alice$ together with its retained state. When $\alice$ outputs its two registers, $\widehat\alice$ appends a classical copy of $\zeta$ to the register sent to $\widehat\bob$ and classical copies of $(\zeta,\Delta)$ to the register sent to $\widehat\charlie$.

After the split, $\widehat\bob$ receives $(u,v,x_\bob,w_\bob)$, runs $\bob$ on $(x_\bob,w_\bob)$, and outputs $\widehat b_\bob+\zeta$. The algorithm $\widehat\charlie$ receives $(u,v,x_\charlie,w_\charlie)$, runs $\charlie$ on $(x_\charlie,w_\charlie)$, and outputs $\widehat b_\charlie+\zeta+\Delta$. Since $b_\bob=e+\zeta$, $b_\charlie=e+\zeta+\Delta$, and $e=\langle u,t\rangle+\langle v,s\rangle$, the winning event in $\hybrid_8$ is identical to the winning event of $\widehat\advmoe$ in the computational $\dcmoe$ experiment. By \Cref{thm:comp-dcmoe-aux}, $\Pr[\hybrid_8=1]\leq\frac12+\negl(\lambda)$. Combining the preceding hybrid transitions and the statistical error $\varepsilon$ proves the theorem.

\paragraph{Changes for the identical-uniform challenge distribution.}
For the $\IDU$-generalized UPO special case, we replace the injective/lossy-function and universal-hashing hybrids $\hybrid_6$--$\hybrid_8$, together with \Cref{lem:correlated-upo-masked-string-simulation}, by the NCE-based hybrids of~\cite[Theorem~6.3]{ABH+26}. In particular, for the common uniform challenge point $x$, we sample one key-robust NCE key, encrypt $\mathsf{msg}(u,v)$ under this key, and use the resulting NCE ciphertext and masked key string in both modified programs. NCE correctness and key robustness give the required functionality equivalence to invoke $\iO$ security, while the non-committing property gives a simulation of the pre-split information that is independent of $(u,v)$. The preceding puncturable PRF and $\io$ hybrids and the final reduction to computational $\dcmoe$ remain unchanged. Since the required NCE can be instantiated from post-quantum one-way functions, the LWE assumption can be replaced with one-way functions in this case. For completeness, we give the formal proof of the ``Moreover'' part in \Cref{sec:alt-proof-upo}.

\end{proof}

\subsubsection{Formal proof of the moreover part of \texorpdfstring{\Cref{thm:correlated-upo-construction}}{Theorem 6}}
\label{sec:alt-proof-upo}
We need the following primitive.

\begin{definition}[Key-robust one-time non-committing encryption]
A key-robust one-time non-committing encryption scheme $\mathsf{NCE}=(\mathsf{NCE.Gen},\mathsf{NCE.Enc},\mathsf{NCE.Dec},\mathsf{NCE.Fake},\mathsf{NCE.Open})$ with message space $\zo^{\ell_{\mathsf{NCE-msg}}}$ and key space $\zo^{\ell_{\mathsf{NCE-key}}}$ consists of five PPT algorithms as follows.
\begin{enumerate}
    \item $k\gets\mathsf{NCE.Gen}(1^\lambda):$ outputs an encryption key $k$.
    \item $\mathsf{ct}\gets\mathsf{NCE.Enc}(k,M):$ takes as input a key $k$ and a message $M$, and outputs a ciphertext $\mathsf{ct}$.
    \item $M'\gets\mathsf{NCE.Dec}(k,\mathsf{ct}):$ is deterministic and outputs a message $M'$ or $\bot$.
    \item $(\mathsf{ct}^\star,\mathsf{st})\gets\mathsf{NCE.Fake}(1^\lambda):$ outputs a fake ciphertext and state.
    \item $k^\star\gets\mathsf{NCE.Open}(\mathsf{st},M):$ takes as input a state and a message, and outputs a key.
\end{enumerate}
A key $k$ is good if, for every message $M$ and every $k'\neq k$,
\begin{equation}
\mathsf{NCE.Dec}(k,\mathsf{NCE.Enc}(k,M))=M
\qquad\text{and}\qquad
\mathsf{NCE.Dec}(k',\mathsf{NCE.Enc}(k,M))=\bot.
\label{eq:nce-good-key}
\end{equation}
A key generated by $\mathsf{NCE.Gen}$ is good with overwhelming probability. Non-committing security requires that, for every QPT algorithm that outputs a message $M$ together with quantum auxiliary information $R$,
\begin{equation}
(R,\mathsf{ct},k)\approx_c(R,\mathsf{ct}^\star,k^\star),
\label{eq:nce-security}
\end{equation}
where $k\gets\mathsf{NCE.Gen}(1^\lambda)$ and $\mathsf{ct}\gets\mathsf{NCE.Enc}(k,M)$ in the first distribution, while $(\mathsf{ct}^\star,\mathsf{st})\gets\mathsf{NCE.Fake}(1^\lambda)$ and $k^\star\gets\mathsf{NCE.Open}(\mathsf{st},M)$ in the second distribution. We may assume that $\mathsf{NCE.Open}$ is deterministic by including its random tape in $\mathsf{st}$.
\end{definition}

\begin{theorem}[Key-robust one-time non-committing encryption from $\io$ and one-way functions~\cite{ABH+26}]
\label{thm:nce-instantiation}
Assume the existence of post-quantum secure $\io$ and post-quantum one-way functions. For every constant $c>0$ and every polynomially bounded input length $\ell_{\mathsf{in}}=\ell_{\mathsf{in}}(\lambda)\geq\lambda^c$, there exists a constant $\epsilon>0$ such that, setting $n:=\lceil\lambda^\epsilon\rceil$, there is a key-robust one-time non-committing encryption scheme with message space $\zo^{8n}$ and key space $\zo^{\ell_{\mathsf{NCE-key}}}$, where $\ell_{\mathsf{NCE-key}}\leq\ell_{\mathsf{in}}(\lambda)$ for all sufficiently large $\lambda$, satisfying \Cref{eq:nce-good-key,eq:nce-security} against QPT distinguishers with quantum auxiliary information.
\end{theorem}
\begin{proof}
We choose $n$ so that the encryption key fits into $\ell_{\mathsf{in}}(\lambda)$ bits. Use the construction of~\cite[Theorem~6.2 and Appendix~A]{ABH+26} for $8n$-bit messages, running key generation and fake-ciphertext generation on input $1^n$. Its key length is bounded by a polynomial $q(n)=O(n^d)$ for some constant $d>0$. Choose a constant $\epsilon>0$ such that $\epsilon d<c$, and set $n:=\lceil\lambda^\epsilon\rceil$. Then $q(n)\leq\lambda^c\leq\ell_{\mathsf{in}}(\lambda)$ for all sufficiently large $\lambda$. An adversary whose running time is polynomial in $\lambda$ also runs in time polynomial in $n$, and an error negligible in $n$ is also negligible in $\lambda$. Thus the required post-quantum security holds, including for distinguishers with quantum auxiliary information.
\end{proof}

\begin{proof}[Detailed proof of the moreover part of \Cref{thm:correlated-upo-construction}]
Fix a QPT triplet of adversaries against $\IDU$-generalized UPO. The challenger samples $x\getsr\mathcal X_\lambda$ and $b\getsr\bit$, gives the same point $x$ to both recipients after the split, and uses the generalized punctured circuit that equals $\mu_\bob(x)$ at $x$ if $b=1$ and equals $C_k(x)$ otherwise. Choose $n=n(\secparam)=\secparam^\epsilon$ for some constant $\epsilon>0$, and a key-robust one-time non-committing encryption scheme with message length $m=8n$ and key length $\ell_{\mathsf{NCE-key}}\leq\ell_{\mathsf{in}}$. Parse every point as $x=x_0\|x_1$, where $|x_0|=\ell_{\mathsf{NCE-key}}$.

Specialize $\hybrid_0$--$\hybrid_5$ above to $x_\bob=x_\charlie=x$, $b_\bob=b_\charlie=b$, $\Delta=0$, and $\mu_\bob$. Thus, after sampling $u\getsr A\setminus\{0\}$, $v\getsr A^\perp\setminus\{0\}$, and $\zeta\getsr\bit$, we have $b=\alpha+\beta+\zeta$, where $\alpha=\langle v,s\rangle$ and $\beta=\langle u,t\rangle$. The puncturable PRF is punctured at the single point $x$, its value at $x$ is replaced by a uniform $y$, and we sample $q\getsr\mathcal Y_\lambda$ and write $y=q\oplus\beta\cdot\mathsf{diff}(x)$, where $\mathsf{diff}(z):=C_k(z)\oplus\mu_\bob(z)$. The two hardcoded values at $x$ are therefore $q\oplus C_k(x)\oplus(\alpha+\zeta)\cdot\mathsf{diff}(x)$ and $q\oplus\beta\cdot\mathsf{diff}(x)$.

Next sample $k_{\mathsf{nce}}\gets\mathsf{NCE.Gen}(1^\lambda)$, $\mathsf{ct}\gets\mathsf{NCE.Enc}(k_{\mathsf{nce}},\mathsf{msg}(u,v))$, and set $\mathsf{str}:=k_{\mathsf{nce}}\oplus x_0$. Replace the two point-hardcoded programs by the following programs, initially using the punctured key $K_{\{x\}}$ in their last lines.

\begin{mdframed}
{\bf \underline{$S_0(a,z)$}}

{\bf Hardcoded: $\obfd{M_0},C_k,\mu_\bob,K_{\{x\}},q,\zeta,\mathsf{ct},\mathsf{str},x_1$}
\begin{enumerate}[label=\arabic*.]
    \item If $\obfd{M_0}(a)=0$, output $\bot$.
    \item Parse $z=z_0\|z_1$, set $k_z:=\mathsf{str}\oplus z_0$, and compute $M_z:=\mathsf{NCE.Dec}(k_z,\mathsf{ct})$.
    \item If $z_1=x_1$ and $M_z\neq\bot$, parse $(u',v'):=M_z$ and output $q\oplus C_k(z)\oplus(\langle v',a\rangle+\zeta)\cdot\mathsf{diff}(z)$.
    \item Output $F_{K_{\{x\}}}(z)\oplus C_k(z)$.
\end{enumerate}
\end{mdframed}

\begin{mdframed}
{\bf \underline{$S_1(a,z)$}}

{\bf Hardcoded: $\obfd{M_1},C_k,\mu_\bob,K_{\{x\}},q,\mathsf{ct},\mathsf{str},x_1$}
\begin{enumerate}[label=\arabic*.]
    \item If $\obfd{M_1}(a)=0$, output $\bot$.
    \item Parse $z=z_0\|z_1$, set $k_z:=\mathsf{str}\oplus z_0$, and compute $M_z:=\mathsf{NCE.Dec}(k_z,\mathsf{ct})$.
    \item If $z_1=x_1$ and $M_z\neq\bot$, parse $(u',v'):=M_z$ and output $q\oplus\langle u',a\rangle\cdot\mathsf{diff}(z)$.
    \item Output $F_{K_{\{x\}}}(z)$.
\end{enumerate}
\end{mdframed}

Condition on the overwhelmingly likely event that $k_{\mathsf{nce}}$ is good. At $z=x$, $k_z=\mathsf{str}\oplus x_0=k_{\mathsf{nce}}$, so correctness gives $M_z=\mathsf{msg}(u,v)$. Hence, for $a\in A+s$, $S_0$ returns $q\oplus C_k(x)\oplus(\alpha+\zeta)\cdot\mathsf{diff}(x)$, and for $a\in A^\perp+t$, $S_1$ returns $q\oplus\beta\cdot\mathsf{diff}(x)$. If $z\neq x$, either $z_1\neq x_1$, or $z_0\neq x_0$ and $k_z\neq k_{\mathsf{nce}}$, in which case key robustness gives $M_z=\bot$. Thus the relevant circuits have the same functionality, and two applications of $\io$-security justify the replacement. We then replace $K_{\{x\}}$ by the full puncturable PRF key $K$ in the last line of each program. The relevant circuits again have the same functionality, and this replacement follows from $\io$-security.

Since $x_0$ is uniform, the same distribution can be sampled by first generating $k_{\mathsf{nce}}$ and $\mathsf{ct}$, then sampling $\mathsf{str}\getsr\zo^{\ell_{\mathsf{NCE-key}}}$ and $x_1\getsr\zo^{\ell_{\mathsf{in}}-\ell_{\mathsf{NCE-key}}}$, and finally setting $x_0:=\mathsf{str}\oplus k_{\mathsf{nce}}$. Apply non-committing security to replace $(\mathsf{ct},k_{\mathsf{nce}})$ by $(\mathsf{ct}^\star,k^\star)$, where $(\mathsf{ct}^\star,\mathsf{st})\gets\mathsf{NCE.Fake}(1^\lambda)$ and $k^\star:=\mathsf{NCE.Open}(\mathsf{st},\mathsf{msg}(u,v))$. Except with negligible probability, $\mathsf{NCE.Dec}(k^\star,\mathsf{ct}^\star)=\mathsf{msg}(u,v)$; otherwise, testing this equality would distinguish the fake/open distribution from the real correct distribution. In the resulting hybrid, the values $(\mathsf{ct}^\star,\mathsf{str},x_1)$ given before the split are independent of $(u,v)$.

We reduce this hybrid to the computational $\dcmoe$ experiment without auxiliary information. The algorithm $\widehat\alice$ receives $\ket{A_{s,t}}$, $\obfd{M_0},\obfd{M_1}$, runs the first stage of $\alice$ to obtain $(k,\mu_\bob,\mu_\charlie)$, samples $(\mathsf{ct}^\star,\mathsf{st})\gets\mathsf{NCE.Fake}(1^\lambda)$, $\mathsf{str}\getsr\zo^{\ell_{\mathsf{NCE-key}}}$, $x_1\getsr\zo^{\ell_{\mathsf{in}}-\ell_{\mathsf{NCE-key}}}$, $q\getsr\mathcal Y_\lambda$, $\zeta\getsr\bit$, and $K\gets\mathsf{PRF.Gen}(1^\lambda)$. It constructs and obfuscates the two programs $S_0,S_1$ with $\mathsf{ct}^\star$ and the full key $K$, gives the resulting UPO state to the continuation of $\alice$, and appends classical copies of $(\mathsf{st},\mathsf{str},x_1,\zeta)$ to both output registers.

After the split, each of $\widehat\bob,\widehat\charlie$ receives $(u,v)$. It computes $M:=\mathsf{msg}(u,v)$, $k^\star:=\mathsf{NCE.Open}(\mathsf{st},M)$, $x_0:=\mathsf{str}\oplus k^\star$, and $x:=x_0\|x_1$, runs the corresponding original recipient on $x$, and adds $\zeta$ to its output. Since the IDU challenge bit is $b=\langle u,t\rangle+\langle v,s\rangle+\zeta$, both original recipients output $b$ if and only if both $\widehat\bob$ and $\widehat\charlie$ output $\langle u,t\rangle+\langle v,s\rangle$. Therefore, \Cref{cor:comp-dcmoe} bounds the winning probability by $\frac12+\negl(\lambda)$. Combining the puncturable PRF, $\io$, non-committing encryption, and negligible bad-key errors proves $\IDU$-generalized UPO security from post-quantum $\io$ and post-quantum one-way functions.
\end{proof}


%% file: upo/applications.tex
\section{Applications of Correlated UPO to Copy-Protection}
\label{sec:correlated-upo-copy-protection-applications}

We recall the UPO-to-copy-protection implications of~\cite{AB24,ABH+26} and adapt them to correlated challenge distributions.  Throughout this subsection, $\mathsf{Circ}=\{C_k:k\in\mathcal K_\lambda\}$ is a keyed circuit class and $\mathcal D_{\mathsf{Circ}}$ is induced by an efficiently samplable key distribution $\mathcal D_K$: sampling $C\gets\mathcal D_{\mathsf{Circ}}(1^\lambda)$ means sampling $k\gets\mathcal D_K(1^\lambda)$ and setting $C:=C_k$.  Except in the point-function, $k$-point-function, and compute-and-compare function, the key $k$ and the challenges (for $\bob$ and $\charlie$) are sampled independently; in such cases, we omit the circuit input from the challenge sampler’s notation.

\begin{definition}[Two-point unpredictability-style puncturing security, adapted from~\cite{ABH+26}]
\label{def:correlated-unpredictability-puncturing}
Let $\mathsf{Circ}$ be a circuit class with output space $\mathcal Y_\lambda=\zo^{m(\lambda)}$. The class $\mathsf{Circ}$ is two-point unpredictability-style puncturable with respect to $(\mathcal D_{\mathsf{Circ}},\mathsf{Samp}_{\mathsf{CP}})$ if there exists a deterministic PPT algorithm $\mathsf{Puncture}$ such that, on input a circuit $C\in\mathsf{Circ}_\lambda$ and a set $\mathcal S\subseteq\mathcal X_\lambda$ of size at most two, it outputs a circuit $C^{\setminus\mathcal S}:\mathcal X_\lambda\to\mathcal Y_\lambda$ satisfying
\begin{equation}
C^{\setminus\mathcal S}(x)=
\begin{cases}
0^{m(\lambda)},&x\in\mathcal S,\\
C(x),&x\notin\mathcal S.
\end{cases}
\label{eq:two-point-puncturing-correctness}
\end{equation}
All circuits compared below are padded to the same size.  Moreover, for every QPT adversary $\mathcal D$, there exists a negligible function $\negl$ such that
\begin{equation}
\Pr\left[
\widehat y_\bob=C(x_\bob)\ \vee\
\widehat y_\charlie=C(x_\charlie):
\begin{array}{l}
C\gets\mathcal D_{\mathsf{Circ}}(1^\lambda),\\
(x_\bob,x_\charlie)\gets\mathsf{Samp}_{\mathsf{CP}}(1^\lambda),\\
\mathcal S:=\{x_\bob,x_\charlie\},\\
C^{\setminus\mathcal S}\gets\mathsf{Puncture}(C,\mathcal S),\\
(\widehat y_\bob,\widehat y_\charlie)
\gets\mathcal D(C^{\setminus\mathcal S},x_\bob,x_\charlie)
\end{array}
\right]
\leq
1-\left(1-2^{-m(\lambda)}\right)^2+\negl(\lambda).
\label{eq:correlated-unpredictability-puncturing}
\end{equation}
\end{definition}

\begin{definition}[Two-point pseudorandomness-style puncturing security with correlated challenges, adapted from~\cite{ABH+26}]
\label{def:correlated-pr-puncturing}
Let $\mathsf{Circ}$ be a deterministic circuit class with output space $\mathcal Y_\lambda=\zo^{m(\lambda)}$. We say that $\mathsf{Circ}$ is two-point pseudorandomness-style puncturable with respect to $(\mathcal D_{\mathsf{Circ}},\mathsf{Samp}_{\mathsf{PRCP}})$ if it has a deterministic PPT puncturing algorithm satisfying \Cref{eq:two-point-puncturing-correctness}, and the following two ensembles are computationally indistinguishable.

Sample $C\gets\mathcal D_{\mathsf{Circ}}(1^\lambda)$ and $(x_\bob,x_\charlie,b_\bob,b_\charlie)\gets\mathsf{Samp}_{\mathsf{PRCP}}(1^\lambda)$, and let $\mathcal S:=\{x_\bob,x_\charlie\}$. In the first ensemble, output
\begin{equation}
(C^{\setminus\mathcal S},
x_\bob,x_\charlie,b_\bob,b_\charlie,
C(x_\bob),C(x_\charlie)).
\label{eq:correlated-pr-puncturing-real}
\end{equation}
In the second ensemble, sample an independent $U_x\getsr\mathcal Y_\lambda$ for every distinct $x\in\mathcal S$, and output
\begin{equation}
(C^{\setminus\mathcal S},
x_\bob,x_\charlie,b_\bob,b_\charlie,
U_{x_\bob},U_{x_\charlie}).
\label{eq:correlated-pr-puncturing-random}
\end{equation}
\end{definition}

Given a UPO scheme $\Pi=(\Obf,\Eval)$, define $\mathsf{CP}_\Pi$ by $\mathsf{CopyProtect}(1^\lambda,C):=\Obf(1^\lambda,C)$ and by letting $\mathsf{CP}_\Pi.\mathsf{Eval}(\rho,x)$ run $(\rho',y)\gets\Eval(\rho,x)$ and output $y$. Whenever we invoke UPO security below, the supported circuit class is understood to contain the base circuits and the circuits produced by the relevant puncturing and programming algorithms, padded to a common size. For \Cref{def:correlated-unpredictability-puncturing}, the corresponding generalized-puncturing algorithm first computes $C^{\setminus\mathcal S}\gets\mathsf{Puncture}(C,\mathcal S)$ and then programs $\mu_\bob(x_\bob)$ and, if $x_\charlie\neq x_\bob$, $\mu_\charlie(x_\charlie)$, using Bob priority on a collision. Thus the punctured event of the UPO game uses the same puncturing procedure as the security assumption.

We first adapt the non-oracular point-function application of generalized UPO to arbitrary high-min-entropy point distributions.

\begin{theorem}[Copy protection for point functions from $\ID_{\mathcal D}$-generalized UPO, adapted from~{\cite[Theorem~89 and Corollary~90]{AB24}}]
\label{thm:idu-upo-to-point-functions}
Let $\mathcal D=\{\mathcal D_\lambda\}_{\lambda\in\NN}$ be a QPT-samplable distribution over $\mathcal X_\lambda=\zo^{\ell_{\mathsf{in}}(\lambda)}$ such that $\Hmin(\mathcal D_\lambda)\in\omega(\log\lambda)$. Suppose that there is a UPO scheme for polynomial-size circuits satisfying $\ID_{\mathcal D}$-generalized UPO security and $\io$-security. Then point functions whose point is sampled according to $\mathcal D_\lambda$ have a copy-protection scheme satisfying identical-challenge point-function security with respect to $\mathcal D$ from \Cref{def:point-function-copy-protection}.
\end{theorem}

\begin{proof}
We adapt the proof of~\cite[Theorem~89]{AB24}. Let $C_0$ be the all-zero circuit and let $G_y$ be obtained by programming $C_0$ to output one at $y$. The $\io$-security of UPO replaces an obfuscation of $P_y$ by an obfuscation of the functionally equivalent circuit $G_y$. In the independent event, replace the desired output by zero. This changes success only when the independent samples $x,y\gets\mathcal D_\lambda$ coincide, which occurs with probability
\begin{equation}
\operatorname{cp}(\mathcal D_\lambda)
:=\sum_x\Pr[\mathcal D_\lambda=x]^2
\leq 2^{-\Hmin(\mathcal D_\lambda)}
=\negl(\lambda).
\end{equation}
The switch from an obfuscation of $G_y$ to an obfuscation of $C_0$ in the independent event is precisely the guarantee given by $\ID_{\mathcal D}$-generalized UPO security with the constant-one replacement circuits. After this switch, both events have challenge-point marginal $\mathcal D_\lambda$, and the final experiment is the $\ID_{\mathcal D}$-generalized UPO experiment. Thus every adversary succeeds with probability at most $1/2+\negl(\lambda)$. By \Cref{lem:point-function-trivial-success}, this implies the required trivial-adversary comparison. The uniform specialization is exactly~\cite[Corollary~90]{AB24}.
\end{proof}

\begin{corollary}[Copy-protection for point functions]
\label{cor:point-function-copy-protection}
Assume post-quantum $\io$ and the quantum hardness of LWE. For every constant $c>0$ and every QPT-samplable distribution ensemble $\mathcal D=\{\mathcal D_\lambda\}_{\lambda\in\NN}$ over $\mathcal X_\lambda$ satisfying $\Hmin(\mathcal D_\lambda)\geq\lambda^c$, point functions whose point is sampled according to $\mathcal D_\lambda$ have a copy-protection scheme satisfying identical-challenge point-function security with respect to $\mathcal D$. Moreover, when $\mathcal D_\lambda$ is uniform over $\zo^{\ell_{\mathsf{in}}(\lambda)}$ and $\ell_{\mathsf{in}}=\ell_{\mathsf{in}}(\secparam)=\secparam^\epsilon$ for some constant $\epsilon>0$, LWE can be replaced with post-quantum one-way functions.
\end{corollary}

\begin{proof}
The first claim follows from \Cref{thm:correlated-upo-construction,thm:correlated-upo-composition,lem:ac-upo-implies-prior}, which give a UPO satisfying $\ID_{\mathcal D}$-generalized UPO security and $\io$-security, and \Cref{thm:idu-upo-to-point-functions}. The uniform claim uses the ``Moreover'' part of \Cref{thm:correlated-upo-construction} together with the add-$\io$ transformation of~\cite[Lemma~3.16]{ABH+26}.
\end{proof}

\begin{corollary}[Copy protection for $k$-point functions]
\label{cor:k-point-function-copy-protection}
Assume post-quantum $\io$ and the quantum hardness of LWE. Let $k=k(\lambda)$ be polynomially bounded, and let $\mathcal D^{(k)}=\{\mathcal D^{(k)}_\lambda\}_{\lambda\in\NN}$ be a QPT-samplable distribution over the $k$-element subsets of $\mathcal X_\lambda$ such that $\Hmin(\mathcal D^{(k)}_\lambda)\geq\lambda^c$ for some constant $c>0$. Let $\mathcal D^{\mathsf{pt}}_\lambda$ be the distribution obtained by sampling $S\gets\mathcal D^{(k)}_\lambda$ and then $x\getsr S$, and suppose that $\Hmin(\mathcal D^{\mathsf{pt}}_\lambda)\geq\lambda^{c'}$ for some constant $c'>0$. Suppose moreover that the resulting distribution over $k$-point functions is preimage-samplable in the sense of~\cite[Definition~84]{AB24}, with the uniform challenge distribution in that definition replaced by $\mathcal D^{\mathsf{pt}}_\lambda$. Then the $k$-point functions sampled according to $\mathcal D^{(k)}$ have a copy-protection scheme satisfying the identical-challenge security (see \Cref{rem:k-point-functions-def}). If $\mathcal D^{(k)}_\lambda$ is uniform over the $k$-element subsets of $\mathcal X_\lambda$ and $k(\lambda)/|\mathcal X_\lambda|$ is negligible, then LWE can be replaced with post-quantum one-way functions.
\end{corollary}

\begin{proof}
For every fixed $x\in\mathcal X_\lambda$,
\begin{equation}
\Pr_{S\gets\mathcal D^{(k)}_\lambda}[x\in S]
=k(\lambda)\Pr[\mathcal D^{\mathsf{pt}}_\lambda=x]
\leq k(\lambda)2^{-\Hmin(\mathcal D^{\mathsf{pt}}_\lambda)}
=\negl(\lambda),
\end{equation}
so the sampled function class is evasive. Under the stated preimage-samplability assumption, the proof of~\cite[Theorem~89]{AB24} applies after replacing each uniform challenge point by a sample from $\mathcal D^{\mathsf{pt}}_\lambda$. By \Cref{thm:correlated-upo-construction,thm:correlated-upo-composition,lem:ac-upo-implies-prior}, the stated assumptions give the required $\ID_{\mathcal D^{\mathsf{pt}}}$-generalized UPO and $\io$-security. The uniform specialization is the setting of~\cite[Theorem~85 and Corollary~90]{AB24}, and the ``Moreover'' part of \Cref{thm:correlated-upo-construction} together with the add-$\io$ transformation of~\cite[Lemma~3.16]{ABH+26} supplies the required $\IDU$-generalized UPO from one-way functions. The trivial success probability is given by \Cref{eq:k-point-function-trivial-success}.
\end{proof}

\begin{theorem}[Copy protection for compute-and-compare functions]
\label{thm:compute-and-compare-copy-protection}
Assume post-quantum $\io$ and the quantum hardness of LWE. Let $\mathsf{Samp}$ be a QPT algorithm that samples $(f,y)\gets\mathsf{Samp}(1^\lambda)$, where $f:\mathcal X_\lambda\rightarrow\mathcal Z_\lambda$ is a polynomial-size circuit, $\mathcal Z_\lambda=\zo^{\ell_{\mathsf{cc}}(\lambda)}$, and $y\in\mathcal Z_\lambda$. Let $(F,Y)$ denote its output distribution and suppose that $\Hmin(Y\mid F)\geq\lambda^c$ for some constant $c>0$. Fix a family of distributions $\{\mathcal Q_{f,z}\}_{f,z}$ as in \Cref{def:compute-and-compare-functions-copy-protection}, and suppose that the resulting identical-challenge experiment is efficiently realizable. Then the compute-and-compare functions sampled by $\mathsf{Samp}$ have a copy-protection scheme satisfying the identical-challenge security of \Cref{def:compute-and-compare-functions-copy-protection}.
Moreover, if \(Y\) conditioned on \(F=f\) is uniform over \(\mathcal Z_\lambda\) for every \(f\) in the support of \(F\), then the LWE assumption can be replaced with post-quantum one-way functions.
\end{theorem}

\begin{proof}

Let $\Pi$ be a UPO, guaranteed by \Cref{thm:correlated-upo-construction,thm:correlated-upo-composition} for any polynomial input and output length circuit class, and let $\mathsf{CP}_{\mathsf{PF}}$ be the point-function scheme obtained from $\Pi$ by the construction used in \Cref{thm:idu-upo-to-point-functions}. Following the generic transformation of~\cite[Construction~4 and Theorem~7]{CMP24}, define
\begin{equation}
\rho_y\gets\mathsf{CopyProtect}_{\mathsf{PF}}(1^\lambda,P_y),
\qquad
\mathsf{CopyProtect}_{\mathsf{CC}}(1^\lambda,\mathsf{CC}_{f,y})
:=(f,\rho_y),
\end{equation}
and
\begin{equation}
\mathsf{Eval}_{\mathsf{CC}}((f,\rho_y),x)
:=
\mathsf{Eval}_{\mathsf{PF}}(\rho_y,f(x)).
\end{equation}
Correctness follows from
$\mathsf{CC}_{f,y}(x)=P_y(f(x))$.

For every $f$ in the support of $F$, let $\mathcal D_f$ denote the conditional distribution of $Y$ given $F=f$. Consider the correlated-UPO sampler that samples
\begin{equation}
(f,z)\gets\mathsf{Samp}(1^\lambda),
\qquad
x\gets\mathcal Q_{f,z},
\qquad
b\getsr\bit,
\end{equation}
where the second output of $\mathsf{Samp}$ is denoted by $z$, and outputs
\begin{equation}\label{eq:first-sampler}
(z_{\mathsf{aux}},x_\bob,x_\charlie,b_\bob,b_\charlie,w_\bob,w_\charlie)
:=(f,z,z,b,b,x,x).
\end{equation}
This sampler satisfies \Cref{def:ac-upo-sampler}: the shared bit is uniform conditioned on $(f,z,x)$ and
\[
\Hmin(z\mid f)=\Hmin(Y\mid F)\geq\lambda^c.
\]

We adapt the proof of \Cref{thm:idu-upo-to-point-functions}. Let $C_0$ be the all-zero circuit and let $G_y$ be obtained by programming $C_0$ to output one at $y$. The $\io$-security of UPO replaces an obfuscation of $P_y$ by an obfuscation of the functionally equivalent circuit $G_y$. In the event $b=0$, replace the desired output by zero. This changes success only if the independent samples $y,z\gets\mathcal D_f$ coincide. Averaged over $f$, this occurs with probability
\begin{equation}
\EE_f[\operatorname{cp}(\mathcal D_f)]
\leq
2^{-\Hmin(Y\mid F)}
=
\negl(\lambda).
\end{equation}

Next, in the event $b=0$, replace the obfuscation of $G_y$ by an obfuscation of $C_0$. To justify this change, consider the correlated-UPO sampler that samples $(f,y)\gets\mathsf{Samp}(1^\lambda)$, a uniform bit $d\getsr\bit$, and outputs
\begin{equation}
(z_{\mathsf{aux}},x_\bob,x_\charlie,b_\bob,b_\charlie,w_\bob,w_\charlie)
:=(f,y,y,d,d,\bot,\bot).
\end{equation}
This is a valid sampler because $\Hmin(y\mid f)\geq\lambda^c$. Correlated UPO security therefore implies that the ensembles consisting of $(f,\Obf(G_y))$ and $(f,\Obf(C_0))$ are computationally indistinguishable: otherwise, $\mathcal R_\alice$ runs the distinguisher and sends a classical copy of its output to each of $\mathcal R_\bob$ and $\mathcal R_\charlie$. In the event $b=0$, the fresh value $z\gets\mathcal D_f$ and the point $x\gets\mathcal Q_{f,z}$ are independent of $y$ conditioned on $f$ and are efficiently sampleable. Hence adjoining $(z,x)$ and applying the QPT algorithms of the compute-and-compare adversary preserves this indistinguishability.

After this change, the resulting experiment is exactly the correlated-UPO experiment for the first sampler (\Cref{eq:first-sampler}) displayed above. The reduction $\mathcal R_\alice$, after receiving $f$ and the UPO state, gives $(f,\rho)$ to $\alice$ and forwards its two output registers; after the split, $\mathcal R_\bob$ and $\mathcal R_\charlie$ ignore the common UPO point $z$ and give the common post-split auxiliary string $x$ to $\bob$ and $\charlie$, respectively. When the shared bit is zero, the state protects $C_0$ and both recipients must output zero. When the shared bit is one, the common UPO challenge point is $z$, the state protects $G_z=P_z$, and both recipients must output one. Since $x\gets\mathcal Q_{f,z}$ and hence $f(x)=z$, their target value is
\begin{equation}
\mathsf{CC}_{f,y}(x)=P_y(z).
\end{equation}
Thus correlated UPO security implies that every adversary in the last experiment succeeds with probability at most $1/2+\negl(\lambda)$. Therefore, we conclude that the success probability of the adversaries in the original compute-and-compare copy protection experiment is at most $1/2+\negl'(\lambda)$ for some negligible function $\negl'(\secparam)$.

Combining the last conclusion with the fact that the trivial adversarial strategy of outputting $1$ (i.e., $\alice$ sends the copy-protected state to $\bob$ and $\charlie$ always outputs $1$) has success probability
\begin{equation}
\frac{1}{2}+\frac12\EE_f[\operatorname{cp}(\mathcal D_f)]
\leq
\frac12+\frac12\,2^{-\Hmin(Y\mid F)},
\end{equation}
which is $1/2$ plus some negligible function, we conclude that the trivial success probability for the experiment is $\frac12$ plus some negligible function. Hence, by the preceding argument, we conclude that the success probability for any adversary is upper bounded by the trivial success probability upto negligible corrections (see \Cref{def:compute-and-compare-functions-copy-protection}), which concludes the proof of the security.

The ``Moreover'' clearly follows by noting that in this special case, we only need correlated UPO secure against identical uniform challenge distribution and then replacing the general statement of \Cref{thm:correlated-upo-construction} in the above proof with the ``Moreover'' part of \Cref{thm:correlated-upo-construction}.
\end{proof}

Next, we turn our attention to unpredictability-style copy protection for puncturable functionalities.

\begin{theorem}[Unpredictability-style copy protection from UPO, adapted from~{\cite[Theorem~57]{AB24}}]
\label{thm:upo-to-non-oracular-unpredictability-copy-protection}
Let $\mathcal D_X$ be the joint distribution induced by $\mathsf{Samp}_{\mathsf{CP}}$, let $m\in\omega(\log\lambda)$, and let $\mathsf{Circ}$ be two-point unpredictability-style puncturable with respect to $(\mathcal D_{\mathsf{Circ}},\mathsf{Samp}_{\mathsf{CP}})$. If $\Pi$ satisfies $\mathcal D_X$-generalized UPO security, then $\mathsf{CP}_\Pi$ satisfies unpredictability-style copy-protection security with respect to $(\mathcal D_{\mathsf{Circ}},\mathsf{Samp}_{\mathsf{CP}})$ (see \Cref{def:unpredictability-copy-protection}).
\end{theorem}

\begin{proof}
We adapt the reduction of~\cite[Theorem~57]{AB24} while preserving the sampled joint pair. In summary, the argument below uses the challenge pair as a jointly sampled object, and hence permits arbitrary correlation among challenge points. 

Fix a copy-protection adversary, and construct a generalized-UPO adversary using the constant replacement circuits $\mu_X(z):=0^{m(\lambda)}$ for $X\in\{\bob,\charlie\}$. Let $p_0$ be the probability that both recipients recover their respective values when the UPO challenge bit is zero, and let $p_1$ be the probability that neither recipient recovers its respective value when the challenge bit is one. By construction, $p_0$ is exactly the success probability in the copy-protection experiment (see \Cref{def:unpredictability-copy-protection}). In the punctured event, the puncturing guarantee \Cref{eq:correlated-unpredictability-puncturing} yields
\begin{equation}
p_1\geq\left(1-2^{-m(\lambda)}\right)^2-\negl(\lambda).
\end{equation}
The reduction outputs zero precisely upon a correct prediction and one otherwise, so generalized-UPO security gives $(p_0+p_1)/2\leq1/2+\negl(\lambda)$. Therefore,
\begin{equation}
p_0\leq1-\left(1-2^{-m(\lambda)}\right)^2+\negl(\lambda),
\end{equation}
which is negligible in $\secparam$. 
\end{proof}

\begin{corollary}[Unpredictability-style copy protection with correlated challenges]
\label{cor:correlated-unpredictability-copy-protection}
Assume post-quantum $\io$ and the quantum hardness of LWE. Let $c>0$, let $\mathsf{Samp}_{\mathsf{CP}}$ be any QPT sampler whose two challenge-point marginals have min-entropy at least $\lambda^c$, and let $\mathsf{Circ}$ have output length $m\in\omega(\log\lambda)$ and satisfy \Cref{def:correlated-unpredictability-puncturing}. Then there exists a copy-protection scheme for $\mathsf{Circ}$ satisfying unpredictability-style security for $(\mathcal D_{\mathsf{Circ}},\mathsf{Samp}_{\mathsf{CP}})$.
\end{corollary}

\begin{proof}
Apply \Cref{thm:correlated-upo-construction,lem:ac-upo-implies-prior} to the joint distribution induced by $\mathsf{Samp}_{\mathsf{CP}}$, and then apply \Cref{thm:upo-to-non-oracular-unpredictability-copy-protection}.
\end{proof}

We note that the above result can be strengthened to oracular security when the two challenge points are identical and are sampled from a high-min-entropy distribution. Let $\mathcal D=\{\mathcal D_\lambda\}_{\lambda\in\NN}$ be a QPT-samplable distribution ensemble over $\mathcal X_\lambda$. We write $\mathsf{Samp}_{\ID_{\mathcal D}}$ for the sampler that samples $x\gets\mathcal D_\lambda$ and outputs $(x_\bob,x_\charlie):=(x,x)$.

\begin{theorem}[Oracular unpredictability-style copy protection for identical challenges]
\label{thm:upo-to-oracular-unpredictability-copy-protection}
Let $\Hmin(\mathcal D_\lambda)\in\omega(\log\lambda)$ and $m\in\omega(\log\lambda)$. Suppose that $\mathsf{Circ}$ is two-point unpredictability-style puncturable with respect to $(\mathcal D_{\mathsf{Circ}},\mathsf{Samp}_{\ID_{\mathcal D}})$ and that $\Pi$ satisfies $\ID_{\mathcal D}$-generalized UPO security. Then $\mathsf{CP}_\Pi$ satisfies oracular unpredictability-style copy-protection security with respect to $(\mathcal D_{\mathsf{Circ}},\mathsf{Samp}_{\ID_{\mathcal D}})$.
\end{theorem}

\begin{proof}
We combine the reduction above with the oracle-hybrid argument of~\cite[Remark~8.3]{ABH+26}, specialized to one common point $x\gets\mathcal D_\lambda$. Let $G_x$ be the punctured circuit, which agrees with $C$ off $x$ and outputs $0^{m(\lambda)}$ at $x$. Conditioned on the event that the UPO challenge bit is one, suppose that $\alice$ makes at most $q=q(\lambda)$ oracle queries, and let $W$ be their total query weight on $x$ when $\alice$ receives $\Obf(G_x)$ but still has oracle access to $C$. We claim that $\mathbb E[W]$ is negligible.

To see this, let $\delta_\lambda:=2^{-\Hmin(\mathcal D_\lambda)}$. For $q\geq1$, construct an $\ID_{\mathcal D}$-generalized UPO adversary $\mathcal{R}=(\mathcal{R}_\alice,\mathcal{R}_\bob,\mathcal{R}_\charlie)$ that samples $k\gets\mathcal D_K(1^\lambda)$, submits $k$ together with the constant replacement circuits $\mu_\bob(z)=\mu_\charlie(z):=0^{m(\lambda)}$, and uses $C:=C_k$ to simulate $\alice$'s $C$-oracle. The algorithm $\mathcal{R}_\alice$ chooses $i\getsr[q]$, measures the input of the $i$-th query if it occurs, sets $x':=\bot$ otherwise, and sends $x'$ to both reduction recipients. Upon receiving the common challenge $x$, each of $\mathcal{R}_\bob$ and $\mathcal{R}_\charlie$ outputs one if $x'=x$ and zero otherwise. Conditioned on the event that the UPO challenge bit is zero, $x'$ is independent of the fresh $x\gets\mathcal D_\lambda$, so both reduction recipients are correct with probability at least $1-\delta_\lambda$. Conditioned on the event that the UPO challenge bit is one, both are correct with probability $\mathbb E[W]/q$. Therefore, $\ID_{\mathcal D}$-generalized UPO security gives
\begin{equation}
\frac12\left(1-\delta_\lambda+\frac{\mathbb E[W]}q\right)
\leq\frac12+\negl(\lambda),
\end{equation}
and hence $\mathbb E[W]\leq q\delta_\lambda+\negl(\lambda)=\negl(\lambda)$. The case $q=0$ is immediate.

By the BBBV query bound~\cite[Theorem~2.1]{ABH+26}, changing $\alice$'s oracle from $C$ to $G_x$ changes the adversarial success probability by at most $O(\sqrt{q\mathbb E[W]})=\negl(\lambda)$. Both post-split recipient oracles are already identical, since $C^{\setminus x}=G_x^{\setminus x}$. The puncturing reduction can consequently simulate the punctured event using $G_x$, and the same $p_0,p_1$ calculation as above gives negligible piracy success and hence the required trivial-strategy bound.
\end{proof}

\begin{corollary}[Oracular unpredictability-style copy protection for identical challenges]
\label{cor:oracular-identical-high-entropy-unpredictability-copy-protection}
Assume post-quantum $\io$ and the quantum hardness of LWE. Let $c>0$ be a constant, let $\mathcal D$ be QPT-samplable with $\Hmin(\mathcal D_\lambda)\geq\lambda^c$, and let $\mathsf{Circ}$ have output length $m\in\omega(\log\lambda)$ and satisfy \Cref{def:correlated-unpredictability-puncturing} with respect to $(\mathcal D_{\mathsf{Circ}},\mathsf{Samp}_{\ID_{\mathcal D}})$. Then there exists a copy-protection scheme for $\mathsf{Circ}$ satisfying oracular unpredictability-style security with respect to $(\mathcal D_{\mathsf{Circ}},\mathsf{Samp}_{\ID_{\mathcal D}})$.
\end{corollary}

\begin{proof}
Apply \Cref{thm:correlated-upo-construction,lem:ac-upo-implies-prior} to $\ID_{\mathcal D}$ and then apply \Cref{thm:upo-to-oracular-unpredictability-copy-protection}.
\end{proof}

For the OWF-based instantiation, we retain the identical-uniform distribution, since the ``Moreover'' part of \Cref{thm:correlated-upo-construction} establishes $\IDU$-generalized UPO security.

\begin{corollary}[Identical-uniform oracular unpredictability-style copy protection]
\label{cor:oracular-idu-unpredictability-copy-protection}
Assume post-quantum $\io$ and post-quantum one-way functions. Let $\ell_{\mathsf{in}}=\secparam^\epsilon$, for some constant $\epsilon>0$ and  $m=m(\secparam)\in\omega(\log(\secparam))$, and suppose that $\mathsf{Circ}$ is two-point unpredictability-style puncturable for the identical-uniform sampler. Then there exists a copy-protection scheme for $\mathsf{Circ}$ satisfying oracular unpredictability-style security for that sampler.
\end{corollary}

\begin{proof}
Combine the ``Moreover'' part of \Cref{thm:correlated-upo-construction} with \Cref{thm:upo-to-oracular-unpredictability-copy-protection}, taking $\mathcal D_\lambda$ to be uniform over $\mathcal X_\lambda$.
\end{proof}

Next, we consider applications to pseudorandomness-style copy protection.

\begin{theorem}[Pseudorandomness-style copy protection from correlated UPO, adapted from~{\cite[Theorem~7.4]{ABH+26}}]
\label{thm:non-oracular-pr-copy-protection-from-correlated-upo}
Let $\mathsf{Samp}_{\mathsf{PRCP}}$ be a correlated pseudorandomness-style challenge sampler, and let $\mathsf{Circ}$ satisfy two-point pseudorandomness-style puncturing security for $(\mathcal D_{\mathsf{Circ}},\mathsf{Samp}_{\mathsf{PRCP}})$. Suppose that $\Pi$ satisfies correlated-UPO security for $\mathsf{Samp}_{\mathsf{PRCP}}$ and $\io$-security. Then $\mathsf{CP}_\Pi$ satisfies pseudorandomness-style copy-protection security for $(\mathcal D_{\mathsf{Circ}},\mathsf{Samp}_{\mathsf{PRCP}})$.
\end{theorem}

\begin{proof}
We adapt the non-oracular portion of~\cite[Theorem~7.4]{ABH+26}. Since the challenge tuple is independent of the sampled circuit and of the adversary's pre-challenge view, we may sample it at the outset and withhold it until after the split without changing the experiment. For a sampled tuple let $\mathcal S:=\{x_\bob,x_\charlie\}$ contain each distinct point once. Define the effective bit $d_{x_\bob}:=b_\bob$ and, when $x_\charlie\neq x_\bob$, $d_{x_\charlie}:=b_\charlie$, and let $T:=\{x\in\mathcal S:d_x=1\}$. This is precisely the Bob-priority convention on a collision. Write $R_x:=C(x)$, and let $U_x,Y_x\getsr\mathcal Y_\lambda$ be independent for each $x\in\mathcal S$. For values $v=(v_x)_{x\in\mathcal S}$, let $\mathsf{Program}(C^{\setminus\mathcal S},v)$ agree with $C^{\setminus\mathcal S}$ off $\mathcal S$ and output $v_x$ at $x\in\mathcal S$.

First, $\io$-security replaces $\Obf(C)$ by $\Obf(\mathsf{Program}(C^{\setminus\mathcal S},R))$. We use the following consequence of two-point pseudorandomness-style puncturing: for every $T\subseteq\mathcal S$ efficiently determined from the public tuple $(x_\bob,x_\charlie,b_\bob,b_\charlie)$,
\begin{equation}
(C^{\setminus\mathcal S},R_{\mathcal S\setminus T},R_T)
\approx_c
(C^{\setminus\mathcal S},R_{\mathcal S\setminus T},U_T).
\label{eq:selective-pr-puncturing}
\end{equation}
Indeed, the real and fully uniform ensembles are indistinguishable by \Cref{def:correlated-pr-puncturing}. Applying to both the map that retains the coordinates outside $T$ and freshly resamples those in $T$ shows that the mixed and fully uniform ensembles are indistinguishable. The triangle inequality gives \Cref{eq:selective-pr-puncturing}. The sampled bits are retained throughout, so this argument also permits arbitrary correlation between the points and bits.

Apply \Cref{eq:selective-pr-puncturing} to replace the protected values $R_T$ by $U_T$, retaining the independent challenge values $Y_T$ as auxiliary information. Swapping the names of the independent uniform families $U_T,Y_T$ is distribution preserving. Apply \Cref{eq:selective-pr-puncturing} a second time, now to replace the uniform challenge values by $R_T$ while retaining the independently programmed values as auxiliary information. Thus the classical challenge is $C(x_X)$ for both recipients, while the protected circuit is programmed with independent uniform values exactly at the points whose effective bit is one. Finally, $\io$-security replaces this circuit by the canonical correlated-punctured circuit.

For completeness, the final correlated-UPO reduction samples $k\gets\mathcal D_K$ and independent $u_\bob,u_\charlie\getsr\mathcal Y_\lambda$ before sending $(k,\mu_\bob,\mu_\charlie)$ to the challenger, where $\mu_X$ is the constant-$u_X$ circuit. After the split, recipient $X$ computes $C_k(x_X)$ and feeds $(x_X,C_k(x_X))$ to its copy-protection recipient. If the points collide, $u_\charlie$ is unused by Bob priority. Hence the simulated experiment is exactly the final hybrid, and correlated-UPO security bounds its success by $1/2+\negl(\lambda)$. No oracle-switch hybrid is used, so no independence between $x_\bob$ and $x_\charlie$ is required.
\end{proof}

\begin{corollary}[Pseudorandomness-style copy protection with correlated challenges]
\label{cor:non-oracular-correlated-pr-copy-protection}
Assume post-quantum $\io$ and the quantum hardness of LWE. Let $c>0$, let $\mathsf{Samp}_{\mathsf{PRCP}}$ have point marginals of min-entropy at least $\lambda^c$, and suppose that $\mathsf{Circ}$ satisfies \Cref{def:correlated-pr-puncturing}. Then $\mathsf{Circ}$ has a copy-protection scheme satisfying pseudorandomness-style security for $(\mathcal D_{\mathsf{Circ}},\mathsf{Samp}_{\mathsf{PRCP}})$.
\end{corollary}

\begin{proof}
Combine \Cref{thm:correlated-upo-construction,thm:correlated-upo-composition} with \Cref{thm:non-oracular-pr-copy-protection-from-correlated-upo}.
\end{proof}

For the oracular pseudorandomness-style result, let $\mathsf{Samp}^{\mathsf{pr}}_{\ID_{\mathcal D}}$ sample $x\gets\mathcal D_\lambda$ and $b\getsr\bit$ and output
\begin{equation}
(x_\bob,x_\charlie,b_\bob,b_\charlie):=(x,x,b,b).
\end{equation}

\begin{theorem}[Oracular pseudorandomness-style copy protection for identical challenges, adapted from~{\cite[Theorem~7.4]{ABH+26}}]
\label{thm:oracular-pr-copy-protection-from-correlated-upo}
Let $\Hmin(\mathcal D_\lambda)\in\omega(\log\lambda)$. Suppose that $\mathsf{Circ}$ is two-point pseudorandomness-style puncturable for $(\mathcal D_{\mathsf{Circ}},\mathsf{Samp}^{\mathsf{pr}}_{\ID_{\mathcal D}})$ and that $\Pi$ satisfies $\ID_{\mathcal D}$-generalized UPO security and $\io$-security. Then $\mathsf{CP}_\Pi$ satisfies oracular pseudorandomness-style copy-protection security for $(\mathcal D_{\mathsf{Circ}},\mathsf{Samp}^{\mathsf{pr}}_{\ID_{\mathcal D}})$.
\end{theorem}

\begin{proof}
Specialize the hybrids of~\cite[Theorem~7.4]{ABH+26} to $x_\bob=x_\charlie=x\gets\mathcal D_\lambda$ and $b_\bob=b_\charlie=b$. The programming and puncturing hybrids use one shared uniform value. Since the protected circuit and $C$ differ only at $x$, the punctured oracles given to $\bob$ and $\charlie$ are identical in every corresponding hybrid, and hence the cross-point oracle switches of~\cite[Claim~7.5]{ABH+26} are unnecessary.

It remains to justify the oracle switch for $\alice$. Let $q=q(\lambda)$ be a polynomial upper bound on the number of its oracle queries and let $W$ be their total query weight on $x$ immediately before this switch. The case $q=0$ is immediate, so suppose $q\geq1$. In the initial experiment, $x$ is sampled independently of $\alice$'s pre-challenge view. Thus, if an index is sampled uniformly from $[q]$ and the input of that query is measured to obtain $x'$, with $x':=\bot$ if the query does not occur, then
\begin{equation}
\Pr[x'=x]\leq 2^{-\Hmin(\mathcal D_\lambda)}=\negl(\lambda).
\end{equation}
The preceding programming and puncturing hybrids remain indistinguishable when this measurement is included. In the hybrid immediately before the oracle switch, the same test succeeds with probability $\mathbb E[W]/q$. Therefore, exactly as in~\cite[Claim~7.6]{ABH+26},
\begin{equation}
\mathbb E[W]
\leq q\cdot2^{-\Hmin(\mathcal D_\lambda)}
+\negl(\lambda)
=\negl(\lambda).
\end{equation}
The BBBV query bound~\cite[Theorem~2.1]{ABH+26} then changes the oracle between $C$ and the programmed circuit at negligible cost. The last $\io$ hybrid replaces the programmed circuit by the functionally equivalent circuit produced by the generalized-puncturing algorithm. The resulting experiment is exactly the $\ID_{\mathcal D}$-generalized UPO experiment, and its winning probability is at most $1/2+\negl(\lambda)$.
\end{proof}

We remark that for independent high-min-entropy challenge points with a common challenge bit, the corresponding oracular pseudorandomness-style result is already shown in~\cite[Theorem~7.4]{ABH+26}, and the oracular unpredictability-style result follows from~\cite[Theorem~8.1 and Remark~8.3]{ABH+26}.

Next \Cref{thm:oracular-pr-copy-protection-from-correlated-upo} combined with our feasibility result for UPO (\Cref{thm:correlated-upo-construction}) gives us the following immediate corollaries.
\begin{corollary}[Oracular pseudorandomness-style copy protection for identical challenges]
\label{cor:oracular-identical-high-entropy-pr-copy-protection}
Assume post-quantum $\io$ and the quantum hardness of LWE. Let $c>0$ be a constant, let $\mathcal D$ be QPT-samplable with $\Hmin(\mathcal D_\lambda)\geq\lambda^c$, and suppose that $\mathsf{Circ}$ is two-point pseudorandomness-style puncturable for $(\mathcal D_{\mathsf{Circ}},\mathsf{Samp}^{\mathsf{pr}}_{\ID_{\mathcal D}})$. Then there exists a copy-protection scheme for $\mathsf{Circ}$ satisfying oracular pseudorandomness-style security for $(\mathcal D_{\mathsf{Circ}},\mathsf{Samp}^{\mathsf{pr}}_{\ID_{\mathcal D}})$.
\end{corollary}

\begin{proof}
Combine \Cref{thm:correlated-upo-construction,thm:correlated-upo-composition,lem:ac-upo-implies-prior} with \Cref{thm:oracular-pr-copy-protection-from-correlated-upo}.
\end{proof}

For the OWF-based instantiation, we again retain identical-uniform challenges.

\begin{corollary}[Identical-uniform oracular pseudorandomness-style copy protection]
\label{cor:oracular-idu-pr-copy-protection}
Assume post-quantum $\io$ and post-quantum one-way functions. Let $\ell_{\mathsf{in}}=\ell_{\mathsf{in}}(\secparam)=\secparam^\epsilon$ for some constant $\epsilon>0$, and suppose that $\mathsf{Circ}$ is two-point pseudorandomness-style puncturable for $(\mathcal D_{\mathsf{Circ}},\mathsf{Samp}^{\mathsf{pr}}_{\IDU})$. Then there exists a copy-protection scheme for $\mathsf{Circ}$ satisfying oracular pseudorandomness-style security for $(\mathcal D_{\mathsf{Circ}},\mathsf{Samp}^{\mathsf{pr}}_{\IDU})$.
\end{corollary}

\begin{proof}
The ``Moreover'' part of \Cref{thm:correlated-upo-construction} gives $\IDU$-generalized UPO security. Add $\io$-security using~\cite[Lemma~3.16]{ABH+26}, and apply \Cref{thm:oracular-pr-copy-protection-from-correlated-upo} with $\mathcal D_\lambda$ uniform over $\mathcal X_\lambda$.
\end{proof}

\begin{corollary}[Oracular pseudorandomness-style copy protection for PRFs]
\label{cor:oracular-prf-copy-protection}
Assume post-quantum $\io$ and post-quantum one-way functions. Then there exists a puncturable PRF family on every domain $\zo^{\ell_{\mathsf{in}}(\lambda)}$ with $\ell_{\mathsf{in}}=\ell_{\mathsf{in}}(\secparam)=\secparam^\epsilon$ for some constant $\epsilon>0$, having a copy-protection scheme satisfying oracular pseudorandomness-style security for $\mathsf{Samp}^{\mathsf{pr}}_{\IDU}$.
\end{corollary}

\begin{proof}
Post-quantum one-way functions imply post-quantum puncturable PRFs. Since the identical sampler yields the singleton set $\mathcal S=\{x\}$, ordinary puncturable-PRF security gives \Cref{def:correlated-pr-puncturing}; apply \Cref{cor:oracular-idu-pr-copy-protection}.
\end{proof}

%% file: 112-EPRINT.tex
\input{111-THIS-PROJECT/abstract}
\newpage\tableofcontents

\newpage\input{intro}
\subsection{Acknowledgments}

\input{111-THIS-PROJECT/acknowledgement}
\newpage\input{tech-overview}

\newpage\section{Preliminaries}\label{sec:prelims}

\input{000-include/boilerplate/notation}
\input{000-include/boilerplate/crypto-definitions}
\subsection{Quantum Information}
We recall some facts about quantum information that we use later.

\input{000-include/boilerplate/quantum-results/gentlemes}
\input{000-include/boilerplate/quantum-definitions/threshold-implementations}
\newpage\part{Definitional Work}

\input{defns/main}
\input{defns/UPO-defn-and-properties}
\newpage\part{Quantum Information Tools}
\input{moe/it-moe}
\newpage\input{moe/comp-moe}

\newpage\part{Copy-Protection Constructions}
\input{sde/main}
\newpage\input{upo/main}
\newpage\input{upo/applications}

\newpage\addrefs 